\documentclass{birkmono}
\usepackage{graphicx}
\usepackage{amsmath}
\usepackage{amsfonts}
\usepackage{amssymb}
\usepackage{imakeidx}
\providecommand{\U}[1]{\protect\rule{.1in}{.1in}}
\newtheorem{theorem}{Theorem}[section]

\newtheorem{corollary}[theorem]{Corollary}

\newtheorem{definition}[theorem]{Definition}
\newtheorem{example}[theorem]{Example}

\newtheorem{lemma}[theorem]{Lemma}

\newtheorem{proposition}[theorem]{Proposition}
\newtheorem{remark}[theorem]{Remark}

\newtheorem{summary}[theorem]{Summary}

\makeindex

\begin{document}

\section*{Symbol correspondences for spin systems}
\thispagestyle{empty}

\vspace{1cm}

{\bf Pedro de M. Rios}\\Departamento de Matem\'atica \\ICMC, Universidade de S\~ao Paulo\\S\~ao Carlos, SP, 13560-970, Brazil\\prios@icmc.usp.br\\

\noindent {\bf Eldar Straume}\\Department of Mathematical Sciences\\Norwegian University of Science and Technology\\N-7491, Trondheim, 
Norway\\eldars@math.ntnu.no\\

\vspace{2cm}

\noindent {\bf Abstract:} The present monograph explores the correspondence between quantum and classical mechanics in the particular context of spin systems, that is, SU(2)-symmetric mechanical systems. Here, a detailed presentation of quantum spin-j systems, with emphasis on the SO(3)-invariant decomposition of their operator algebras, is  followed by an introduction to the Poisson algebra of the classical spin system and a similarly detailed presentation of its SO(3)-invariant decomposition. Subsequently, this monograph proceeds with a detailed and systematic study of general quantum-classical symbol correspondences for spin-j systems and their induced twisted products of functions on the 2-sphere. This original systematic presentation culminates with the study of twisted products in the asymptotic limit of high spin numbers. In the context of spin systems, it shows how classical mechanics may or may not emerge as an asymptotic limit of quantum mechanics.\\

\setcounter{page}{1}

\tableofcontents
\cleardoublepage

\section*{Preface}
\addcontentsline{toc}{chapter}{Preface}
\markboth{Preface}{Preface}\thispagestyle{empty}
This is a monograph intended for people interested
in a more complete understanding of the relation between quantum and
classical mechanics, here explored in the particular context of spin
systems, that is, $SU(2)$-symmetric mechanical systems. As such, this book is aimed both at mathematicians with 
interest in dynamics and quantum theory and physicists with a
mathematically oriented mind.

The mathematical prerequisites for reading this monograph are mostly 
elementary. On the one hand, some knowledge of Lie groups and their
representations can be helpful, but since the group considered here is $SU(2)
$, with its finite dimensional representations, not much more than
lower-division college mathematics is actually required. In this sense,
Chapters 2 and 3 can almost be used as a particular introduction to some aspects of
Lie group representation theory, for students of mathematical sciences. On
the other hand, some knowledge of symplectic differential geometry can be
helpful, but again, since the symplectic manifold considered here is the
two-sphere with its standard area form, some of Chapter 4 can be seen as a
concise presentation of the concepts of Hamiltonian vector field and Poisson
algebra. As the remaining chapters of the book are built on these first three, the
diligent upper-division undergraduate student in mathematics, physics or
engineering, with sufficient knowledge of calculus and linear algebra
already acquired, should find no serious difficulty in reading throughout
most of this book.

Motivation is another story, however, and this monograph may not appeal to
those not yet somewhat familiar with quantum mechanics and classical
Hamiltonian dynamics. Rather, in fact, for those who have already taken a
regular course in quantum mechanics with its standard treatment of angular
momentum, some of the contents in Chapter 3 will look familiar, although
most likely they will find the material covered in this chapter more
precisely and explicitly treated than in other available texts. In
particular, our detailed presentation of the rotationally invariant
decomposition of the operator algebra of spin systems is not so easily found
in other books. Similarly, a person with minimal knowledge of symplectic
geometry and Hamiltonian dynamics may also find Chapter 4 too  
straightforward, and perhaps somewhat amusing when he looks at the detailed
presentation of the rotationally invariant decomposition of the Poisson
algebra on the two-sphere, but again, we were not able to find this last part  in any other book.

Thus, this monograph is mainly addressing a person - student or researcher -
who has pondered on the question of whether or how quantum mechanics and
classical Hamiltonian dynamics (also called Poisson dynamics) can be
precisely related. However, since this rather general question has already
been asked and studied in various ways in many books and research papers, it
is important to clarify how this question is addressed here and what is
actually new in this book.

To this end, we must first emphasize that most ways of studying the above
question are of two major types: the first one, proceeding from quantum
mechanics to Poisson dynamics, and the second, proceeding in reverse order
from Poisson dynamics to quantum mechanics, the latter also widely known as
the process of \textquotedblleft quantization\textquotedblright\ of a
Poisson manifold. This monograph proceeds in the first way, which in its own
turn can be approached according to two distinct methodologies: one which
studies the so-called \textquotedblleft semiclassical
limit\textquotedblright\ directly from the quantum formalism, and another
and more detailed approach, which first translates the quantum formalism to
the classical formalism, in a process also referred to as \textquotedblleft
dequantization\textquotedblright , and then studies the semiclassical limit
in the dequantized setting. This more detailed methodology is the one we
shall follow in the present monograph, albeit here we are limiting ourselves to the
context of spin systems.

The advantage of working in the context of spin systems for this purpose is
twofold: on the one hand, all quantum spaces are finite dimensional, which
greatly simplifies the quantum formalism, and, on the other hand, since each
process of dequantization depends on a choice of symbol correspondence, it
is important to classify and make explicit all such possible choices, which
is easier in the finite dimensional context, and this is fully done for the
first time in Chapter 6.

The comparison of all such choices of symbol correspondences is continued in
Chapter 7, which presents a detailed study of spherical symbol products, and 
Chapter 8, which starts the study of their asymptotic limit of high spin numbers.
Together, Chapters 6, 7 and 8 of this book present for the first time in the literature a
systematic study of general symbol correspondences for spin systems, with
their symbol products and asymptotic limit of high spin numbers. In so
doing, as far as we know these chapters also present the first systematic study, in books as
well as in research papers, of how classical mechanics may or may not emerge as
an asymptotic limit of quantum mechanics, in a rather simple and precise
context.

%----------------------------------------------------------------------------------------------------------------------------------------------------------------------------------------------------------------------
%-------------------------------------------------------------------------------------- Chapter 1 -----------------------------------------------------------------------------------------------------------------
%----------------------------------------------------------------------------------------------------------------------------------------------------------------------------------------------------------------------

\chapter{Introduction}

Non-relativistic quantum mechanics was formulated in the first half of the
$20^{th}$ century by various people, most prominently by Heisenberg
\cite{Heis} and Schr\"{o}dinger \cite{Schr} who independently followed, in the
mid 1920's, on the preliminary work of Bohr \cite{Bohr1}.  The latter attempted to
modify non-relativistic classical mechanics, as introduced by Newton in the
$17^{th}$ century and developed in the following two centuries by Euler,
Lagrange, Hamilton, Jacobi and Poisson, among others, in order to be
compatible with the energy and momentum quantization postulates presented
first by Planck and then by Einstein and de Broglie in the 1900's.

While the pioneering work of Heisenberg with the help of Born and
Jordan introduced what at first became known as ``matrix mechanics'' \cite{BJ, HBJ}, the
work of Schr\"{o}dinger produced a partial differential equation for the ``wave function''. Although
Schr\"{o}dinger was soon able to link his approach to Heisenberg's, the two
approaches were fully brought together a little later by von Neumann \cite{vN,
vNN}. His final formalism for non-relativistic quantum mechanics is based on
the concepts of vectors and operators on a complex Hilbert space. In many cases these
are infinite dimensional spaces and, in fact, the complete formalism had to
expand on these concepts in order to accommodate distributions and functions
which are not square integrable, thus leading to Gelfand's rigged Hilbert
spaces \cite{Gel}, quite later.

From the start, Bohr emphasized the importance of relating the measurable
quantities in quantum mechanics to the measurable quantities in classical
mechanics. However, his so-called \textquotedblleft correspondence
principle\textquotedblright\ was not so easily implemented at the level of
relating the two mathematical formalisms in a coherent way. At first, the
basic mathematical concepts for describing classical conservative dynamics,
functions on a phase space of positions and momenta (a symplectic affine space),
were brought in a very contrived way to the quantum formalism through a
series of cooking recipes called \textquotedblleft
quantization\textquotedblright. Conversely, the \textquotedblleft classical
limit\textquotedblright\ of a quantum dynamical system, where classical
dynamics should prevail, is often a singular limit and in the initial formulations
of quantum mechanics it is not the phase space, but rather either the space
of positions or the space of momenta, that is present in a more explicit way.

Some of these problems were addressed by Dirac \cite{Dirac1, Dirac} already in his
PhD thesis, but a clearer approach for relating the classical and the quantum
mathematical formalisms was first introduced by Weyl \cite{Weyl} via a
so-called symbol map, or symbol correspondence from bounded operators on
Hilbert space $L^{2}_{\mathbb{C}}(\mathbb{R}^{n})$ to functions on phase space
$\mathbb{R}^{2n}$, these latter functions depending on Planck's constant
$\hbar$. Soon after, Wigner \cite{Wig} expanded on Weyl's idea to produce an
$\hbar$-dependent function $\mathcal{W}_{\psi}\in L^{1}_{\mathbb{R}%
}(\mathbb{R}^{2n})$ as the phase-space representation of a wave function
$\psi\in L^{2}_{\mathbb{C}}(\mathbb{R}^{n})$. While real and integrating to
$1$ over phase space, $\mathcal{W}_{\psi}$ is only a pseudo probability
distribution on $\mathbb{R}^{2n}$ because it can take on negative values, as
opposed to $|\psi|^{2}$ which is a true probability distribution on
$\mathbb{R}^{n}$.

Despite this shortcoming, Wigner's function inspired Moyal \cite{Moy} to
develop an alternative formulation of quantum mechanics as a ``phase-space
statistical theory'', following on Weyl's correspondence, where the Poisson
bracket of functions on $\mathbb{R}^{2n}$ was replaced by an $\hbar$-dependent
bracket of $\hbar$-dependent functions on $\mathbb{R}^{2n}$, whose ``classical
limit'' is the Poisson bracket. Moyal's skew-symmetric bracket still satisfies Jacobi's identity and, in fact,
can be seen as the commutator of an $\hbar$-dependent associative product on
the space of $\hbar$-dependent Weyl symbols. Although Moyal's product is
written in terms of bi-differential operators, an integral formulation of this
product had been previously developed, first by von Neumann \cite{vN1} soon
after Weyl's work, then re-discovered by Groenewold \cite{Groen}. This
Weyl-Moyal approach was further developed by Hormander \cite{Horm, Horm1, D-H}, 
among others, into the
calculus of pseudo-differential and Fourier-integral operators. It also
inspired the deformation quantization approach started by Bayen, Flato, Frondsal, 
Lichnerowicz and Sternheimer for general Poisson manifolds \cite{Betal}.

However, while the Hilbert space used to describe the dynamics of particles in ${3}%
$-dimensional configuration-space of positions is generally infinite
dimensional, if one restricts attention to rotations around a point, only, the
corresponding Hilbert space is finite dimensional. Historically, the necessity
to introduce an independent, or intrinsic finite dimensional Hilbert space for
studying the dynamics of atomic and subatomic particles stemmed from the
understanding that a particle such as an electron has an intrinsic
``spinning'' which is independent of its dynamics in $3$-space. An extra
degree of freedom, therefore.

After the hypothesis of an extra degree of freedom was first posed by Pauli in
1924 \cite{Pauli} and one year later identified by Goudsmit and Uhlenbeck
\cite{G-U, U-G} as an intrinsic ``electron spin'', the spin theory was
developed  around 1930, mostly  by Wigner \cite{Wig1, Wig11, Wig111}, but also by Weyl
\cite{Weyl1}, through a careful study of the group of rotations and its
simply-connected double cover $SU(2)$, and their representations. 

Because this intrinsic quantum dynamics in a finite-dimensional Hilbert
space had no obvious classical counterpart, at that time, the necessity to
relate it with a corresponding classical dynamics was not present from the
start. In fact, to the extra degree of freedom of spin, there should
correspond a $2$-dimensional phase space and, by $SU(2)$ invariance, this
phase space must be the homogeneous $2$-sphere $S^{2}$. Thus, the fact that
its classical configuration-space cannot be products of euclidean $3$-space or
the real $2$-sphere was at first interpreted as an indication that a physical
correspondence principle for spin systems was impossible.

This may explain why the mathematical formulation of a Weyl's style
correspondence for spin systems was developed much later. Thus, while Weyl's
to Moyal's works are dated from the mid 1920's to the late 1940's, an
incomplete version of a symbol correspondence for spin systems was first set
forth in mid 1950's by Stratonovich \cite{Strat} and a first more complete
version was presented in mid 1970's by Berezin \cite{Berezin1, Berezin2, Berezin}. 
Then, in the late 1980's, the work of Varilly and Gracia-Bondia \cite{VG-B} finally
completed and expanded the draft of Stratonovich. Its contemporary work of
Wildberger \cite{Wild} was followed by the relevant works of Madore
\cite{Madore}, in the 1990's, and of Freidel and Krasnov \cite{F-K}, in early
2000's. All these works produced a very good, but in our view, still incomplete
understanding of symbol correspondences and symbol products, for spin systems.

From the mathematical point of view, such a late historical development is
surprising because quantum mechanics in finite-dimensional Hilbert space, on
top of being far closer to Heisenberg's original matrix mechanics, is far
simpler than quantum mechanics in infinite-dimensional Hilbert space. On the
other hand, the corresponding phase space for the classical dynamics of spin
systems, the $2$-sphere $S^{2}$, has nontrivial topology, as opposed to
$\mathbb{R}^{2n}$. Nonetheless, one should expect that a mathematical
formulation of Bohr's correspondence principle ``\`a la Weyl'' would have been developed for
spin systems in a more systematic and complete way than it had been achieved
for mechanics of particles in $k$-dimensional Euclidean configuration space. However,
despite the many developments in the last 50 years, such a systematic
mathematical formulation of correspondences in spin systems was still lacking.
In particular, we wanted a solid mathematical presentation unifying the main
scattered papers on the subject (some of them based also on heuristic
arguments) and completing the various gaps, so as to clarify the whole
landscape and open new paths to still unexplored territories.

Here we attempt to fulfill this goal, at least to some extent. This  monograph, written to be as self-contained as possible, is organized
as follows.

In Chapter 2 we review generalities on Lie groups and their representations, with special attention to the groups $SU(2)$ and $SO(3)$.  

In Chapter 3 we present the elements for quantum dynamics of spin systems, that is, $SU(2)$-symmetric quantum mechanical systems.
First, we carefully review its
representation theory which defines spin-$j$ systems and their standard basis
of Hilbert space. Then, after reviewing the tensor product and presenting the space of
operators, with its irreducible summands and standard  coupled basis, we remind 
how the operator product is related to quantum dynamics via Heisenberg's equation 
involving the commutator and then we present  in a very detailed way the
$SO(3)$-invariant decomposition of the operator product, its multiplication rule in the coupled basis of operators, with  
its parity property. 

Almost all choices, conventions and notations used in this 
chapter are standard ones used in the vast mathematics and physics literature on this subject. 
In particular, we introduce and manipulate Clebsch-Gordan coefficients and 
the various Wigner symbols in a traditional way, without resorting to the more 
modern language of ``spin networks'' which, in our understanding, would require 
the introduction of additional definitions, etc, which nonetheless would not 
contribute significantly to the main results to be used in the rest of the book.  

In Chapter 4 we present the elements of the classical $SU(2)$-symmetric mechanical system.
After reviewing some basic facts about the $2$-sphere as a symplectic manifold and presenting the key concepts of classical 
Hamiltonian dynamics  and Poisson algebra of smooth functions on $S^{2}$, defining the classical spin system, 
we present in some detail the space of polynomial functions on $S^{2}$,
the spherical harmonics, followed by a detailed presentation of 
the $SO(3)$-invariant  decomposition of the pointwise product and the Poisson
bracket of functions on $S^{2}$.

Chapter 5, Intermission, pauses the study of spin systems. Here we present a brief historical overview of symbol correspondences in affine mechanical systems, in preparation for the remaining and most novel chapters of the book.  

Thus, in Chapter 6 we present the $SO(3)$-equivariant symbol correspondences between
operators on a finite-dimensional Hilbert space and (polynomial) functions on $S^{2}$.
After defining general symbol correspondences and determining their moduli space, 
we distinguish the isometric
ones, the so-called Stratonovich-Weyl symbol correspondences. Then, we present
explicit constructions of general symbol correspondences, 
introducing the key concept of
characteristic numbers of a symbol correspondence, and the general 
covariant-contravariant duality for non-isometric  correspodences. Besides
correspondences of Stratonovich-Weyl type, special attention is  devoted to the
non-isometric correspondence defined by Berezin.

In Chapter 7 we study in great detail the products of symbols induced from the operator 
product via symbol correspondences, the so-called twisted products. After definition and 
basic properties, we produce explicit expressions for some twisted products of cartesian 
symbols, valid for all $n=2j\in\mathbb{N}$. Then, we describe the formulae for general 
twisted products of spherical harmonics, $Y_{l}^{m}$, discussing some of their common 
properties.  Next, we present a detailed study of  integral trikernels, which define twisted 
products via integral equations.  We state the general properties and produce various 
explicit formulae (including some integral ones) for these trikernels. In so doing, we arrive 
at formulae for various functional transforms that generalize the Berezin transform and we also 
see that it is not so easy to infer a simple closed formula for the Stratonovich trikernel in terms of 
midpoint triangles (in a form first inquired by Weinstein \cite{Wein} based on the analogy with the 
Groenewold-vonNeumann trikernel) unless, perhaps, asymptotically.    

Then, in Chapter 8 we start the study of the asymptotic $j\rightarrow\infty$ limit of
symbol correspondence sequences and their sequences of twisted products. In this monograph, 
we focus on the high-$j$ asymptotics for finite $l$, here called low-$l$ high-$j$%
-asymptotics, leaving high-$l$-asymptotics to a later opportunity. We show 
that Poisson (anti-Poisson) dynamics emerge in the asymptotic $j\to\infty$ limit of the 
standard (alternate) Stratonovich twisted product as well as the standard (alternate) 
Berezin twisted product, but this is not the generic case for sequences of twisted products, not even 
in the restricted subclass of twisted products induced from isometric correspondence sequences. 
Thus, we characterize some kinds of symbol correspondence sequences based on their 
asymptotic properties and also discuss some measurable consequences.   

In Chapter 9 we present some concluding thoughts. 
For spin systems, adding to  Rieffel's old theorem on $SO(3)$-invariant strict 
deformation quantizations of the $2$-sphere \cite{Rieffel1}, one now also 
has to take into consideration the  fact that  generic symbol correspondence sequences do not yield Poisson 
dynamics in the asymptotic  limit of high spin numbers. In light of these results, old and new,  
we reflect on  the peculiar nature of the classical-quantum correspondence. 

Finally, in the Appendix we gather proofs of some of the propositions and a
theorem, which were stated in the main text.

\

\noindent {\bf Acknowledgements:} During work on this project, 
we have benefitted from several mutual visits. 
We
thank FAPESP (scientific sponsor for the state of S\~{a}o Paulo, Brazil) and
USP, as well as NTNU and the Norwegian NSF, for support of these visits. We are also grateful to UC Berkeley's Math. Dept. for hospitality during some of the periods  when we were working  on this project. Again, we thank the above sponsors for financial support of these stays.  
Many
of the original results in this monograph were first presented at the conference
\textquotedblleft Geometry and Algebra of PDE's\textquotedblright, Troms\o ,
27-31 August 2012. We thank the organizers for the opportunity and some of the
people in the audience for the interest, which stimulated us to wrap up this
work in its present form. We also thank,  in particular, 
Robert Littlejohn and Austin Hedeman for discussions, 
Marc Rieffel for taking time to read and comment on some parts of the 
monograph and Nazira Harb for collaboration in Appendix \ref{parity prop} 
and for sharing with us some of her results in Examples \ref{exF} and \ref{exL}. Finally, we thank Alan Weinstein for invaluable suggestions 
on a preliminary version of this monograph, dated December  2012. Since then, 
we've also benefitted from various suggestions and comments from some other people.

%----------------------------------------------------------------------------------------------------------------------------------------------------------------------------------------------------------------------
%-------------------------------------------------------------------------------------- Chapter 2 -----------------------------------------------------------------------------------------------------------------
%----------------------------------------------------------------------------------------------------------------------------------------------------------------------------------------------------------------------

\chapter{Preliminaries}

This chapter presents basic material on Lie groups and their representations, with emphasis on the Lie groups $SO(3)$ and $SU(2)$, as a preparation for next chapters. The reader already very familiar with the subject may skip or just glance through this chapter. For the reader unfamiliar or vaguely familiar with the subject, we refer to \cite{Adams, Baker, D-K, Gilmore,  Knapp, Straume2, Vilekin}, for instance, for more detailed presentations. 

\section{On Lie groups and their representations}

The notion of ``symmetry'', expressed mathematically in terms of groups of
transformations, plays a fundamental role in classical as well as quantum
mechanics. 

Recall that a group\index{Groups ! definition} is a set $G$ together with a map $ \mu : G\times G\to G \ , (g_1,g_2)\mapsto \mu(g_1,g_2)=g_1g_2 $, satisfying  
$$(i) \ (g_1g_2)g_3 = g_1(g_2g_3)  , \forall g_1,g_2,g_3\in G ,$$  $$(ii) \  \exists  e\in G \ \text{s.t.} \   eg=ge=g   , \ \forall g\in G  ,$$ $$(iii) \ \forall g\in G, \exists  g^{-1}\in G \ \text{s.t.} \ gg^{-1}=g^{-1}g=e. $$
\noindent If the set $G$ is a smooth manifold and the map $\Lambda : G\times G\to G \ , (g_1,g_2)\mapsto g_1^{-1}g_2$ is smooth, then $G$ is called  a \emph{Lie group}\index{Lie groups ! definition}\index{Lie groups |(}. The concept of subgroup is the natural one, so,  $H\subset G$ is a \emph{subgroup}\index{Groups ! subgroup} of $G$ if $g_1^{-1}g_2\in H , \ \forall g_1, g_2 \in H$.  

Recall also that given two groups $G,H$, a group homomorphism\index{Groups ! homomorphism} from $G$ to $H$ is a map $\Psi:G\to H$ satisfying $\Psi(g_1g_2)=\Psi(g_1)\Psi(g_2)$.  If $G,H$ are Lie groups and $\Psi$ is smooth, then it is a Lie group homomorphism.\index{Lie groups ! homomorphism} If $\Psi$ is also bijective, then it is a (Lie) group isomorphism.\index{Groups ! isomorphism}

Now, we say that a group $G$ acts (from the left) on a set $M$ if there is a map $\Phi
:G\times M\rightarrow M $, \  such that\index{Groups ! action}  
$$
(i) \ em=m \ \  \text{and} \ \ (ii) \ g_{2}(g_{1}m)=(g_{2}g_{1})m  \  , \ \text{for all} \ m\in M, \ g_{i}\in G, 
$$
\noindent where we have used
the simpler notation $gm$ for $\Phi (g,m)$. In other words, there is a group
homomorphism%
\begin{equation}
\bar{\Phi}:G\rightarrow S(M),\text{ }\bar{\Phi}(g):m\mapsto gm \ ,
\label{group1}
\end{equation}%
where $S(M)$ is the group of invertible maps $M\rightarrow M$. 

Associated
with the above action of $G$ on $M$ there is the induced (left) action on the set of
scalar functions on $M$: 
\begin{equation}
g\in G:f\rightarrow f^{g},\text{ }f^{g}(m)=f(g^{-1}m)  \ . \label{group2}
\end{equation}

In the typical applications in geometry and physics, $G$ is a Lie group, $M$
is a smooth manifold, and the above maps and functions are smooth. $M$ might
also have some specific algebraic or geometric structure and $S(M)$ is the
group of ``symmetries'', that is, transformations preserving the relevant
structure.

\subsubsection{Classical groups over the classical fields}

Many of the familiar examples of Lie groups arise from considering
\emph{matrix groups}, which are subgroups of $GL_{\mathbb{K}}(n)$, the group
of invertible matrices over $\mathbb{K=R},\mathbb{C},$ or $\mathbb{H}$
(quaternions), which is isomorphic to the group $GL(V)$ of invertible linear
transformations of an $n$-dimensional $\mathbb{K}$-vector space $M=V\simeq
\mathbb{K}^{n}$, via the usual left action of matrices on column vectors.
Note, however, in the case $\mathbb{K=H}$ scalar multiplication must act on
the right side of column vectors to ensure that matrix multiplication from the
left side is an $\mathbb{H}$- linear operator. $\mathbb{K}^{n}$ has the usual
Hermitian inner product $\left\langle u,v\right\rangle =\sum_{j=1}^{n}\bar
{u}_{j}v_{j}$ and norm $\left\Vert u\right\Vert =$ $\left\langle
u,u\right\rangle ^{1/2}$.

The natural extension of scalars, from real to complex to quaternion, has
several important consequences for representations of groups. First of all,
for $n\geq1$ there are natural inclusion maps $c,q$ and linear isomorphisms
$r,c^{\prime}$%
\begin{align}
\mathbb{R}^{n} &  \rightarrow^{c}\mathbb{C}^{n}=\mathbb{R}^{n}+i\mathbb{R}%
^{n}\rightarrow^{r}\mathbb{R}^{n}\oplus\mathbb{R}^{n},\label{c}\\
\text{ }\mathbb{C}^{n} &  \rightarrow^{q}\mathbb{H}^{n}=\mathbb{C}%
^{n}+j\mathbb{C}^{n}\rightarrow^{c^{\prime}}\mathbb{C}^{n}\text{ }%
\oplus\mathbb{C}^{n}\text{,}\nonumber
\end{align}
where $c$ and $r$ (resp. $q$ and $c^{\prime})$ are $\mathbb{R}$-linear (resp.
$\mathbb{C}$-linear) maps\footnote{The mappings $c,q,r,c^{\prime}$ are
generally referred to as complexification, quaternionification,
real-reduction, and complex-reduction, respectively.}. Then there are
corresponding injective homomorhisms of groups%
\begin{align}
GL_{\mathbb{R}}(n) &  \rightarrow^{c}GL_{\mathbb{C}}(n)\rightarrow
^{r}GL_{\mathbb{R}}(2n)\label{d}\\
GL_{\mathbb{C}}(n) &  \rightarrow^{q}GL_{\mathbb{H}}(n)\rightarrow^{c^{\prime
}}GL_{\mathbb{C}}(2n)\nonumber
\end{align}
Again, $c,q$ in (\ref{d}) denote inclusions, whereas $r$ and $c^{\prime}$
replace each entry of a matrix by a 2-block, as follows:%
\begin{equation}
r:(x+iy)\rightarrow\left(
\begin{array}
[c]{cc}%
x & -y\\
y & x
\end{array}
\right)  ,\text{ \ }c\prime:(z_{1}+jz_{2})\rightarrow\left(
\begin{array}
[c]{cc}%
z_{1} & -\bar{z}_{2}\\
z_{2} & \bar{z}_{1}%
\end{array}
\right)  \label{e}%
\end{equation}
The terminology explained in the last footnote also makes sense for the above
groups and homomorphisms.

Of particular interest to us is the unitary group $U(V)\subset GL(V)$
consisting of operators preserving the norm of vectors. Its conjugacy class
consists, in fact, of all maximal compact subgroups of $GL(V)$, as follows
directly from Lemma \ref{Lemma1} (i) below. In terms of matrices, the $\mathbb{K}$-unitary
group $U_{\mathbb{K}}(n)\subset GL_{\mathbb{K}}(n)$ consists of the matrices
$A$ whose inverse $A^{-1}$ is the adjoint $A^{\ast}=\bar{A}$ $^{T}$ (conjugate
transpose), or equivalently, $\left\Vert Au\right\Vert =\left\Vert
u\right\Vert $ for all vectors $u$. For the three cases of scalar field
$\mathbb{K}$, these are the compact classical groups
\begin{equation}
O(n)\subset U(n)\subset Sp(n) \label{a}%
\end{equation}
called the orthogonal, unitary, and symplectic group\footnote{The reader should not confuse the symplectic group $Sp(n)\equiv Sp_{\mathbb H}(n)$, which is a compact group over the quaternions, with the group $Sp_{\mathbb R}(2n)\subset GL_{\mathbb R}(2n)$, which is a noncompact group over the reals, often also called the symplectic group, in short for the \emph{real symplectic group} - we will briefly mention a few basic things about this latter in the intermission, Chapter 5.}, respectively. The
special orthogonal and special unitary groups $SO(n)\subset SU(n)$ are
constrained by the additional condition $\det(A)=1$. Observe that the map $r$
in (\ref{e}) yields an isomorphism $U(1)\simeq SO(2)$ by restriction to
$x^{2}+y^{2}=1$, whereas the map $c^{\prime}$ yields an isomorphism
$Sp(1)\simeq SU(2)$ by restriction to $|z_{1}|^{2}+|z_{2}|^{2}=1$.

In fact, these are particular cases of the following observation: since at each ``doubling procedure'' $\mathbb R^n\to\mathbb C^n\to\mathbb  H^n$ a new algebraic structure is introduced, which needs to be preserved by the symmetry group, we also have the inclusions 
$$Sp_{\mathbb H}(n)\equiv Sp(n)\subset U(2n)\equiv U_{\mathbb C}(2n),$$  $$U_{\mathbb C}(n)\equiv U(n)\subset O(2n)\equiv O_{\mathbb R}(2n) .$$

\subsubsection{Linear representations of a group}

Let us recall the basic definitions and results about finite dimensional
linear representations of a group $G$. Namely, a \emph{representation} of a
(Lie) group $G$ on $V\simeq\mathbb{K}^{n}$ is a (Lie) group homomorphism%
\begin{equation}
\varphi:G\rightarrow GL(V)\simeq GL_{\mathbb{K}}(n) . \label{b}%
\end{equation}
We also say that $\varphi$ is a $\mathbb{K}$-representation to emphasize that $G$
acts on $V$ by $\mathbb{K}$-linear transformations. Recall from elementary
linear algebra that an isomorphism of the groups $GL(V)$ and $GL_{\mathbb{K}%
}(n)$, as indicated in (\ref{b}), amounts to the choice of a basis for $V$. The
representation is said to be \emph{orthogonal}, \emph{unitary}, or
\emph{symplectic} if the image $\varphi(G)$, viewed as a subgroup of
$GL_{\mathbb{K}}(n)$, lies in the corresponding unitary group (\ref{a}).

Two representations of $G$ on vector spaces $V_{1}$ and $V_{2}$ are said to be
equivalent (or isomorphic) if there is a $G$-equivariant linear isomorphism
$F:V_{1}\rightarrow V_{2}$, namely satisfying $F(gv)=gF(v)$ for all $g\in G,$
$v\in V_{1}$. In terms of matrices, if $\varphi_{i}:G\rightarrow
GL_{\mathbb{K}}(n),i=1,2$, are two matrix representations, then their 
equivalence $'\simeq \ '$ (as $\mathbb{K}$-representations) is defined with respect to some fixed 
$A\in GL_{\mathbb{K}}(n)$, as follows: 
\begin{equation}\label{equivreps}
\varphi_{1}\simeq\varphi_{2} \iff \varphi_{2}(g)=A\varphi_{1}(g)A^{-1}, \forall g\in G\text{.}%
\end{equation}
However, the shorthand $\varphi_{1}=\varphi_{2}$ is often used to mean $\varphi_{1}\simeq\varphi_{2}$. 

For simplicity, a representation $\varphi$ of $G$ on $V$ is sometimes denoted
by the pair $(G,V)$, with $\varphi$ tacitly understood. A subspace $U\subset
V$ is said to be $G$-invariant if $gv\in U$ for each $v\in U$, and the
representation (\ref{b}) is irreducible if there is no $G$-invariant
subspace strictly between $\{0\}$ and $V$. In representation theory, the
classification of all irreducible representations of $G$, up to equivalence,
is a central problem. More generally, the $\mathbb{K}$-representation $(G,V)$
is said to be \emph{completely reducible} if $V$ decomposes into a direct sum
$V=V_{1}\oplus V_{2}\oplus...\oplus V_{k}$ of irreducible $G$-invariant
$\mathbb{K}$-subspaces $V_{i}$, each defining an irreducible representation
$(G,V_{i})$ with homomorphism $\varphi_{i}:G\rightarrow GL(V_{i})$.

Henceforth, we shall assume the group $G$ is compact, which in view of the
following lemma simplifies the representation theory considerably.

\begin{lemma}\label{Lemma1} For a compact group $G$ the following hold:

%\begin{itemize}
%\noindent 
(i) all $\mathbb{K}$-representations $\varphi$ are $\mathbb{K}$-unitary,
namely $\varphi:G\rightarrow U(V)$ for a suitable Hermitian inner product on
$V$.

%\noindent 
(ii) all representations are completely reducible.
%\end{itemize}
\end{lemma}

The standard proof is often referred to as ``Weyl's unitary trick", using the
fact that $(G,V)$ has a $G$-invariant Hermitian inner product $\left\langle
,\right\rangle $. It is found by averaging a given inner product $(,)$ over
$G$, namely we set%

\begin{equation}
\left\langle u,v\right\rangle =%
%TCIMACRO{\dint \limits_{G}}%
%BeginExpansion
{\displaystyle\int\limits_{G}}
%EndExpansion
(gu,gv)dg \label{k}%
\end{equation}
Here $dg$ denotes the (normalized) Haar measure on $G$, which is bi-invariant
in the sense that\ the functions $f(g),f(kg),f(gk)$ on $G$ have the same
integral, for $k\in G$ fixed. Now property (ii) follows from the observation
that for any $G$-invariant subspace $U\subset V$ , the orthogonal complement
$U^{\perp}$ is also $G$-invariant.

From completely reducibility one can show that any representation has a unique
decomposition $\varphi=%
%TCIMACRO{\dsum }%
%BeginExpansion
{\displaystyle\sum}
%EndExpansion
\varphi_{i}$ into a (internal) direct sum of irreducibles, often written as an
ordinary sum $\varphi=\varphi_{1}+\varphi_{2}+...+\varphi_{k}$. Some of the
$\varphi_{i}$ may be equivalent, of course, so the notation such as
$\varphi=2\varphi_{1}+3\varphi_{2}$ should be clear enough.

Note, however, that the internal splitting (or direct sum) as explained above is
different from the notion of an external splitting of a representation
$(G,V),$ or equivalently an external direct sum of representations
$(G_{i},V_{i})$. Namely, $G$ is the product of the groups $G_{i}$ and $V$ is
the direct sum of the $V_{i}$, each $G_{i}$ acting nontrivially only on the
summand $V_{i}$. For example, with two factors we have
\[
(G,V)=(G_{1}\times G_{2},V_{1}\oplus V_{2})=(G_{1},V_{1})\oplus(G_{2}%
,V_{2})\text{ },\text{ }\varphi=\varphi_{1}\oplus\varphi_{2}%
\]
Moreover, in the case $G_{1}=G_{2}=H$ so that $G=H\times H$, restriction of
the above external splitting to the diagonal subgroup $\Delta H\simeq H$ of
$G$ yields an internal splitting $\varphi|_{\Delta H}=\varphi_{1}+\varphi_{2}$
of the representation $(H,V)$.

Clearly, an irreducible representation may become reducible by extension of
the scalar field $\mathbb{K}$ . For example, a real irreducible representation
$\psi$ may split after complexification,
\begin{equation}
c\psi:G\rightarrow^{\psi}GL_{\mathbb{R}}(n)\rightarrow^{c}GL_{\mathbb{C}}(n),
\label{g}%
\end{equation}
say $c\psi=\varphi_{1}+\varphi_{2}$. As an example, consider the standard
representation $\mu_{1}$ of the circle group $U(1)=\{z\in\mathbb{C};|z|=1\}$
acting by scalar multiplication on the complex line $\mathbb{C}$.
Realification of $\mu_{1}$ yields the standard representation $\rho_{2}$ of
$SO(2)$ $\simeq U(1)$, acting by rotations on $\mathbb{R}^{2}$. Next,
complexification of $\rho_{2}$ means regarding the rotation matrices as
elements of $GL_{\mathbb{C}}(2)$, and here they can be diagonalized
simultaneous. Refering to (\ref{d}) we illustrate the effect of the two
``operations" $r$ and $c$ as follows:%
\begin{equation}
(x+iy)\rightarrow^{r}\left(
\begin{array}
[c]{cc}%
x & -y\\
y & x
\end{array}
\right)  \rightarrow^{c}\left(
\begin{array}
[c]{cc}%
x+iy & 0\\
0 & x-iy
\end{array}
\right)  \text{, }\mu_{1}\rightarrow\rho_{2}\rightarrow\mu_{1}+\bar{\mu}_{1}
\label{f}%
\end{equation}
Here, in the last step we have changed the basis of $\mathbb{C}^{2}$ so that
the matrices becomes diagonal. Thus, in particular, the irreducible representation
$\rho_{2}$ becomes reducible when complexified.

Next, note that conjugation of complex matrices yields a group isomorphism%
\[
t:GL_{\mathbb{C}}(n)\rightarrow GL_{\mathbb{C}}(n)\text{, \ }A\rightarrow
\bar{A}%
\]
Therefore, a $\mathbb{C}$-representation $\varphi$ composed with $t$ is the
\emph{complex conjugate} representation $t\varphi=\bar{\varphi}:g\rightarrow
\ \overline{\varphi(g)}$; for example see (\ref{f}). In a similar vein, there
is a closely related construction, namely the \emph{dual }(or contragradient)
representation $\varphi^{T}$ of $\varphi$, acting on the dual space $V^{\ast
}=Hom(V,\mathbb{C)}$ and defined by%
\begin{equation}
\varphi^{T}(g)=\varphi(g^{-1})^{T}, \label{m}%
\end{equation}
where $\Phi^{T}$ denotes the dual of an operator $\Phi$ on $V$. However, since
$G$ is compact it follows that the dual representation is the same as the
complex conjugate representation. In fact, in terms of matrices with respect
to dual bases in $V$ and $V^{\ast}$, for $\varphi(g)$ unitary the left side of
(\ref{m}) is the matrix $\bar{\varphi}(g)$, hence $\varphi^{T}=\bar{\varphi}$.

Together with the homomorphisms in (\ref{d}) we now have the six ``operations"
$r,c,c^{\prime},q,t,1$, where $1$ denotes the identity, being performed on
representations of the appropriate kind, and it is not difficult to verify the
following relations
\begin{align}
cr  &  =1+t,\ c^{\prime}q=1+t,\text{ }rc=2,\text{ }qc^{\prime}%
=2, \ tr=r,\label{h}\\
tc  &  =c, \ tc^{\prime}=c^{\prime}, \ qt=q, \ tc=c, \ t^{2}=1\nonumber
\end{align}
For example, returning to the above splitting example, $c\psi=\varphi
_{1}+\varphi_{2}$ in (\ref{g}), the relation $rc=2$ tells us that
$2\psi=r\varphi_{1}+r\varphi_{2}$, namely $\psi=r\varphi_{1}=r\varphi_{2}$.

As we have seen, a real representation $\psi\longleftrightarrow(G,U$) can be
extended to a complex representation $\varphi\longleftrightarrow(G,V)$ by the
process of complexification, namely we set $V=U^{\mathbb{C}}=U+iU\ $as in
(\ref{c}), and we set $\varphi=c\psi.$ Conversely, we say that a complex
representation $\varphi\longleftrightarrow(G,V)$ has a \emph{real form}
$\psi\longleftrightarrow(G,U)$ if $c\psi=\varphi$, that is, $(G,V)$ is
equivalent to $(G,U^{\mathbb{C}}).$ In the same vein, a complex representation
$\varphi\longleftrightarrow(G,V)$ has a \emph{quaternionic form}
$\eta\longleftrightarrow(G,W)$ if $c^{\prime}\eta=\varphi$, that is $(G,W)$ is
equivalent to $(G,V^{\mathbb{H}})$ where $G$ acts by $\mathbb{H}$-linear
transformations on $V^{\mathbb{H}}=V+jV$.

Note that distinct real representations are mapped to distinct complex
representations when they are complexified, and distinct quaternionic
representations have distinct complex-reduction. Consequently, all the
representation theory is actually contained in the realm of complex
representations $\varphi$, regarding $\varphi$ as \emph{real} (resp.
\emph{quaternionic}) if it has a real (resp. quaternionic) form. 
The following general result is valid (at least) for compact groups\ $G$ (for a
proof, see e.g. \cite{Adams}):

\begin{lemma}\label{Lemma 2} (a) Let $\varphi$ be a complex representation of $G$. If
$\varphi$ is real or quaternionic, then $\varphi$ is self-conjugate, that is,
$t\varphi(=\bar{\varphi})\simeq\varphi.$ Conversely, if $\varphi$ is irreducible and
self-conjugate, then it is either real or quaternionic (but not both!).

(b) $\varphi$ is real (resp. quaternionic) if and only if $(G,V)$ has a
$G$-invariant non-singular symmetric (resp. skew-symmetric) bilinear form
$V\times V\rightarrow\mathbb{C}$.
\end{lemma}

For any two representations of $G$ on $\mathbb{K}$-vector spaces $V_1$ and $V_2$
respectively, set $Hom^{G}(V_1,V_2)$ to be the vector space of $G$-equivariant
linear maps $F:V_1\rightarrow V_2$. The following result is classical but of
central importance.    

\begin{lemma}[Schur's lemma]\label{Schur's lemma}
Let $\varphi _{i}:G\rightarrow GL_{\mathbb C}(V_{i})$, $i=1,2$, be two 
irreducible representations of $G$ on $\mathbb{C}$-vector spaces $V_i$. Then $Hom^{G}(V_{1},V_{2})\simeq \mathbb{C}$ if $\varphi _{1}\simeq $ $\varphi _{2}$, and $Hom^{G}(V_{1},V_{2})=0$
otherwise.
\end{lemma}
\begin{proof}
As $F:V_{1}\rightarrow V_{2}$ is equivariant, it is
easy to see that both $\ker (F)$ and Im$(F)$ are $G$-invariant
subspaces and, from irreducibility, $F=0$ or $F$ is an isomorphism. In the latter case the
two representations are equivalent, so let us assume $V_{1}=V_{2}$. Since $F$ has at least one eigenvalue and  its corresponding 
eigenspace is $G$-invariant, it follows from irreducibility that $F$ is a nonzero
multiple of the identity. 
\end{proof}\index{Lie groups ! representations |)} 

\subsubsection{The infinitesimal version of Lie groups and their
representations}

Next we turn to the infinitesimal aspect of a connected Lie group $G$,
namely the fundamental reduction to its \emph{Lie algebra}\index{Lie algebras |(} $\mathcal{G}$, which is an $\mathbb{R}$-vector space, identified with the tangent space of $G$ at the
identity element $e$, endowed with a specific bilinear product $$[ \ , \ ] \ : \ \mathcal{G}\times\mathcal{G}\to\mathcal{G}$$
called the \emph{Lie bracket}\index{Lie algebras ! Lie bracket}, which is \emph{skew-symmetric} and satisfies
the \emph{Jacobi identity} 
\[
\lbrack X,Y]=-[Y,X],\text{ \ }[[X,Y],Z]+[[Y,Z],X]+[[Z,X],Y]=0.
\]%

The Lie bracket is, in fact, the ``linearization'' of the product in $G$,
and this relation is more precisely understood via the so-called
\emph{exponential map}\index{Lie algebras ! exponential map |(} $\exp $: $\mathcal{G\rightarrow }G$, which provides the
linkage between a connected group and its Lie algebra. It is a local
diffeomorphism mapping a neighborhood of $X=0$ onto a neighborhood of $e\in G
$. This also explains why the Lie algebra is determined by any small neighborhood of $e\in G$, 
and consequently locally isomorphic groups have isomorphic Lie algebras.
In particular, the Lie algebra of a disconnected Lie group $G$ depends only
on the identity component subgroup $G^{\circ }$ of $G$.

Now, let us again focus attention on linear or matrix groups. The set $gl(V)$
of all $\mathbb{K}$-linear operators $T:V\rightarrow V$, resp. the set 
$gl_{\mathbb{K}}(n)=M_{\mathbb{K}}(n)$ of matrices, are vector spaces as well as
associative algebras over $\mathbb{R}$ with regard to the usual product of
operators, resp. matrices, and with the commutator product $[S,T]=ST-TS$
they are the Lie algebras of $GL(V)$ and $GL_{\mathbb{K}}(n)$, respectively.
In these cases the exponential map is defined by the (usual) power series
expansion 
\begin{equation}
\exp (X)=e^{X}=I+X+\frac{1}{2}X^{2}+\frac{1}{3!}X^{3}+...+\frac{1}{k!}%
X^{k}+.....  \label{exp}
\end{equation}%
The inverse map is the logarithm $$\log (I+Y)=\displaystyle{\sum_{k=1}^{\infty }}%
\frac{(-1)^{k+1}}{k}Y^{k}, \ \text{for} \  \left\vert Y\right\vert <1.$$ It follows,
for example, that $\det (e^{X})=e^{tr(X)}$ , where $tr(X)$ is the trace of $X
$. 

For closed matrix groups $G$ $\subset GL_{\mathbb{K}}(n)$, there is the
following commutative diagram of groups, Lie algebras, and vertical
exponential maps $\hat{e}$ 
\[
\begin{array}{ccc}
\mathcal{G} & \subset  & M_{\mathbb{K}}(n) \\ 
\hat{e}\downarrow  &  & \downarrow \hat{e} \\ 
G & \subset  & GL_{\mathbb{K}}(n),%
\end{array}%
\]%
Namely, the exponential map for $G$ is the restriction of the matrix
exponential map (\ref{exp}), and moreover, the Lie algebra $\mathcal{G}$ can
be determined as follows%
\begin{equation}
\mathcal{G=}\left\{ X\in M_{\mathbb{K}}(n); \ \exp (tX)\in G\text{ for all }%
t\right\}.   \label{LieAlg}
\end{equation}%

As we shall mainly focus attention on orthogonal groups $SO(n)\subset
O(n)\subset GL_{\mathbb{R}}(n)$ and unitary groups $SU(n)\subset U(n)\subset
GL_{\mathbb{C}}(n)$, we first note that $O(n)$ has two connected components; the component different
from $SO(n)$ consisting of the matrices with determinant $-1$. By (\ref{LieAlg}%
), the Lie algebras of $SO(n)$ and $SU(n)$ consist of traceless skew-symmetric and
skew-Hermitian matrices, respectively. 

For compact connected groups $G$, such as $SO(n),SU(n),U(n),$ the
exponential map is, in fact, surjective. But more generally, a connected Lie
group $G$ is still generated by its 1-parameter subgroups \{$\exp (tX),t\in 
\mathbb{R\}}$ and hence is determined by its Lie algebra, up to local isomorphisms. In fact, for a
given basis \{$X_{1},..,X_{r}\}$ of $\mathcal{G}$, $G$ is generated by the
corresponding 1-parameter groups $\{\exp (tX_{i})\}$. Therefore, the
elements $X_{i}$ are sometimes referred to as the infinitesimal generators
of the group $G$, a terminology dating back to Sophus Lie in the 19th
century.\index{Lie algebras ! exponential map |)}

Finally, an n-dimensional unitary representation $\varphi $ of $G$ yields by
differentiation an associated representation $\varphi _{\ast }$ of the Lie
algebra $\mathcal{G}$, and there is the following commutative diagram  
\begin{equation}
\begin{array}{ccc}
\mathcal{G} & \longrightarrow ^{\varphi \ast } & \mathcal{U}(n) \\ 
\downarrow \hat{e} &  & \downarrow \hat{e} \\ 
G & \longrightarrow ^{\varphi } & U(n)%
\end{array}
\label{diag2}
\end{equation}%
where $\mathcal{U}(n)$ is the Lie algebra of $U(n)$ consisting of the
skew-Hermitian matrices, and $\varphi _{\ast }$ is a \emph{Lie algebra
homomorphism}\index{Lie algebras ! homomorphism}, namely a linear map preserving the Lie bracket,  
\[
\varphi _{\ast }([X,Y])=[\varphi _{\ast }(X),\varphi _{\ast }(Y)] \ .
\]%

In effect, for a given set of infinitesimal generators $X_{i}\in $ $\mathcal{%
G}$ for $G$, the representation $\varphi $ is uniquely determined by the
skew-Hermitian matrices $\tilde{X}_{i}=$ $\varphi _{\ast }(X_{i})$. In many
cases, $\ker \varphi $ is a finite group, so that $G$ is locally isomorphic
with the image group $\tilde{G}=\varphi (G)$ in $U(n)$, wheras $\varphi
_{\ast }:\mathcal{G\rightarrow \tilde{G}}=\varphi _{\ast }(\mathcal{G)}$ is
an isomorphism between $\mathcal{G}$ and the Lie subalgebra $\mathcal{\tilde{%
G}}$ of $\mathcal{U}(n)$.

Conversely, a Lie algebra homomorphism of $\mathcal{G}$ is not always
derived from a Lie group homomorphism $\varphi $ of $G$, unless $G$ is simply
connected, such as the group $SU(n)$. In the general case, let us first
replace $G$ in (\ref{diag2}) by the unique simply connected group $\bar{G}$
which is locally isomorphic to $G$, the so-called \emph{universal covering group}.\index{Lie groups ! universal covering group} Then $G=%
\bar{G}/K$ for some discrete group $K$ of $\bar{G}$, and the Lie algebra of $%
\bar{G}$ is still the same $\mathcal{G}$. Now, every homomorphism of $%
\mathcal{G}$ is derived from some $\varphi $ for $\bar{G}$ as in the diagram (\ref{diag2}%
), and this Lie algebra homomorphism  will also exponentiate to a Lie group homomorphism of $G$ if and only if $%
K\subset \ker \varphi $.

\subsubsection{The adjoint and the coadjoint representations}

Let $G$ be any Lie group with Lie algebra $\mathcal{G}$. Then $G$ acts on
itself by inner automorphisms $\mathfrak{i}_{g}:G\rightarrow G$,
$\mathfrak{i}_{g}(h)=ghg^{-1}$, and differentiation of $\mathfrak{i}_{g}$ at
the identity $e\in G$ yields a homomorphism
\begin{equation}
Ad_{G}:G\rightarrow GL(\mathcal{G)}\text{, \ }Ad_{G}(g)=(\mathfrak{i}%
_{g})_{\ast}:\mathcal{G}\rightarrow\mathcal{G}\text{, }\label{Ad}%
\end{equation}
where in fact $Ad_{G}(g)$ is a Lie algebra isomorphism. $Ad_{G}$ is called the
\emph{adjoint} representation of $G$ and the dual representation
\[
Ad_{G}^{T}:G\rightarrow GL(\mathcal{G}^{\ast})
\]
is often called the \emph{coadjoint} representation. These are real
representations, and they are not necessarily equivalent, unless $G$ is
compact. 

But in the latter case one can construct an equivalence, that is, a
$G$-equivariant isomorphism $\Phi:\mathcal{G\rightarrow G}^{\ast}$, by
defining $\Phi(v)$ $=v^{\ast}$ to be the linear functional $w\rightarrow
\left\langle v,w\right\rangle $, where $\left\langle ,\right\rangle $ denotes
a $G$-invariant inner product on $\mathcal{G}$. The equivariance condition
$\Phi(gv)=g\Phi(v)$ amounts to show $\Phi(Ad(g))v=Ad^{T}(g)\Phi(v)$, and this
is easily checked by applying both sides to a vector $w\in\mathcal{G}$.

For a compact connected Lie group, a characteristic property of the adjoint
(or coadjoint) representation is that all orbits $\mathcal O$ are of type $G/H$, where $H$
is a subgroup containing a maximal torus $T$. In fact, the principal orbits,
filling an open dense subset of $\mathcal{G}$ (or $\mathcal{G}^{\ast}$), are
all of type $G/T$. For a description of all these subgroups $H$ of
$G=SO(n),SU(n),$ or $Sp(n)$, we refer to Table 1 in \cite{Straume}. 

Finally, we remind that every coadjoint orbit $\mathcal{O}$ (and hence also adjoint orbit when $\mathcal{O}=G/H$) naturally carries a $G$-invariant 
symplectic structure $\omega$, as follows:  For any $\nu\in\mathcal O\subset \mathcal{G}^{\ast}$ and 
$v,w\in\mathcal{G}$, via $({Ad_{G}^{T}})_*$ identify $v_{\nu},w_{\nu}\in T_{\nu}\mathcal O$. Then,  
\begin{equation}\label{canonicalsympform}
\omega(v_{\nu},w_{\nu})=\nu([v,w]) 
\end{equation}
defines a $G$-invariant nondegenerate closed $2$-form $\omega$ on $\mathcal O$, uniquely up to sign.

%-------------------------------------------------------------------------------------- Section 2.2 -----------------------------------------------------------------------------------------------------------------

\section{On the Lie groups SU(2) and SO(3)}\label{genrotgroup} 

In this section we shall explore in greater detail the Lie
groups $SU(2)$ and $SO(3)$, which have isomorphic Lie algebras;
they are locally isomorphic groups since $SO(3)\simeq SU(2)/\{\pm Id\}$, 
with $SU(2)$ being the universal covering group of $SO(3)$.\index{Lie groups |)}\index{Lie algebras |)} 

\subsubsection{Basic definitions}

Let $M_{\mathbb{K}}(n)$ be the $n\times n$-matrix algebra over $\mathbb{K}%
=\mathbb{R} $ or $\mathbb{C}$. First we focus attention on particular elements
of $M_{\mathbb{C}}(2)$ and $M_{\mathbb{R}}(3)$ in order to exhibit the
relationship between the groups $SU(2)$ and $SO(3)$, especially from the
viewpoint of quantum mechanics. We remind that these groups are defined as
follows:
\[
SU(2)=\{g\in M_{\mathbb{C}}(2) \ | \ g^{*}g=gg^{*}=Id, \ \det g = 1\} \ ,
\]
\[
SO(3)=\{g\in M_{\mathbb{R}}(3) \ | \ g^{T}g=gg^{T}=Id, \ \det g = 1\} \ .
\]
However, from the viewpoint of Lie group theory, the crucial fact is that
$SO(3)\simeq SU(2)/\mathbb{Z}_{2}$, where $\mathbb{Z}_{2}=\{\pm Id\}$ is the
center of $SU(2)$.
\index{$SU(2)\to SO(3)$ ! definitions}

Thus, we introduce the 3-vectors $\mathbf{\sigma=(}\sigma_{1},\sigma
_{2},\sigma_{3})$, $\mathbf{L}=(L_{1},L_{2},L_{3})$ whose components are
specific matrices
\begin{equation}
\text{ }\mathbf{\ }\sigma_{1}=\left(
\begin{array}
[c]{cc}%
0 & 1\\
1 & 0
\end{array}
\right)  ,\text{ }\sigma_{2}=\left(
\begin{array}
[c]{cc}%
0 & -i\\
i & 0
\end{array}
\right)  ,\text{ }\sigma_{3}=\left(
\begin{array}
[c]{cc}%
1 & 0\\
0 & -1
\end{array}
\right)  \text{ \ \ (Pauli spin matrices)\index{$SU(2)\to SO(3)$   ! Pauli spin matrices}} \label{Pauli}%
\end{equation}
\qquad\ \ \ \ \
\begin{equation}
L_{1}=\left(
\begin{array}
[c]{ccc}%
0 & 0 & 0\\
0 & 0 & -1\\
0 & 1 & 0
\end{array}
\right)  ,\text{ }L_{2}=\left(
\begin{array}
[c]{ccc}%
0 & 0 & 1\\
0 & 0 & 0\\
-1 & 0 & 0
\end{array}
\right)  ,\text{ }L_{3}=\left(
\begin{array}
[c]{ccc}%
0 & -1 & 0\\
1 & 0 & 0\\
0 & 0 & 0
\end{array}
\right)  \label{L}%
\end{equation}
which provide natural bases for the Hermitian matrices and two Lie algebras,
\begin{align*}
\mathcal{H(}2)  &  :I, \sigma_{1},\sigma_{2},\sigma_{3}\\
\mathcal{SU(}2)  &  :i\sigma_{1},i\sigma_{2},i\sigma_{3}\\
\mathcal{SO}(3)  &  :L_{1},L_{2},L_{3}%
\end{align*}
We will also use the notation
\[
\mathbf{n\cdot\sigma=}\sum n_{k}\sigma_{k},\text{ \ \ }\mathbf{n\cdot L=}\sum
n_{k}L_{k}%
\]
to denote any element of $\mathcal{H(}2)$ or $\mathcal{SO}(3),$ where the
vector $\mathbf{n}\in\mathbb{R}^{3}$ is to be interpreted as pointing in the
direction of the axis of rotation, in the case of $SO(3)$. To make this more
precise, first observe that the commutation rules\index{$SU(2)\to SO(3)$ ! Lie algebras}
\begin{equation}
\left\{
\begin{array}
[c]{c}%
\left[  \text{ }\sigma_{j},\text{ }\sigma_{k}\right]  =2i\epsilon_{jkl}%
\sigma_{l}\\
\left[  L_{j},L_{k}\right]  =\epsilon_{jkl}L_{l}%
\end{array}
\right.  \label{comm1}%
\end{equation}
lead to a Lie algebra isomorphism\index{$SU(2)\to SO(3)$ ! Lie algebra isomorphism}
\begin{equation}
d\psi:\text{\ }\mathcal{SU(}2)\rightarrow\mathcal{SO}(3);\text{ }\left\{
\begin{array}
[c]{c}%
\varepsilon_{1}\frac{i}{2}\sigma_{1}\rightarrow L_{1}\\
\varepsilon_{2}\frac{i}{2}\sigma_{2}\rightarrow L_{2}\\
\varepsilon_{3}\frac{i}{2}\sigma_{3}\rightarrow L_{3}%
\end{array}
\right.  \label{iso}%
\end{equation}
for any choice of signs $\varepsilon_{j}=\pm1$ with $\prod\varepsilon_{j}=-1$.
Our standard choice will be $\varepsilon_{j}=-1$ for all $j$, and in terms of
the exponential map, $A\rightarrow\exp(A)=e^{A}$, there is the geometrically
suggestive notation\index{$SU(2)\to SO(3)$ ! (spinor) rotation matrices |(}
\[
U(\mathbf{n,}\theta)=\exp(-\frac{i}{2}\theta(\mathbf{n\cdot\sigma))}\text{,
\ }^{\ }R(\mathbf{n},\theta)=\exp(\theta(\mathbf{n\cdot L))}\text{ }\
\]
which yields a ``standard'' homomorphism\index{$SU(2)\to SO(3)$ ! Lie group homomorphism} 
\begin{equation}
\psi:SU(2)\rightarrow SO(3),\text{ }U(\mathbf{n},\theta\mathbf{)\ \rightarrow
}R(\mathbf{n},\theta\mathbf{)}\ \label{homo}%
\end{equation}
We observe that, given a unit vector $\mathbf{n}$, $R(\mathbf{n},\theta)$ is
the rotation in euclidean 3-space through the angle $\theta$ (in the
right-handed sense) around the axis directed along $\mathbf{n}$, namely%
\[
R(\mathbf{n},\theta\mathbf{)}:\mathbf{v}\rightarrow(\cos\theta)\mathbf{v}%
+(1-\cos\theta)(\mathbf{n\cdot v)n}+(\sin\theta)\mathbf{n}\times\mathbf{v}%
\]
In view of this, $U(\mathbf{n,}\theta)\mathbf{\in}SU(2)$ receives the quantum
mechanical interpretation of a \emph{spinor rotation} with respect to the axis
determined by $\mathbf{n}$.\index{$SU(2)\to SO(3)$ ! (spinor) rotation matrices |)}

$\ $Since the homomorphism $\psi$ is one-to-one for small $\theta$ (and
$\mathbf{n}$ fixed), one can easily deduce the \emph{adjoint formulas}
\begin{equation}
U(\mathbf{n\cdot\sigma})U^{-1}=(R\mathbf{n)\cdot\sigma}\text{, \ \ \ }%
R(\mathbf{n\cdot L)}R^{-1}=(R\mathbf{n)}\cdot\mathbf{L}, \label{adjoint}%
\end{equation}
valid for any pair $U\in SU(2)$ and $R=\psi(U).$ In particular, this
establishes a natural equivalence between the adjoint representation of
$G=SU(2)$ or $SO(3),$ acting on its Lie algebra, and the standard
representation of $SO(3)$, acting by rotations on euclidean 3-space
$\mathbb{R}^{3}$.

Now, let $\mathbf{e}_{i},i=1,2,3$, be the usual orthonormal basis of
$\mathbb{R}^{3}$, for which $\mathbf{e}_{3}=(0,0,1)$ is identified with the
north pole of the unit sphere
\begin{equation}
S^{2}\subset\mathbb{R}^{3}:x^{2}+y^{2}+z^{2}=1 \label{Sphere2}%
\end{equation}
The classical way of expressing a rotation in terms of Euler angles
$(\alpha,\beta,\gamma)$ is frequently used in quantum mechanics. We refer to
\cite{Rose}, Chapter 13 or \cite{VMK}, Chapter 1.4. The idea is to express a
rotation as a product of three "simple" rotations, namely rotations around two
chosen coordinate axes. A widely used definition amounts to setting\index{$SU(2)\to SO(3)$ ! Euler angles}
\begin{align}
R(\alpha,\beta,\gamma)  &  =R(\mathbf{e}_{3},\alpha)R(\mathbf{e}_{2}%
,\beta)R(\mathbf{e}_{3},\gamma)=e^{\alpha L_{3}}e^{\beta L_{2}}e^{\gamma
L_{3}}\label{RotR}\\
U(\alpha,\beta,\gamma)  &  =U(\mathbf{e}_{3},\alpha)U(\mathbf{e}_{2}%
,\beta)U(\mathbf{e}_{3},\gamma)=e^{-i\frac{\alpha}{2}\sigma_{3}}%
e^{-i\frac{\beta}{2}\sigma_{2}}e^{-i\frac{\gamma}{2}\sigma_{3}} \label{RotU}%
\end{align}
and then the above homomorphism (\ref{homo}) expresses as
\[
\psi:U(\alpha,\beta,\gamma)\rightarrow R(\alpha,\beta,\gamma)
\]

We point out, however, that there is no canonical choice of homomorphism
between $SU(2)$ and $SO(3)$. However, the various choices only differ by an
automorphism of $SU(2)$ (or $SO(3))$. For example, complex conjugation in
$SU(2)$ is the automorphism $\mathfrak{\varsigma}:U(\alpha,\beta
,\gamma)\rightarrow U(-\alpha,\beta,-\gamma)$, which composed with $\psi$
yields the homomorphism\index{$SU(2)\to SO(3)$ ! Lie group homomorphism} 
\begin{equation}
\psi^{\prime}=\psi\circ\mathfrak{\varsigma}\ :SU(2)\rightarrow SO(3),\text{
}U(\alpha,\beta,\gamma)\rightarrow R(-\alpha,\beta,-\gamma) \label{homo2}%
\end{equation}
In the following subsection we shall derive the two homorphisms $\psi$ and
$\psi^{\prime}$ in a more geometric way, in terms of equivariant maps between spaces.

\subsubsection{Hopf map and stereographic projection}

Now, $SU(2)$ acts by its standard (unitary) representation on the 3-sphere 
\begin{equation}
S^{3}\subset\mathbb{C}^{2}:\left\vert z_{1}\right\vert ^{2}+\left\vert
z_{2}\right\vert ^{2}=1, \label{S3}%
\end{equation}
and the following diffeomorphism
\begin{equation}
\Psi:SU(2)\rightarrow S^{3},\text{ \ }g=\left(
\begin{array}
[c]{cc}%
z_{1} & -\bar{z}_{2}\\
z_{2} & \bar{z}_{1}%
\end{array}
\right)  \mapsto \Psi(g)= \left(\begin{array}[c]{c} z_1 \\  z_2 \end{array}\right) \sim (z_{1},z_{2})=\mathbf{z} \label{C-K}%
\end{equation}
which identifies $SU(2)$ with the 3-sphere, is equivariant when the group acts
on itself by left translation (in equation (\ref{C-K}) above, the symbol $\sim$ means that we identify a column vector with a line vector whenever the distinction is irrelevant, in order to save space). More precisely, equivariance under left action means that, for all $g$ and $\Psi(g)$ as in (\ref{C-K}) above, and for all 
$$h=\left(
\begin{array}
[c]{cc}%
w_{1} & -\bar{w}_{2}\\
w_{2} & \bar{w}_{1}%
\end{array}
\right) \in SU(2) \ , \Psi(h)= \left(
\begin{array}
[c]{c}%
w_{1} \\
w_{2} %
\end{array}
\right) \in S^3\subset\mathbb C^2 \ ,$$
we have that 
$$g\Psi(h)= \left(
\begin{array}
[c]{cc}%
z_{1} & -\bar{z}_{2}\\
z_{2} & \bar{z}_{1}%
\end{array}
\right) \left(
\begin{array}
[c]{c}%
w_{1} \\
w_{2} %
\end{array}
\right) = \Psi \left( \left(
\begin{array}
[c]{cc}%
z_{1} & -\bar{z}_{2}\\
z_{2} & \bar{z}_{1}%
\end{array}
\right)  \left(
\begin{array}
[c]{cc}%
w_{1} & -\bar{w}_{2}\\
w_{2} & \bar{w}_{1}%
\end{array}
\right) \right) = \Psi(gh) \ .$$

Let us recall the classical Hopf map\index{Hopf map}
\begin{equation}
\pi:S^{3}\rightarrow S^{2}\simeq S^{3}/U(1)=\mathbb{C}P^{1},\text{ \ }\left\{
\begin{array}
[c]{c}%
\mathbf{z}=(z_{1},z_{2})\rightarrow\mathbf{n}=(x,y,z)\\
x+iy=2\bar{z}_{1}z_{2},\text{ }z=\left\vert z_{1}\right\vert ^{2}-\left\vert
z_{2}\right\vert ^{2}%
\end{array}
\right.  \label{hopf}%
\end{equation}
which yields a fibration of the 3-sphere (\ref{S3}) over the 2-sphere
(\ref{Sphere2}).$\ $As indicated in (\ref{hopf}), $\pi$ also identifies
$S^{2}$ with the orbit space of $U(1)=\left\{  e^{i\theta}\right\}  $ where
$e^{i\theta}$ acts by scalar multiplication on vectors $\mathbf{z\in}$
$\mathbb{C}^{2}$.

Now, $SO(3)$ acts by rotations on the 2-sphere, and associated with the map
$\pi$ is a distinguished homomorphism\index{$SU(2)\to SO(3)$ ! Lie group homomorphism} 
\begin{equation}
\bar{\psi}:SU(2)\rightarrow SO(3) \label{psi}%
\end{equation}
defined by the constraint that $\pi$ is $\bar{\psi}$-equivariant, namely
\begin{equation}
\pi(g\mathbf{z)}=\bar{\psi}(g)\pi(\mathbf{z})\text{, \ for }g\in SU(2)
\label{hopf2}%
\end{equation}
In fact, $\bar{\psi}$ coincides with the homomorphism $\psi$ in (\ref{homo}).
Moreover, by writing $z_{1}=x_{1}+iy_{1}$, $z_{2}=x_{2}+iy_{2}$, a
straightforward calculation of the homomorphism (\ref{psi}) gives the explicit
expression
\begin{equation}
\psi({g})=\left(
\begin{array}
[c]{ccc}%
(x_{1}^{2}-x_{2}^{2}-y_{1}^{2}+y_{2}^{2}) & 2(x_{1}y_{1}-x_{2}y_{2}) & x\\
-2(x_{1}y_{1}+x_{2}y_{2}) & (x_{1}^{2}+x_{2}^{2}-y_{1}^{2}-y_{2}^{2}) & y\\
-2(x_{1}x_{2}-y_{1}y_{2}) & -2(x_{1}y_{2}+x_{2}y_{1}) & z
\end{array}
\right)  \label{SO(3)}%
\end{equation}
where the third column is the vector $\mathbf{n}$ in (\ref{hopf}). For
historical reasons, the pair $(z_{1},z_{2})$, subject to the condition
$\left\vert z_{1}\right\vert ^{2}+\left\vert z_{2}\right\vert ^{2}=1$, is also
referred to as \emph{Cayley-Klein}\index{$SU(2)\to SO(3)$ ! Cayley-Klein coordinates} coordinates for $SU(2)$ and $SO(3)$.

Next, let us also recall the stereographic projection\index{$2$-sphere ! stereographic projection} to the complex plane
$\mathbb{C}$
\begin{equation}
\pi^{\prime}:S^{2}-\left\{  \mathbf{e}_{3}\right\}  \rightarrow\mathbb{C}%
\text{, }(x,y,z)\rightarrow\xi=\frac{x+iy}{1-z}\text{\ } \label{stereogr}%
\end{equation}
and$\ $the action of $SU(2)$ on $\mathbb{C}$ by fractional linear (or
M\"{o}bius) transformations
\begin{equation}
\xi\rightarrow\frac{z_{1}\xi-\bar{z}_{2}}{z_{2}\xi+\bar{z}_{1}}. \label{Moeb}%
\end{equation}
Associated with the map (\ref{stereogr}) is a homomorphism\index{$SU(2)\to SO(3)$ ! Lie group homomorphism}  $SU(2)\rightarrow
SO(3)$ which makes (the inverse of) the 1-1 correspondence $\pi^{\prime}$ in
(\ref{stereogr}) equivariant, and by straightforward calculations the
homomorphism is found to be $\psi^{\prime}$ in (\ref{homo2}).

The identity $Id\in SU(2)$ corresponds via $\Psi$ in (\ref{C-K}) to the basic
vector $\mathbf{z}=(1,0)$, which by $\pi$ is mapped to the north pole
$\mathbf{e}_{3}$ of $S^{2}$, and the corresponding isotropy groups
$H=U(1)\simeq SO(2)$ in $G=SU(2)$ and $SO(3)$ are related by
\begin{equation}
SU(2)\supset\left\{  e^{-i\frac{\theta}{2}\sigma_{3}}\right\}
=U(1)\rightarrow^{\psi}SO(2)=\left\{  e^{\theta L_{3}}\right\}  \subset SO(3)
\label{U(1)}%
\end{equation}
This also realizes the 2-sphere as a homogeneous space in two ways,\index{$2$-sphere ! as a homogeneous space}
\begin{equation}%
\begin{array}
[c]{cc}%
SU(2)=S^{3}\searrow^{\pi} & \ \\
\downarrow\psi & S^{2}=G/H=SU(2)/U(1)=SO(3)/SO(2)\\
SO(3)\nearrow_{\bar{\pi}} & \
\end{array}
\label{Hopf2}%
\end{equation}
where $\bar{\pi}:SO(3)\to S^{2}$ projects a matrix to its third column
$\mathbf{n}=(x,y,z).$

$SO(3)$ acts by orthogonal transformations on $\mathbb{R}^{3}$; this is the
standard representation $\rho_{3}$. The adjoint representation $Ad_{SO(3)}$ is
the action by conjugation, $Ad(g)S=gSg^{-1}$, on skew symmetric matrices
$S\in\mathcal{SO(}3)$. We set up the following 1-1 correspondence between
$\mathcal{SO(}3)$ and $\mathbb{R}^{3}$
\[
S=\left(
\begin{array}
[c]{ccc}%
0 & z & y\\
-z & 0 & x\\
-y & -x & 0
\end{array}
\right)  \longleftrightarrow\left(
\begin{array}
[c]{c}%
x\\
y\\
z
\end{array}
\right)
\]
which is, in fact, $SO(3)$-equivariant, and this proves that $\rho_{3}$
identifies with the adjoint, hence also the coadjoint representation of
$SO(3)$. 

The orbits are therefore all concentric spheres $S^{2}(r)$ of radius
$r$ $\geq0$, which for $r>0$ are of type $SO(3)/SO(2)$, as in (2.26).

It is often convenient to have expressions for functions and various other structures defined on $S^{2}$ written in local coordinates. 
Here we shall mostly write them in local spherical polar coordinates \index{$2$-sphere ! polar coordinates} $(\theta,\varphi)$ where, for simplicity, we
identify $S^{2}$ with the unit sphere in $\mathbb{R}^{3}$ according to (\ref{Sphere2}), 
so that $(\theta,\varphi)$ are defined by 
\begin{equation}
x=\sin\varphi\cos\theta,\text{ \ }y=\sin\varphi\sin\theta,\text{ \ }%
z=\cos\varphi\label{polar}%
\end{equation}
where $\varphi$ is the colatitude, $\theta$ is the longitude and the origin of the polar coordinate system is the north pole
$\mathbf{e}_3=\mathbf{n}_{0}=(0,0,1)\in S^{2}(1)\subset\mathbb{R}^{3}$.

Note that, since the radius of any $2$-sphere in $\mathbb R^3$ is an $SO(3)$-invariant quantity, we are free to rescale all spheres of radius $r> 0$ to the unit sphere. Or equivalently, by using angular coordinates $(\theta, \varphi)$ we can simply forget about radii.     

In particular, the
$G$-invariant symplectic (or area) form, cf. (\ref{canonicalsympform}), is (locally) expressed in terms of
spherical polar coordinates as  \index{$2$-sphere ! symplectic form} 
\begin{equation}
\omega=\sin\varphi d\varphi\wedge d\theta\ . \label{S2a}%
\end{equation}

 Of course, instead of polar coordinates we could use complex coordinates on $S^2$, 
 but a relation between these two local coordinate systems is straightforwardly obtained from (\ref{polar}) and (\ref{stereogr}).

\subsubsection{Prelude to the irreducible unitary representations of SU(2)} 

The irreducible (unitary) representations of $SU(2)$ are finite dimensional and typically denoted by
$\left[  j\right]  $ in the physics literature, but we shall also use the
notation\index{$SU(2)\to SO(3)$ ! irreducible representations |(}%
\begin{equation}
\text{ }\varphi_{j}=\left[  j\right]  ,\text{ \ }j=0,1/2,1,3/2,...:\dim
_{\mathbb{C}}\varphi_{j}=2j+1, \label{irrep}%
\end{equation}
These are $SO(3)$-representations only when $j$ is an integer $l$, in which
case they have a real form
\begin{equation}
\psi_{l},\text{ \ }l=0,1,2,3,...,\dim_{\mathbb{R}}\psi_{l}=2l+1,
\label{realrep}%
\end{equation}
that is, $\varphi_{l}=\left[  \psi_{l}\right]  ^{\mathbb{C}}$ is the
complexification of $\psi_{l}$.  However, if $j$ has 
half-integral value, $\varphi_{j}$ has a quaternionic form, so, in both cases $\varphi_{j}$ is actually seff-dual, that is, $\bar{\varphi}_{j}\simeq {\varphi}_{j}$ for all $j$ (another issue is the precise relation of a given basis with its dual and the precise form of this equivalence). 

It is a basic fact about compact connected Lie groups $G$ that its maximal
tori $T$ constitute a single conjugacy class $(T)$, and each element $g\in G$ can be conjugated into a fixed torus $T$, say $hgh^{-1}=t\in T$ for some $h\in G$.
In particular, for the matrix groups $SU(n),U(n)$ this is
 the diagonalization of matrices, when $T$ is chosen to be the diagonal
matrices with entries $e^{i\theta _{k}}$. As a consequence of this, a
representation of $G$ is uniquely determined by its restriction to the torus 
$T$. 

In the case of $G=SU(2)$, the diagonal group $U(1)$ in (\ref{U(1)}) is our chosen maximal torus $T$. By Schur's
lemma (cf. Lemma \ref{Schur's lemma}), a unitary representation $\varphi_j$ of $SU(2)$, when
restricted to $U(1)$, splits into 1-dimensional representations and clearly
each one is determined by a homomorhism $U(1)\rightarrow \mathbb{C}^{\ast }$ of
type    
\[
diag(e^{i\theta },e^{-i\theta })\rightarrow e^{qi\theta }\text{, }q\in 
\mathbb{Z}
\]%
where the ``functional" $q\theta $ is called the w\emph{eight }of the above $%
U(1)$-representation. Thus, there are altogether $(n+1)=\dim \varphi_j$  weights $%
q_{i}\theta $, and the totality of weights%
\begin{equation}
\Omega (\varphi_j)=\{q_{0}\theta, q_{1}\theta ,q_{2}\theta ,...,q_{n}\theta \}
\label{w-system}
\end{equation}%
completely determines $\varphi_j$ and is referred to as the \emph{weight system }%
of the $SU(2)$-representation (with respect to $U(1)$). 

\begin{remark}\label{multisetremark} The
collection (\ref{w-system}) must be regarded as a \emph{multiset}, namely the
elements are counted with multiplicity, since the $n+1$ weights are not
necessarily distinct but the total multiplicity equals $n+1.$
For example, as multisets we write identities of the following type:  
\[
\{a,a,b,c,b,a,c,b,d,e\}=3\{a,b\}+2\{c\}+\{d,e\}.
\]%
\end{remark}

Now, a well-known way to build a $(n+1)$-dimensional representation of $SU(2)$ explicitly from its standard $2$-dimensional representation  is by mapping $\mathbb C^2$ to the space of complex homogeneous polynomials of degree $n$ in two variables: 
$$\mathbf z=( z_1, z_2)\in \mathbb C^2 \to  \mathbb C^{n+1}\simeq { }^h\!P^n_{\mathbb C}( z_1,  z_2)= Span_{\mathbb C}\{{ z_1^{n-k}}{ z_2^{k}}\}_{0\leq k\leq n} \ .$$  
Then, as $SU(2)$ acts on $\mathbb C^2$ via its standard representation, this induces an action of $SU(2)$ on 
${ }^h\!P^n_{\mathbb C}( z_1,  z_2)$ and different but equivalent $(n+1)$-dimensional representations of $SU(2)$ are related by different choices 
of basis for ${ }^h\!P^n_{\mathbb C}( z_1,  z_2)$. 
A choice of orthonormal basis for ${ }^h\!P^n_{\mathbb C}( z_1,  z_2)$ is the ordered set $\{\mathbf v(n,k)\}$, where 
\begin{equation}\label{standmon}  
\mathbf v(n,k)=\sqrt{\binom{n}{k}}z_1^{n-k} z_2^{k} \ , k=0,1,\cdots, n \ . 
\end{equation}
However, such a basis is often written in terms of $j=n/2$ and $m=j-k$. 
 
 As will be made clearer further below, equation (\ref{standmon}) defines the basis $\{\mathbf v(n,k)\}$, for each $n=2j$, only up to an overall phase factor. 
 In fact, accounting for a scaling freedom  for the inner product on ${ }^h\!P^n_{\mathbb C}( z_1,  z_2)$ implies that two such  bases given by (\ref{standmon}) can be identified, for each $j$, if they differ from each other by an overall nonzero complex number (cf. Schur's Lemma \ref{Schur's lemma}), and the  basis introduced by Bargmann \cite{Barg} differs from (\ref{standmon}) above by the $\sqrt{n!}$ factor.   

Finally, note  that one can map these concrete representations of $SU(2)$ on ${ }^h\!P^n_{\mathbb C}( z_1,  z_2)$ to concrete representations on the space of $n$-degree holomorphic polynomials  on $S^2$, $\mathcal Hol^n(S^2)$, by composing with the projective maps $$\mathbb C^2\to\mathbb CP^1\simeq S^2 \ , \ (z_1,z_2)\mapsto (1,\xi=z_2/z_1)  , \ \text{or} \   (\zeta=z_1/z_2, 1).$$
Composing (\ref{standmon}) with the projective map $(z_1,z_2)\mapsto\zeta$ yields the basis used by Berezin \cite{Berezin}, while the map $(z_1,z_2)\mapsto\xi$ and Bargmann's normalization convention yields the basis used in 
the book of Vilenkin and Klimyk \cite{Vilekin}, where explicit expressions  for the $(n+1)$-dimensional representations of $SU(2)$ are presented.     

In the next chapter, the irreducible representations of $SU(2)$ shall be studied in great detail  in terms of infinitesimal techniques involving
weights, operators and Lie algebras, which 
 is the approach commonly
used in quantum mechanics.

%----------------------------------------------------------------------------------------------------------------------------------------------------------------------------------------------------------------------
%-------------------------------------------------------------------------------------- Chapter 3 -----------------------------------------------------------------------------------------------------------------
%----------------------------------------------------------------------------------------------------------------------------------------------------------------------------------------------------------------------

\chapter{Quantum spin systems and their operator algebras}\label{quantumspinsystems} 

This chapter presents the basic mathematical framework for quantum mechanics
of spin systems. Much of the material can be found in texts in representation
theory (some found within the list of references at the beginning of Chapter 2) and quantum theory of angular momentum (e.g. \cite{BL, BL2, Rose, Bohm, CT-D-L, Saku}, the last three being textbooks in quantum mechanics which can also be used by the reader not too familiar with the subject as a whole). Our emphasis here is to provide a
self-contained presentation of quantum spin systems where, in particular,  
the combinatorial role of Clebsch-Gordan coefficients and various kinds of Wigner symbols is elucidated, leading to the $SO(3)$-invariant decomposition of
the operator product which, strangely enough, we have not found explicitly done anywhere.

%-------------------------------------------------------------------------------------- Section 3.1 -----------------------------------------------------------------------------------------------------------------

\section{Basic definitions of quantum spin systems}

In line with the standard formulation of quantum affine mechanical systems, we define quantum spin systems as follows: 

\begin{definition}
A \emph{spin-j quantum mechanical system}, or \emph{spin-j system},\index{Spin-j system ! definition} is a complex Hilbert space
$\mathcal{H}_{j}\simeq\mathbb{C}^{n+1}$ together with an irreducible unitary
representation
\begin{equation}
\varphi_{j}:SU(2)\rightarrow G\subset U(\mathcal{H}_{j})\simeq U(n+1),\text{
\ }n=2j \in \mathbb N \ , \label{basic}%
\end{equation}
where $G$ denotes the image of $SU(2)$ and hence is isomorphic to $SU(2)$ or
$SO(3)$ according to whether $j$, called the \emph{spin number}\index{Spin-j system ! spin number}, is half-integral or integral.
\end{definition}

\begin{remark}
\label{conjugatefirstentry} Unless otherwise stated (as in Appendix
\ref{proofWild}), throughout this monograph we shall always assume the Hermitian
inner product of a Hilbert space is skew-linear (conjugate linear) in the
first variable.
\end{remark}

A vector in $\mathcal{H}_{j}$ is also called a $j$\emph{-spinor}.\index{Spin-j system ! j-spinor}  For our
description of bases, operators and matrix reprentations we will use familiar
terminology and notation from quantum mechanics. The representation
(\ref{basic}) is normally described at the infinitesimal level by Hermitian
operators $J_{1},J_{2},J_{3}$ satisfying the standard commutation relations
for angular momentum, namely
\begin{equation}
\left[  J_{a},J_{b}\right]  =i\epsilon_{abc}J_{c}\ \label{comm2}%
\end{equation}
together with the basic relation
\begin{equation}
J^{2}=J_{1}^{2}+J_{2}^{2}+J_{3}^{2}=j(j+1)I \ . \label{J2sum}%
\end{equation}

Indeed, by (\ref{comm2}) the operator\ sum $J^{2}$ commutes with each $J_{i}$,
so by irreducibility it must be a multiple of the identity. All this is
equivalent to saying that the group $G\ $with the Lie algebra
\begin{equation}
\mathcal{G=}\text{ }lin_{\mathbb{R}}\left\{  iJ_{1},iJ_{2},iJ_{3}\right\}
\simeq\mathcal{SU(}2) \label{span}%
\end{equation}
of skew-Hermitian operators acts irreducibly on $\mathcal{H}_{j}$.

In analogy with $\mathbf{\sigma}$ and $\mathbf{L}$ in (\ref{Pauli}),
(\ref{L}), the vector of operators
\begin{equation}
\mathbf{J=}\ (J_{1},J_{2},J_{3}) \label{spin}%
\end{equation}
is referred to as the \emph{total angular momentum (or spin)}\index{Spin-j system ! total angular momentum}
operator\emph{\ }of the quantum system, and its components satisfy the
standard commutation relations (\ref{comm2}) for angular momentum. The
commutation relations are also symbolically expressed as $\mathbf{J\times
J}=i\mathbf{J}$, and the square $\mathbf{J}^{2}$ is the operator
sum$\ $(\ref{J2sum}).

\begin{remark}
\label{hbar=1} In the physics literature, one usually finds Planck's constant
$\hbar$, resp. $\hbar^{2}$, explicitly multiplying the r.h.s. of equation
(\ref{comm2}), resp. equation (\ref{J2sum}), which also guarantees that
$\mathbf{J}$ has the dimensions of angular momentum. However, since this
factor can be removed by an appropriate scaling of $\mathbf{J}$, we will omit
it throughout almost the whole book.
\end{remark}

Note that the infinitesimal generators $-iJ_{k},k=1,2,3,$ of the operator
group $G$ satisfy the same commutation relations as the operators $L_{k}$ in
(\ref{comm1}), that is, the correspondence $L_{k}\rightarrow-iJ_{k},1\leq
k\leq3,$ is a Lie algebra isomorphism, and for a unit vector $\mathbf{n}$ in
euclidean 3-space we shall refer to the operator
\[
J_{\mathbf{n}}=\mathbf{n\cdot J=}\sum n_{i}J_{i}%
\]
as the angular momentum in the direction of $\mathbf{n.}$ The corresponding
homomorphism in (\ref{basic}) is
\begin{equation}
\varphi_{j}:e^{\frac{1}{2}i\theta(\mathbf{n}^{.}\mathbf{\sigma)}}\rightarrow
e^{i\theta(\mathbf{n}^{.}\mathbf{J)}}\ \ \label{homo3}%
\end{equation}

\subsection{Standard basis and standard matrix representations}\index{Spin-j system ! standard representation |(}

The above operators on $\mathcal{H}_{j}$ will be represented by matrices with
respect to a suitable choice of orthonormal basis, unique up to a common phase
factor. In particular, the basis diagonalizes the operator $J_{3}$ and will be
referred to as a \emph{standard basis}. It will be characterized below in
terms of the action of the angular momentum operators.

Starting with the simplest case, for a spin-$\frac{1}{2}$ quantum system the
angular momentum (\ref{spin}) is defined to be the following vector of
operators
\begin{equation}
\mathbf{J=}\frac{1}{2}\mathbf{\sigma=}\frac{1}{2}\mathbf{(}\sigma_{1}%
,\sigma_{2},\sigma_{3})\text{,\ } \label{Pauli2}%
\end{equation}
namely the Pauli matrices with the factor 1/2. Then $G=$ $SU(2)$ in
(\ref{basic}), and $\varphi_{1/2}$ is the identity. In particular, the
cartesian basis $\{\mathbf{e}_{1,}\mathbf{e}_{2}\}$ of $\mathcal{H}%
_{1/2}=\mathbb{C}^{2}$ diagonalizes $J_{3}$ with eigenvalues $\pm1/2$, and
this is a standard basis, see below.

Next, for a spin-$1$ quantum system the angular momentum is defined to be
\[
\mathbf{J}=i\mathbf{L}
\]
and a standard basis is typically chosen to be
\begin{equation}\label{standard-n=2}
-\frac{1}{\sqrt{2}}(\mathbf{e}_{1}+i\mathbf{e}_{2}),\mathbf{e}_{3},\frac
{1}{\sqrt{2}}(\mathbf{e}_{1}-i\mathbf{e}_{2})
\end{equation}
(note that the first and last elements in this standard basis are {\it complex} linear combinations of $\mathbf{e}_{1}$ and $\mathbf{e}_{2}$, cf. Remark  \ref{complexrepr}, below).

In general, for a spin-$j$ system one would like to \textquotedblleft
measure\textquotedblright\ the angular momentum $J_{\mathbf{n}}$ in a chosen
direction $\mathbf{n}$, usually by choosing the angular momentum $J_{3}$ in
the (positive) z-axis direction. The eigenvalues $m$ of $J_{3}$ are sometimes
referred to as \emph{magnetic quantum numbers}\index{Spin-j system ! magnetic number}. In fact, $\mathcal{H}_{j}$ has
an orthonormal basis
\begin{equation}
\mathbf{u}(j,m)=\left\vert jm\right\rangle ,\text{ \ }m=j,j-1,...,-j+1,-j
\label{u}%
\end{equation}
consisting of eigenvectors of $J_{3}$ whose eigenvalues constitute the string
of numbers $m$ as indicated in (\ref{u}), where Dirac's ``ket'' notation for
the vectors is displayed. Thus, with the above ordering of the basis, $J_{3}$
has the matrix representation
\begin{equation}
J_{3}=\left[
\begin{array}
[c]{ccccc}%
j & 0 & 0 & 0 & 0\\
0 & j-1 & 0 & 0 & 0\\
: & : & : & : & :\\
0 & 0 & 0 & -j+1 & 0\\
0 & 0 & 0 & 0 & -j
\end{array}
\right]  \label{J3}%
\end{equation}
We shall use the term \emph{standard basis}\index{Spin-j system ! standard basis} for the above basis (\ref{u}). So
far, however, the vectors are only determined modulo an individual phase factor.

In terms of weights we consider the circle subgroup $\left\{  e^{i\theta
\sigma_{3}}\right\}  $ of $SU(2)$ (cf. (\ref{U(1)})), consisting of the spinor
rotations around the z-axis, acting on $\mathcal{H}_{j}$ with the vectors
(\ref{u}) as weight vectors and $2m\theta$ as the associated weight, namely
\begin{equation}
\varphi_{j}(e^{i\theta\sigma_{3}})=e^{i2\theta J_{3}}:\mathbf{u}%
(j,m)\rightarrow e^{2mi\theta}\mathbf{u}(j,m) \label{weight}%
\end{equation}
Therefore, by definition, the weight system of $\varphi_{j}$ is the set
\begin{equation}
\Omega(\varphi_{j})=\left\{  2j\theta,2(j-1)\theta,...,-2j\theta\right\}
\label{weight1}%
\end{equation}
%\qquad\

In order to fix a phase convention for a standard basis (\ref{u}), which also
fixes our \emph{standard} matrix representation of the operators
$J_{i},i=1,2,3$, let us first invoke the strucure of the algebra (\ref{span}),
expressed by the commutation rules (\ref{comm2}). To this end one introduces
the mutually adjoint pair of operators
\begin{equation}
J_{+}=J_{1}+iJ_{2},\text{ \ }J_{-}=J_{1}-iJ_{2} \label{lower}%
\end{equation}
called the \emph{raising}\ and \emph{lowering}\ operators, respectively, whose
commutation rules
\begin{equation}
\left[  J_{+},J_{-}\right]  =2J_{3,}\text{ \ }\left[  J_{3},J_{\pm}\right]
=\pm J_{\pm}, \label{rel}%
\end{equation}
yield the following identity between nonnegative Hermitian operators
\begin{equation}
J_{-}J_{+}=J^{2}-J_{3}(J_{3}+I). \label{rel1}%
\end{equation}

The relations (\ref{rel}) also imply
\begin{equation}
J_{+}\mathbf{u}(j,m)=\alpha_{j,m}\mathbf{u}(j,m+1),\text{ \ }J_{-}%
\mathbf{u}(j,m)=\beta_{j,m}\mathbf{u}(j,m-1) \label{gen}%
\end{equation}
for some constants $\alpha_{j,m},\beta_{j,m}$ which are nonzero, except that
$\alpha_{jj}=\beta_{j,-j}=0$ (since there is no eigenvector outside the range
(\ref{u})).

\begin{definition}
\label{standard}A standard basis\index{Spin-j system ! standard basis} $\left\{  \mathbf{u}(j,m)\right\}  $ of
$\mathcal{H}_{j}$, ordered as in (\ref{u}), is defined by choosing the first
(and highest weight) unit vector $\mathbf{u}(j,j)$ and inductively fixing the
phase of $\mathbf{u}(j,m-1)$ so that $\beta_{j,m}$ in (\ref{gen}) is always nonnegative.
\end{definition}

Consequently, a standard basis is unique up to one common phase factor
$e^{i\omega}$. The above commutation rules yield the formulae
\begin{equation}
\alpha_{j,m}=\sqrt{(j-m)(j+m+1)},\text{ \ }\beta_{j,m}=\sqrt{(j+m)(j-m+1)}
\label{c-d}%
\end{equation}
With respect to a standard basis the mutually adjoint matrices representing
$J_{\pm}$ have nonnegative entries, so they are the transpose of each other,
with all nonzero entries on a subdiagonal, as illustrated (where $n=2j)$:
\begin{equation}
J_{-}= J_{+}^{T}=\left[
\begin{array}
[c]{cccccc}%
0 & 0 & 0 & 0 & 0 & 0\\
\sqrt{n^{.}1} & 0 & 0 & 0 & 0 & 0\\
0 & \sqrt{(n-1)^{.}2} & 0 & 0 & 0 & 0\\
: & : & : & : & : & :\\
0 & 0 & 0 & \sqrt{2^{.}(n-1)} & 0 & 0\\
0 & 0 & 0 & 0 & \sqrt{1^{.}n} & 0
\end{array}
\right]  \text{ \ } \label{Jn}%
\end{equation}
From this, one calculates the Hermitian matrices
\[
J_{1}=\frac{1}{2}(J_{+}+J_{-})\text{ , \ }J_{2}=\frac{1}{2i}(J_{+}-J_{-}),
\]
and in the initial case $j=1/2$ the identity (\ref{Pauli2}) is recovered.

Finally, consider also the induced action of $SU(2)$ on the dual space
\begin{equation}
\mathcal{{H}}_{j}^*=Hom(\mathcal{H}_{j},\mathbb{C)\simeq}\mathcal{H}_{j}
\label{dual}%
\end{equation}
namely the dual representation $\bar{\varphi}_{j}$. This is isomorphic to
$\varphi_{j}$, so we may regard the two spaces in (\ref{dual}) to be the same
underlying Hilbert space. Then the two cases are distinguished by the actions,
namely $g\in SU(2)$ acts by its complex conjugate $\bar{g}$ in the dual case.
The resulting effect on infinitesimal generators is that $J_{2}$ is invariant,
whereas $J_{1}$ and $J_{3}$ are multiplied by -1, and thus the raising and
lowering operators in the dual case are
\[
\check{J}_{+}=-J_{-},\text{ \ \ }\check{J}_{-}=-J_{+}%
\]
Consequently, in view of (\ref{gen}) a standard basis of $(\varphi
_{j},\mathcal{H}_{j})$ is not a standard basis of the dual representation $(\check{\varphi}%
_{j},\check{\mathcal{H}}_{j})\equiv (\bar{\varphi}%
_{j},\mathcal{H}_{j}^*)$. However, the standard basis with the vectors
in the opposite order and with alternating sign changes is a dual standard
basis. Our choice of sign convention is specified as follows:

\begin{definition}
\label{dualbasis}The dual standard basis,\index{Spin-j system ! dual standard basis} dual to $\left\{  \mathbf{u}%
(j,m)\right\} $, is the ordered collection of vectors
\begin{equation}
\mathbf{\check{u}}(j,m)=(-1)^{j+m}\mathbf{u}(j,-m),\text{ \ }-j\leq m\leq j
\label{dual2}%
\end{equation}
\end{definition}

\begin{remark}
Observe that the standard duality (\ref{dual2}) is not \textquotedblleft
involutive\textquotedblright\ when $j$ is half-integral, since applying the
dual construction twice amounts to multiplying the original vectors
$\mathbf{u}(j,m)$ by $(-1)^{2j}$.
\end{remark}

The unitary operators $\varphi_{j}(g)$ on $\mathcal{H}_{j}$ are represented by
well defined unitary matrices $D^{j}(g)$. Using the notation (\ref{RotU}) for
elements $g\in SU(2)$,\ we shall denote the corresponding unitary operators
$\varphi_{j}(g)$ on $\mathcal{H}_{j}$ by $\hat{D}^{j}(g)$ or $\hat{D}%
^{j}(\alpha,\beta,\gamma)$, namely there is the homomorphism
\begin{equation}
\label{WigD}\varphi_{j}:U(\alpha,\beta,\gamma)\ \rightarrow\hat{D}^{j}%
(\alpha,\beta,\gamma)=e^{-i\alpha J_{3}}e^{-i\beta J_{2}}e^{-i\gamma J_{3}}%
\end{equation}
The associated matrix of $\hat{D}^{j}$ with respect to a standard basis
(\ref{u}) is the matrix $D^{j}=\left(  D_{m_{1},m_{2}}^{j}\right)  $ whose
entries are the following functions on $SU(2)$ $\ $
\begin{equation}
D_{m_{1},m_{2}}^{j}(g)=\left\langle \mathbf{u}(j,m_{1}),\hat{D}^{j}%
(g)\mathbf{u}(j,m_{2})\right\rangle , \label{Wigner-D}%
\end{equation}
also called the \emph{Wigner D-functions}\index{Spin-j system ! Wigner D-functions }\index{Wigner D-functions }, cf. \cite{VMK}, Chap. 4.

\begin{remark}\label{complexrepr} The reader should be aware that even when $j$ is an integer, so that  the $SU(2)$-representation is effectively a representation of $SO(3)$, the standard representation is a complex representation. Thus, for instance, the standard representation of $SU(2)$ for $j=1$ consists of complex $3\times 3$ matrices. Namely, by conjugation with the unitary transition matrix from the basis $\mathbf{e}_{i},i=1,2,3$, to the basis (\ref{standard-n=2}), the real matrix group with Lie algebra generated by (\ref{L}) becomes  a complex matrix group.\index{$SU(2)\to SO(3)$ ! irreducible representations |)}\index{Spin-j system ! standard representation |)} 
\end{remark}

%-------------------------------------------------------------------------------------- Section 3.4 -----------------------------------------------------------------------------------------------------------------

\section{The tensor product and the space of operators}

For a given spin-j quantum mechanical system $\mathcal{H}_{j}$, let us
identify the Hilbert space with the complex (n+1)-space $\mathbb{C}%
^{n+1},n=2j$, by the correspondence
\begin{equation}
\mathbf{e}_{k}=\left\vert j,j-k+1\right\rangle ,\text{ \ }k=1,2,...,n+1,
\label{standard2}%
\end{equation}
which identifies a given standard basis (\ref{u}) of $\mathcal{H}_{j}$ with
the usual standard basis of $\mathbb{C}^{n+1}$, namely the column matrices
\[
\mathbf{e}_{1}=(1,0,...,0)^{T},\mathbf{e}_{2}=(0,1,0,...,0)^{T},\text{ etc.}%
\]
For two systems $\mathcal{H}_{j_{1}}$,$\mathcal{H}_{j_{2}}$, the space of
linear operators $Hom(\mathcal{H}_{j_{2}},\mathcal{H}_{j_{1}})$ identifies
with the full matrix space $M_{\mathbb{C}}(n_{1}+1,n_{2}+1)$, linearly spanned
by the one-element matrices
\begin{equation}
\mathcal{E}_{k,l}=\mathbf{e}_{k}\mathbf{e}_{l}^{T}\text{ (matrix product),
\ }(\mathcal{E}_{kl})_{pq}=\delta_{kp}\delta_{lq}\text{, } \label{elem}%
\end{equation}
and there is the linear isometry\index{Spin-j system ! tensor product |(}
\begin{equation}
\mathbb{C}^{n_{1}+1}\otimes\mathbb{C}^{n_{2}+1}\rightarrow M_{\mathbb{C}%
}(n_{1}+1,n_{2}+1),\text{ \ \ }\mathbf{e}_{k}\otimes\mathbf{e}_{l}%
\rightarrow\mathcal{E}_{kl} \label{equiv}%
\end{equation}
where the matrix space has the (Hilbert-Schmidt) Hermitian inner product
\begin{equation}
\left\langle P,Q\right\rangle =trace(P^{\ast}Q)\text{ }=\text{\ }%
\operatorname{Re}\left\langle P,Q\right\rangle +i\operatorname{Im}\left\langle
P,Q\right\rangle \text{ } \label{hilb}%
\end{equation}
and $P^{\ast}=\overline{P}^{T}$ is the adjoint of $P$. The real part in
(\ref{hilb}) is a euclidean metric for the matrix space viewed as a real
vector space.

We are primarily interested in the case $n_{1}=n_{2}=n$, in which case the
matrix space, denoted by $M_{\mathbb{C}}(n+1)$, is also an algebra and has the
orthogonal decomposition
\begin{equation}
M_{\mathbb{C}}(n+1)=\mathcal{AS}ym(n+1)\oplus\mathcal{S}ym(n+1) \label{decomp2}%
\end{equation}
into skew-symmetric and symmetric matrices. Moreover, as a real vector space
there is the real orthogonal decomposition (w.r.t. the real part in
(\ref{hilb}))
\begin{equation}
M_{\mathbb{C}}(n+1)=\mathcal{U(}n+1)\otimes_{\mathbb{R}}\mathbb{C=}%
\mathcal{\ \mathcal{U}}(n+1)\oplus\mathcal{H}(n+1), \label{decomp3}%
\end{equation}
where $\mathcal{U(}n+1)$ is the Lie algebra of $U(n+1)$ consisting of
skew-Hermitian matrices and $\mathcal{H}(n+1)=i\,\mathcal{U(}n+1)$ is the
space of Hermitian matrices.

\subsection{SU(2)-invariant decomposition of the tensor
product\label{decomp tensor}}

Let $\mu_{n+1}$ be the standard representation of $U(n+1)$ on $\mathbb{C}%
^{n+1}$, and let $\check{\mu}_{n+1}$ be its dual with $g\in U(n+1)$ acting by
the complex conjugate matrix $\bar{g}$ on $\mathbb{C}^{n+1}$. Consider the
tensor product representations $\mu_{n_{1}+1}\otimes\mu_{n_{2}+1}$ and
$\mu_{n_{1}+1}\otimes\check{\mu}_{n_{2}+1}$ of $U(n_{1}+1)\times U(n_{2}+1)$
acting on $\mathbb{C}^{n_{1}+1}\otimes\mathbb{C}^{n_{2}+1}$. The matrix model\index{Spin-j system ! tensor product ! matrix models} 
of these representations follows from the isometry (\ref{equiv}), when the
group acts on matrices $P$ $\in M_{\mathbb{C}}(n_{1}+1,n_{2}+1)$ by matrix
multiplication, as follows :
\begin{align}
(i)\text{ \ }\mu_{n_{1}+1}\otimes\mu_{n_{2}+1}\  &  :(g,h)P\rightarrow
gPh^{T}\text{ }\label{tensorprod2}\\
(ii)\text{ \ }\mu_{n_{1}+1}\otimes\check{\mu}_{n_{2}+1}\  &
:(g,h)P\rightarrow gPh^{-1}\nonumber
\end{align}
Composing with irreducible representations $\varphi_{j_{i}}:SU(2)\rightarrow
U(n_{i}+1),$ yields the following tensor product representations of $SU(2)$
and its action on matrices:
\begin{align}
(i)\text{ }\varphi_{j_{1}}\otimes\varphi_{j_{2}}  &  :(g,P)\rightarrow
\varphi_{j_{1}}(g)P\varphi_{j_{2}}(g)^{T}\label{tensorprod3}\\
(ii)\text{ }\varphi_{j_{1}}\otimes\bar{\varphi}_{j_{2}}  &
:(g,P)\rightarrow\varphi_{j_{1}}(g)P\varphi_{j_{2}}(g)^{-1}\nonumber
\end{align}
However, since the $SU(2)$-representations $\bar{\varphi}_{j}$ and
$\varphi_{j}$ are equivalent for any $j$, so are the two tensor products and
their equivariant matrix models (\ref{tensorprod3}). Combining
(\ref{standard2}), (\ref{equiv}), and Definition \ref{dualbasis}, we are led
to the following :

\begin{definition}
\label{uncoupled1}For the two matrix models (\ref{tensorprod3}) of the tensor
product $\mathbb{C}^{n_{1}+1}\otimes\mathbb{C}^{n_{2}+1}$, the \emph{uncoupled
standard basis }is the following collection of one-element matrices (cf.
(\ref{dual2}))%
\begin{align}
\text{model (i)}  &  \text{:\ \ }\left\vert j_{1}m_{1}j_{2}m_{2}\right\rangle
=\mathbf{u}(j_{1},m_{1})\otimes\mathbf{u}(j_{2},m_{2})=\mathcal{E}%
_{j_{1}-m_{1}+1,j_{2}-m_{2}+1}\label{uncouple}\\
\text{model (ii)}  &  \text{: }\left\vert j_{1}m_{1}j_{2}m_{2}\right\rangle
=\mathbf{u}(j_{1},m_{1})\otimes\mathbf{\check{u}}(j_{2},m_{2})=(-1)^{j_{2}%
+m_{2}}\mathcal{E}_{j_{1}-m_{1}+1,j_{2}+m_{2}+1}\nonumber
\end{align}
where $-j_{i}\leq m_{i}\leq j_{i},$ and all $j_{i}-m_{i}$ are integers.$\ $
\end{definition}

At the infinitesimal level the angular momentum operators $J_{k},k=1,2,3,$ of
$SU(2)$ act on matrices $P$ in the two models by
\begin{equation}
(i)\text{ \ }J_{k}\cdot P=J_{k}^{(j_{1})}P+P(J_{k}^{(j_{2})})^{T}\text{,
\ \ }(ii)\text{ }J_{k}\cdot P=J_{k}^{(j_{1})}P-PJ_{k}^{(j_{2})}, \label{ad}%
\end{equation}
where the matrix $J_{k}^{(j)}\in M_{\mathbb{C}}(2j+1)$ represents $J_{k}$
acting on $\mathcal{H}_{j}=\mathbb{C}^{2j+1}$. In particular, for $1\leq k\leq
n_{1}+1,1\leq l\leq n_{2}+1,$
\begin{equation}
(i)\text{ \ }J_{3}\cdot\mathcal{E}_{kl}=(j_{1}+j_{2}+2-k-l)\mathcal{E}%
_{kl}\text{ , \ }(ii)\text{ \ }J_{3}\cdot\mathcal{E}_{kl}=(j_{1}%
-j_{2}+l-k)\mathcal{E}_{kl} \label{action}%
\end{equation}
and this tells us that the uncoupled basis of $M_{\mathbb{C}}(n_{1}%
+1,n_{2}+1)$ diagonalizes $J_{3}$, with the eigenvalues as shown in
(\ref{action}).

\begin{remark}\label{diagonal1} Thus, in model (i) the eigenspace of \emph{quantum magnetic number} $m=m_{1}+m_{2}$
consists of the ``anti-subdiagonal'' matrices\ spanned by matrices
$\mathcal{E}_{kl}$ with $k+l=(j_{1}+j_{2}+2-m)$, whereas in model (ii) the
eigenspace is the ``subdiagonal'' spanned by the matrices $\mathcal{E}_{kl}$
with $l-k=j_{2}-j_{1}+m$, so that in model (ii) $\mathcal{E}_{kl}$ is the actual $k$-th subdiagonal when $j_1=j_2$. 
\end{remark}

Next, let us decompose the tensor product (3.30) into irreducible summands,
by first calculating the weight system of the tensor product and then
determine its decomposition, using the fact that a representation $\phi $ is
uniquely determined by its weight system $\Omega (\phi )$. For convenience,
let us formally define 
\[
\{a_{1},a_{2},..,a_{m}\}\otimes
\{b_{1},b_{2},..,b_{n}\}=\{(a_{i}+b_{j});1\leq i\leq m,1\leq j\leq n\}
\]%
to be the ``tensor product'' of two multisets (cf. Remark \ref{multisetremark}). Observe
that the vector$\ \left\vert j_{1}m_{1}j_{2}m_{2}\right\rangle $ in
(\ref{uncouple}) is a weight vector of weight $2\theta(m_{1}+m_{2})$ in the
tensor product (\ref{tensorprod3}). Setting $\lambda =j_{1}+j_{2}-|j_{1}-j_{2}|+1$ and omitting (for
convenience) the factor $2\theta $ of the weights, then,  
 by
writing a union of multisets additively (again, see Remark \ref{multisetremark}), we have by (\ref{weight1})
\begin{align}
\Omega(\varphi_{j_{1}}\otimes\varphi_{j_{2}})  &  =\left\{  j_{1}%
,j_{1}-1,..,-j_{1}\right\}  \otimes\left\{  j_{2},j_{2}-1,..,-j_{2}\right\}
\nonumber\\
&  =%
%TCIMACRO{\dsum \limits_{k=1}^{\delta}}%
%BeginExpansion
{\displaystyle\sum\limits_{k=1}^{\lambda}}
%EndExpansion
k\left\{  \pm(j_{1}+j_{2}-k+1)\right\}  =%
%TCIMACRO{\dsum \limits_{j=|j_{1}-j_{2}|}^{j_{1}+j_{2}}}%
%BeginExpansion
{\displaystyle\sum\limits_{j=|j_{1}-j_{2}|}^{j_{1}+j_{2}}}
%EndExpansion
\left\{  j,j-1,..,-j\right\} \label{weightdecomp}\\
&  =%
%TCIMACRO{\dsum \limits_{j=|j_{1}-j_{2}|}^{j_{1}+j_{2}}}%
%BeginExpansion
{\displaystyle\sum\limits_{j=|j_{1}-j_{2}|}^{j_{1}+j_{2}}}
%EndExpansion
\Omega(\varphi_{j})\nonumber
\end{align}

\begin{remark}
\label{triangle}A neat way to describe the range of $j$ in the sum
(\ref{weightdecomp}) is to state $\delta(j_{1},j_{2},j)=1$. This is the
``triangle inequality'' condition, involving three nonnegative integral or
half-integral numbers . Namely,
\[
\delta(j_{1},j_{2},j_{3})=1 \ \iff
\]
\begin{equation}
\label{triangleineq}(i)\text{ }|j_{1}-j_{2}|\text{ }\leq j_{3}\leq j_{1}%
+j_{2}\text{ \ and \ }(ii)\text{ }j_{1}+j_{2}+j_{3}\in\mathbb{Z},
\end{equation}
and $\delta(j_{1,}j_{2},j_{3})=0$ otherwise. We must note that the condition
is symmetric, that is, independent of the order of the numbers.
\end{remark}

It follows from (\ref{weightdecomp}) that%
\begin{equation}
\varphi_{j_{1}}\otimes\varphi_{j_{2}}=%
%TCIMACRO{\dsum \limits_{\Delta(j_{1},j_{2},j)=1}}%
%BeginExpansion
{\displaystyle\sum\limits_{\delta(j_{1},j_{2},j)=1}}
%EndExpansion
\varphi_{j} \label{split3}%
\end{equation}
and each of the summands $\varphi_{j}$ in (\ref{split3}) has its own standard
basis, denoted by
\begin{equation}
\left\vert (j_{1}j_{2})jm\right\rangle ,m=j,j-1,...,-j+1,-j \label{uncoupled2}%
\end{equation}
in the literature, and \textit{this basis is unique up to a phase factor for
each} $j$. By (\ref{uncouple}) these vectors are identified with specific
matrices in $M_{\mathbb{C}}(n_{1}+1,n_{2}+1)$, and as pointed out this can be
done naturally in two different ways depending on the choice of matrix model.
In any case, there is the orthogonal decomposition%
\[
M_{\mathbb{C}}(n_{1}+1,n_{2}+1)=%
%TCIMACRO{\dsum \limits_{j=|j_{1}-j_{2}|}^{j_{1}+j_{2}}}%
%BeginExpansion
{\displaystyle\sum\limits_{j=|j_{1}-j_{2}|}^{j_{1}+j_{2}}}
%EndExpansion
M_{\mathbb{C}}(\varphi_{j})
\]
where $M_{\mathbb{C}}(\varphi_{j})$ has the standard orthonormal basis
(\ref{uncoupled2}), for each $j$. The totality of these vectors (or matrices)
constitute the \emph{coupled standard basis} of the tensor product (or matrix space).\index{Spin-j system ! tensor product |)}

\subsubsection{Clebsch-Gordan coefficients}

To describe the connection between the coupled and uncoupled basis the
following definition is crucial.\index{Spin-j system ! Clebsch-Gordan coefficients |(}\index{Clebsch-Gordan coefficients |(}

\begin{definition}
\label{C-G}The \emph{Clebsch-Gordan }coefficients are, by definition, the
entries of the unitary transition matrix relating the uncoupled and coupled
standard basis, namely the inner products
\begin{equation}
C_{m_{1},m_{2},m}^{\text{ }j_{1},\text{ }j_{2},\text{ }j}=\ \left\langle
(j_{1}j_{2})jm|j_{1}m_{1}j_{2}m_{2}\right\rangle \ \label{Clebsch1}%
\end{equation}
which are the coefficients in the expansion
\begin{equation}
\left\vert j_{1}m_{1}j_{2}m_{2}\right\rangle =\sum_{j=|j_{1}-j_{2}|}%
^{j_{1}+j_{2}}\sum_{m=-j}^{j}C_{m_{1},m_{2},m}^{\text{ }j_{1},\text{ }%
j_{2},\text{ }j}\text{ }\left\vert (j_{1}j_{2})jm\right\rangle ,
\label{Clebsch2}%
\end{equation}
\end{definition}

\begin{remark}\index{Clebsch-Gordan coefficients ! nonvanishing conditions}
\label{non-vanish}Clebsch-Gordan coefficients, also called Wigner
coefficients,\index{Wigner coefficients | see {Clebsch-Gordan coefficients} } have been extensively studied in the physics literature; we
refer to \cite{BL, Condon, VMK} for surveys of their properties. First, they
satisfy the following non-vanishing conditions:%
\begin{equation}
\text{ }C_{m_{1},m_{2},m}^{\text{ }j_{1},\text{ }j_{2},\text{ }j}
\neq0\Longrightarrow\left\{
\begin{array}
[c]{c}%
m=m_{1}+m_{2}\\
\delta(j_{1},j_{2},j)=1
\end{array}
\right.  \label{C-nonvanish}%
\end{equation}

Also, they are uniquely determined once a phase connvention for the coupled
basis is chosen, and on the other hand, such a convention follows by choosing
the phase of some of the coefficients. We shall follow the generally accepted
convention (cf. e.g. \cite{BL, Condon, VMK})
\begin{equation}
C_{j_{1},j-j_{1},j}^{\text{ }j_{1},\text{ }j_{2},j}>0\text{ \ whenever }%
\delta(j_{1},j_{2},j)=1\text{ },\text{ } \label{conv1}%
\end{equation}
which, in fact, also implies that \underline{\emph{all the coefficients are
real}}. Then, it follows from their definition that they satisfy the following
orthogonality equations:
\begin{equation}\index{Clebsch-Gordan coefficients ! orthogonality equations} 
\label{orthogCG}\displaystyle{\sum_{m_{1},m_{2}}}C_{m_{1},m_{2},m}^{\text{
}j_{1},\text{ }j_{2},\text{ }j}C_{m_{1},m_{2},m^{\prime}}^{\text{ }%
j_{1},\text{ }j_{2},\text{ }j^{\prime}} = \delta_{j,j^{\prime}}\delta
_{m,m^{\prime}} \ , \ \displaystyle{\sum_{j,m}}C_{m_{1},m_{2},m}^{\text{
}j_{1},\text{ }j_{2},\text{ }j}C_{m_{1}^{\prime},m_{2}^{\prime},m}^{\text{
}j_{1},\text{ }j_{2},\text{ }j} = \delta_{m_{1},m_{1}^{\prime}}\delta
_{m_{2},m_{2}^{\prime}}%
\end{equation}
\end{remark}

Consequently, the unitary transition matrix in the above definition is
orthogonal, so the inversion of (\ref{Clebsch2}) is the formula
\begin{equation}
\left\vert (j_{1}j_{2})jm\right\rangle =%
%TCIMACRO{\dsum \limits_{m_{1}+m_{2}=m}}%
%BeginExpansion
{\displaystyle\sum\limits_{(m_{1}+m_{2}=m)}}
%EndExpansion
C_{m_{1},m_{2},m}^{\text{ }j_{1},\text{ }j_{2},\text{\ }j}\text{ }\left\vert
j_{1}m_{1}j_{2}m_{2}\right\rangle \label{Invers}%
\end{equation}
In particular, when $m=j_{1}+j_{2}$ there is only one term in the expansion
(\ref{Clebsch2}), so $C_{j_{1},j_{2},,j_{1}+j_{2}}^{\text{ }j_{1},\text{
}j_{2},j_{1}+j_{2}}=1$. Moreover, when $j_{1}=0$ or $j_{2}=0$, there is no
reason to distinguish between $\varphi_{j},\varphi_{0}\otimes\varphi_{j}$, and
$\varphi_{j}\otimes\varphi_{0}$, and so $C_{m,0,m}^{j,0,j}$ $=C_{0,m,m}%
^{0,j,j}=1$.

\ 

On the other hand, there is a connection between Clebsch-Gordan coefficients
and the Wigner D-functions defined by (\ref{Wigner-D}), given by the following
\textit{coupling rule}:\index{Spin-j system ! Wigner D-functions }\index{Wigner D-functions ! coupling rule}

\begin{proposition}
\label{WigD-CG} For a fixed $g\in SU(2)$,
\begin{equation}
D_{\mu_{1},m_{1}}^{j_{1}}D_{\mu_{2},m_{2}}^{j_{2}}=%
%TCIMACRO{\dsum \limits_{j}}%
%BeginExpansion
{\displaystyle\sum\limits_{j}}
%EndExpansion
C_{\mu_{1},\mu_{2},\mu_{1}+\mu_{2}}^{j_{1},j_{2},j}C_{m_{1},m_{2},m_{1}+m_{2}%
}^{j_{1},j_{2},j}D_{\mu_{1}+\mu_{2},m_{1}+m_{2}}^{j} \label{couple1}%
\end{equation}
whose inversion formula reads
\begin{equation}
D_{\mu m}^{j}=%
%TCIMACRO{\dsum \limits_{\mu_{1}}}%
%BeginExpansion
{\displaystyle\sum\limits_{\mu_{1}}}
%EndExpansion
\sum_{m_{1}}C_{m_{1},m_{2},m}^{j_{1,}j_{2},j}C_{\mu_{1},\mu_{2,}\mu}%
^{j_{1},j_{2},j}D_{\mu_{1},m_{1}}^{j_{1}}D_{\mu_{2},m_{2}}^{j_{2}} \ .
\label{couple2}%
\end{equation}
\end{proposition}

We refer to Appendix \ref{ProofWDCG} for a proof of Proposition \ref{WigD-CG}

\begin{remark}
\label{Wignercomputation} The above proposition shows that the Wigner
D-functions are essentially determined by the Clebsch-Gordan coefficients and vice-versa.

Thus, starting from the trivially available four functions $\{D_{kl}^{1/2}\}$
one can use the formula (\ref{couple2}) to determine successively the
functions $D_{\mu m}^{j}$ for all $j$.

Conversely, starting from the expression (\ref{Jn}) and formulas
(\ref{lower})-(\ref{rel}) and (\ref{WigD})-(\ref{Wigner-D}) which produce
formulas for the $D_{\mu m}^{j}$, one can use formula (\ref{couple2}) to
compute all Clebsch-Gordan coefficients explicitly. This is the way these
coefficients were first explicitly computed, by Wigner in 1927 \cite{Wig1, Wig11}.
\end{remark}

In fact, iterating the recursive equation (\ref{basis}) below, properly
generalized to $|(j_{1}j_{2})jm\rangle$, gives another way to obtain explicit
expressions for all Clebsch-Gordan coefficients, as explored by Racah. Thus,
there are various equivalent explicit expressions for the Clebsch-Gordan
coefficients, see for instance \cite{BL, VMK}. Here we list a rather symmetric
one, first obtained by van der Waerden in 1932:\index{Clebsch-Gordan coefficients ! explicit formulae}
\begin{align}
&  C_{m_{1},m_{2},m_{3}}^{j_{1},j_{2},j_{3}}\ =\ \delta_{m_3,m_1+m_2}\sqrt{2j+1}\ \Delta(j_{1},j_{2},j_{3})\ S_{m_{1},m_{2},m_{3}}%
^{\ j_{1},\ \ j_{2},\ \ j_{3}}\label{explicitCG2}\\
&  \cdot{\sum_{z}}\frac{(-1)^{z}}{z!(j_{1}+j_{2}-j_{3}-z)!(j_{1}%
-m_{1}-z)!(j_{2}+m_{2}-z)!(j_{3}-j_{2}+m_{1}+z)!(j_{3}-j_{1}-m_{2}%
+z)!}\nonumber
\end{align}
where by definition
\begin{equation}
\Delta(j_{1},j_{2},j_{3})=\sqrt{\frac{(j_{1}+j_{2}-j_{3})!(j_{3}+j_{1}%
-j_{2})!(j_{2}+j_{3}-j_{1})!}{(j_{1}+j_{2}+j_{3}+1)!}}\ \ , \label{DeltaW}%
\end{equation}%
\begin{equation}
S_{m_{1},m_{2},m_{3}}^{\ j_{1},\ j_{2},\ j_{3}}=\sqrt{(j_{1}+m_{1}%
)!(j_{1}-m_{1})!(j_{2}+m_{2})!(j_{2}-m_{2})!(j_{3}+m_{3})!(j_{3}-m_{3}%
)!}\ \ \label{Sjjj}%
\end{equation}

\begin{remark}
\label{summation} We remind that the Clebsch-Gordan coefficients are nonzero and satisfy equation (\ref{explicitCG2}) above only if the conditions (\ref{C-nonvanish}) are satisfied. Also, in the sum ${\sum_{z}}$ of formula (\ref{explicitCG2}), the
summation index $z$ is asumed to take all integral values for which all
factorial arguments are nonnegative, with the usual convention $0!=1$. In the
sequel we shall also encounter similar summations, and the same convention on
the summation index is tacitly assumed unless otherwise stated.
\end{remark}

By inspection of this and other equivalent formulae, one obtains the symmetry
properties for the Clebsch-Gordan coefficients. Here we list some of these:\index{Clebsch-Gordan coefficients ! symmetry properties}
\begin{equation}
C_{m_{1},m_{2},m_{3}}^{j_{1},j_{2},j_{3}}\ =\ (-1)^{j_{1}+j_{2}-j_{3}%
}C_{-m_{1},-m_{2},-m_{3}}^{j_{1},j_{2},j_{3}}\ =\ (-1)^{j_{1}+j_{2}-j_{3}%
}C_{m_{2},m_{1},m_{3}}^{j_{2},j_{1},j_{3}}\ \label{symCG1}%
\end{equation}%
\begin{equation}
C_{m_{1},m_{2},m_{3}}^{j_{1},j_{2},j_{3}}\ =\ (-1)^{j_{2}+m_{2}}\sqrt
{\frac{2j_{3}+1}{2j_{1}+1}}C_{-m_{3},m_{2},-m_{1}}^{j_{3},j_{2},j_{1}%
}\ \label{symCG2}%
\end{equation}%
\begin{equation}
C_{m_{1},m_{2},m_{3}}^{j_{1},j_{2},j_{3}}\ =\ (-1)^{j_{1}-m_{1}}\sqrt
{\frac{2j_{3}+1}{2j_{2}+1}}C_{m_{1},-m_{3},-m_{2}}^{j_{1},j_{3},j_{2}%
}\ \label{symCG3}%
\end{equation}
\index{Spin-j system ! Clebsch-Gordan coefficients |)}\index{Clebsch-Gordan coefficients |)}

%-------------------------------------------------------------------------------------- Section 3.5 -----------------------------------------------------------------------------------------------------------------

\section{SO(3)-invariant decomposition of the operator algebra}\index{Spin-j system ! operator algebra invariant decomposition |(}

We shall further investigate the special case $j_{1}=j_{2}=j$ and resume the
terminology from the previous section. In particular, the matrix algebra
$M_{\mathbb{C}}(n+1)$ represents the space of linear operators on
$\mathcal{H}_{j}=\mathbb{C}^{n+1}$, on which the unitary group $U(n+1)$ acts
by two different (inner) tensor product representations\index{Spin-j system ! tensor product ! matrix models}\index{Spin-j system ! tensor product |(} 
\begin{align}
(i)\text{ \ }\mu_{n+1}\otimes\mu_{n+1}  &  \simeq\Lambda^{2}\mu_{n+1}+S^{2}%
\mu_{n+1}:(g,P)\rightarrow gPg^{T}\text{ }\label{split}\\
(ii)\text{ \ }\mu_{n+1}\otimes\check{\mu}_{n+1}  &  \simeq Ad_{U(n+1)}%
^{\mathbb{C}}=_{\mathbb{R}}2Ad_{U(n+1)}:(g,P)\rightarrow gPg^{-1}\nonumber
\end{align}
The splitting in the two cases corresponds to the $U(n+1)$-invariant
decompositions (\ref{decomp2}) and (\ref{decomp3}), respectively. In case (ii)
the splitting is over $\mathbb{R}$ and $U(n+1)$ acts by its (real) adjoint
representation $Ad_{U(n+1)}$ on both $\mathcal{\mathcal{U}}(n+1)$ and
$\mathcal{H}(n+1)$.

Composition of the above representations with the irreducible representation
$\varphi_{j}:SU(2)\rightarrow U(n+1)$ yields the following two equivalent
representations
\begin{align}
(i)\text{ \ }\varphi_{j}\otimes\varphi_{j}  &  \simeq(\Lambda^{2}\mu
_{n+1}+S^{2}\mu_{n+1})\circ\varphi_{j}:(g,P)\rightarrow\varphi_{j}%
(g)P\varphi_{j}(g)^{T}\label{split2}\\
(ii)\text{ \ }\varphi_{j}\otimes\bar{\varphi}_{j}  &  \simeq Ad_{U(n+1)}%
^{\mathbb{C}}\circ\varphi_{j}:(g,P)\rightarrow\varphi_{j}(g)P\varphi
_{j}(g)^{-1}\nonumber
\end{align}
and according to (\ref{split3}) this representation splits into an integral
string of irreducibles
\begin{equation}
\varphi_{j}\otimes_{\mathbb{C}}\varphi_{j}=\varphi_{0}+\varphi_{1}%
+....+\varphi_{n}\text{, \ }n=2j \label{split4}%
\end{equation}
Let us denote the corresponding decomposition of the matrix space as follows
\begin{equation}
M_{\mathbb{C}}(n+1)=\sum_{l=0}^{n}M_{\mathbb{C}}(\varphi_{l}), \label{sum}%
\end{equation}
where the summands consist of either symmetric or skew-symmetric matrices,
depending on the parity of $l$ and according to the splitting
\begin{align*}
S^{2}\mu_{n+1}|SU(2)  &  =\varphi_{n}+\varphi_{n-2}+\varphi_{n-4}+..\\
\Lambda^{2}\mu_{n+1}|SU(2)  &  =\varphi_{n-1}+\varphi_{n-3}+\varphi_{n-5}+...
\end{align*}
The above tensor product (\ref{split4}) is, in fact, a representation of
$SO(3)=SU(2)/\mathbb{Z}_{2}$ and hence it has a real form. Such a real form
can be embedded in $M_{\mathbb{C}}(n+1)$ in different ways; for example as the
space of real matrices
\begin{equation}
M_{\mathbb{R}}(n+1)=\sum_{l=0}^{n}M_{\mathbb{R}}(\psi_{l}), \label{sum1}%
\end{equation}
where the irreducible summands consist of either symmetric or skew-symmetric
matrices, depending on the parity of $l$ as in (\ref{sum}), cf. (\ref{irrep}), (\ref{realrep}).

\begin{definition}\label{defmodelused}
In order to agree with the standard framework in quantum mechanics\index{Spin-j system ! tensor product ! quantum matrix model}  (see Remark
\ref{Heisenbdynspinj}), we henceforth stick to the \emph{\textbf{matrix model
(ii)}} in (\ref{split2}) and therefore $SO(3)$ acts via the adjoint action of
$U(n+1)$ on $M_{\mathbb{C}}(n+1)$.
\end{definition}

Thus, the above representation of $SO(3)$ on the matrix space (\ref{sum})
splits into real invariant subspaces
\begin{eqnarray}
\mathcal{\mathcal{U}}(n+1)=\sum_{l=0}^{n}\mathcal{U}(\psi_{l}) & , & \mathcal{H}(n+1)=\sum_{l=0}^{n}\mathcal{H}(\psi_{l}) \ , \label{summands} \\
%\end{equation}%
%\[
\mathcal{H}(\psi_{l})=\mathcal{H}(n+1)\cap M_{\mathbb{C}}(\varphi_{l}) & , &  \mathcal{U}(\psi_{l})=\mathcal{\mathcal{U}}(n+1)\cap M_{\mathbb{C}}%
(\varphi_{l}).\text{\ } \nonumber
\end{eqnarray}\index{Spin-j system ! tensor product |)}

At the infinitesimal level the angular momentum operators $J_{k}$, represented
as matrices in $\mathcal{H}(n+1)$, act on $M_{\mathbb{C}}(n+1)$ via the
commutator product
\begin{align*}
\text{ }J_{k}\cdot P  &  =ad_{J_{k}}(P)=[J_{k},P]=J_{k}P-PJ_{k},\text{
\ }k=1,2,3\\
J_{3}\cdot\mathcal{E}_{kl}  &  =\left[  J_{3},\mathcal{E}_{kl}\right]
=(l-k)\mathcal{E}_{kl},\text{ \ }1\leq k,l\leq n+1
\end{align*}
For example, the summand $M_{\mathbb{C}}(\varphi_{0})$ (resp. $M_{\mathbb{C}%
}(\varphi_{1})$) is linearly spanned by the identity matrix $I$ (resp. the
matrices $J_{k})$.

Let us introduce the $J_{3}$-eigenspace decompositions
\begin{equation}
M_{\mathbb{C}}(n+1)=\sum_{m=-n}^{n}\Delta_{\mathbb{C}}(m\mathbb{)}\text{,
\ \ }M_{\mathbb{R}}(n+1)=\sum_{m=-n}^{n}\Delta_{\mathbb{R}}(m\mathbb{)}%
\text{\ } \label{Wm}%
\end{equation}
where $\Delta(m)$ consists of \ the $m$\emph{-subdiagonal}\ matrices, spanned
by the one-element matrices $\mathcal{E}_{k,l}$ with $l-k=m$ (cf.
(\ref{elem}) and Remark \ref{diagonal1}). Clearly
\[
\dim_{\Bbbk}\Delta_{\Bbbk}(m)=n+1-\left\vert m\right\vert .
\]
In particular, the zero weight space $\Delta(0)$ consists of the main diagonal
matrices, and $\Delta(n)$ (resp. $\Delta(-n)$) is spanned by the one-element
matrix $\mathcal{E}_{1,n+1}$ (resp. $\mathcal{E}_{n+1,1}$) with its non-zero
entry positioned at the upper right (resp. lower left) corner. It is sometimes
convenient to denote an m-subdiagonal matrix $P=(P_{ij})$ with m-subdiagonal
entries $x_{i}$ as a coordinate vector
\begin{equation}
P=(x_{1},x_{2},..,x_{k})_{m}\text{, }k=n+1-\left\vert m\right\vert
\label{m-sub}%
\end{equation}

\subsection{The irreducible summands of the operator algebra}

We shall further investigate how the irreducible summands $M_{\mathbb{C}%
}(\varphi_{l})$ (resp. $\mathcal{U}(\psi_{l}))$ and $\mathcal{H}(\psi_{l})$)
are embedded in the operator (or matrix) algebra $M_{\mathbb{C}}(n+1).$ It
suffices to consider the Hermitian operators since $\mathcal{U}(\psi
_{l}))=i\,\mathcal{H}(\psi_{l})$ and
\[
M_{\mathbb{C}}(\varphi_{l})=\mathcal{H}(\psi_{l})+i\mathcal{H}(\psi_{l})
\]
To this end,$\ $consider the subspace
\[
\mathcal{H}(n+1)_{l}\subset\mathcal{H}(n+1)\
\]
of Hermitian operators formally expressible as real homogeneous polynomials
$P,Q..$ of degree $l$ in the non-commuting \textquotedblleft
variables\textquotedblright\ $J_{k}$. As generators of the Lie algebra
$\mathcal{SO}(3)$ the operators $L_{k}=-iJ_{k}$ act as derivations on
polynomials,
\[
ad_{L_{k}}(PQ)=ad_{L_{k}}(P)Q+Pad_{L_{k}}(Q);\text{ \ \ }ad_{L_{a}}%
(J_{b})=\varepsilon_{abc}J_{c}%
\]
and this action preserves the degree of a polynomial, leaving $\mathcal{H}%
(n+1)_{l}$ invariant.

The non-commutativity of the operators $J_{k}$ can be handled by considering
ordered 3-partitions of $l$%
\[
\pi=(l_{1},l_{2},l_{3}),\text{ }l_{i}\geq0,\sum l_{i}=l\text{.}%
\]
For each such partition $\pi$ there is a symmetric polynomial expression in
the symbols $J_{i}$
\begin{equation}
P_{\pi}=J_{1}^{l_{1}}J_{2}^{l_{2}}J_{3}^{l_{3}}+...+ \label{symm}%
\end{equation}
obtained from the leading monomial by symmetrization, as indicated in
(\ref{symm}), that is, $P_{\pi}$ is the sum of all monomials with the same
total degree $l_{k}$ for each $J_{k}$. For example, associated with
$\pi=(1,2,0)$ is the polynomial $J_{1}J_{2}^{2}+J_{2}J_{1}J_{2}+J_{2}^{2}%
J_{1}$. The monomials in (\ref{symm}) are actually equal as operators, modulo
a polynomial of lower degree, due to the basic commutation relation
(\ref{comm2}). Also, for $l=0$ the only operator of type (\ref{symm}) is
$P=Id$.

Let us denote by $\delta_{l}$ the number of partitions $\pi$ of the above
kind. We claim that for $l\leq n$ the operators $P_{\pi}$ constitute a basis
for $\mathcal{H}(n+1)_{l}$, and consequently
\[
\dim\mathcal{H}(n+1)_{l}=\delta_{l}=\frac{1}{2}(l+1)(l+2)\mathcal{\ }%
\]
The crucial reason is that the (Casimir) operator $J^{2}=J_{1}^{2}+J_{2}%
^{2}+J_{3}^{2}$, which by (\ref{J2sum}) acts as a multiple of the identity,
yields by multiplication an $\mathcal{SO}(3)$-invariant imbedding
\[
\mathcal{H}(n+1)_{l-2}\rightarrow J^{2}\mathcal{H}(n+1)_{l-2}\subset
\mathcal{H}(n+1)_{l}%
\]
and there is a complementary and invariant subspace $\mathcal{V}_{l}$
$\subset\mathcal{H}(n+1)_{l}$, namely $\mathcal{V}_{l}=\mathcal{H}(\psi_{l})$,
of dimension
\[
\dim\mathcal{V}_{l}=\delta_{l}-\delta_{l-2}=2l+1
\]
Thus, we have
\[
\sum_{l=0}^{n}\dim\mathcal{V}_{l}=\sum_{l=0}^{n}(2l+1)=(n+1)^{2}%
=\dim\mathcal{H}(n+1)
\]
and in fact,
\[
\mathcal{H}(n+1)=\sum_{l=0}^{n}\mathcal{V}_{l}=\sum_{l=0}^{n}\mathcal{H}%
(\psi_{l})
\]
is the splitting (\ref{summands}). In particular, each monomial of degree
$n+1$ is actually an operator expressible as a linear combination of monomials
of degree $\leq n.$

\begin{example}
For fixed $n>1$ the operators
\begin{align*}
A_{1}  &  =J_{2}J_{3}+J_{3}J_{2},\text{ }A_{2}=J_{3}J_{1}+J_{1}J_{3},\text{
}A_{3}=J_{1}J_{2}+J_{2}J_{1}\\
B_{1}  &  =J_{2}^{2}-J_{3}^{2},\text{ }B_{2}=J_{3}^{2}-J_{1}^{2},\text{ }%
B_{3}=J_{1}^{2}-J_{2}^{2}\text{ \ \ \ \ }(B_{1}+B_{2}+B_{3}=0)
\end{align*}
in the 6-dimensional space $\mathcal{H}(n+1)_{2}$ span a 5-dimensional
$\mathcal{SO}(3)$-invariant subspace $\mathcal{V}_{2}$, whose orthogonal
complement is a trivial summand, namely the line spanned by the operator
$J^{2}$.
\end{example}

\subsection{The coupled standard basis of the operator algebra}\index{Spin-j system ! coupled standard basis |(}
\label{invmatrixbasis}

The last example does not provide any clue to the calculation of standard
bases for the irreducible summands $M_{\mathbb{C}}(\varphi_{l})$ of
$M_{\mathbb{C}}(n+1)$. What we seek is a collection of matrices \emph{\ }
\begin{equation}
\mathbf{e}^{j}(l,m)=\left\vert (jj)lm\right\rangle ,\text{ }-l\leq m\leq
l,\text{ }0\leq l\leq n=2j \label{couple}%
\end{equation}
such that for each $l$ the matrices $\mathbf{e}^{j}(l,m)$ constitute a
standard basis for $M_{\mathbb{C}}(\varphi_{l})$. Moreover, with the phase
convention (\ref{conv1}), namely the positivity condition
\begin{equation}
C_{j,l-j,l}^{\text{ }j,\text{ }j,l}>0\text{ \ for each }l \ , \label{phase2}%
\end{equation}
the basis will be uniquely determined.

\begin{remark}
In what follows, we sometimes drop the superscript $j$ for the coupled basis
vectors $\mathbf{e}^{j}(l,m)$, when there is no ambiguity about the total spin
$j$.
\end{remark}

Now, we shall construct the coupled standard basis (\ref{couple}), expressed
in terms of the matrices $J_{\pm}$, see (\ref{Jn}). Consider the repeated
product of $J_{+}$ with itself
\[
(J_{+})^{l}\in\Delta_{\mathbb{R}}(l)\subset M_{\mathbb{R}}(n+1)\text{, \ }%
\]
noting that the product vanishes for $l=n+1$, and by definition, $(J_{\pm
})^{0}=I$. For a fixed $n$, the norms of the above matrices yield a sequence
of positive integers depending on $n$ $\ $
\begin{equation}
\mu_{l}^{n}=\left\Vert (J_{+})^{l}\right\Vert =\sqrt{trace((J_{-})^{l}%
(J_{+})^{l})}\text{, \ }0\leq l\leq n, \label{my1}%
\end{equation}
\ where $\ $
\begin{equation}
\mu_{0}^{n}=\sqrt{n+1}\text{, \ }\mu_{n}^{n}=n!, \label{my4}%
\end{equation}
and moreover, for $l\geq0$ there is the general formula
\begin{equation}
(\mu_{l}^{n})^{2}=\frac{(l!)^{2}}{(2l+1)!}(n-l+1)(n-l+2)\cdot\cdot\cdot
n(n+1)(n+2)\cdot\cdot\cdot(n+l+1) \label{my5}%
\end{equation}
As a function of $n$ this is, in fact, a polynomial of degree $2l+1$.

\begin{definition}
\label{normalized}For fixed $l\geq0,$ define recursively a string of real
matrices of norm 1
\begin{equation}
\mathbf{e}^{j}(l,m)\in\Delta_{\mathbb{R}}(m)\text{, \ \ }-l\leq m\leq
l\ \label{l-basis}%
\end{equation}
by setting
\begin{equation}
\mathbf{e}^{j}(l,l)=\frac{(-1)^{l}}{\mu^{n}_{l}}(J_{+})^{l},\ \text{\ \ }%
\mathbf{e}^{j}(l,m-1)=\frac{1}{\beta_{l,m}}\left[  J_{-},\mathbf{e}%
^{j}(l,m)\right]  , \label{basis}%
\end{equation}
where $\beta_{l,m}$ is the number defined in (\ref{c-d}).
\end{definition}

\begin{proposition}
\label{standard basis}The above family of real matrices $\mathbf{e}^j(l,m)$
constitute a coupled standard orthonormal basis
\[
\left\vert (jj)lm\right\rangle =\mathbf{e}^{j}(l,m);\text{ \ \ }0\leq l\leq
n=2j \ , \ -l\leq m\leq l \ ,
\]
for $M_{\mathbb{C}}(n+1)$ in agreement with the phase convention
(\ref{phase2})
\begin{equation}
C_{j,l-j,l}^{\text{ }j,\text{ }j,l}=(-1)^{l}\left\langle \mathbf{e}^j(l,l),\mathcal{E}_{1,l+1}\right\rangle =(-1)^{l}\mathbf{e}^j(l,l)_{1,l+1}>0.
\label{conv}%
\end{equation}
For fixed $l$, the family of vectors $\mathbf{e}^j(l,m)$ is a
standard orthonormal basis for the irreducible tensor summand $M_{\mathbb{C}%
}(\varphi_{l})$, and as matrices satisfy the relations
\begin{equation}
\mathbf{e}^j(l,-m)=(-1)^{m}\mathbf{e}^j(l,m)^{T},\text{ \ }-l\leq m\leq l,
\label{sym}%
\end{equation}
and in particular,%
\[
\mathbf{e}^{j}(l,-l)=\frac{1}{\mu^{n}_{l}}(J_{-})^{l}%
\]
\end{proposition}

\begin{proof}
For a fixed $l$, if we have a unit vector $\mathbf{u}(l,l)\in M_{\mathbb{C}%
}(\varphi_{l})$ with maximal $J_{3}$-eigenvalue $m=l$ and thus belonging to
the $l$-subdiagonal $\Delta_{\mathbb{C}}(l)$, then successive application of
the lowering operator $J_{-}$ and normalization of the vectors will generate
an orthonormal basis $\left\{  \mathbf{u}(l,m)\right\}  $ for $M_{\mathbb{C}%
}(\varphi_{l})$ which, by Definition \ref{standard}, is a standard basis. The
\textquotedblleft higher level\textquotedblright\ vectors $\mathbf{u}%
(l+1,l),...,\mathbf{u}(n,l)$ span a hyperplane of $\Delta_{\mathbb{C}}(l)$, so
knowledge of these vectors would uniquely determine $\mathbf{u}(l,l)$ up to a
choice of phase.

We claim that one can take $\mathbf{u}(l,m)$ to be the above matrix
$\mathbf{e}(l,m)$ for all $(l,m)$. The point is that $\ $
\begin{equation}
\Delta_{\mathbb{R}}(l)=lin\left\{  \mathbf{e}(l,l),\mathbf{e}%
(l+1,l),...,\mathbf{e}(n,l)\right\}  \label{Wbas}%
\end{equation}
where the listed vectors, indeed, constitute an orthonormal basis. To see
this, we may assume inductively that the \textquotedblright higher
level\textquotedblright\ vectors $\mathbf{e}(l+k,l)$ in (\ref{Wbas}) are
already known to be perpendicular, and then it remains to check that
$\mathbf{e}(l,l)$ is perpendicular to all $\mathbf{e}(l+k,l)$, $k>0$. However,
their inner product with $\mathbf{e}(l,l)$ is (modulo a factor $\neq0)$
\[
trace((J_{+}^{\ l})^{T}ad(J_{-}^{\ })^{k}(J_{+}^{l+k}))=trace(J_{-}^{l}\left[
J_{-},...\left[  J_{-},J_{+}^{l+k}\right]  ...\right]  )=0,
\]
by successive usage of the rule $trace(XY)=trace(YX)$.

Observe that the real matrix $(J_{+})^{l}$ is $l$-subdiagonal\ and with
positive entries. In particular, its inner product with $\mathcal{E}_{1,l+1}$
is positive. On the other hand, by model (ii) in (\ref{uncouple})
\[
\ e(j,l-j)=(-1)^{l}\mathcal{E}_{1,l+1}%
\]
and thus the factor $(-1)^{l}$ in (\ref{basis}) is needed because of the sign
convention (\ref{conv}).

Finally, the identity (\ref{sym}) can be seen from symmetry considerations
using that $J_{-}$ is the transpose of $J_{+}$, but with due regard to the
sign convention.
\end{proof}

Now, we shall obtain an explicit general expression for the coupled basis
vectors $\mathbf{e}(l,m)$. But first, let us look at the unnormalized matrices
$E(l,m)$, namely
\begin{equation}
E(l,m)=(-1)^{l}\mu_{l,m}^{n}\mathbf{e}(l,m),\ \text{ }\mu_{l,m}^{n}=\left\Vert
E(l,m)\right\Vert \text{\ }. \label{scale}%
\end{equation}
They are constructed recursively, as follows. For $0\leq l\leq n,-l\leq m\leq
l$, define
\begin{align}
E(l,l)  &  =J_{+}^{l},\text{ \ }E(0,0)=Id\label{5a}\\
E(l,m-1)  &  =[J_{-},E(l,m)]=ad_{J_{-}}(E(l,m))\nonumber
\end{align}
and hence there is the general formula%
\begin{equation}
E(l,m)=(ad_{J_{-}})^{l-m}(J_{+}^{l})=%
%TCIMACRO{\dsum \limits_{i=0}^{l-m}}%
%BeginExpansion
{\displaystyle\sum\limits_{k=0}^{l-m}}
%EndExpansion
(-1)^{k}\binom{l-m}{k}J_{-}^{l-m-k}J_{+}^{l}J_{-}^{k} \label{5}%
\end{equation}

It remains to determine the norm $\mu_{l,m}^{n}$ of $E(l,m)$, see
(\ref{scale}). First, we remark there is the following identity
\begin{equation}
\lbrack J_{+},E(l,m)]=\alpha_{l,m}^{2}E(l,m+1)\text{\ , cf. }(\ref{c-d})
\label{5c}%
\end{equation}
and next, for fixed $l$ define positive numbers $p_{l,m}$ recursively by
\[
p_{l,0}=1,\text{ }p_{l,m+1}=\alpha_{l,m}^{2}p_{l,m}\text{, \ }0<m\leq
l\text{.}%
\]
Then, the matrices $E(l,m)$ and $E(l,-m)$ are related via transposition by
\begin{equation}
E(l,-m)=(-1)^{m}p_{l,m}E(l,m)^{T},\text{ \ }m\geq0. \label{5e}%
\end{equation}
Finally, it follows from (\ref{5a}), (\ref{my1}), and (\ref{c-d})%
\[
\ \text{\ }\mu_{l,l}^{n}=\mu_{l}^{n}\ \text{, \ }\mu_{l,m-1}^{n}=\mu_{l,m}%
^{n}\beta_{l,m}\quad
\]
and consequently%
\begin{equation}
\mu_{l,m}^{n}=\frac{l!}{\sqrt{2l+1}}\sqrt{\frac{(n+l+1)!}{(n-l)!}}\sqrt
{\frac{(l-m)!}{(l+m)!}}\ ,\text{\ }0\leq m\leq l, \label{my6}%
\end{equation}
where we used equation (\ref{my5}) for $\mu_{l}^{n}$. Thus, from (\ref{5}) we obtain:

\begin{theorem}
\label{explicitbasis} The coupled standard basis vectors of $M_{\mathbb{C}%
}(n+1)$ are given by
\[
\mathbf{e}^{j}(l,-m)=(-1)^{m}\mathbf{e}^{j}(l,m)^{T}\ \text{for}\ -l\leq
m\leq0\text{,}\ \text{where for }\ 0\leq m\leq l,
\]%
\begin{equation}
\mathbf{e}^{j}(l,m)=\frac{(-1)^{l}}{\mu_{l,m}^{n}}\ {\displaystyle\sum
\limits_{k=0}^{l-m}}(-1)^{k}\binom{l-m}{k}J_{-}^{l-m-k}J_{+}^{l}J_{-}^{k}\ ,
\label{explicitelm1}%
\end{equation}%
\[
\mbox{with}\ \ \mu_{l,m}^{n}\ \mbox{given by (\ref{my6}) and}\ \ J_{\pm}\in
M_{\mathbb{R}}(n+1)\ \mbox{given by (\ref{Jn}).}
\]
\end{theorem}

\begin{remark}\label{superop} In accordance with equation (\ref{basis}) in Definition \ref{normalized}, the matrices $\mathbf{e}^j(l,m)$ given explicitly by (\ref{explicitelm1}) above satisfy
\begin{equation}\label{relelm1}  [J_+,\mathbf{e}^j(l,m)]=\alpha_{l,m}\mathbf{e}^j(l,m) \ , \ [J_-,\mathbf{e}^j(l,m)]=\beta_{l,m}\mathbf{e}^j(l,m) \ , \end{equation}
with $\alpha_{l,m}$ and $\beta_{l,m}$ as in  (\ref{c-d}) and  $J_{\pm}$ given by  (\ref{Jn}). However, one can also verify the following relations:  
\begin{equation}\label{relelm2} [J_3,\mathbf{e}^j(l,m)]=m  \mathbf{e}^j(l,m) \ , \ \sum_{k=1}^3 \ [J_k,[J_k,\mathbf{e}^j(l,m)]]=l (l+1) \mathbf{e}^j(l,m) \ , \end{equation} where $J_1=(J_++J_-)/2$, $J_2=(J_+-J_-)/2i$ and $J_3$ is given by (\ref{J3}). 
%\end{remark}

Now, for any operator $A$, if we define a ``superoperator'' $\mathbb A$ acting on an operator $B$ by 
\begin{equation} \mathbb A \cdot B = [A,B]  \ , \nonumber\end{equation}
then we can rewrite (\ref{relelm2}) as 
\begin{equation}\label{relelm3} \mathbb J_3 \cdot \mathbf{e}^j(l,m) = m\mathbf{e}^j(l,m) \ , \  \mathbb J^2\cdot \mathbf{e}^j(l,m) = l(l+1)\mathbf{e}^j(l,m) \ , \end{equation}
in other words, the $(2l+1)$-dimensional subspace $\mathcal{V}_l\subset M_{\mathbb C}(n+1)$ is the eigenspace of the superoperator $\mathbb J^2 =  \mathbb J_1^2 +  \mathbb J_2^2 +  \mathbb J_3^2$ of eigenvalue $l(l+1)$. In the physics literature, one usually interprets the tensor product $\varphi_{j_1}\otimes\varphi_{j_2}$ as a ``sum'', so that the coupled invariant spaces are eigenspaces of  ``addition of angular momenta''. But in our context, we take the tensor product $\varphi_j\otimes\bar{\varphi_j}$ and therefore ``$\mathbb J$'' is better interpreted as the  ``difference of angular momenta'' (we thank Robert Littlejohn for this point).
\end{remark}

\subsubsection{Illustrations} 

We shall illustrate the general formula  (\ref{explicitelm1}) above by generating the matrices $\mathbf{e}^j(l,m)$ in the
lower dimensional cases, showing explicitly that each $\mathbf{e}^{j}(l,m)$ is in fact an
$m$-subdiagonal matrix.

\begin{example}
\label{half-j}j=1/2:
\[
\mathbf{e}(1,1)=\left(
\begin{array}
[c]{cc}%
0 & -1\\
0 & 0
\end{array}
\right)  \ \mathbf{e}(1,0)=\frac{1}{\sqrt{2}}\left(
\begin{array}
[c]{cc}%
1 & 0\\
0 & -1
\end{array}
\right)  \ \mathbf{e}(1,-1)=\left(
\begin{array}
[c]{cc}%
0 & 0\\
1 & 0
\end{array}
\right)
\]

\end{example}

\begin{example}
\label{j=1}j=1:
\begin{align*}
\mathbf{e}(2,2)  &  =\left(
\begin{array}
[c]{ccc}%
0 & 0 & 1\\
0 & 0 & 0\\
0 & 0 & 0
\end{array}
\right)  \ \mathbf{e}(2,1)=\frac{1}{\sqrt{2}}\left(
\begin{array}
[c]{ccc}%
0 & -1 & 0\\
0 & 0 & 1\\
0 & 0 & 0
\end{array}
\right)  \ \mathbf{e}(2,0)=\frac{1}{\sqrt{6}}\left(
\begin{array}
[c]{ccc}%
1 & 0 & 0\\
0 & -2 & 0\\
0 & 0 & 1
\end{array}
\right) \\
\mathbf{e}(1,1)  &  =\frac{-1}{\sqrt{2}}\left(
\begin{array}
[c]{ccc}%
0 & 1 & 0\\
0 & 0 & 1\\
0 & 0 & 0
\end{array}
\right)  \ \mathbf{e}(1,0)=\frac{1}{\sqrt{2}}\left(
\begin{array}
[c]{ccc}%
1 & 0 & 0\\
0 & 0 & 0\\
0 & 0 & -1
\end{array}
\right)
\end{align*}

\end{example}

\begin{example}
\label{j=3/2}{\small { j=3/2:%
\begin{align*}
& \mathbf{e}(3,3)\!=\!\left(
\begin{array}
[c]{cccc}%
0 & 0 & 0 & -1\\
0 & 0 & 0 & 0\\
0 & 0 & 0 & 0\\
0 & 0 & 0 & 0
\end{array}
\right)  \ \mathbf{e}(3,2)\!=\!\frac{1}{\sqrt{2}}\!\left(
\begin{array}
[c]{cccc}%
0 & 0 & 1 & 0\\
0 & 0 & 0 & -1\\
0 & 0 & 0 & 0\\
0 & 0 & 0 & 0
\end{array}
\right) \\
& \mathbf{e}(3,1)\!=\!\frac{1}{\sqrt{5}}\!\left(
\begin{array}
[c]{cccc}%
0 & -1 & 0 & 0\\
0 & 0 & \sqrt{3} & 0\\
0 & 0 & 0 & -1\\
0 & 0 & 0 & 0
\end{array}
\right)  \ \mathbf{e}(3,0)\!=\!\frac{1}{\sqrt{20}}\!\left(
\begin{array}
[c]{cccc}%
1 & 0 & 0 & 0\\
0 & -3 & 0 & 0\\
0 & 0 & 3 & 0\\
0 & 0 & 0 & -1
\end{array}
\right) \\
& \mathbf{e}(2,2)\!=\!\frac{1}{\sqrt{2}}\!\left(
\begin{array}
[c]{cccc}%
0 & 0 & 1 & 0\\
0 & 0 & 0 & 1\\
0 & 0 & 0 & 0\\
0 & 0 & 0 & 0
\end{array}
\right)   \mathbf{e}(2,1)\!=\!\frac{1}{\sqrt{2}}\!\left(
\begin{array}
[c]{cccc}%
0 & -1 & 0 & 0\\
0 & 0 & 0 & 0\\
0 & 0 & 0 & 1\\
0 & 0 & 0 & 0
\end{array}
\right)   \mathbf{e}(2,0)\!=\!\frac{1}{2}\!\left(
\begin{array}
[c]{cccc}%
1 & 0 & 0 & 0\\
0 & -1 & 0 & 0\\
0 & 0 & -1 & 0\\
0 & 0 & 0 & 1
\end{array}
\right) \\
& \mathbf{e}(1,1)\!=\!\frac{-1}{\sqrt{10}}\!\left(
\begin{array}
[c]{cccc}%
0 & \sqrt{3} & 0 & 0\\
0 & 0 & 2 & 0\\
0 & 0 & 0 & \sqrt{3}\\
0 & 0 & 0 & 0
\end{array}
\right)  \ \mathbf{e}(1,0)\!=\!\frac{1}{\sqrt{20}}\!\left(
\begin{array}
[c]{cccc}%
3 & 0 & 0 & 0\\
0 & 1 & 0 & 0\\
0 & 0 & -1 & 0\\
0 & 0 & 0 & -3
\end{array}
\right)
\end{align*}
} }
\end{example}

\begin{remark}
\label{CGentries} Let us also make the observation that all non-zero
Clebsch-Gordan coefficients (\ref{Clebsch1}) with $j_{1}=j_{2}$ can be read
off from the entries of the matrices $\mathbf{e}^{j}(l,m)$, namely by
(\ref{uncouple}) and (\ref{Invers}) there is the expansion%
\begin{equation}
\mathbf{e}^{j}(l,m)=%
%TCIMACRO{\dsum \limits_{k}}%
%BeginExpansion
{\displaystyle\sum\limits_{k}}
%EndExpansion
(-1)^{m+k-1}C_{j-k+1,m-j+k-1,m}^{j,j,l}\mathcal{E}_{k,m+k} \label{Clebsch}%
\end{equation}
where the summation is in the range
\[
\left\{
\begin{array}
[c]{c}%
k=1,2,..,n+1-m,\text{ \ when }m\geq0\\
k=|m|+1,|m|+2,..,n+1,\text{ \ when }m<0
\end{array}
\right.
\]
Therefore, the m-subdiagonal matrix (\ref{Clebsch}), presented as in
(\ref{m-sub}), is
\begin{equation}
\mathbf{e}^{j}(l,m)=(e_{1,}e_{2,...,}e_{n+1-|m|})_{m}\text{ } \label{subdiag1}%
\end{equation}
where%
\begin{equation}
e_{k}=\left\{
\begin{array}
[c]{c}%
\mathbf{\ }(-1)^{m+k-1}C_{j-k+1,m-j+k-1,m}^{j,j,l}\text{ , when }m\geq0\\
(-1)^{k-1}C_{j-|m|-k+1,-j+k-1,\text{ }m}^{j,j,l}\text{ , when }m<0
\end{array}
\right.  \label{CGexplicitentries}%
\end{equation}

\end{remark}

Let us illustrate this for $j=1$ (see above), where there are 17 coefficients
$C_{m_{1},m_{2},m\text{ }}^{1,1,l}$for the appropriate range of indices,
namely
\[
C_{m_{1},m_{2},m\text{ }}^{1,1,l},\text{ \ }|m_{i}|\leq1,m=m_{1}+m_{2}\text{,
}|m|\leq l\leq2
\]
Thus, for example,
\begin{align*}
\mathbf{e}(1,1)  &  =(-1/\sqrt{2},-1\sqrt{2})_{1}=(-C_{1,0,1}^{1,1,1}%
,C_{0,1,1}^{1,1,1})_{1}\\
\mathbf{e}(1,0)  &  =(1/\sqrt{2},0,-1/\sqrt{2})_{0}=(C_{1,-1,0}^{1,1,1}%
,-C_{0,0,0}^{1,1,1},C_{-1,1,0}^{1,1,1})_{0}\\
\mathbf{e}(0,0)  &  =(1/\sqrt{3},1/\sqrt{3},1/\sqrt{3})_{0}=(C_{1,-1,0}%
^{1,1,0},-C_{0,0,0}^{1,1,0},C_{-1,1,0}^{1,1,0})_{0}%
\end{align*}
and among the above 17 coefficients only $C_{0,0,0}^{1,1,1}$ vanishes.\index{Spin-j system ! coupled standard basis |)}

\subsection{Decomposition of the operator product}\index{Spin-j system ! coupled standard basis ! product rule |(}
\label{covariant mult}

Linearly, the algebra $M_{\mathbb{C}}(n+1)$ is spanned by the matrices
$\mathbf{e}(l,m)$, so it is natural to inquire about their multiplication
laws. But these are normalized matrices, namely the scaled version of the
matrices $E(l,m)$, cf. (\ref{scale}). 

\subsubsection{The parity property}

Let us first state a \emph{parity
property},\index{Spin-j system ! parity property} to be established later, for the commutator $[P,Q]=PQ-QP$ and
anti-commutator $[[P,Q]]=PQ+QP$ of these unnormalized basis vectors $E(l,m)$,
as follows:

\begin{proposition}\index{Parity property ! for operators}
(The Parity Property for operators)\label{Emult} The following multiplication
rules hold for the matrices $E(l,m):$
\begin{align}
(i)  &  :\left[  E(l_{1},m_{1}),E(l_{2},m_{2})\right]  =\sum_{l\ \equiv
\ l_{1}+l_{2}+1}K_{m_{1},m_{2}}^{l_{1},l_{2},l}\text{ }E(l,m_{1}%
+m_{2}),\label{5g}\\
(ii)  &  :\left[  \left[  E(l_{1},m_{1}),E(l_{2},m_{2})\right]  \right]
=\sum_{l\ \equiv\ l_{1}+l_{2}}K_{m_{1},m_{2}}^{l_{1},l_{2},l}\text{ }%
E(l,m_{1}+m_{2}), \label{5k}%
\end{align}
with sums restricted by $\delta(l_{1},l_{2},l)=1$, and $l\equiv k$ means
congruence modulo $2$.
\end{proposition}

\begin{remark}\index{Heisenberg's dynamical equation}
\label{Heisenbdynspinj} The importance of commutators in quantum mechanics
stemmed from Heisenberg-Born-Jordan's matrix mechanics. Given a preferred
hermitian matrix H, called the hamiltonian matrix, it generates a flow in
Hilbert space defining the dynamics of a time-dependent matrix $M$ by
\emph{Heisenberg's equation}:
\begin{equation}
\label{Heisenbeq}\frac{dM}{dt}= -\frac{i}{\hbar}[H,M] + \frac{\partial
M}{\partial t} \ ,
\end{equation}
where $\hbar$ is Planck's constant. This equation extends to define quantum
dynamics of a bounded operator $M$ on any (finite or infinite)-dimensional
Hilbert space.
\end{remark}

\begin{remark}
For a fixed $n$, the multi-indexed coefficients $K$ in the expansions
(\ref{5g})-(\ref{5k}) are seen to be rational numbers. In fact, they are
polynomials in $n$ whose degree increases stepwise by 2 as $l$ decreases by 2.
\end{remark}

\begin{example}
For $n\geq5$%
\begin{align*}
E(3,2)E(2,1)  &  =k_{3}E(3,3)+k_{4}E(4,3)+k_{5}E(5,3)\\
E(2,1)E(3,2)  &  =k_{3}E(3,3)-k_{4}E(4,3)+k_{5}E(5,3)
\end{align*}
where $k_{3}=\frac{2}{3}n^{2}+\frac{4}{3}n-22$, $k_{4}=3/2$, and $k_{5}=4/15.$
\end{example}

We shall give two proofs of Proposition \ref{Emult}; one at the end of this
chapter, relying on the product rule developed below, and an independent one
in Appendix \ref{parity prop}.

\subsubsection{The full product rule}

Indeed, it is possible to obtain in a very straightforward way the full
multiplication rule for the standard coupled basis vectors given by Thorem
\ref{explicitbasis}, conveniently denoted in differentent ways such as
\[
\mathbf{e}(l,m)=\mathbf{e}^{j}(l,m)=|(jj)lm\rangle\ \ ,\ \ \langle
(jj)l^{\prime}m^{\prime}|(jj)lm\rangle=\delta_{l,l^{\prime}}\delta
_{m,m^{\prime}},\
\]
from their Clebsch-Gordan expansions in terms of the uncoupled basis vectors
$|j_{1}m_{1}j_{2}m_{2}\rangle$. Thus, from Definition \ref{uncoupled1} (model
(ii)) and the equations (\ref{Clebsch2}) and (\ref{Invers}), together with the
multiplication rule for one-element matrices
\begin{equation}
\mathcal{E}_{i,j}\mathcal{E}_{k,l}=\delta_{j,k}\mathcal{E}_{i,l},
\label{usualproduct}%
\end{equation}
we are straightforwardly led to the following result:

\begin{theorem}
\label{emult1 copy(1)} The operator product of the standard coupled basis
vectors decomposes in the standard coupled basis according to the following
formula:
\begin{equation}
\mathbf{e}^{j}(l_{1},m_{1})\mathbf{e}^{j}(l_{2},m_{2})=\sum_{l=0}%
^{2j}\mathcal{M}[j]_{m_{1},m_{2},m}^{\ l_{1},\ l_{2},\ l}\mathbf{e}^{j}(l,m)
\label{coupledproduct1}%
\end{equation}
where $m=m_{1}+m_{2}$ and the \emph{product coefficients} can be expressed as
\begin{equation}
\mathcal{M}[j]_{m_{1},m_{2},m}^{\ l_{1},\ l_{2},\ l}=\sum_{\mu_{1}=-j}^{j}%
\sum_{\mu_{2}=-j}^{j}\sum_{\mu_{3}=-j}^{j}(-1)^{j+\mu_{2}}C_{\mu_{1},\mu
_{2},m_{1}}^{\ j,\ j,\ l_{1}}C_{-\mu_{2},\mu_{3},m_{2}}^{\ j,\ j,\ l_{2}%
}C_{\mu_{1},\mu_{3},m}^{\ j,\ j,\ l} \label{coupledproduct2}%
\end{equation}

\end{theorem}

\begin{remark}
\label{nonvanishingM} At first sight, there should be a summation in $m$ from
$-l$ to $l$ in equation (\ref{coupledproduct1}). However, it follows
straightforwardly from equation (\ref{coupledproduct2}) and the non-vanishing
conditions for the Clebsch-Gordan coefficients (cf. (\ref{C-nonvanish})), that
the \emph{product coefficients} $\mathcal{M}[j]$ vanish unless $m=m_{1}+m_{2}%
$, implying no summation in $m$ in equation (\ref{coupledproduct1}), in
agreement with equations (\ref{5g})-(\ref{5k}).

In fact, as we shall see more clearly below (cf. equations (\ref{defWPS}) and
(\ref{WprodNonvanish})), the product coefficients satisfy non-vanishing
conditions similar to those of the Clebsch-Gordan coefficients, (cf.
(\ref{C-nonvanish})), namely
\begin{equation}
\mathcal{M}[j]_{m_{1},m_{2},m}^{\ l_{1},\ l_{2},\ l}\neq0\Longrightarrow
\left\{
\begin{array}
[c]{c}%
m=m_{1}+m_{2}\\
\delta(l_{1},l_{2},l)=1
\end{array}
\right.  , \label{M-nonvanish}%
\end{equation}
and clearly $\delta(j,j,l_{i})=\delta(j,j,l)=1$ also holds.
\end{remark}

\subsubsection{Wigner symbols and the product rule}

It is convenient to introduce the \emph{Wigner $3jm$ symbols} (cf.
\cite{VMK}), which are closely related to the Clebsch-Gordan coefficients.\index{Wigner $3jm$ symbols |(}\index{Wigner $3jm$ symbols ! relation with Clebsh-Gordan coefficients}\index{Clebsch-Gordan coefficients ! relation with Wigner $3jm$ symbols}\index{Spin-j system ! Wigner $3jm$ symbols |(}

\begin{definition}
\label{defnW3jm} The \emph{Wigner $3jm$ symbol} is defined to be the rightmost
symbol of the following identity
\begin{equation}
C_{m_{1},m_{2},-m_{3}}^{j_{1},j_{2},j_{3}}=(-1)^{j_{1}-j_{2}-m_{3}}%
\sqrt{2j_{3}+1}\left(
\begin{array}
[c]{ccc}%
j_{1} & j_{2} & j_{3}\\
m_{1} & m_{2} & m_{3}%
\end{array}
\right)  \label{CG-W}%
\end{equation}

\end{definition}

As a substitute for the Clebsch-Gordan coefficients, the relevance of the
Wigner $3jm$ symbols is largely due to their better symmetry properties:\index{Wigner $3jm$ symbols ! symmetry properties}
\begin{align}
&  \left(
\begin{array}
[c]{ccc}%
j_{1} & j_{2} & j_{3}\\
m_{1} & m_{2} & m_{3}%
\end{array}
\right)  =\left(
\begin{array}
[c]{ccc}%
j_{2} & j_{3} & j_{1}\\
m_{2} & m_{3} & m_{1}%
\end{array}
\right)  =\left(
\begin{array}
[c]{ccc}%
j_{3} & j_{1} & j_{2}\\
m_{3} & m_{1} & m_{2}%
\end{array}
\right) \label{cyclic}\\
&  =(-1)^{j_{1}+j_{2}+j_{3}}\left(
\begin{array}
[c]{ccc}%
j_{2} & j_{1} & j_{3}\\
m_{2} & m_{1} & m_{3}%
\end{array}
\right)  =(-1)^{j_{1}+j_{2}+j_{3}}\left(
\begin{array}
[c]{ccc}%
j_{1} & j_{2} & j_{3}\\
-m_{1} & -m_{2} & -m_{3}%
\end{array}
\right) \nonumber
\end{align}
which are obtained directly from (\ref{CG-W}) and the symmetry properties
(\ref{symCG1})-(\ref{symCG3}) for the Clebsch-Gordan coefficients.
Furthermore, the non-vanishing conditions for the Wigner $3jm$ symbols
analogous to (\ref{C-nonvanish}) are more symmetric:
\begin{equation}
\left(
\begin{array}
[c]{ccc}%
j_{1} & j_{2} & j_{3}\\
m_{1} & m_{2} & m_{3}%
\end{array}
\right)  \neq0\ \ \Rightarrow\left\{
\begin{array}
[c]{c}%
m_{1}+m_{2}+m_{3}=0\\
\delta(j_{1},j_{2},j_{3})=1\ .
\end{array}
\right.  \ \ \label{W3non-vanish}%
\end{equation}

Thus, let us also introduce another symbol, in analogy with (\ref{CG-W}).\index{Wigner product symbol |(}\index{Spin-j system ! Wigner product symbol |(}

\begin{definition}
\label{defnWPS} The \emph{Wigner product symbol} is defined to be the
rightmost symbol of the following identity
\begin{equation}
\mathcal{M}[j]_{m_{1},m_{2},-m_{3}}^{l_{1},l_{2},l_{3}}=(-1)^{2j-m_{3}}\left[
\begin{array}
[c]{ccc}%
l_{1} & l_{2} & l_{3}\\
m_{1} & m_{2} & m_{3}%
\end{array}
\right]  \!\!{[j]} \label{defWPS}%
\end{equation}

\end{definition}

From equations (\ref{coupledproduct2}), (\ref{CG-W}) and (\ref{defWPS}), the
\emph{Wigner product symbol} is also defined by its relation to the Wigner
$3jm$ symbols according to the following identity:%
\begin{align}
&  \left[
\begin{array}
[c]{ccc}%
l_{1} & l_{2} & l_{3}\\
m_{1} & m_{2} & m_{3}%
\end{array}
\right]  \!\!{[j]}\ \label{coupledproduct5}\\
&  =\sqrt{(2l_{1}+1)(2l_{2}+1)(2l_{3}+1)}\sum_{\mu_{1},\mu_{2},\mu_{3}=-j}%
^{j}(-1)^{3j-\mu_{1}-\mu_{2}-\mu_{3}}\ \nonumber\\
&  \cdot\left(
\begin{array}
[c]{ccc}%
j & l_{1} & j\\
\mu_{1} & m_{1} & -\mu_{2}%
\end{array}
\right)  \left(
\begin{array}
[c]{ccc}%
j & l_{2} & j\\
\mu_{2} & m_{2} & -\mu_{3}%
\end{array}
\right)  \left(
\begin{array}
[c]{ccc}%
j & l_{3} & j\\
\mu_{3} & m_{3} & -\mu_{1}%
\end{array}
\right) \nonumber
\end{align}
With the above definition, equation (\ref{coupledproduct1}) can be restated
as:
\begin{equation}
\mathbf{e}^{j}(l_{1},m_{1})\mathbf{e}^{j}(l_{2},m_{2})=\sum_{l=0}%
^{2j}(-1)^{2j+m}\left[
\begin{array}
[c]{ccc}%
l_{1} & l_{2} & l\\
m_{1} & m_{2} & -m
\end{array}
\right]  \!\!{[j]}\ \mathbf{e}^{j}(l,m) \label{coupledproductW}%
\end{equation}

In order to re-interpret the Wigner product symbol defined by equation
(\ref{coupledproduct5}) via a sum of triple products of Wigner $3jm$ symbols,
it is convenient to introduce another kind of Wigner symbol already studied in
the literature. These are the Wigner $6j$ symbols, which are related to the
re-coupling coefficients that appear when taking triple tensor products of
irreducible $SU(2)$-representations.\index{Wigner $6j$ symbols |(}\index{Spin-j system ! Wigner $6j$ symbols |(}

For simplicity, let us refer to three representations $\varphi_{j_{1}}%
,\varphi_{j_{2}},\varphi_{j_{3}}$, as $j_{1},j_{2},$ $j_{3}$, respectively.
Thus, if $|j_{1}m_{1}j_{2}m_{2}j_{3}m_{3}\rangle$ is an uncoupled basis vector
of the triple tensor product of $j_{1},j_{2},$ $j_{3}$, a coupled basis vector
for the tensor product can be obtained by first coupling $j_{1}$ and $j_{2}$
and then coupling with $j_{3}$, or first coupling $j_{2}$ and $j_{3}$ and then
coupling $\ $with $j_{1}$, or still, first coupling $j_{3}$ and $j_{1}$ and
then with $j_{2}$. Symbolically we describe the three coupling schemes by
\[
(i)\ \ \ \ \ \ j_{1}+j_{2}=j_{12}\ ,\ \ \ j_{12}+j_{3}=j
\]%
\[
(ii)\ \ \ \ \ j_{2}+j_{3}=j_{23}\ ,\ \ \ j_{23}+j_{1}=j
\]%
\[
(iii)\ \ \ \ j_{3}+j_{1}=j_{31}\ ,\ \ \ j_{31}+j_{2}=j
\]
and the coupled basis vectors arising from the three coupling schemes $(i)$,
$(ii)$ and $(iii)$, respectively, are given by%
\begin{equation}
\label{coupling}
\end{equation}
\[
|(j_{12}j_{3})jm\rangle=\sum_{m_{1}=-j_{1}}^{j_{1}}\sum_{m_{2}=-j_{2}}^{j_{2}%
}\sum_{m_{3}=-j_{3}}^{j_{3}}C_{m_{1},m_{2},m_{12}}^{j_{1},j_{2},j_{12}%
}C_{m_{12},m_{3},m}^{j_{12},j_{3},j}|j_{1}m_{1}j_{2}m_{2}j_{3}m_{3}\rangle
\]%
\[
|(j_{2}{}_{3}j_{1})jm\rangle=\sum_{m_{1}=-j_{1}}^{j_{1}}\sum_{m_{2}=-j_{2}%
}^{j_{2}}\sum_{m_{3}=-j_{3}}^{j_{3}}C_{m_{2},m_{3},m_{23}}^{j_{2},j_{3}%
,j_{23}}C_{m_{23},m_{1},m}^{j_{23},j_{1},j}|j_{1}m_{1}j_{2}m_{2}j_{3}%
m_{3}\rangle
\]%
\[
|(j_{31}j_{2})jm\rangle=\sum_{m_{1}=-j_{1}}^{j_{1}}\sum_{m_{2}=-j_{2}}^{j_{2}%
}\sum_{m_{3}=-j_{3}}^{j_{3}}C_{m_{3},m_{1},m_{31}}^{j_{3},j_{1},j_{31}%
}C_{m_{31},m_{2},m}^{j_{31},j_{2},j}|j_{1}m_{1}j_{2}m_{2}j_{3}m_{3}\rangle
\]
where in case (i), for example, the left side is a coupled basis vector
arising from the tensor product of Hilbert spaces with standard basis
$\{|(j_{1}j_{2})j_{12}m_{12}\rangle\}$ and $\{|j_{3}m_{3}\rangle\}$. The range
of $j_{12}$ is determined by the condition $\delta(j_{1},j_{2},j_{12})=1$.

\begin{definition}
\label{defnW6j} The inner products between these coupled basis vectors are the
\emph{re-coupling coefficients} and they define the \emph{Wigner $6j$ symbol}
as the rightmost symbol of the following identity
\begin{equation}
\langle(j_{12}j_{3})jm|(j_{23}j_{1})jm\rangle= \label{W6j1}%
\end{equation}%
\[
(-1)^{j_{1}+j_{2}+j_{3}+j}\sqrt{(2j_{12}+1)(2j_{23}+1)}\left\{
\begin{array}
[c]{ccc}%
j_{1} & j_{2} & j_{12}\\
j_{3} & j & j_{23}%
\end{array}
\right\}
\]

\end{definition}

Therefore, using (\ref{coupling}), the \emph{Wigner $6j$ symbols} can be
written as
\begin{equation}
\left\{
\begin{array}
[c]{ccc}%
j_{1} & j_{2} & j_{12}\\
j_{3} & j & j_{23}%
\end{array}
\right\}  =\frac{(-1)^{j_{1}+j_{2}+j_{3}+j}}{\sqrt{(2j_{12}+1)(2j_{23}+1)}}
\label{W6j3}%
\end{equation}%
\[
\cdot\sum_{m_{1}=-j_{1}}^{j_{1}}\sum_{m_{2}=-j_{2}}^{j_{2}}\sum_{m_{3}=-j_{3}%
}^{j_{3}}C_{m_{1},m_{2},m_{12}}^{j_{1},j_{2},j_{12}}C_{m_{12},m_{3},m}%
^{j_{12},j_{3},j}C_{m_{2},m_{3},m_{23}}^{j_{2},j_{3},j_{23}}C_{m_{23},m_{1}%
,m}^{j_{23},j_{1},j}%
\]
with similar equations for the Wigner $6j$ symbols obtained by the re-coupling
coefficients $\langle(j_{23}j_{1})jm|(j_{31}j_{2})jm\rangle$ and
$\langle(j_{31}j_{2})jm|(j_{12}j_{3})jm\rangle$.

Replacement of the Clebsch-Gordan coefficients in (\ref{W6j3}) by the Wigner
3jm symbols using equation (\ref{CG-W}) yields the following symmetric
expression for the Wigner $6j$ symbols (cf. \cite{VMK}):\index{Wigner $6j$ symbols ! relation with Wigner $3jm$ symbols |(}
\begin{equation}
\left\{
\begin{array}
[c]{ccc}%
a & b & c\\
d & e & f
\end{array}
\right\}  =\sum(-1)^{d+e+f+\delta+\epsilon+\phi}{}_{\cdot} \label{W6j2}%
\end{equation}%
\[
\cdot\left(
\begin{array}
[c]{ccc}%
a & b & c\\
\alpha & \beta & \gamma
\end{array}
\right)  \left(
\begin{array}
[c]{ccc}%
a & e & f\\
\alpha & \epsilon & -\phi
\end{array}
\right)  \left(
\begin{array}
[c]{ccc}%
d & b & f\\
-\delta & \beta & \phi
\end{array}
\right)  \left(
\begin{array}
[c]{ccc}%
d & e & c\\
\delta & -\epsilon & \gamma
\end{array}
\right)
\]
where the sum is taken over all possible values of $\alpha,\beta,\gamma
,\delta,\epsilon,\phi$, remembering that only three of these are independent.
For example, since the sum of the numbers in the second row of a $3jm$ symbol
is zero, we have the relations
\[
\alpha=-\epsilon+\phi,\beta=\delta-\phi,\gamma=\epsilon-\delta
\]
We should note that the six numbers in a Wigner $6j$ symbol are on an equal
footing, representing total spins and not projection quantum number $m_{i}$.

Now, we shall use an identity (cf. \cite{VMK}) which can be obtained from the
orthonormality relation for Wigner $3jm$ symbols. The latter is derived from
the one for Clebsch-Gordan coefficients (\ref{orthogCG}) and combined with the
symmetry properties (\ref{cyclic}) of the Wigner $3jm$ symbols one obtains the
identity
\[
\left(
\begin{array}
[c]{ccc}%
a & b & c\\
-\alpha & -\beta & -\gamma
\end{array}
\right)  \left\{
\begin{array}
[c]{ccc}%
a & b & c\\
d & e & f
\end{array}
\right\}  =
\]%
\[
\sum_{\delta,\epsilon,\phi}(-1)^{d-\delta+e-\epsilon+f-\varphi}\left(
\begin{array}
[c]{ccc}%
e & a & f\\
\epsilon & \alpha & -\phi
\end{array}
\right)  \left(
\begin{array}
[c]{ccc}%
f & b & d\\
\phi & \beta & -\delta
\end{array}
\right)  \left(
\begin{array}
[c]{ccc}%
d & c & e\\
\delta & \gamma & -\epsilon
\end{array}
\right)
\]
Together with equation (\ref{coupledproduct5}), this yields the following result:
\index{Wigner $6j$ symbols ! relation with Wigner $3jm$ symbols |)}

\begin{proposition}
\label{WPSWSS} The Wigner product symbol is proportional to the product of a
Wigner $3jm$ symbol and a Wigner $6j$ symbol, precisely as follows\index{Wigner product symbol ! relation with Wigner $6j$ and $3jm$ symbols}
\begin{align}
&  \left[
\begin{array}
[c]{ccc}%
l_{1} & l_{2} & l_{3}\\
m_{1} & m_{2} & m_{3}%
\end{array}
\right]  \!\!{[j]}\ \label{coupledproduct6}\\
&  =\sqrt{(2l_{1}+1)(2l_{2}+1)(2l_{3}+1)}\left(
\begin{array}
[c]{ccc}%
l_{1} & l_{2} & l_{3}\\
-m_{1} & -m_{2} & -m_{3}%
\end{array}
\right)  \left\{
\begin{array}
[c]{ccc}%
l_{1} & l_{2} & l_{3}\\
j & j & j
\end{array}
\right\} \nonumber
\end{align}

\end{proposition}

Then, from equations (\ref{coupledproductW}) and (\ref{coupledproduct6}) we
have immediately:

\begin{corollary}
\label{emult2 copy(1)} The operator product of the standard coupled basis
vectors $\mathbf{e}^{j}(l,m)$, stated in Theorem \ref{emult1 copy(1)}, is
given by
\begin{align}
\mathbf{e}^{j}(l_{1},m_{1})\mathbf{e}^{j}(l_{2},m_{2})  &  =\sum_{l=0}%
^{2j}(-1)^{2j+m}\sqrt{(2l_{1}+1)(2l_{2}+1)(2l+1)}\nonumber\\
&  \cdot\left(
\begin{array}
[c]{ccc}%
l_{1} & l_{2} & l\\
-m_{1} & -m_{2} & m
\end{array}
\right)  \left\{
\begin{array}
[c]{ccc}%
l_{1} & l_{2} & l\\
j & j & j
\end{array}
\right\}  \mathbf{e}^{j}(l,m)\nonumber
\end{align}

\end{corollary}

\begin{remark}
\label{equivariantdecomposition} (i) By linearity, given two operators,
$F={\sum_{l=0}^{n}\sum_{m=-l}^{l}}F_{lm}\mathbf{e}^{j}(l,m)$ and
$G={\sum_{l=0}^{n}\sum_{m=-l}^{l}}G_{lm}\mathbf{e}^{j}(l,m)$, with
$F_{lm},G_{lm}\in\mathbb{C}$, their operator product decomposes as
$FG={\sum_{l=0}^{n}\sum_{m=-l}^{l}}(FG)_{lm}\mathbf{e}^{j}(l,m)$, with
$(FG)_{lm}\in\mathbb{C}$ given by
\begin{equation}
(FG)_{lm}\ =\ (-1)^{n+m}\displaystyle{\sum_{l_{1},l_{2}=0}^{n}\sum
_{m_{1}=-l_{1}}^{l_{1}}}\left[
\begin{array}
[c]{ccc}%
l_{1} & l_{2} & l_{{}}\\
m_{1} & m_{2} & -m_{{}}%
\end{array}
\right]  \!\!{[j]}F_{l_{1}m_{1}}G_{l_{2}m_{2}}\nonumber
\end{equation}
where $m_{2}=m-m_{1}$ and the sum in $l_{1},l_{2}$ is restricted by
$\delta(l_{1},l_{2},l)=1$.

(ii) One should compare the above equation with the equation for the product
of operators $F$ and $G$ decomposed in the orthonormal basis of one-element
matrices $\mathcal{E}_{i,j}$ so that, for $F=\sum_{i,j=1}^{n}F_{ij}%
\mathcal{E}_{i,j}$, $G=\sum_{i,j=1}^{n}G_{ij}\mathcal{E}_{i,j}$,
$F_{ij},G_{ij}\in\mathbb{C}$, from (\ref{usualproduct}) one has the usual much
simpler expression for the matrix product
\[
(FG)_{ij}={\sum_{k=1}^{n+1}}F_{ik}G_{kj}\ .
\]
Of course, the problem with this familiar and simple  product decomposition is that it is
not $SO(3)$-invariant. However, our main justification for going through the
$SO(3)$-invariant decomposition of the operator product in the coupled basis
$\mathbf{e}^{j}(l,m)$ will become clear later on, in Chapter
\ref{multisymbols}.
\end{remark}

\subsubsection{Explicit formulae}\index{Wigner product symbol ! explicit formulae |(}

Now, in view of the above remark, it is interesting to have some explicit
expressions for the Wigner product symbol. However, from Proposition
\ref{WPSWSS}, this amounts to having explicit formulae for the Wigner $3jm$
symbol and the particular Wigner $6j$ symbol appearing in equation
(\ref{coupledproduct6}).

The first ones are obtained from (\ref{CG-W}) and the expressions for the
Clebsch-Gordan coefficients. For completeness, we list here the one obtained
from (\ref{explicitCG2}):\index{Wigner $3jm$ symbols ! explicit formulae}
\begin{equation}
\left(
\begin{array}
[c]{ccc}%
l_{1} & l_{2} & l_{3}\\
m_{1} & m_{2} & m_{3}%
\end{array}
\right)  \ =\ \Delta(l_{1},l_{2},l_{3}%
)\ S_{m_{1},m_{2},m_{3}}^{\ \ l_{1},\ l_{2},\ l_{3}}\ N_{m_{1},m_{2},m_{3}%
}^{\ l_{1},\ l_{2},\ l_{3}}\ \ , \label{explicitW3jm}%
\end{equation}
where $\Delta(l_{1},l_{2},l_{3})$ and $S_{m_{1},m_{2},m_{3}}^{\ \ l_{1}%
,\ l_{2},\ l_{3}}$ are given respectively by (\ref{DeltaW}) and (\ref{Sjjj})
via substitution of $j_{i}$ for $l_{i}$ and
\begin{align}
&N_{m_{1},m_{2},m_{3}}^{\ l_{1},\ l_{2},\ l_{3}}= \label{N123} \\
%\end{equation}%
%\begin{equation}
&{\sum_{z}}\frac{(-1)^{l_{1}-l_{2}-m_{3}+z}}{z!(l_{1}+l_{2}-l_{3}%
-z)!(l_{1}-m_{1}-z)!(l_{2}+m_{2}-z)!(l_{3}-l_{2}+m_{1}+z)!(l_{3}-l_{1}%
-m_{2}+z)!}\nonumber
\end{align}
with the usual summation convention (cf. Remark \ref{summation}).

We remind that equation (\ref{explicitW3jm}) holds only under the non-vanishing 
conditions (\ref{W3non-vanish}) for the Wigner $3jm$ symbols. 
Also, it is immediate from equations (\ref{DeltaW}) and (\ref{Sjjj}) that
$\Delta(l_{1},l_{2},l_{3})$ and $S_{m_{1},m_{2},m_{3}}^{\ \ l_{1}%
,\ \ l_{2},\ \ l_{3}}$ are invariant under any permutation of the columns in
the Wigner $3jm$ symbol and any change in sign of the magnetic numbers $m_{i}%
$. Therefore, it is the latter function $N_{m_{1},m_{2},m_{3}}^{\ l_{1}%
,\ l_{2},\ l_{3}}$ that carries the symmetry properties of the Wigner $3jm$
symbol, namely,
\begin{align}
N_{m_{1},m_{2},m_{3}}^{\ l_{1},\ l_{2},\ l_{3}}  &  =N_{m_{3},m_{1},m_{2}%
}^{\ l_{3},\ l_{1},\ l_{2}}=(-1)^{l_{1}+l_{2}+l_{3}}N_{m_{2},m_{1},m_{3}%
}^{\ l_{2},\ l_{1},\ l_{3}}\label{symmN}\\
&  =(-1)^{l_{1}+l_{2}+l_{3}}N_{-m_{1},-m_{2},-m_{3}}^{\ l_{1},\ l_{2},\ l_{3}%
}\ .\nonumber
\end{align}\index{Wigner $3jm$ symbols |)}\index{Spin-j system ! Wigner $3jm$ symbols |)}

Explicit expressions for the Wigner $6j$ symbol in equation
(\ref{coupledproduct6}) can be obtained as a particular case of the known
explicit expressions for general Wigner $6j$ symbols which have been obtained
from equations like (\ref{W6j3}) and (\ref{W6j2}) and are listed (cf. eg.
\cite{BL, VMK}). The expression obtained by Racah in 1942 yields:
\begin{align}
\left\{
\begin{array}
[c]{ccc}%
l_{1} & l_{2} & l_{3}\\
j & j & j
\end{array}
\right\}   &  =l_{1}!l_{2}!l_{3}!\Delta(l_{1},l_{2},l_{3})\sqrt{\frac
{(n-l_{1})!(n-l_{2})!(n-l_{3})!}{(n+l_{1}+1)!(n+l_{2}+1)!(n+l_{3}+1)!}%
}\nonumber\label{explicitW6j}\\
&  \cdot\displaystyle{\sum_{k}}\frac{(-1)^{n+k}(n+1+k)!}{(n+k-l_{1}%
-l_{2}-l_{3})!R(l_{1},l_{2},l_{3};k)}\quad\quad,
\end{align}
where $n=2j$, and
\begin{align}
&  R(l_{1},l_{2},l_{3};k)\label{explicitW6j2}\\
&  =(k-l_{1})!(k-l_{2})!(k-l_{3})!(l_{1}+l_{2}-k)!(l_{2}+l_{3}-k)!(l_{3}%
+l_{1}-k)!\quad.\nonumber
\end{align}

We note that the function $\Delta(l_{1},l_{2},l_{3})$ given by (\ref{DeltaW})
appears in both expressions (\ref{explicitW3jm}) and (\ref{explicitW6j}).
Therefore, by Proposition \ref{WPSWSS} and equation (\ref{coupledproduct6})
this function appears squared in the expression for the Wigner product symbol.\index{Wigner product symbol ! explicit formulae |)}

\subsubsection{Symmetry properties of the product rule}

Now, the following is immediate from (\ref{explicitW6j})-(\ref{explicitW6j2}%
):\index{Wigner $6j$ symbols ! symmetry properties}

\begin{proposition}
\label{symWPS} $\ \left\{
\begin{array}
[c]{ccc}%
l_{1} & l_{2} & l_{3}\\
j & j & j
\end{array}
\right\}  \ \mbox{is invariant by any permutation of} \ \ (l_{1}l_{2}l_{3})$
\end{proposition}

This property can also be obtained directly by algebraic manipulations
starting from equation (\ref{W6j2}), as shown in Appendix \ref{proofsymWPS}.

\begin{remark}
In fact, Proposition \ref{symWPS} is a particular case of the more general
statement: every Wigner $6j$ symbol is invariant under any permutation of its
columns and under any exchange of the upper and lower numbers in any given
column. A proof of this statement, which follows from the symmetries of the
$3jm$ symbols and the associativity of the triple tensor product, is found in
\cite{BL}.
\end{remark}\index{Wigner $6j$ symbols |)}\index{Spin-j system ! Wigner $6j$ symbols |)}

From Proposition \ref{symWPS} and equation (\ref{coupledproduct6}) in
Proposition \ref{WPSWSS}, we have:

\begin{corollary}\index{Wigner product symbol ! symmetry properties}
The Wigner product symbols (\ref{coupledproduct5}) have the same symmetry and
nonvanishing properties as the Wigner $3jm$ symbols, namely (cf.
(\ref{cyclic})),%
\begin{align}
&  \left[
\begin{array}
[c]{ccc}%
l_{1} & l_{2} & l_{3}\\
m_{1} & m_{2} & m_{3}%
\end{array}
\right]  \!\!{[j]} =\left[
\begin{array}
[c]{ccc}%
l_{2} & l_{3} & l_{1}\\
m_{2} & m_{3} & m_{1}%
\end{array}
\right]  \!\!{[j]} =\left[
\begin{array}
[c]{ccc}%
l_{3} & l_{1} & l_{2}\\
m_{3} & m_{1} & m_{2}%
\end{array}
\right]  \!\!{[j]}\label{WprodSym}\\
&  =(-1)^{l_{1}+l_{2}+l_{3}}\left[
\begin{array}
[c]{ccc}%
l_{2} & l_{1} & l_{3}\\
m_{2} & m_{1} & m_{3}%
\end{array}
\right]  \!\!{[j]} =(-1)^{l_{1}+l_{2}+l_{3}}\left[
\begin{array}
[c]{ccc}%
l_{1} & l_{2} & l_{3}\\
-m_{1} & -m_{2} & -m_{3}%
\end{array}
\right]  \!\!{[j]}\nonumber
\end{align}
\begin{equation}
\left[
\begin{array}
[c]{ccc}%
l_{1} & l_{2} & l_{3}\\
m_{1} & m_{2} & m_{3}%
\end{array}
\right]  \!\!{[j]} \neq0\ \ \Rightarrow\left\{
\begin{array}
[c]{c}%
m_{1}+m_{2}+m_{3}=0\\
\delta(l_{1},l_{2},l_{3})=1
\end{array}
\right.  \ \ \label{WprodNonvanish}%
\end{equation}

\end{corollary}\index{Wigner product symbol |)}\index{Spin-j system ! Wigner product symbol |)}

\ 

\noindent\textbf{Proof of the parity property}: Together with equation
(\ref{coupledproductW}), the identities (\ref{WprodSym}) and
(\ref{WprodNonvanish}) imply the parity property for operators, as stated in
Proposition \ref{Emult}.

\begin{remark}
By comparison with the direct proof of Proposition \ref{Emult} presented in
Appendix \ref{parity prop}, a look at the above proof is enough to indicate
the great amount of combinatorics that is encoded by the Wigner $3jm$ and $6j$
symbols. We refer to \cite{BL, VMK} for overviews of their further properties.
\end{remark}\index{Spin-j system ! coupled standard basis ! product rule |)}\index{Spin-j system ! operator algebra invariant decomposition |)}

%----------------------------------------------------------------------------------------------------------------------------------------------------------------------------------------------------------------------
%-------------------------------------------------------------------------------------- Chapter 4 -----------------------------------------------------------------------------------------------------------------
%----------------------------------------------------------------------------------------------------------------------------------------------------------------------------------------------------------------------

\chapter{The Poisson algebra of the classical spin system}\label{classicalspinsystems} 

This chapter presents the basic mathematical framework for classical mechanics
of a spin system. Practically all the material in the introductory section below can be found in basic textbooks on classical mechanics and we refer to some of these, e.g.  \cite{AM, Arnold, Gold, G-S, MR}, for the reader not yet too familiar with the subject, or for further details, examples, etc (reference \cite{Gold} is more familiar to physicists, while the others are more mathematical and closer in style to our brief introduction below).  Our
emphasis here is to provide a self-contained presentation of the
$SO(3)$-invariant decomposition of the pointwise product and the Poisson
bracket of polynomials, which are not easily found elsewhere (specially the latter).

%-------------------------------------------------------------------------------------- Section 4.1 -----------------------------------------------------------------------------------------------------------------

\section{Basic definitions of the classical spin system}

In this section we collect some basic facts and definitions concerning the Poisson algebra of functions on the 2-sphere  $G/H=S^{2}$, whose
homogeneous space structure was exploited in (\ref{Hopf2}), namely $G$
$=SO(3)$ or $SU(2)$ acts by rotations, and in the latter case $G$ acts via a
(fixed) covering homomorphism $\psi:SU(2)\rightarrow SO(3)$, cf. (\ref{homo}),
(\ref{psi}). We remind that the homogeneous $2$-sphere carries a $G$-invariant symplectic (or area) form $\omega$, cf. also (\ref{canonicalsympform}), which is (locally) expressed in terms of
spherical polar coordinates (\ref{polar}) as in (\ref{S2a}), that is, \index{$2$-sphere ! symplectic form} 
\begin{equation}
\omega=\sin\varphi d\varphi\wedge d\theta\ . \label{S2}%
\end{equation}

\begin{remark}
\label{signchoice} We often write $\omega=dS$ to indicate that $\omega$ is the
surface element, but we emphasize that this is a shorthand notation: $\omega$
is not an exact form. On the other hand, $\omega$ is a nondegenerate closed $2$-form, i.e. a \emph{symplectic} form. 
The local expression (\ref{S2}) for the $SO(3)$%
-invariant symplectic form $\omega$ on the sphere is canonical up to a choice
of orientation, or sign $\pm$. For the standard choice of orientation on
$\mathbb{R}^{3} = \mathcal{SO}(3)=\mathcal{SU}(2)$, $dx\wedge dy\wedge dz >
0$, the induced orientation for $\omega$ is via the identification 
\begin{equation}\label{omega2}
\omega=
(xdy\wedge dz+ydz\wedge dx+zdx\wedge dy)_{x^{2}+y^{2}+z^{2}=1}\end{equation} 
and is the one
with the choice of $+$ sign given by equation (\ref{S2}). Formula (\ref{omega2}) above also provides a direct way to verify that $\omega$ is symplectic and $G$-invariant.
\end{remark}

For a complex valued continuous function on $S^{2}$, its normalized integral
over $S^{2}$ equals its integral over $G$, namely
\begin{equation}
\int_{G}F(g\mathbf{n}_{0})dg=\frac{1}{4\pi}\int_{S^{2}}F(\mathbf{n}%
)dS=\frac{1}{4\pi}\int_{0}^{2\pi}\int_{0}^{\pi}F(\theta,\varphi)\sin\varphi
d\varphi d\theta\label{integr}%
\end{equation}
where $dg$ denotes the normalized Haar integral, and\ $\mathbf{n}_{0}\in
S^{2}$ is the \textquotedblleft north pole\textquotedblright\ fixed by
$H=SO(2)$, cf. (\ref{U(1)}). For functions $F_{i}$ on $S^{2}$ the $L^{2}%
$-inner product is
\begin{equation}
\left\langle F_{1},F_{2}\right\rangle =\frac{1}{4\pi}\int_{S^{2}}%
\overline{F_{1}(\mathbf{n})}F_{2}(\mathbf{n)}dS, \label{L2}%
\end{equation}
in particular, the constant $1$ has norm $1$. On the other hand, the metric  \index{$2$-sphere ! metric} on
$S^{2}$ and the gradient of a function $f$ are given in local spherical
coordinates by%
\begin{equation}
ds^{2}=d\varphi^{2}+\sin^{2}\varphi d\theta^{2}\text{, \ }\nabla
f=\frac{\partial f}{\partial\varphi}\frac{\partial}{\partial\varphi}+\frac
{1}{\sin^{2}\varphi}\frac{\partial f}{\partial\theta}\frac{\partial}%
{\partial\theta}\ . \label{metric}%
\end{equation}

Besides the ordinary pointwise multiplication of functions on the sphere,
which is commutative, another classical product of smooth functions on the sphere is
defined using the symplectic form and this turns out to be anti-commutative, or skew-symmetric.

In what follows, if $\alpha$ is an $n$-form and $v$ is a vector field, let 
$v\lrcorner \alpha$
denote the $(n-1)$-form obtained via interior product of $v$ and $\alpha$. Then, we have the following:  
\begin{definition} For any smooth function $f$ on $S^{2}$, its \emph{Hamiltonian
vector field} \index{$2$-sphere ! Hamiltonian vector field} $X_{f}$ is defined by
\begin{equation}\label{hvf}
X_{f}\lrcorner\ \omega+df=0\ .
\end{equation}
\end{definition}
In local spherical coordinates, the Hamiltonian vector field has the expression
\begin{equation}
X_{f}=\frac{1}{\sin\varphi}\left(  \frac{\partial f}{\partial\varphi}%
\frac{\partial}{\partial\theta}-\frac{\partial f}{\partial\theta}%
\frac{\partial}{\partial\varphi}\right)  \ .
\end{equation}
\begin{definition}
The \emph{Poisson bracket} \index{$2$-sphere ! Poisson bracket} of two smooth functions is, by definition,%
\begin{equation}
\left\{  f_{1},f_{2}\right\}  =X_{f_{1}}(f_{2})=-X_{f_{2}}(f_{1}%
)=\omega(X_{f_{1}},X_{f_{2}})\ . \label{Poisson}%
\end{equation}
\end{definition}
In local spherical coordinates, it is written as
\begin{equation}\label{pbsc}
\left\{  f_{1},f_{2}\right\}  =\frac{1}{\sin\varphi}\left(  \frac{\partial
f_{1}}{\partial\varphi}\frac{\partial f_{2}}{\partial\theta}-\frac{\partial
f_{1}}{\partial\theta}\frac{\partial f_{2}}{\partial\varphi}\right)  \ .
\end{equation}
In particular, it follows that
\begin{equation}
\left\{  x,y\right\}  =z,\left\{  y,z\right\}  =x,\left\{  z,x\right\}  =y
\end{equation}
and, furthermore, it follows immediately from (\ref{Poisson}) the following:   
\begin{proposition} The Poisson bracket is a derivation with respect to the
ordinary pointwise product of functions:
\begin{equation}
\label{Poissonderivation}\left\{  f_{1},f_{2}f_{3}\right\}  =\left\{
f_{1},f_{2}\right\}  f_{3}+f_{2}\left\{  f_{1},f_{3}\right\}  \ .
\end{equation}
\end{proposition}
Finally, while the ordinary commutative product is associative, for the skew-symmetric Poisson
bracket we have: 
\begin{proposition} The Poisson bracket 
 satisfies the Jacobi identity:
\begin{equation}
\label{PoissonJacobi}\left\{  \left\{  f_{1},f_{2}\right\}  ,f_{3}\right\}
+\left\{  \left\{  f_{2},f_{3}\right\}  ,f_{1}\right\}  +\left\{  \left\{
f_{3},f_{1}\right\}  ,f_{2}\right\}  =0
\end{equation}
and thus defines a Lie algebra on the space of smooth functions on the sphere.\index{Lie algebras ! Poisson algebra}
\end{proposition}
\begin{proof} First, note that for any Hamiltonian vector field $X_f$ defined by (\ref{hvf}), it follows from $d\omega=0$ and Cartan's ``magic'' formula $\mathcal L_{X_f}\omega=X_f\lrcorner d\omega +d(X_f\lrcorner\omega)$ that  $\mathcal L_{X_f}\omega=0$. It follows that $\mathcal L_{X_f}(\{g,h\})=X_f(\{g,h\})=\{X_f(g),h\}+\{g,X_f(h)\}$. But, from (\ref{Poisson}), this equation is equivalent to $\{f,\{g,h\}\}=\{\{f,g\},h\}+\{g,\{f,h\}\}$, which is equivalent to (\ref{PoissonJacobi}). 
\end{proof} 
\begin{corollary}
The Jacobi identity (\ref{PoissonJacobi})  is equivalent to 
\begin{equation}\label{PoissonJacobi2} [X_f,X_g]=X_{\{f,g\}} \ .\end{equation}
\end{corollary}
\begin{proof} 
By definition, $[X_f,X_g](h)=X_f(X_g(h))-X_g(X_f(h))=X_f(\{g,h\})-X_g(\{f,h\})=\{f,\{g,h\}\}-\{g,\{f,h\}\}$. But from (\ref{PoissonJacobi}), this is equal to $\{\{f,g\},h\}=X_{\{f,g\}}(h)$. And similarly, from  (\ref{PoissonJacobi2}) we get (\ref{PoissonJacobi}).
\end{proof}

Because of (\ref{PoissonJacobi})-(\ref{PoissonJacobi2}), the Poisson algebra is also called the Poisson-Lie algebra of smooth functions on a symplectic manifold. 

\begin{definition}
\label{Poissonalgebra} The \emph{Poisson algebra} of $S^{2}$\index{$2$-sphere ! Poisson algebra ! definition} is the space of smooth complex functions on $S^{2}$ with its
commutative pointwise product $\cdot$ and anti-commutative Poisson bracket $\{ \ , \ \}$ defined by the $SO(3)$-invariant symplectic form $\omega$ given by (\ref{S2})-(\ref{omega2}), via (\ref{hvf})-(\ref{pbsc}), satisfying (\ref{Poissonderivation})-(\ref{PoissonJacobi}), which shall be denoted $\{ \mathcal{C}_{\mathbb{C}}^{\infty}(S^{2}), \omega\}$.
\end{definition}

\begin{remark}
\label{hamdynsphere} The importance of Hamiltonian vector fields on $S^{2}$ is
that they generate dynamics of time-dependent functions on the sphere in the
same sense of usual Hamilton-Poisson dynamics  \index{$2$-sphere ! Hamilton-Poisson dynamics} derived from Newton's laws.
Thus, given a preferred differentiable function $h:S^{2}\to\mathbb{R}$,
usually called the Hamiltonian function, it generates a flow on the $2$-sphere 
that defines the dynamics of a differentiable function $f:S^{2}\times
\mathbb{R}\to\mathbb{R}$ via \emph{Hamilton's equation}:
\begin{equation}
\label{Hameq1}\frac{df}{dt}=X_{h}(f)+\frac{\partial f}{\partial t} \ ,
\end{equation}
which can be rewritten, using (\ref{Poisson}), as
\begin{equation}
\label{Hameq2}\frac{df}{dt}=\{h,f\}+\frac{\partial f}{\partial t} \ .
\end{equation}
These equations extend naturally to define the dynamics of a complex-valued
function $f: S^{2}\times\mathbb{R}\to\mathbb{C}$. One must remark the close
resemblance of equation (\ref{Hameq2}) to Heisenberg's equation
(\ref{Heisenbeq}), a resemblance which is at the core of Bohr's correspondence principle.
\end{remark}

We remind that Poisson algebras can also be defined on spaces of smooth functions on manifolds of a more general type, called Poisson manifolds. 
In the context of  $G$-invariant algebras,
$G=SU(2)$, let us now take a closer look at the Poisson algebra of smooth functions on  $\mathcal G^*\simeq\mathcal G=\mathbb R^3$ (cf. \cite{MR}, for instance, for more details). Referring to the extended Hopf map $\mathbb{C}^{2}\rightarrow \mathbb{R}^{3}$
given by (\ref{hopf}), and regarding $\mathbb{R}^{3}$ as $\mathcal{G}^{\ast }$,  $(x,y,z)$ can be seen as the ``angular momentum''  coordinates, so that 
the Poisson bi-vector field on $\mathcal G^*$ is given up to a choice of sign by 
\begin{equation}\label{Poissonbivector} 
\Pi=x\partial_y\wedge\partial_z+y\partial_z\wedge\partial_x+z\partial_x\wedge\partial_y \nonumber
\end{equation}  
and the Poisson bracket of two smooth functions $F,H:\mathcal G^*=\mathbb R^3\to\mathbb C$ is given by 
$$\{F,H\}=\Pi(dF,dH) \ .$$

Now, except at the origin, the foliation of this Poisson manifold $(\mathcal{G}^{\ast },\Pi )$ by symplectic leaves is regular: all symplectic leaves are
$2$-spheres centered at the origin, which are G-invariant submanifolds of $\mathcal{G}^{\ast }=\mathbb{R}^{3}$ via the coadjoint action. That is, each
of the spheres is a coadjoint orbit of $G=SU(2)$ with a $G$-invariant
symplectic form, defined as in (\ref{canonicalsympform}). Moreover, from Hamilton's equation
(\ref{Hameq2}), each of these spheres is also an invariant space under the Poisson
dynamics, in other words,  the classical dynamics of a $SU(2)$-symmetric mechanical system defined on $(\mathcal{G}^{\ast },\Pi )$ 
restricts to Poisson dynamics given by (\ref{Hameq2}) on each $G$-invariant sphere.

From another viewpoint, let us decompose $\mathcal{G}^{\ast }-\{0\}=\mathbb{R}^{3}-\{0\}\simeq S^{2}\times \mathbb{R}^{+}$ and consider the Poisson
algebra $\{C_{\mathbb{C}}^{\infty }(S^{2}\times \mathbb{R}^{+}), \Pi \}$.
It follows that this ``extended'' Poisson algebra decomposes under the
coadjoint action of $SU(2)$ into a  (continuous) sum of $G$-invariant Poisson subalgebras 
$\{\mathcal{C}_{\mathbb{C}}^{\infty}(S^{2}), \omega_r\}$, $r\in\mathbb R^+$, which are all isomorphic, i.e. $\{\mathcal{C}_{\mathbb{C}}^{\infty}(S^{2}), \omega_r\}\simeq \{\mathcal{C}_{\mathbb{C}}^{\infty}(S^{2}), \omega\}$, $\forall r\in\mathbb R^+$, under rescaling of $\omega_r$. 

Therefore, taking appropriate cautions, as restricting to the space of smooth functions $f:S^2\times\mathbb R^+\to\mathbb C$ with compact support in $\mathbb R^+$, etc, we can write:  
$$\{\mathcal{C}_{\mathbb{C}}^{\infty}(S^{2}\times\mathbb R^+), \Pi\}\simeq \{\mathcal{C}_{\mathbb{C}}^{\infty}(S^{2}), \omega\}\otimes\mathcal C^{\infty}_{\mathbb C}(\mathbb R^+).$$

Thus, for a classical spin system, i.e. $G=SU(2)$-symmetric classical mechanical system, not much is gained by extending in such a trivial way $\{\mathcal{C}_{\mathbb{C}}^{\infty}(S^{2}), \omega\}$,  the $G$-invariant Poisson algebra of $S^2$, to the Poisson algebra $\{\mathcal{C}_{\mathbb{C}}^{\infty}(S^{2}\times\mathbb R^+), \Pi\}$. 

On the other hand, as $G=SU(2)$ acts through $SO(3)$ on $S^2$, this action extends to a $G$-action on $T^*S^2$ which is symplectic for the canonical symplectic form $d\eta$ on $T^*S^2$ and this defines another $G=SU(2)$-invariant Poisson algebra, denoted $\{\mathcal C^{\infty}_{\mathbb C}(T^*S^2), d\eta\}$. But this algebra is ``too big'' to be the Poisson algebra of the classical spin system because  $T^*S^2$ is a real $4$-dimensional symplectic space and this corresponds to a classical dynamical system with $2$ degrees of freedom, while quantum spin-j systems are dynamical systems with $1$ degree of freedom.  Similarly or worse in dimensional counting, if we consider $S^2\times S^2$, $T^*S^3$, etc...

In view of the above discussion, it is natural to make the following
definition:

\begin{definition}\label{csp} The \emph{classical spin mechanical system}, or the \emph{classical spin system},  is the homogeneous $2$-sphere with its Poisson algebra $\{\mathcal{C}_{\mathbb{C}}^{\infty}(S^{2}), \omega\}$. \end{definition}   

Following standard physics terminology, the $2$-sphere with its $SO(3)$-invariant area form, $(S^2,\omega)$, is called the \emph{phase space} of the classical spin system.  

\begin{remark}
The homogeneous $2$-sphere $S^{2}$ with its $SO(3)$-invariant area form $\omega$ and metric, given
by (\ref{S2}) and (\ref{metric}) respectively, is also a K\"{a}hler (or complex)
manifold and this reflects in the fact that
\[
\omega(X_{f_{1}},X_{f_{2}}) = \omega(\nabla f_{1},\nabla f_{2}) \ .
\]
\end{remark} 

We note that the phase space of an affine mechanical system, $(\mathbb R^{2k},\omega)$, can also be seen as a K\"{a}hler manifold, $\mathbb C^k$. However, this identification depends on the choice of a complex structure and the full group of symmetries of a classical affine mechanical system does not preserve a complex structure.\footnote{Refer to the next chapter, Intermission, for definition of a classical affine mechanical system and the affine symplectic form $\omega$ in $(\mathbb R^{2k},\omega)$, as well as its symmetry group, etc.} This situation contrasts with the case of the classical spin system: the full symmetry group $SU(2)$, which acts on $S^2$ via $SO(3)$, preserves the symplectic form  (\ref{S2}) and the metric  (\ref{metric}) and therefore it preserves the complex structure that is compatible with both.

%-------------------------------------------------------------------------------------- Section 4.2 -----------------------------------------------------------------------------------------------------------------

\section{SO(3)-invariant decomposition of the Poisson algebra}

Having defined the Poisson algebra of the classical spin system, we now study how it  decomposes under the action of $SO(3)$, as this will be  fundamental for what will follow later on. To this end, we must look at the polynomial algebra on $S^2$.

\subsection{The irreducible summands of the polynomial algebra}\index{Spherical harmonics |(}\index{$2$-sphere ! Poisson algebra ! invariant decomposition |(}

Let $\mathbb{R}\left[  x,y,z\right]  $ be the algebra of real polynomials on
$\mathbb{R}^{3}.$ Their restriction to $S^{2}$ defines a distinguished class
of functions densely approximating smooth functions,
\begin{equation}
\mathbb{R}\left[  x,y,z\right]  \rightarrow\mathbb{R}\left[  x,y,z\right]
/\left\langle x^{2}+y^{2}+z^{2}-1\right\rangle \simeq Poly_{\mathbb{R}}%
(S^{2})\subset C_{\mathbb{R}}^{\infty}(S^{2}) \label{Pol}%
\end{equation}
We regard these spaces as the real form of the spaces of $\mathbb{C}$-valued
functions, that is, of their complexified versions
\begin{equation}
\mathbb{C}\left[  x,y,z\right]  /\left\langle x^{2}+y^{2}+z^{2}-1\right\rangle
\simeq Poly_{\mathbb{C}}(S^{2})\subset C_{\mathbb{C}}^{\infty}(S^{2})
\label{PolC}%
\end{equation}
These function spaces are $SO(3)$-modules with the induced action
\begin{equation}
\ F\rightarrow F^{g},\text{ }F^{g}(\mathbf{n})=F(g^{-1}\mathbf{n})\text{
}\mathbf{\ }\text{ } \label{action2}%
\end{equation}
In particular, $g\in SO(3)$ transforms a polynomial $Y$ by substituting the
variables
\[
Y(x,y,z)\rightarrow Y(x^{\prime},y^{\prime},z^{\prime})=Y^{g}(x,y,z)
\]
where $(x^{\prime},y^{\prime},z^{\prime})=(x,y,z)g$ (matrix product). Then
with the basis $\{x,y,z\}$, linear forms transform by $g\in SO(3)$ according
to the standard representation $\psi_{1}$ on $\mathbb{R}^{3}$, that is,
$\psi_{1}(g)=g$, and the action on forms of degree $l$ is the $l$-th symmetric
tensor product of $\psi_{1}$, with the following splitting
\[
\mathbb{R}\left[  x,y,z\right]  _{l}:S^{l}\psi_{1}=\psi_{l}+\psi_{l-2}%
+\psi_{l-4}+...
\]

On the other hand, multiplication by $(x^{2}+y^{2}+z^{2})$ is injective in the
polynomial ring, and clearly the lower summands lie in the subspace
\[
\mathbb{R}\left[  x,y,z\right]  _{l-2}(x^{2}+y^{2}+z^{2})\subset
\mathbb{R}\left[  x,y,z\right]  _{l}:\psi_{l-2}+\psi_{l-4}+...
\]
Consequently, when we restrict functions to $S^{2}$ we identify them according
to (\ref{Pol}), and then there will be a unique irreducible summand of type
$\psi_{l}$ for each $l$, spanned by polynomials of proper (minimal) degree
$l$. In this way, we have obtained the following result:

\begin{proposition}
\label{splitting}Real (resp. complex) polynomial functions of proper degree
$\leq n$ on $S^{2}$ constitute a $SO(3)$-representation with the same
splitting into irreducibles as the space of Hermitian matrices (resp. the full
matrix space) in dimension $n+1:$
\begin{equation}
Poly_{\mathbb{R}}(S^{2})_{\leq n}=\sum_{l=0}^{n}Poly(\psi_{l})\ \simeq
\sum_{l=0}^{n}\mathcal{H}(\psi_{l})=\mathcal{H}(n+1),\text{ \ cf.
(\ref{summands})} \label{polyR}%
\end{equation}%
\begin{equation}
Poly_{\mathbb{C}}(S^{2})_{\leq n}=\sum_{l=0}^{n}Poly(\varphi_{l})\ \simeq
\sum_{l=0}^{n}M_{\mathbb{C}}(\varphi_{l})=M_{\mathbb{C}}(n+1),\text{ \ cf.
(\ref{sum})} \label{PolyC}%
\end{equation}

\end{proposition}

\begin{definition}
$Poly(\psi_{l})$ (resp. its complex extension $Poly(\varphi_{l})$) denotes the
space of \emph{spherical harmonics} of type $\psi_{l}$ (resp. $\varphi_{l}$),
$0\leq l\leq n$. In view of the relation $x^{2}+y^{2}+z^{2}=1$ these are
polynomial functions of proper degree $l$.
\end{definition}

\subsection{The standard basis of spherical harmonics}

\label{genspherharm}

For the sake of completeness, let us also give a procedure for the calculation
of spherical harmonics, namely polynomials of proper degree $l$ which
constitute a standard orthonormal basis
\begin{equation}
Y_{l,l},Y_{l,l-1},...,Y_{l,0},...,Y_{l,-l+1},Y_{l,-l} \label{standard1}%
\end{equation}
for the representation space $Poly(\varphi_{l})\simeq\mathbb{C}^{2l+1}$,
$l\geq1$.$\ $Sometimes we also use the notation $Y_{l}^{m}$ for $Y_{l,m}$.

The basic case is $l=1$, where $g\in SO(3)$ transforms linear forms, expressed
in the basis $\{x,y,z\}$, by the same matrix $g$. Namely, the infinitesimal
generators $L_{k}$ $\in\mathcal{SO}(3)$ are the matrices in (\ref{L}), and the
angular momentum operators act linearly on $\mathbb{C}\left\{  x,y,z\right\}
\simeq\mathbb{C}^{3}$ with the matrix representation $\ $
\[
J_{k}=iL_{k}\text{ , }k=1,2,3 \ , \ J_{\pm}=J_{1}\pm iJ_{2} \ ,
\]
for $L_{k}$ given by (\ref{L}). Then,
%In particular, \
%\[
%J_{3}=\left(
%\begin{array}
%[c]{ccc}%
%0 & -i & 0\\
%i & 0 & 0\\
%0 & 0 & 0
%\end{array}
%\right)  ,\text{ }J_{-}=J_{1}-iJ_{2}=\left(
%\begin{array}
%[c]{ccc}%
%0 & 0 & 1\\
%0 & 0 & -i\\
%-1 & i & 0
%\end{array}
%\right)  ,
%\]
for example, $J_{3}(x)=iy,J_{3}(y)=-ix$. The operators $J_{k}$ and $J_{\pm}$
act as derivations on functions in general, for example
\[
J_{3}((x+iy)^{l})=l(x+iy)^{l}%
\]
and consequently $Y_{l,l}$ must be proportional to $(x+iy)^{l}$.

\begin{remark}
\label{phaseY} As in the definition of the standard coupled basis of operators
(cf. section 2.4.2), the definition of a standard orthonormal basis for each
$Poly(\varphi_{l})\simeq\mathbb{C}^{2l+1}$ depends on a choice of overall
phase (one for each $Poly(\varphi_{l})$), see Remark \ref{non-vanish} and
Proposition \ref{standard basis}.
\end{remark}

We set $Y_{0,0}=1$ and choose the phase convention by setting
\begin{equation}
Y_{l,l}=\frac{(-1)^{l}}{\lambda_{l,l}}\ (x+iy)^{l}, \ \forall l\in\mathbb N, \ Y_{l,m-1}%
=\frac{1}{\beta_{l,m}}J_{-}(Y_{l,m}),\text{ \ for \ }0<\left\vert m\right\vert
<l, \label{Y_functions}%
\end{equation}
where $\lambda_{l}=\lambda_{l,l}>0$ is calculated \ by
\begin{equation}
\lambda_{l}^{2}=\left\Vert (x+iy)^{l}\right\Vert ^{2}=\frac{1}{4\pi}%
\int_{S^{2}}(1-z^{2})^{l}dS=\sum_{k=0}^{l}\binom{l}{k}\frac{(-1)^{k}}%
{2k+1}=\frac{(l!)^{2}2^{2l}}{(2l+1)!} \label{lambda}%
\end{equation}
With reference to (\ref{gen}), (\ref{c-d}), (\ref{my6}) and (\ref{explicitelm}%
), we also conclude%
\begin{equation}
\label{explicitYlm}Y_{l,m}=\frac{(-1)^{l}}{\lambda_{l,m}}J_{-}^{l-m}%
(x+iy)^{l}\text{, \ \ \ }\lambda_{l,m}=\frac{l! \ 2^{l}}{\sqrt{2l+1}}%
\sqrt{\frac{(l-m)!}{(l+m)!}}\text{, \ }0\leq m\leq l\text{\ ,}%
\end{equation}
and changing $m\rightarrow-m$ has the effect
\begin{equation}
Y_{l,-m}=(-1)^{m}\overline{Y_{l,m}}\text{ , \ \ }0\leq m\leq l. \label{Ylm1}%
\end{equation}
Finally, we have the following relations, that set $Y_{l,m}$ as eigenvector of $J_3$ and $J^2$: 
\begin{equation} J_3Y_{l,m}=mY_{l,m} \ , \ \sum_{k=1}^3 J_k^2Y_{l,m} = l(l+1)Y_{l,m} \ . \end{equation}

It should be mentioned that $Y_{l,0}$ depends only on $z$. For easy reference
we list the resulting functions for $l=1,2,3:$
\begin{align}
Y_{1,1}  &  =-\sqrt{\frac{3}{2}}(x+iy),\text{ }Y_{1,0}=\sqrt{3}z,\text{
\ }Y_{1,-1}=\sqrt{\frac{3}{2}}(x-iy)\nonumber\\
Y_{2,2}  &  =\sqrt{\frac{15}{8}}(x+iy)^{2},\text{ }Y_{2,1}=-\sqrt{\frac{15}%
{2}}(x+iy)z,\text{ }Y_{2,0}=\frac{\sqrt{5}}{2}(3z^{2}-1)\label{harmonics}\\
Y_{3,3}  &  =-\frac{\sqrt{35}}{4}(x+iy)^{3},\text{ }Y_{3,2}=\frac{3}{2}%
\sqrt{\frac{35}{6}}(x+iy)^{2}z,\text{ }\nonumber\\
Y_{3,1}  &  =-\frac{\sqrt{21}}{4}(x+iy)(5z^{2}-1),\text{ \ }Y_{3,0}%
=\frac{\sqrt{7}}{2}(5z^{3}-3z)\nonumber
\end{align}

\begin{remark}
\label{homogenization} As exemplified above, we remind that all spherical
harmonics may be homogenized by using the relation $x^{2}+y^{2}+z^{2}=1$.
Thus, if $\mathbf{n}\in S^{2}$ has cartesian coordinates $(x,y,z)$, so that
$Y_{l,m}\equiv Y_{l,m}(x,y,z)=Y_{l,m}(\mathbf{n})$, and $-\mathbf{n}\in S^{2}$
denotes the antipodal point to $\mathbf{n}$, with coordinates $(-x,-y,-z)$,
then
\begin{equation}
\label{antipodalY}Y_{l,m}(-\mathbf{n}) \ = \ (-1)^{l} \ Y_{l,m}(\mathbf{n}).
\end{equation}
\end{remark}

Because it is not always so easy to determine the proper degree of a spherical
polynomial expressed as a polynomial in the cartesian coordinates $(x,y,z)$,
in view of the relation $x^{2}+y^{2}+z^{2}=1$, it is useful to have formulae
for the spherical harmonics which are expressed in terms of spherical
coordinates $(\varphi,\theta)$, cf. (\ref{polar}).

These are well known and, with our previous scaling convention, we have
\begin{align}
\ Y_{l,m}  &  =Y_{l}^{m}=\sqrt{2l+1}\sqrt{\frac{(l-m)!}{(l+m)!}}P_{l}^{m}%
(\cos\varphi)e^{im\theta},\label{spherical}\\
Y_{l,m}(0,\theta)  &  =\delta_{m,0}\sqrt{2l+1} \label{spherical1}%
\end{align}
where the functions $P_{l}^{m}$ are the so-called associated Legendre
polynomials,\index{Associated Legendre polynomials |(} which are ``classical'' well-known polynomials in $\cos\varphi=z$
and $\sin\varphi=(1-z^{2})^{1/2}$ and the identity (\ref{spherical1}) simply
means that $Y_{l,m}$ vanishes at the north pole, except when $m=0$, in which
case the value is $\sqrt{2l+1}$.

More precisely, the functions $P_{l}=P_{l}^{0}$ are the Legendre polynomials,\index{Legendre polynomials |(}
normalized so that $P_{l}(1)=1$. As polynomials in $z$, they are defined by
\begin{equation}
\label{defPl}P_{l}(z)=\frac{1}{2^{l}}\displaystyle{\sum_{k=0}^{[l/2]}}%
(-1)^{k}\frac{(2l-2k)! \ z^{l-2k}}{k!(l-k)!(l-2k)!} \ ,
\end{equation}
where $[q]$ denotes the integral part of $q\in\mathbb{Q}$, \ or by Rodrigues'
formula:
\begin{equation}
P_{l}(z)=\frac{1}{2^{l}l!}\frac{d^{l}}{dz^{l}}(z^{2}-1)^{l} \ .
\end{equation}

They can also be written as polynomials in $(1+z)$ by
\begin{equation}
P_{l}(z)=\displaystyle{\sum_{k=0}^{l}}(-1)^{l+k}\binom{l+k}{k}\binom{l}%
{k}\left(  \frac{1+z}{2}\right)  ^{k}=\displaystyle{\sum_{k=0}^{l}}%
\frac{(-1)^{l+k}}{(k!)^{2}}\frac{(l+k)!}{(l-k)!}\left(  \frac{1+z}{2}\right)
^{k} \label{1NN}%
\end{equation}
and they satisfy the following differential equation:
\begin{equation}
\frac{d}{dz}P_{l+1}(z)=(2l+1)P_{l}(z)+(2(l-2)+1)P_{l-2}(z)+(2(l-4)+1)P_{l-4}%
(z)+\cdots
\end{equation}%
\begin{equation}
\iff\quad(2l+1)P_{l}(z)=\frac{d}{dz}[P_{l+1}(z)-P_{n-1}(z)]\ .
\end{equation}
Furthermore, the Legendre polynomials are orthogonal in the interval $(-1,1)$
\begin{equation}
\int_{-1}^{1}P_{l}(z)P_{k}(z)dz=\delta_{l,k}\frac{2}{2l+1}\ . \label{orthogPl}%
\end{equation}
\index{Legendre polynomials |)}

The associated Legendre polynomials $P_{l}^{m}$ are defined for $m\geq0$ by
\begin{equation}
P_{l}^{m}(z)=(-1)^{m}(1-z^{2})^{m/2}\frac{d^{m}}{dz^{m}} P_{l}(z) \ , \ m\geq0
\end{equation}
and for negative $m$ by the identity
\begin{equation}
P_{l}^{-m}=(-1)^{m}\frac{(l-m)!}{(l+m)!}P_{l}^{m} .
\end{equation}

By setting $P_{0}^{0}=1$, $P_{l}^{m}=0$ for $l<m$, and
\begin{equation}
P_{l}^{l}(z)=(-1)^{l}(2l-1)!!(1-z^{2})^{l/2}\text{, }%
\end{equation}
the polynomials $P_{l}^{m}$ can also be derived by the recurrence formula%
\begin{equation}
(l-m)P_{l}^{m}(z)=(2l-1)zP_{l-1}^{m}(z)-(l+m-1)P_{l-2}^{m}(z) \ . \label{rec1}%
\end{equation}
The first polynomials are given by:
\begin{align*}
P_{1}^{0}  &  =z\text{, }P_{1}^{1}=-(1-z^{2})^{1/2}\text{, }P_{2}^{0}=\frac
{1}{2}(3z^{2}-1)\text{, }P_{2}^{1}=-3z(1-z^{2})^{1/2}\text{, }P_{2}%
^{2}=3(1-z^{2})\text{, }\\
P_{3}^{0}  &  =\frac{1}{2}z(5z^{2}-3)\text{, }P_{3}^{1}=\frac{3}{2}%
(1-5z^{2})(1-z^{2})^{1/2}\text{, }P_{3}^{2}=15z(1-z^{2})\text{, }P_{3}%
^{3}=-15(1-z^{2})^{3/2}%
\end{align*}
\index{Spherical harmonics |)}
\index{Associated Legendre polynomials |)}

\subsection{Decompositions of the classical products}

The space of complex polynomials on the $2$-sphere defined by \eqref{PolC}, $Poly_{\mathbb C}(S^2)$, densely approximates the space of smooth functions on the $2$-sphere, $\mathcal C^{\infty}_{\mathbb C}(S^2)$, and therefore it densely approximates the Poisson algebra of the classical spin system, $\{\mathcal C^{\infty}_{\mathbb C}(S^2), \omega\}$, cf. Definitions \ref{Poissonalgebra} and \ref{csp}, by letting $n\to\infty$ for   
$$Poly_{\mathbb C}(S^2)_{\leq n}\subset  \mathcal C^{\infty}_{\mathbb C}(S^2) \ .$$

Thus, although the classical products of functions on the sphere, the pointwise
product and the Poisson bracket, are defined for general smooth functions on $S^2$,  for what follows it is useful to have formulas for these classical products as
decomposed in the standard orthonormal basis of spherical harmonics, $$\langle
Y_{l,m},Y_{l^{\prime},m^{\prime}}\rangle= \delta_{l,l^{\prime}}\delta
_{m,m^{\prime}} \ .$$

The formula for the pointwise product has long been known, although it is not
so commonly presented with its proof. Here we prove it following the approach
outlined in \cite{Rose}, which uses a connection between the above functions
and the Wigner $D$-functions. On the other hand, as far as we know, the
formula for the Poisson bracket appeared for the first time only in 2002, in a
paper by Freidel and Krasnov \cite{F-K}, on which our proof below is based.

\subsubsection{Decomposition of the pointwise product}

\begin{proposition}
\label{PP} The pointwise product of spherical harmonics decomposes in the
basis of spherical harmonics according to the following formula:
\begin{equation}
Y_{l_{1},m_{1}}Y_{l_{2},m_{2}}=%
%TCIMACRO{\dsum \limits_{\substack{l=|l_{1}-l_{2}|\\l\equiv l_{1}%
%+l_{2}(\operatorname{mod}2)}}^{l_{1}+l_{2}}}%
%BeginExpansion
{\displaystyle\sum\limits_{\substack{l=|l_{1}-l_{2}|\\l\equiv l_{1}%
+l_{2}(\operatorname{mod}2)}}^{l_{1}+l_{2}}}
%EndExpansion
\sqrt{\frac{(2l_{1}+1)(2l_{2}+1)}{2l+1}}C_{m_{1},m_{2},m}^{l_{1},l_{2}%
,l}C_{0,0,0}^{l_{1},l_{2},l}Y_{l,m} \label{prod2}%
\end{equation}

\end{proposition}

\begin{remark}
\label{C000} We note that the particular rightmost Clebsch-Gordan coefficient appearing
in equation (\ref{prod2}) has a simple closed formula (cf. \cite{VMK}):
\begin{equation}
C_{0,0,0}^{l_{1},l_{2},l_{3}}=\frac{(-1)^{(l_{1}+l_{2}-l_{3})/2}\sqrt
{2l_{3}+1}\Delta(l_{1},l_{2},l_{3})((l_{1}+l_{2}+l_{3})/2)!}{((-l_{1}%
+l_{2}+l_{3})/2)!((l_{1}-l_{2}+l_{3})/2)!((l_{1}+l_{2}-l_{3})/2)!}\ ,
\label{C_000}%
\end{equation}
for $L=l_{1}+l_{2}+l_{3}$ even, and $C_{0,0,0}^{l_{1},l_{2},l_{3}}\equiv0$ for
$L$ odd, where $\Delta(l_1,l_2,l_3)$ is given by equation (\ref{DeltaW}).
We also note that the vanishing condition above follows from one of the
symmetries of the Clebsch-Gordan coefficients (cf. (\ref{symCG1})), namely,
\begin{equation}
\label{vanishCG000}C^{l_{1},l_{2},l_{3}}_{0,0,0} = (-1)^{l_{1}+l_{2}+l_{3}%
}C^{l_{1},l_{2},l_{3}}_{-0,-0,-0} \ .
\end{equation}

\end{remark}

\begin{proof}
The elements $g\in SO(3)$ rotate the points $p$ in the euclidean 3-space and
its unit sphere $S^{2}$, but fixing the points we can also view a rotation as
a transformation of the coordinate system $(x,y,z)$ to another system
$(x^{\prime},y^{\prime},z^{\prime})$. Then a function $Y$ on the sphere is
transformed to a function $Y^{g}$ such that%
\[
Y^{g}(x,y,z)=Y(x^{\prime},y^{\prime},z^{\prime})\text{, }Y^{g}(\varphi
,\theta)=Y(\varphi^{\prime},\theta^{\prime}),
\]
and on the other hand
\begin{equation}
Y_{l,m}^{g}=%
%TCIMACRO{\dsum \limits_{\mu}}%
%BeginExpansion
{\displaystyle\sum\limits_{\mu}}
%EndExpansion
D_{\mu,m}^{l}(g)Y_{l,\mu} \label{trans1}%
\end{equation}
since $g\rightarrow D^{l}(g)$ is, by definition, the matrix representation of
$SO(3)$ on the space spanned by the functions $\{Y_{l,m}\}$. Now, for a fixed
positive integer $l$ and two points $p_{1}=(\varphi_{1},\theta_{1}%
),p_{2}=(\varphi_{2},\theta_{2})$ on the sphere, we claim that the quantity%
\begin{equation}
\Theta=%
%TCIMACRO{\dsum \limits_{m}}%
%BeginExpansion
{\displaystyle\sum\limits_{m}}
%EndExpansion
\overline{Y}_{l,m}(\varphi_{1},\theta_{1})Y_{l,m}(\varphi_{2},\theta_{2})
\label{sigma}%
\end{equation}
is invariant under rotations. In fact, by (\ref{trans1}) and the unitary
property of the matrix $D^{l}=D^{l}(g)$,
\begin{align*}
\Theta^{g}  &  =%
%TCIMACRO{\dsum \limits_{m}}%
%BeginExpansion
{\displaystyle\sum\limits_{m}}
%EndExpansion
\
%TCIMACRO{\dsum \limits_{\mu_{1},\mu_{2}}}%
%BeginExpansion
{\displaystyle\sum\limits_{\mu_{1},\mu_{2}}}
%EndExpansion
\overline{D}_{\mu_{1},m}^{l}D_{\mu_{2},m}^{l}\ \overline{Y}_{l,\mu_{1}}%
(\theta_{1},\varphi_{1})Y_{l,\mu_{2}}(\varphi_{2},\theta_{2})\\
&  =%
%TCIMACRO{\dsum \limits_{\mu_{1},\mu_{2}}}%
%BeginExpansion
{\displaystyle\sum\limits_{\mu_{1},\mu_{2}}}
%EndExpansion
\delta_{\mu_{1},\mu_{2}}\ \overline{Y}_{l,\mu_{1}}(\varphi_{1},\theta
_{1})Y_{l,\mu_{2}}(\varphi_{2},\theta_{2})=%
%TCIMACRO{\dsum \limits_{\mu}}%
%BeginExpansion
{\displaystyle\sum\limits_{\mu}}
%EndExpansion
\overline{Y}_{l,\mu}(\varphi_{1},\theta_{1})Y_{l,\mu}(\varphi_{2},\theta
_{2})=\Theta
\end{align*}

In particular, let $\varphi$ be the spherical distance between $p_{1}$ and
$p_{2}$, and choose $g$ to rotate the coordinate system so that $\varphi
_{1}^{\prime}=0$ and $\theta_{2}^{\prime}=0$. Then $p_{1}$ is the new north
pole and hence $\varphi_{2}^{\prime}=\varphi$, and by (\ref{spherical1}) the
quantity (\ref{sigma}) becomes
\[
\Theta=\overline{Y}_{l,0}(0,\theta_{1}^{\prime})Y_{l,0}(\varphi,0)=\sqrt
{2l+1}Y_{l,0}(\varphi,0)
\]
Combined with the first formula (\ref{sigma}) of $\Theta$ this yields the
general formula
\begin{equation}
Y_{l,0}(\varphi,0)=\frac{1}{\sqrt{2l+1}}%
%TCIMACRO{\dsum \limits_{m}}%
%BeginExpansion
{\displaystyle\sum\limits_{m}}
%EndExpansion
\overline{Y}_{l,m}(\varphi_{1},\theta_{1})Y_{l,m}(\varphi_{2},\theta_{2})
\label{F1}%
\end{equation}

On the other hand, let us evaluate the identity (\ref{trans1}) with $m=0$ at
the point $p_{2}$ and obtain the following identity similar to (\ref{F1})
\begin{equation}
Y_{l,0}(\varphi_{2}^{\prime},0)=Y_{l,0}(\varphi_{2}^{\prime},\theta
_{2}^{\prime})=%
%TCIMACRO{\dsum \limits_{m}}%
%BeginExpansion
{\displaystyle\sum\limits_{m}}
%EndExpansion
D_{m,0}^{l}(g)Y_{l,m}(\varphi_{2},\theta_{2}) \label{F2}%
\end{equation}
Now, choose the rotation $g=R(\alpha,\beta,0)$ with Euler angle $\gamma=0$
(cf. (\ref{RotR})), which rotates the (old) north pole $(0,0,1)$ to the point
\[
p_{1}=(x_{1},y_{1},z_{1})=(\cos\alpha\sin\beta,\sin\alpha\sin\beta,\cos
\beta),
\]
namely the new north pole is $p_{1}$ $=$ $(\varphi_{1},\theta_{1}%
)=(\beta,\alpha)$ and therefore its spherical distance to $p_{2}$ is
$\varphi_{2}^{\prime}$. Consequently, the left sides of (\ref{F1}) and
(\ref{F2}) are identical, both expressing the same general coupling rule, so
we conclude%
\[
Y_{l,m}(\beta,\alpha)=\sqrt{2l+1}\overline{D}_{m,0}^{l}(\alpha,\beta,0)\
\]

Combing the above identity with the coupling rule (\ref{couple1}) for
$D$-functions, we deduce the formula (\ref{prod2}) as follows :
\begin{align*}
Y_{l_{1},m_{1}}(\varphi,\theta)Y_{l_{2},m_{2}}(\varphi,\theta)  &
=\sqrt{(2l_{1}+1)(2l_{2}+1)}\overline{D}_{m_{1},0}^{l_{1}}(\theta
,\varphi,0)\overline{D}_{m_{2},0}^{l_{2}}(\theta,\varphi,0)\\
&  =\sqrt{(2l_{1}+1)(2l_{2}+1)}%
%TCIMACRO{\dsum \limits_{l}}%
%BeginExpansion
{\displaystyle\sum\limits_{l}}
%EndExpansion
C_{m_{1},m_{2},m}^{l_{1},l_{2},l}C_{0,0,0}^{l_{1},l_{2},l}\overline{D}%
_{m_{1}+m_{2},0}^{l}(\theta,\varphi,0)\\
&  =\sum_{l=|l_{1}-l_{2}|}^{l_{1}+l_{2}}\sqrt{\frac{(2l_{1}+1)(2l_{2}%
+1)}{2l+1}}\ C_{m_{1},m_{2},m}^{l_{1},l_{2},l}C_{0,0,0}^{l_{1},l_{2}%
,l}\ Y_{l,m}(\varphi,\theta)
\end{align*}
In the above sum the range of $l$ is also restricted to $l\equiv l_{1}%
+l_{2}(\operatorname{mod}2)$, since the coefficient $C_{0,0,0}^{l_{1},l_{2}%
,l}$ vanishes when $l+l_{1}+l_{2}$ is odd.
\end{proof}

\subsubsection{Decomposition of the Poisson bracket}

\begin{proposition}
[\cite{F-K}]\label{DP} The Poisson bracket of spherical harmonics decomposes
in the basis of spherical harmonics according to the following formula:
\begin{equation}
i\left\{  Y_{l_{1}}^{m_{1}},Y_{l_{2}}^{m_{2}}\right\}  ={\displaystyle\sum
\limits_{_{\substack{l=|l_{1}-l_{2}|+1\\l\equiv l_{1}+l_{2}-1}}}^{l_{1}%
+l_{2}-1}}\sqrt{\frac{(2l_{1}+1)(2l_{2}+1)}{2l+1}}\ C_{m_{1},m_{2},m}%
^{l_{1},\ l_{2},\ l}P({l_{1},l_{2},l})\ Y_{l}^{m} \label{Poisson333}%
\end{equation}
where, by definition,
\begin{align}
&  P(l_{1},l_{2},l_{3})\label{PPP}\\
&  =\frac{(-1)^{(l_{1}+l_{2}-l_{3}+1)/2}\sqrt{2l_{3}+1}\Delta(l_{1}%
,l_{2},l_{3})(L+1)((L-1)/2)!}{((-l_{1}+l_{2}+l_{3}-1)/2)!((l_{1}-l_{2}%
+l_{3}-1)/2)!((l_{1}+l_{2}-l_{3}-1)/2)!}\ ,\nonumber
\end{align}
for $L=l_{1}+l_{2}+l_{3}$ odd, and $P(l_{1},l_{2},l_{3})\equiv0$ for $L$ even.
\end{proposition}

\begin{remark}
One should note the close resemblance between formula (\ref{prod2}) for the
pointwise product and formula (\ref{Poisson333}) for the Poisson bracket, in
view of the close resemblance between formulas (\ref{C_000}) and (\ref{PPP})
for the coefficients (which depend only on the l's) multiplying $C_{m_{1}%
,m_{2},m}^{l_{1},l_{2},l}$ in each case.

Note also that, since the Poisson bracket is skew symmetric, the summation in
(\ref{Poisson333}) starts with $l$ $=|l_{2}-l_{1}|+1=$ $\max\{$ $|l_{1}%
-l_{2}-1|$, $|l_{2}-l_{1}-1|\}$.
\end{remark}

\begin{proof}
The calculation for the decomposition of the Poisson bracket
\begin{equation}
\left\{  Y_{l_{1}}^{m_{1}},Y_{l_{2}}^{m_{2}}\right\}  =\frac{(-1)^{l_{1}%
+l_{2}}}{\lambda_{l_{1},m_{1}}\lambda_{l_{2},m_{2}}}\left\{  J_{-}%
^{l_{1}-m_{1}}(x+iy)^{l_{1}},J_{-}^{l_{2}-m_{2}}(x+iy)^{l_{2}}\right\}
\label{Poisson1}%
\end{equation}
can be considerably simplified by the appropriate choices of coordinates on
$\mathbb{R}^{3}$, perhaps also complex coordinates since the functions are
complex. Thus, in addition to $(x,y,z)$ and spherical polar coordinates
\[
\ (\rho,\varphi,\theta)\,:x=\rho\sin\varphi\cos\theta,y=\rho\sin\varphi
\sin\theta,z=\cos\varphi,
\]
following \cite{F-K} we shall also express the various vector fields (or
infinitesimal operators) in terms of the coordinate system
\begin{equation}
(u,v,z):u=x+iy,\ v=x-iy,\text{ }z=z \ \label{uv-coord}%
\end{equation}
(the main difficulty with using only spherical polar coordinates for this
calculation lies in handling the derivatives of the associated Legendre polynomials).

Now, via the action of $SO(3)$ on $\mathbb{R}^{3}$ the angular momentum
operators $J_{k}$ act as derivations of functions,  yielding the
following (complex valued) vector fields
\begin{align}
J_{1}  &  =i(z\frac{\partial}{\partial y}-y\frac{\partial}{\partial z}%
)=i(\sin\theta\frac{\partial}{\partial\varphi}+\cot\varphi\cos\theta
\frac{\partial}{\partial\theta})\nonumber\\
J_{2}  &  =i(x\frac{\partial}{\partial z}-z\frac{\partial}{\partial
x})=i(-\cos\theta\frac{\partial}{\partial\varphi}+\cot\varphi\sin\theta
\frac{\partial}{\partial\theta})\nonumber\\
J_{3}  &  =i(y\frac{\partial}{\partial x}-x\frac{\partial}{\partial
y})\ =u\frac{\partial}{\partial u}-v\frac{\partial}{\partial v}=-i\frac
{\partial}{\partial\theta}\label{fields}\\
J_{+}  &  =J_{1}+iJ_{2}=2z\frac{\partial}{\partial v}-u\frac{\partial
}{\partial z}=e^{i\theta}(\frac{\partial}{\partial\varphi}+i\cot\varphi
\frac{\partial}{\partial\theta})\nonumber\\
J_{-}  &  =J_{1}-iJ_{2}=-2z\frac{\partial}{\partial u}+v\frac{\partial
}{\partial z}=e^{-i\theta}(-\frac{\partial}{\partial\varphi}+i\cot\varphi
\frac{\partial}{\partial\theta})\nonumber
\end{align}
which are also tangential to the unit sphere $S^{2}=(\rho=1)$. Let us also
express the coordinate vector fields of the system (\ref{uv-coord}) in terms
of spherical coordinates
\begin{align}
\frac{\partial}{\partial u}  &  =\frac{1}{2}(\sin\varphi e^{-i\theta}%
\frac{\partial}{\partial\rho}+\cos\varphi e^{-i\theta}\frac{1}{\rho}%
\frac{\partial}{\partial\varphi}-\frac{ie^{-i\theta}}{\rho\sin\varphi}%
\frac{\partial}{\partial\theta})\nonumber\\
\frac{\partial}{\partial v}  &  =\frac{1}{2}(\sin\varphi e^{i\theta}%
\frac{\partial}{\partial\rho}+\cos\varphi e^{i\theta}\frac{1}{\rho}%
\frac{\partial}{\partial\varphi}+\frac{ie^{i\theta}}{\rho\sin\varphi}%
\frac{\partial}{\partial\theta})\label{uvz-fields}\\
\frac{\partial}{\partial z}  &  =\cos\varphi\ \frac{\partial}{\partial\rho
}-\sin\varphi\frac{1}{\rho}\frac{\partial}{\partial\varphi}\nonumber
\end{align}
In particular, along the sphere $S^{2}$ these operators also have a component
in the normal direction $\frac{\partial}{\partial\rho}$, which is simply
ignored when we calculate the Poisson bracket (\ref{Poisson}) on $S^{2}$. The
following lemma turns out to be very useful.

\begin{lemma}
The Poisson bracket $\left\{  F,G\right\}  $ on the 2-sphere can be expressed
by the formula%
\begin{equation}
i\left\{  F,G\right\}  =(\frac{\partial F}{\partial u})(J_{+}G)+(\frac
{\partial F}{\partial v})(J_{-}G)+(\frac{\partial F}{\partial z})(J_{3}G)
\label{Poisson2}%
\end{equation}

\end{lemma}

\begin{proof}
It is straightforward to calculate the right hand side of the identity in
terms of the coordinates $(\varphi,\theta)$, using the expressions
(\ref{fields}) and (\ref{uvz-fields}). Then one arrives at the expression
(\ref{Poisson}) multiplied by $i$.
\end{proof}

%\ 

The spherical harmonics $Y_{l}^{m}$ are generated by the successive
application of the operator $J_{-}$ to the monomial $u^{l}$, and calculation
of their Poisson bracket (\ref{Poisson1}) amounts to applying operator
products of type $\frac{\partial}{\partial\varphi}J_{-}^{k}$ and
$\frac{\partial}{\partial\theta}J_{-}^{k}$ \ to $u^{l}$. However, since the
commutation relations between $J_{-}^{k}$ and $\frac{\partial}{\partial
\varphi}$ or $\frac{\partial}{\partial\theta}$ are rather intricate, the
coordinates (\ref{uv-coord}) suggest themselves as more suitable for the
calculation of the bracket of these particular functions.

In fact, the operator $\frac{\partial}{\partial u}$ commutes with $J_{-}$.
This is, indeed, the motivation for the above lemma, cf. also formula (B14) in
\cite{F-K}. Using the expressions (\ref{fields}) the following commutation
identities are easily proved by induction
\begin{align*}
\frac{\partial}{\partial v}J_{-}^{k}  &  =J_{-}^{k}\frac{\partial}{\partial
v}+kJ_{-}^{k-1}\frac{\partial}{\partial z}-k(k-1)J_{-}^{k-2}\frac{\partial
}{\partial u}\\
\frac{\partial}{\partial z}J_{-}^{k}  &  =J_{-}^{k}\frac{\partial}{\partial
z}-2kJ_{-}^{k-1}\frac{\partial}{\partial u},\quad J_{3}J_{-}^{k}=J_{-}%
^{k}J_{3}-kJ_{-}^{k}%
\end{align*}
Consequently,
\begin{align*}
\frac{\partial}{\partial u}Y_{l}^{m}  &  =-\frac{1}{2}\sqrt{\frac{2l+1}{2l-1}%
}\sqrt{(l+m)(l+m-1)}Y_{l-1}^{m-1}\\
\frac{\partial}{\partial v}Y_{l}^{m}  &  =\frac{1}{2}\sqrt{\frac{2l+1}{2l-1}%
}\sqrt{(l-m)(l-m-1)}Y_{l-1}^{m+1}\\
\frac{\partial}{\partial z}Y_{l}^{m}  &  =\sqrt{\frac{2l+1}{2l-1}}%
\sqrt{(l+m)(l-m)}Y_{l-1}^{m}\\
J_{-}Y_{l}^{m}  &  =\sqrt{(l+m)(l-m+1)}Y_{l}^{m-1},\\
J_{+}Y_{l}^{m}  &  =\sqrt{(l-m)(l+m+1)}Y_{l}^{m+1},\ \ \ J_{3}Y_{l}^{m}%
=mY_{l}^{m}%
\end{align*}
and substitution into formula (\ref{Poisson2}) yields%
\begin{align*}
i\left\{  Y_{l_{1}}^{m_{1}},Y_{l_{2}}^{m_{2}}\right\}   &  = \frac{1}{2}\sqrt{\frac{2l_{1}+1}{2l_{1}-1}\ }\sqrt{(l_{1}-m_{1}%
)(l_{1}-m_{1}-1)(l_{2}+m_{2})(l_{2}-m_{2}+1)}Y_{l_{1}-1}^{m_{1}+1}Y_{l_{2}%
}^{m_{2}-1}\\
&  -\frac{1}{2}\sqrt{\frac{2l_{1}+1}{2l_{1}-1}}\sqrt{(l_{1}+m_{1})(l_{1}%
+m_{1}-1)(l_{2}-m_{2})(l_{2}+m_{2}+1)}Y_{l_{1}-1}^{m_{1}-1}Y_{l_{2}}^{m_{2}%
+1}\\
&  +\sqrt{\frac{2l_{1}+1}{2l_{1}-1}\ }m_{2}\sqrt{(l_{1}-m_{1})(l_{1}+m_{1}%
)}Y_{l_{1}-1}^{m_{1}}Y_{l_{2}}^{m_{2}}%
\end{align*}

Combining this with the product formula (\ref{prod2}) we arrive at
\begin{equation}
i\left\{  Y_{l_{1}}^{m_{1}},Y_{l_{2}}^{m_{2}}\right\}  ={\displaystyle\sum
\limits_{_{\substack{l=|l_{1}-l_{2}|+1\\l\equiv l_{1}+l_{2}-1}}}^{l_{1}%
+l_{2}-1}}\sqrt{\frac{(2l_{1}+1)(2l_{2}+1)}{2l+1}}\ C_{0,0,0}^{l_{1}%
-1,l_{2},l}P_{m_{1},m_{2},m}^{l_{1}-1,l_{2},l}Y_{l}^{m} \label{Poisson3}%
\end{equation}
where $m=m_{1}+m_{2}$, and%
\begin{align}
P_{m_{1},m_{2},m}^{l_{1}-1,l_{2},l}  &  = \frac{1}{2}\sqrt{(l_{1}-m_{1})(l_{1}-m_{1}-1)(l_{2}+m_{2})(l_{2}-m_{2}%
+1)}C_{m_{1}+1,m_{2}-1,m}^{l_{1}-1,l_{2},l}\nonumber\\
&  -\frac{1}{2}\sqrt{(l_{1}+m_{1})(l_{1}+m_{1}-1)(l_{2}-m_{2})(l_{2}+m_{2}%
+1)}C_{m_{1}-1,m_{2}+1,m}^{l_{1}-1,l_{2},l}\nonumber\\
&  +m_{2}\sqrt{(l_{1}-m_{1})(l_{1}+m_{1})}C_{m_{1},m_{2},m}^{l_{1}-1,l_{2},l} \label{P}
\end{align}
(We mention that equations (\ref{Poisson3})-(\ref{P}) can be put in a more
symmetric form by writing similar equations for $\left\{  Y_{l_{2}}^{m_{2}%
},Y_{l_{1}}^{m_{1}}\right\}  $ and using the skew symmetry of the Poisson
bracket to write $\left\{  Y_{l_{1}}^{m_{1}},Y_{l_{2}}^{m_{2}}\right\}
=\frac{1}{2}\left(  \left\{  Y_{l_{1}}^{m_{1}},Y_{l_{2}}^{m_{2}}\right\}
-\left\{  Y_{l_{2}}^{m_{2}},Y_{l_{1}}^{m_{1}}\right\}  \right)  $).

Now, in order to obtain (\ref{Poisson333}) we use equivariance under the group
action. Let us introduce the symbol
\begin{equation}
K_{m_{1},m_{2},m}^{l_{1},l_{2},l}=\sqrt{\frac{(2l_{1}+1)(2l_{2}+1)}{2l+1}%
}\ C_{0,0,0}^{l_{1}-1,l_{2},l}P_{m_{1},m_{2},m}^{l_{1}-1,l_{2},l}\ \label{KKK}%
\end{equation}
From equation (\ref{trans1}), we have on the one hand
\begin{align}
\{Y_{l_{1},m_{1}}^{g},Y_{l_{2},m_{2}}^{g}\}  &  ={\displaystyle\sum
\limits_{\mu_{1},\mu_{2}}}D_{\mu_{1},m_{1}}^{l_{1}}(g)D_{\mu_{2},m_{2}}%
^{l_{2}}(g)\{Y_{l_{1},\mu_{1}},Y_{l_{2},\mu_{2}}\}\nonumber\\
&  ={\displaystyle\sum\limits_{\mu_{1},\mu_{2},l}}D_{\mu_{1},m_{1}}^{l_{1}%
}(g)D_{\mu_{2},m_{2}}^{l_{2}}(g)K_{\mu_{1},\mu_{2},\mu}^{l_{1},l_{2}%
,l}\ Y_{l,\mu}\ , \label{trans111}%
\end{align}
where we have used (\ref{Poisson3}) and (\ref{KKK}), with summation in $l$
under the appropriate restriction indicated in (\ref{Poisson3}). On the other
hand,
\begin{align}
\{Y_{l_{1},m_{1}}^{g},Y_{l_{2},m_{2}}^{g}\}  &  ={\displaystyle\sum
\limits_{l}}K_{m_{1},m_{2},m}^{l_{1},\ l_{2},\ l}\ Y_{l,m}^{g}\nonumber\\
&  ={\displaystyle\sum\limits_{l,\mu}}D_{\mu,m}^{l}(g)K_{m_{1},m_{2},m}%
^{l_{1},\ l_{2},\ l}\ Y_{l,\mu}\ . \label{trans1111}%
\end{align}
But by the coupling rule, equation (\ref{couple1}), we can rewrite equation
(\ref{trans111}) as
\begin{equation}
\{Y_{l_{1},m_{1}}^{g},Y_{l_{2},m_{2}}^{g}\}={\displaystyle\sum\limits_{\mu
_{1},\mu_{2},l,l^{\prime}}}C_{\mu_{1},\mu_{2},\mu^{\prime}}^{l_{1}%
,l_{2},l^{\prime}}C_{m_{1},m_{2},m^{\prime}}^{l_{1},l_{2},l^{\prime}}%
D_{\mu^{\prime},m^{\prime}}^{l^{\prime}}(g)\ K_{\mu_{1},\mu_{2},\mu}%
^{l_{1},\ l_{2},\ l}\ Y_{l,\mu}\ . \label{trans11}%
\end{equation}
Then, using the orthonormality relations (\ref{orthogCG}) of the
Clebsch-Gordan coefficients, we conclude that ``solutions'' of (\ref{trans1111})
$=$ (\ref{trans11}) are given by
\begin{equation}
K_{m_{1},m_{2},m}^{l_{1},\ l_{2},\ l}=F(l_{1},l_{2},l)\ C_{m_{1},m_{2}%
,m}^{l_{1},l_{2},l}\ , \label{KKKK}%
\end{equation}
where, in principle, $F$ could be any function of $l_{1},l_{2},l$. \ However,
writing
\[
F(l_{1},l_{2},l)=\sqrt{\frac{(2l_{1}+1)(2l_{2}+1)}{2l+1}}\ P(l_{1},l_{2},l)
\]
and substituting (\ref{KKKK}) into (\ref{KKK}), we see that the function
$P(l_{1},l_{2},l)$ is determined by
\begin{equation}
C_{0,0,0}^{l_{1}-1,l_{2},l}P_{m_{1},m_{2},m}^{l_{1}-1,l_{2},l}=C_{m_{1}%
,m_{2},m}^{l_{1},\ l_{2},\ l}\ P(l_{1},l_{2},l)\ . \label{PP1}%
\end{equation}

Now, first we note that the l.h.s. of (\ref{PP1}) vanishes if $l_{1}+l_{2}+l$
is even (cf. (\ref{vanishCG000})), and therefore
\begin{equation}
P(l_{1},l_{2},l)\equiv0\ \text{, if }l_{1}+l_{2}+l\ \text{is even,}
\label{vanishP}%
\end{equation}
which agrees with the sum in (\ref{Poisson333}) being restricted to $l\equiv
l_{1}+l_{2}-1$ (mod 2).

Second, we note that equation (\ref{PP}) must hold for any values of $m_{1}$
and $m_{2}$; thus, in particular, for $m=m_{1}+m_{2}=l$ and $m_{1}=l_{1}$ we
have that
\begin{equation}
C_{0,0,0}^{l_{1}-1,l_{2},l}P_{l_{1},l-l_{1},l}^{l_{1}-1,l_{2},l}%
=C_{l_{1},l-l_{1},l}^{l_{1},\ l_{2},\ l}\ P(l_{1},l_{2},l)\ . \label{PPPP}%
\end{equation}

The Clebsch-Gordan coefficient $C_{0,0,0}^{l_{1}-1,l_{2},l}$ has the closed
formula given by equation (\ref{C_000}), with $l_{1}$ replaced by $l_{1}-1$,
but similarly, the Clebsch-Gordan coefficients $C_{l_{1},l-l_{1},l}%
^{l_{1},\ l_{2},\ l}$ also has a well-known simple closed formula (cf. \cite{VMK}):
\begin{equation}
C_{l_{1},l-l_{1},l}^{l_{1},l_{2},l}=\sqrt{\frac{(2l_{1})!(2l+1)!}{(l_{1}%
+l_{2}+l+1)!(l_{1}-l_{2}+l)!}} \label{closedCGl}%
\end{equation}

Then, equations (\ref{PPPP}) and (\ref{P}), together with equations
(\ref{C_000}) and (\ref{closedCGl}) straightforwardly yield equation
(\ref{PPP}), when $l_{1}+l_{2}+l$ is odd.
\end{proof}

\begin{remark}
Of course, by linearity, for $f$ and $g$ decomposed in the basis of spherical
harmonics, one obtains the coefficients of the expansion in spherical
harmonics of $fg$ and $\{f,g\}$ straightforwardly from (\ref{prod2}%
)-(\ref{C_000}) and (\ref{Poisson333})-(\ref{PPP}).
\end{remark}\index{$2$-sphere ! Poisson algebra ! invariant decomposition |)}

%----------------------------------------------------------------------------------------------------------------------------------------------------------------------------------------------------------------------
%-------------------------------------------------------------------------------------- Chapter 5 -----------------------------------------------------------------------------------------------------------------
%----------------------------------------------------------------------------------------------------------------------------------------------------------------------------------------------------------------------

\chapter{Intermission}

\subsubsection{Brief historical overview of symbol correspondences in affine mechanical systems} 

The names of Pythagoras, Euclid and Plato can perhaps best summarize the dawning of the ``mathematization of nature'' process, that took place in ancient Greek civilization when the Pythagorean school boosted the philosophy that numbers and mathematical concepts were the key to understand the divine cosmic order. While Euclid's {\it Elements} set forth the axiomatization of geometry, what we now call ``flat'' geometry, it was Plato, however, who first applied the Pythagorean philosophy to the empirical science of the time:  astronomy. Inspired by the ``divine geometrical perfection'' of spheres and circles, in his dialogue {\it Timaeus} Plato set forth the notion of uniform circular motions to be the natural motions of heavenly bodies, as each of which being in possess of its own anima mundi, or ``world soul'',  would continue on such an eternal motion. 

Plato's ``circular inertial motions'' had to come to terms with the empirical astronomical data of his historical period, however, and thus a whole mathematical model was developed, based on uniform circular motions, in which circles upon circles upon circles... were used to describe the apparent motions of the sun, moon and known planets, culminating in the treatise of Ptolemy called the  {\it Almagest}. The latter was used by professional astronomers for centuries with adequate precision, as it realized the first ``perturbation theory'' to be used in history, wherein new tinier circles could always be added to better fit new and more precise data, in a process  akin to our  well-known Fourier series decomposition of periodic functions. 

Many centuries later, when Galileo discovered empirically  that  a principle of inertial motions applied to earthly motions as well, he then followed Plato's principle, so that Galileo's  inertial motions were still the circular inertial motions of Plato, only extended to sub-celestial motions (after all, a uniform straight motion on the surface of the earth is actually circular) \cite{Koyre}. It was Descartes who reformulated the universal principle of inertia in terms of uniform linear motions. But then, the approximately circular planetary motions had to be re-explained as resulting from the action of a ``force'', as emphasized by Hook, culminating in Newton's mechanical theory of universal gravitation set forth in his {\it Principia}.                    

Thus, following Descartes, Newton placed uniform linear motion, or the straight line,  as the first mathematical axiom of his new mechanics, greatly increasing in  importance Euclid's axiomatics of geometry, so that Euclidean $3$-space became synonymous with universal space.  And as Newtonian mechanics developed mathematically throughout  the decades, straight lines retained their primordial  predominance so that, at the turn of the $20^{th}$ century, when Planck first set forth the hypothesis of quantized energies, the mechanics of conservative systems was modeled on what we now understand as the Poisson algebra of functions on a symplectic  affine space, which doubled the dimensions of the configuration space of positions, this latter seen as some product of Euclidean $3$-spaces, or in simplified versions as an Euclidean $k$-space $\mathbb R^k_{\epsilon}$, including $k=1$, a straight line (in our notation, $\mathbb R^k$ is the $k^{th}$ power of $\mathbb R$, while $\mathbb R^k_{\epsilon}$ also carries the Euclidean metric $\epsilon_k$).   

In this way, the most fundamental symmetry of these spaces are straight linear motions, or affine translations, and so we can call mechanical systems which are symmetric under such affine translations as  {\it affine mechanical systems}.  

For an affine symplectic space $(\mathbb R^{2k}, \omega)\equiv \mathbb R^{2k}_{\omega}$, where $\mathbb R^{2k}\ni(p,q) \ , \ p,q\in\mathbb R^k$,  $$\omega=\sum_{i=1}^k dp_i\wedge dq_i \ , $$
the group of affine translations is $\mathbb R^{2k}$ with usual vector addition as the group product, and such that for any fixed $\xi=(a,b)\in\mathbb R^{2k}$ the affine translation is 
\begin{equation}\label{translation} T({\xi}) : \mathbb R^{2k}_{\omega}\to\mathbb R^{2k}_{\omega} \ , \ (p,q)\mapsto (p+a,q+b) \ .\end{equation} 
In fact, it is well known that the the full group of symmetries of an affine symplectic space $\mathbb R^{2k}_{\omega}$ is actually  much larger and is called the affine symplectic group, $aSp_{\mathbb R}(2k)= \mathbb R^{2k} \rtimes Sp_{\mathbb R}(2k)$, the semi-direct product of $\mathbb R^{2k}$ and $Sp_{\mathbb R}(2k)$, the latter being the group of linear symplectic transformations of $\mathbb R^{2k}_{\omega}$, which is the maximal subgroup of $GL_{\mathbb R}(2k)$ that preserves the symplectic structure $\omega$. 

Therefore, when a mathematical formulation of quantum mechanical systems began taking shape in early $20^{th}$ century, it also followed the form of an affine mechanical system.  However, in contrast to classical affine mechanical systems,  which are defined by the Poisson algebra of functions on a finite dimensional affine symplectic space,  quantum affine mechanical systems  are defined by the algebra of operators on an infinite dimensional  Hilbert space, $\mathcal H=L^2_{\mathbb C}(\mathbb R^k_{\epsilon})$, where $\mathbb R^k_{\epsilon}$ is usually identified with either the  space of positions $q$ or the space of momenta $p$.  

Because, while quantum symmetries have to be implemented by unitary operators, according to the mathematical framework established by von Neumann, some of these arise from the symmetries of $\mathbb R^k_{\epsilon}$ itself, whose symmetry group is the Euclidean group $E_{\mathbb R}(k)=\mathbb R^k\rtimes O(k)$, the semi-direct product of  $\mathbb R^k$ and the orthogonal group $O(k)$,  the latter being the maximal subgroup  of $GL_{\mathbb R}(2k)$  that preserves distances (and hence also angles).    
Now, $O(k)$ is compact and thus admits finite dimensional unitary representations, but $\mathbb R^k$ and hence $E_{\mathbb R}(k)$ are  noncompact, so that both only admit infinite dimensional unitary representations.\footnote{Except, in the case of $\mathbb R^k$, through its toroidal compactification $\mathbb R^k\to  T^k=(\mathbb R/\mathbb Z)^k$.} This is the geometrical explanation of  why the Hilbert spaces of quantum affine mechanical systems have to be infinite dimensional. 

Letting  $E_{\mathbb R}(k)$ act ``in the same way'', or diagonally, on the $p$ and $q$ subspaces of $\mathbb R^{2k}_{\omega}$,  $E_{\mathbb R}(k)$ is naturally a subgroup of  $aSp_{\mathbb R}(2k)$, the latter group being much larger. Let us denote this embedding of $E_{\mathbb R}(k)$ into $aSp_{\mathbb R}(2k)$ by $\tilde{E}_{\mathbb R}(k)$. Then, $$\tilde{E}_{\mathbb R}(k) \subset E_{\mathbb R}(k\oplus k) , \ \ \mathbb R^{2k} \subset E_{\mathbb R}(k\oplus k),$$ 
$$E_{\mathbb R}(k\oplus k) \subset E_{\mathbb C}(k) \subset aSp_{\mathbb R}(2k), $$
where $E_{\mathbb R}(k\oplus k)=\mathbb R^{2k}\rtimes\tilde{O}(k)$, with $\tilde{O}(k)$ denoting the embedding of $O(k)$ into $Sp_{\mathbb R}(2k)$ obtained by diagonal action  on the $p$ and $q$ subspaces, and where $E_{\mathbb C}(k)=\mathbb C^k\rtimes U(k)$ is the ``complex Euclidean group'', which is the subgroup of $aSp_{\mathbb R}(2k)$ that preserves a complex structure on $\mathbb R^{2k}\simeq \mathbb C^k$, with $U(k)\simeq Sp_{\mathbb R}(2k)\cap O(2k)$.         

By first restricting attention to affine translations and noting that the group of translations of $\mathbb R^{2k}_{\omega}$ has the double dimension of $\mathbb R^k$, a natural question was how to extend the unitary action of $\mathbb R^k$ on $\mathcal H$ to a unitary action of  $\mathbb R^{2k}$ in a way to account for Heisenberg's canonical commutation relations $[\hat{p}_i,\hat{q}_j]=i\hbar\delta_{ij}$. 

A nice way to appreciate one solution to this problem  is by presenting the Heisenberg group $H_0^{2k}$, which is a $U(1)$ central extension of $\mathbb R^{2k}$, with product 
\begin{equation}\label{Heisenberggroup} (\xi_1,\exp(i\theta_1))\cdot(\xi_2,\exp(i\theta_2))=(\xi_1+\xi_2,\exp(i(\theta_1+\theta_2+(\omega(\xi_1,\xi_2)/2\hbar)))) \ , \ \end{equation}
where $\xi_j\in\mathbb R^{2k}$, $\exp(i\theta_j)\in U(1)$, $j=1,2$.  
Working with a unitary action of $H_0^{2k}$  on $\mathcal H$, and the Fourier transform, we can formally set up Weyl's  symbol correspondence between operators on $\mathcal H$ and functions on $\mathbb R^{2k}_{\omega}$, the latter ``almost identified'' with $H_0^{2k}$, as the commutator of  (\ref{Heisenberggroup}) does not  depend on phases $\theta_j$. 

More concretely, let $\xi=(a,b)\in\mathbb R^{2k}_{\omega}$, as in (\ref{translation}), and remind that  $\epsilon_k$ denotes the Euclidean metric on $\mathbb R^k_{\epsilon}$, in other words, the usual scalar product. Defining  
\begin{equation}\nonumber  T_0^{\hbar}(\xi,\exp(i\theta)) : {\mathcal H}\to{\mathcal H} \ ,\end{equation} 
\begin{equation}\label{Heisgrouprep}   \psi(q)\mapsto \psi(q-b)\exp(i\theta + \epsilon_k(a,q)/\hbar -\epsilon_k(a,b)/2\hbar) \ , 
\end{equation} 
one can verify that $T_0^{\hbar}(\xi,\exp(i\theta))$ is  an element in a unitary representation of $H_0^{2k}$, acting on $\mathcal H$ (see e.g. \cite{Lit0}).  Choosing $T_0^{\hbar}(\xi,1)$ to represent the equivalence class $[T_0^{\hbar}(\xi)]$ defined by the equivalence relation $T_0^{\hbar}(\xi,\exp(i\theta_1))\approx T_0^{\hbar}(\xi,\exp(i\theta_2))$, then 
via Fourier transform we can formally define the ``reflection'' operator  on $\mathcal H$ by   
\begin{equation}\label{reflectops} {R}_0^{\hbar}(x)  =  \frac{1}{(2\pi{\hbar})^k} \int {T}_0^{\hbar}({\xi},1)\exp(i\omega(x,\xi)/\hbar)d\xi \ : \ {\mathcal H}\to{\mathcal H} \ , \end{equation}
where $d\xi$ is the Liouville volume element on $\mathbb R^{2k}_{\omega}\ni x,\xi$. 

Then, if ${A}$ is an operator  on $\mathcal H$, the Weyl symbol of ${A}$ is the complex function $W^{\hbar}_A$ on $\mathbb R^{2k}_{\omega}$ formally defined by (see \cite{Royer}, also \cite{Ozo}):
\begin{equation}\label{Weylrule}    W^{\hbar}_A(x)=trace({A}{R}_0^{\hbar}(x)) \ , \end{equation} 
in such a way that, if the action of $A$ on $\mathcal H$ is described in integral form by  
\begin{equation}\label{Schwarzkernel} A\psi(q') =  \int {\mathcal S}_A^{\hbar}(q',q'')\psi(q'')d^kq'' \ , \end{equation} 
 then, writing $x=(p,q)$,  the integral kernel ${\mathcal S}_A^{\hbar}$ and the symbol $W^{\hbar}_A$ are related by 
\begin{equation}\label{Weyl-Schwarz-relation}  W^{\hbar}_A(p,q) = \frac{1}{(2\pi{\hbar})^k} \int {\mathcal S}_A^{\hbar}(q-v/2,q+v/2)\exp(i\epsilon_k(v,p)/\hbar)d^kv \ , \end{equation}
\begin{equation}\label{Schwarz-Weyl-relation} {\mathcal S}_A^{\hbar}(q',q'') = \frac{1}{(2\pi{\hbar})^k} \int  W^{\hbar}_A(p,(q'+q'')/2)\exp(i\epsilon_k(p,q'-q'')/\hbar)d^kp \ , \end{equation} 
where  $d^kq'', d^kv, d^kp$ denote the usual Euclidean volume $d^kp=dp_1dp_2\cdots dp_k$.

Whenever well defined, Weyl's correspondence, seen as a map $A\mapsto W^{\hbar}_A$ that assigns to an operator $A$ a unique function $W^{\hbar}_A$ on affine symplectic space, satisfies:
%\begin{center} 
$$ (i)\quad \text{it is a linear injective map}; $$  $$(ii)\quad W^{\hbar}_{A^*}=\overline{W^{\hbar}_A} \ , \ \text{where} \ A^* \ \text{is the adjoint of} \ A ;$$ 
$$ (iii)\quad \text{it is equivariant under action of} \  aSp_{\mathbb R}(2k),$$
 %\end{center}  
where  the action on functions is the usual one, while the action on operators is the effective action of $aSp_{\mathbb R}(2k)$ obtained via adjoint action of a ($\infty$-dimensional) unitary representation of the affine metaplectic group, the latter being a $U(1)$ central extension of $aSp_{\mathbb R}(2k)$ containing as subgroups $H^k_0$ and the metaplectic group $Mp(2k)$, which is a  double covering group of  $Sp_{\mathbb R}(2k)$ \cite{Veil};
\begin{equation}\label{Weyltrace} {(iv)} \quad trace(A)= \frac{1}{(2\pi{\hbar})^k} \int W^{\hbar}_A(x)dx \ , \end{equation}
where  again,  $dx$ stands for the Liouville volume on $\mathbb R^{2k}_{\omega}$, $dx=d^kpd^kq$. \ Furthermore, Weyl's symbol correspondence in fact also satisfies the stronger property:     
\begin{equation}\label{Weyltrace2} {(v)} \quad trace(A^*B)= \frac{1}{(2\pi{\hbar})^k} \int \overline{W^{\hbar}_A}(x)W^{\hbar}_B(x)dx \ . \end{equation}
 
Clearly, a symbol correspondence satisfying all above properties is a very powerful tool in relating the quantum and classical formalisms of affine mechanical systems, especially because property (v) allows us to compute quantum measurable quantities entirely within the classical formalism of functions on  affine symplectic space $\mathbb R^{2k}_{\omega}$ in a way that is equivariant under the full group of symmetries of this space, in accordance with  (iii). Furthermore, we can  ``import'' the product of operators to a  noncommutative associative product $\star$ of functions on  $\mathbb R^{2k}_{\omega}$,  
\begin{equation}\label{weylproduct} W^{\hbar}_{AB}=W^{\hbar}_A\star W^{\hbar}_B \ , \end{equation}
whose commutator is precisely the bracket originally introduced by Moyal \cite{Moy}. 

With such a powerful tool in hands, one can thus proceed to study quantum dynamics in the asymptotic limit of high quantum numbers, where classical Poisson dynamics should prevail, entirely within the mathematical framework of functions on  $\mathbb R^{2k}_{\omega}$. That is, in this way one first proceeds with  a ``dequantization'' of the quantum mechanical formalism and then study its semiclassical limit. 
For affine mechanical systems, this semiclassical limit is often studied by treating $\hbar$ formally as a variable and looking at the asymptotic  expressions for the symbols, their products and commutators, measurable quantities, etc, as $\hbar\to 0$.       

However, when trying to make a more rigorous mathematical sense of Weyl's symbol correspondence obtained via (\ref{Heisenberggroup})-(\ref{Schwarz-Weyl-relation}),  it became clear that one has to be very careful as to which classes of operators on $\mathcal H$ and  functions on  $\mathbb R^{2k}_{\omega}$ should be considered. In fact, although Weyl's correspondence has been presented as a map from operators to functions, in practice equations like  (\ref{Schwarz-Weyl-relation}), with (\ref{Schwarzkernel}), have often been used in the opposite direction, that is, of defining new classes of operators, as pseudo-differential or Fourier-integral operators (see Hormander \cite{Horm, Horm1, D-H}, or in the proper Weyl context see \cite{GLS} and also \cite{Voros} for asymptotics). Therefore, this inverse direction of using a symbol correspondence, often also referred to as ``quantization'', has actually become more familiar to many people. 

Moreover,  it was soon realized that the symbol correspondence rule obtained via (\ref{Heisenberggroup})-(\ref{Schwarz-Weyl-relation}) is not unique (and it was not the one originally used by Hormander). In fact, note that choosing $T_0^{\hbar}(\xi,1)$ to represent the equivalence class $[T_0^{\hbar}(\xi)]$ seems to be arbitrary, but more importantly, note also that  
$\mathbb R^{2k}_{\omega}$ can  be more generally  ``almost identified'' with a large  family of Heisenberg groups $H^{2k}_{\alpha}$, which are  defined by modifying the product  (\ref{Heisenberggroup}) to the more general one given by 
$$(\xi_1,\exp(i\theta_1))\cdot(\xi_2,\exp(i\theta_2))=(\xi_1+\xi_2,\exp(i(\theta_1+\theta_2+([\omega+\alpha](\xi_1,\xi_2)/2\hbar)))) \ ,$$  
where $\alpha$ is a symmetric bilinear form on $\mathbb R^{2k}$ and different choices of $\alpha$ are related to different choices of ``orderings'' for products of operators $\hat{p}_j$ and $\hat{q}_j$. In this respect, Weyl's ordering is the symmetric one, $p_jq_j\leftrightarrow(\hat{p}_j\hat{q}_j+\hat{q}_j\hat{p}_j)/2$, but other popularly used orderings are $p_jq_j\leftrightarrow\hat{p}_j\hat{q}_j$ and $p_jq_j\leftrightarrow\hat{q}_j\hat{p}_j$ (normal ordering). Similarly, one can define different symbol correspondences using complex coordinates $z_j=q_j+ip_j$ and $\bar{z}_j= q_j-ip_j$ on $\mathbb R^{2k}_{\omega}\simeq \mathbb C^k$, by considering different orderings for products of $a_j= \hat{q}_j+i\hat{p}_j$ and ${a}_j^{\dagger}=\hat{q}_j-i\hat{p}_j$, or by using coherent states, etc... 

Furthermore, for  affine mechanical systems it is not obvious which group to impose equivariance for all possible symbol correspondences. If ${W^{\hbar}}'$ is another symbol correspondence  satisfying (i)-(iv) but not (v), then from (\ref{weylproduct})  we see that it also satisfies the weaker property 
\begin{equation}\label{anytrace} (v') \quad trace(A^*B)= \frac{1}{(2\pi{\hbar})^k} \int {\overline{W^{\hbar}_A}}'\star {W^{\hbar}_B}' (x) dx \ , \end{equation} 
so that this property (v') can be used, instead of (v), to compute quantum expectation values within the classical formalism of functions on  affine symplectic space.  Then, one can consider relaxing (iii) to an equivariance under some 
subgroup of  $aSp_{\mathbb R}(2k)$ containing $E_{\mathbb R}(k)$ and all affine translations, 
like $E_{\mathbb R}(k\oplus k)$ or $E_{\mathbb C}(k)$, 
as long as (v') is still invariant under the whole affine symplectic group  $aSp_{\mathbb R}(2k)$ (symbol correspondences define via $\alpha\neq 0$, like Hormander's normal ordering, are usually not equivariant under the full group $aSp_{\mathbb R}(2k)$ in the strong sense of (iii)).

In other words, there are many other symbol correspondences in affine mechanical systems satisfying all or most of properties (i)-(v) above and it would be desirable to  classify all such correspondences, study their semiclassical asymptotic limit, see how these correspondences agree or disagree in this limit, etc...  

For affine mechanical systems, such a complete and systematic study is not known to us. However, in contrast to quantum affine mechanical systems, for quantum spin-j systems the Hilbert spaces are finite dimensional, allowing for an independent mathematical formulation of such systems, as was done in Chapter  \ref{quantumspinsystems}, and there is never any ambiguity about which classes of operators to consider. Moreover, 
 for spin systems there is no doubt about the natural group of symmetries to impose equivariance: 
 the $3$-sphere  $SU(2)$, acting effectively via $SO(3)$. 

Therefore, going back in time over two millennia, so to speak, and replacing straight lines and $k$-planes by circles and spheres as the fundamental geometrical objects, 
we can provide the complete classification and a systematic study of symbol correspondences for spin systems, as is presented below in 
Chapters \ref{Chaptercorrespondences}, \ref{multisymbols} and \ref{AsympChapter}.  
But it is somewhat curious, perhaps, that while ancient Greek philosophers looked outwards to the far sky in search of circles and spheres, modern physicists found them by looking deeply inwards into matter.

%----------------------------------------------------------------------------------------------------------------------------------------------------------------------------------------------------------------------
%-------------------------------------------------------------------------------------- Chapter 6 -----------------------------------------------------------------------------------------------------------------
%----------------------------------------------------------------------------------------------------------------------------------------------------------------------------------------------------------------------

\chapter{Symbol correspondences for a spin-j system}\index{Symbol correspondences |(}
\label{Chaptercorrespondences}

In Chapters \ref{quantumspinsystems} and \ref{classicalspinsystems}, quantum spin-j mechanical 
systems and the classical spin mechanical system, respectively, were defined and studied in fully 
independent ways.  In this chapter, the two formulations are brought together via spin-j symbol correspondences. 
Inspired by Weyl's correspondence in affine mechanical systems we investigate,
for a spin-j system, symbol correspondences that associate operators on
Hilbert space to functions on phase space, satisfying certain properties. 

Here 
we define, classify and study such correspondences, presenting explicit
constructions. 
Our cornerstone is the concept of characteristic numbers of a
symbol correspondence, cf. Definition \ref{charact}, which provides coordinates 
on the moduli space of spin-j symbol correspondences. As we shall see below,  
for any $j$ a (quite smaller) subset of characteristic numbers can be 
distinguished in terms of a stricter requirement for an isometric correspondence. However, 
a more subtle distinction is obtained in the asymptotic limit $n=2j\to\infty$, 
to be explored in  Chapter \ref{AsympChapter}.  
As we shall see there in detail, not all $n$-sequences of characteristic numbers lead to classical Poisson dynamics in this asymptotic limit.

%-------------------------------------------------------------------------------------- Section 6.1 -----------------------------------------------------------------------------------------------------------------

\section{General symbol correspondences for a spin-j system}

\subsection{Definition of spin-j symbol correspondences}

Following Stratonovich, Varilly and Gracia-Bondia (cf. \cite{Strat, VG-B}), in the spirit of Weyl we introduce the following main definition:

\begin{definition}\index{Symbol correspondences ! general definition} 
\label{symbol corr}A \emph{symbol correspondence} for a spin-j quantum mechanical system
$\mathcal{H}_{j}$ $\simeq\mathbb{C}^{n+1}$, where $n=2j$, is a rule which
associates to each operator $P\in\mathcal{B}(\mathcal{H}_{j})$ a smooth
function $W_{P}^{j}$ on the 2-sphere $S^{2}$ with the following properties:%
\begin{equation}%
\begin{array}
[c]{ll}%
(i) & \text{Linearity : The map }P\rightarrow W_{P}^{j}\text{ is linear and
injective}\\
(ii) & \text{Equivariance : }W_{P^{g}}^{j}=(W_{P}^{j})^{g}\text{, for each
}g\in SO(3)\\
(iii) & \text{Reality}\ \text{: \ }W_{P^{\ast}}^{j}(\mathbf{n})=\overline
{W_{P}^{j}(\mathbf{n})}\\
(iv) & \text{Normalization : }\frac{1}{4\pi}%
%TCIMACRO{\tint \limits_{S^{2}}}%
%BeginExpansion
{\textstyle\int\limits_{S^{2}}}
%EndExpansion
W_{P}^{j}dS=\frac{1}{n+1}trace(P)\text{ \ }%
\end{array}
\label{axiom1}%
\end{equation}
\end{definition}

In this way, we characterize a distinguished family of \emph{symbol maps}:
\[
W^{j}:\mathcal{B}(\mathcal{H}_{j}) \to C_{\mathbb{C}}^{\infty}(S^{2}) .
\]
However, the linear injection requirement in (i) can be transformed to a
linear bijection requirement by reducing the target space of every symbol map
\begin{equation}
W^{j}:\mathcal{B}(\mathcal{H}_{j}) \simeq M_{\mathbb{C}}(n+1)\rightarrow
Poly_{\mathbb{C}}(S^{2})_{\leq n}\subset C_{\mathbb{C}}^{\infty}(S^{2})
\label{symbol}%
\end{equation}
which associates to each operator $P$ its \emph{symbol} $W^{j}_{P}$. We shall
assume the symbols $W^{j}_{P}$ are polynomial functions, as indicated in
(\ref{symbol}), unless otherwise stated.

\begin{remark}\index{Symbol correspondences ! isometry condition}
\label{axiom}In addition to the above four axioms (\ref{axiom1}) one may also
impose the axiom
\begin{equation}%
\begin{array}
[c]{cc}%
(v) & \text{Isometry : }\left\langle W_{P}^{j},W_{Q}^{j}\right\rangle
=\left\langle P,Q\right\rangle _{j}%
\end{array}
\label{axiom2}%
\end{equation}
as a \textquotedblleft metric normalization\textquotedblright\ condition,
where the right side of equation (\ref{axiom2}) is the normalized
Hilbert-Schmidt inner product of two operators, given by
\begin{equation}
\left\langle P,Q\right\rangle _{j}=\frac{1}{n+1}\left\langle P,Q\right\rangle
=\frac{1}{n+1}trace(P^{\ast}Q)\ , \label{innerprod}%
\end{equation}
$n=2j$, and the left side of (\ref{axiom2}) is the normalized $L^{2}$ inner
product of two functions on the sphere, given by (\ref{L2}). Thus, condition
(iv) is just a special case of (v), namely (iv) can be stated as
\[
(iv)^{\prime}\text{ \ }\left\langle 1,W_{P}^{j}\right\rangle =\left\langle
I,P\right\rangle _{j}\ .
\]

\end{remark}

Berezin \cite{Berezin1, Berezin2, Berezin} was the first to investigate symbol correspondences for
spin-j systems in a more systematic way. His correspondence satisfies axioms $(i)-(iv)$, but not axiom
$(v)$. Varilly and Gracia-Bondia \cite{VG-B} were the first to systematically  investigate the
rules $P\mapsto W^{j}_{P}$ satisfying all five axioms, as previously
outlined by Stratonovich \cite{Strat} for the spin-j-system version of the
Weyl correspondence.

\begin{definition}\index{Symbol correspondences ! Stratonovich-Weyl symbol }
A \emph{Stratonovich-Weyl correspondence} is a symbol correspondence that also
satisfies the isometry axiom (v).
\end{definition}

\begin{remark}
\label{expectations} Varilly and Gracia-Bondia's justification for axiom (v)
is the need \textquotedblleft to assure that quantum mechanical expectations
can be calculated by taking integrals over the sphere\textquotedblright.
However, this justification is not entirely proper, as these latter integrals
can be defined for any symbol correspondence (cf. equation (\ref{CCduality})
and Remark \ref{CCdual}, as well as Remark \ref{inducedinnprod} and equation
(\ref{inducedinnerproduct}), below). This situation mimics, in fact, the situation for symbol correspondences in affine mechanical systems: while Weyl's symbol correspondence satisfies property (v) given in (\ref{Weyltrace2}), many other  useful symbol correspondences in affine mechanical systems do not.  
\end{remark}

\subsection{The moduli space of spin-j symbol correspondences}

We recall that the action of $SU(2)$ on operators in $M_{\mathbb C}(n+1)$, as explained in Chapter 2, is by
conjugation via the representation $\varphi_{j}:SU(2)\rightarrow SU(n+1)$,
\[
g\in SU(2):A\rightarrow A^{g}=\varphi_{j}(g)A\varphi_{j}(g^{-1}), 
\]
which factors to an effective action of $SO(3)$, while the action of $SU(2)$ on polynomial functions in $Poly_{\mathbb C}(S^2)_{\leq n}$, as explained in Chapter 3, is obtained from the standard action of $SO(3)$ on the two-sphere:   
\[
g\in SU(2): F\rightarrow F^g \ , \ F^g(\mathbf n)=F(g^{-1}\mathbf n), 
\]
so that in both spaces we have an effective left action of $G=SO(3)$. 

Furthermore, by Proposition \ref{splitting} these two spaces are isomorphic and have the same splitting into $G$-invariant subspaces, 
both in the complex and the real case. Therefore, by the classical Schur's lemma (Lemma \ref{Schur's lemma}) applied to (\ref{PolyC}), \ 
\begin{equation}
Hom^{G}(M_{\mathbb C}(n+1),Poly_{\mathbb{C}}(S^{2})_{\leq n})\simeq%
%TCIMACRO{\dprod \limits_{l=0}^{n}}%
%BeginExpansion
{\displaystyle\prod\limits_{l=0}^{n}}
%EndExpansion
Hom^{G}(M_{\mathbb C}(\varphi_{l}),Poly(\varphi_{l}))\simeq\mathbb{C}^{n+1}\text{,}
\label{SchurC}%
\end{equation}
and similarly, by restricting to the real case  (\ref{polyR}), \
\begin{equation}
Hom^{G}(\mathcal{H}(n+1),Poly_{\mathbb{R}}(S^{2})_{\leq n})\simeq%
%TCIMACRO{\dprod \limits_{l=0}^{n}}%
%BeginExpansion
{\displaystyle\prod\limits_{l=0}^{n}}
%EndExpansion
Hom^{G}(\mathcal{H}(\psi_{l}),Poly(\psi_{l}))\simeq\mathbb{R}^{n+1}\text{.}
\label{Schur}%
\end{equation}

Now, note that (i) $+$ (ii) in Definition \ref{symbol corr} is the condition that $W^{j}$ is a 
$G$-isomorphism, and (iii) assures that $W^{j}$ preserves the real structure.
Therefore, we have the following:

\begin{corollary}
\label{C}Each $\mathbb{R}$-linear $G$-map which takes Hermitian matrices to
real polynomials may be identified with a unique $(n+1)$-tuple
\begin{equation}
(c_{0},c_{1},...,c_{n})\in\mathbb{R}^{n+1} \label{tuple}%
\end{equation}
In particular, the tuple corresponds to a $G$-isomorphism
\begin{equation}%
\begin{array}
[c]{ccc}%
\mathcal{H}(n+1)\  & \simeq & Poly_{\mathbb{R}}(S^{2})_{\leq n}\\
\cap &  & \cap\\
M_{\mathbb{C}}(n+1)\  & \simeq & Poly_{\mathbb{C}}(S^{2})_{\leq n}%
\end{array}
\ \label{Equiv}%
\end{equation}
if and only if each $c_{l}\neq0 \ , \ l=0,1,...,n$ .
\end{corollary}

\begin{remark}
\label{multitude} The above correspondence is not canonical since there is no
natural choice of isomorphism by which one may identify matrices with
functions on $S^{2}$ equivariantly. This multitude of choices is a central
topic in the next section, where the numbers $c_{l}$ will
be defined in a precise way.
\end{remark}

But note also that the symbol $W_{I}^{j}$ of the identity operator is a
constant function, say equal to $c_{0}$, and then condition (iv) clearly
implies
\begin{equation}
c_{0}=1 \ , \label{c0=1}%
\end{equation}
namely, the symbol map respects the unit elements of the two rings in
(\ref{symbol}). In this way, we can identify each $W^{j}$ by its real
$n$-tuple representation:
\begin{equation}
W^{j}\leftrightarrow(c_{1},...,c_{n})\in(\mathbb{R}^{\ast})^{n}\text{, where
}\mathbb{R}^{\ast}=\mathbb{R}-\left\{  0\right\}  \text{,} \label{Wc}%
\end{equation}
(cf. Remark \ref{multitude}). For an explicit definition of the numbers
$c_{l}$ we refer to Theorem \ref{kernel} and Definition \ref{charact} below. 
To summarize, we have obtained the following: 
\begin{theorem}\index{Symbol correspondences ! moduli space} The moduli space of all spin-j symbol correspondences satisfying 
conditions $(i)-(iv)$ in Definition \ref{symbol corr} is $(\mathbb R^*)^n$. \end{theorem}

Note that nothing has yet been said with respect to the isometric spin-j symbol correspondences, 
i.e. the Stratonovich-Weyl symbol correspondences, which also satisfy condition $(v)$. 

In our general setting to be developed below, we fix a representation (\ref{Wc}), either  by coupled basis decomposition or more systematically  in
terms of operator kernels, and
relate the metric properties of $W^{j}$ to the numbers $c_{l}$ (further below,
we also present a direct way to define the symbol map introduced by Berezin). 

Then, it will be seen below that condition $(v)$ determines each $c_{l}$ up to sign.

%-------------------------------------------------------------------------------------- Section 6.2 -----------------------------------------------------------------------------------------------------------------

\section{Explicit constructions of spin-j symbol correspondences}

We consider the general category of symbol correspondences. As a basis point
on $S^{2}$ we choose the north pole $\mathbf{n}_{0}=(0,0,1)$, assuming that
its isotropy group is the circle group $U(1)$ $\subset SU(2)$ in (\ref{U(1)})
whose fixed point set in $M_{\mathbb{C}}(n+1)$ consists of the diagonal
matrices, namely the $m=0$ eigenspace of $J_{3}$ (cf. equation (\ref{action}) 
and the discussion following Definition \ref{defmodelused}). From now on, we shall often write
$W$ for $W^{j}$, for simplicity, whenever suppressing the spin number $j=n/2$
is not a cause for confusion.

\subsection{Symbol correspondences via coupled standard basis}\label{scvcsb}

Following on the same reasoning that led to Corollary \ref{C} above, a natural and simple 
way to establish a symbol correspondence between operators and polynomial functions 
is obtained by relating appropriately chosen basis for each space. 

Thus, remind that for  a given $n=2j$, the operator space $\mathcal{B}(\mathcal{H}_{j})$ has the
orthogonal decomposition
\[
M_{\mathbb{C}}(n+1)=%
%TCIMACRO{\dsum \limits_{l=0}^{n}}%
%BeginExpansion
{\displaystyle\sum\limits_{l=0}^{n}}
%EndExpansion
M_{\mathbb{C}}(\varphi_{l})
\]
where each summand $M_{\mathbb{C}}(\varphi_{l})$ has its standard basis
$\mathbf{e}^{j}(l,m),-l\leq m\leq l$, in accordance with Proposition
\ref{standard basis} and Theorem \ref{explicitbasis}.

On the other hand, the space of polynomial functions on $S^{2}$ of proper
degree $\leq n$ has the orthogonal decomposition
\[
Poly_{\mathbb{C}}(S^{2})_{\leq n}= \displaystyle{\sum_{l=0}^{n}}
Poly(\varphi_{l})
\]
where each summand $Poly(\varphi_{l})$ of polynomials of proper degree $l$ has
its standard basis of spherical harmonics $Y_{l,m}, -l\leq m\leq l$.

Consequently, for a given value of $n$ and signs $\varepsilon^{n}_{l}=\pm1$,
$l=1,2,..,n$, we can set up the 1-1 correspondence%
\begin{equation}
 \mu_{0}\mathbf{e}^{j}(l,m)\longleftrightarrow\varepsilon^{n}_{l}%
Y_{l,m}\text{ };\text{ }-l\leq m\leq l\leq n,\text{ }\ \mu_{0}=\sqrt{n+1}.
\label{symbol1}%
\end{equation}
Such a correspondence is obviously isometric, therefore, if it extends linearly to a symbol correspondence in the sense of Definition
\ref{symbol corr}, then clearly all the $2^{n}$
Stratonovich-Weyl symbol correspondences are obtained in this way.
In fact, by relaxing on the isometric requirement and allowing for scaling freedom, we have the more general result, whose formal proof is deferred to the next section:

\begin{proposition}\label{coupledcorrespondence}
Any symbol correspondence $W^j$ satisfying Definition \ref{symbol corr} is uniquely determined by non-zero real numbers
$c_{l}^{n}$, $l=1,2,..,n$, which yield the explicit 1-1 correspondence%
\begin{equation}
W^j=W_{\vec{c}}^{j}:\mu_{0}\mathbf{e}^{j}(l,m)\mapsto c_{l}%
^{n}Y_{l,m}\text{ };\text{ }-l\leq m\leq l\leq n,\text{ }c_{l}^{n}%
\in\mathbb{R}^{\ast} , \label{symbol2}%
\end{equation}
where $\vec{c}$ is a shorthand notation for the $n$-string $(c_1^n,c_2^n,...,c_n^n)$. 
Furthermore, $W^j$ is a Stratonovich-Weyl symbol correspondence if and only if $c^n_l=\varepsilon^n_l=\pm 1$, $l=1,2,...,n$.
\end{proposition}
\begin{remark}\index{Symbol correspondences ! moduli space}
\label{explcharact} We can now understand better the moduli space of symbol
correspondences $W$. Starting with the isometric ones, recall from Remark \ref{phaseY} that
each respective standard basis of $M_{\mathbb{C}}(\varphi_{l})$ and
$Poly(\varphi_{l})$ is uniquely defined up to an arbitrary overall phase.
Therefore, an isometric correspondence between $M_{\mathbb{C}}(\varphi_{l})$
and $Poly(\varphi_{l})$ is uniquely defined modulo a relative phase $z_{l}\in
S^{1}\subset\mathbb{C}$. However, the reality requirement $(iii)$ in
Definition \ref{symbol corr} fixes this phase to be real, thus $z_{l}%
=\varepsilon_{l}=\pm1$. For the non-isometric correspondences, we have the
further freedom of a relative scaling $\rho_{l}\in\mathbb{R}^{+}$ and,
therefore, a correspondence between $M_{\mathbb{C}}(\varphi_{l})$ and
$Poly(\varphi_{l})$ is uniquely defined modulo a number $c_{l}=\rho
_{l}\varepsilon_{l}\in\mathbb{R}^{\ast}$. 
\end{remark}

\subsection{Symbol correspondences via operator kernel}

In order to study general symbol correspondences more systematically, observe that a diagonal matrix $K$ gives rise to a function $K(\mathbf{n})$ on
$S^{2}$ such that $K(\mathbf{n}_{0})=K$, and $K(\mathbf{n})=K^{g}$ for
$\mathbf{n}=g\mathbf{n}_{0}$.

\begin{proposition}
\label{OK1} For each symbol correspondence $W=W^j$ there is a unique operator
$K\in M_{\mathbb{C}}(n+1)$ such that
\begin{equation}
W_{P}(g\mathbf{n}_{0})=trace(PK^{g}) \label{W}%
\end{equation}
or equivalently,
\[
W_{P}(\mathbf{n})=trace(PK(\mathbf{n}))=\left\langle P^{\ast},K(\mathbf{n}%
)\right\rangle
\]
Moreover, $K$ is a diagonal matrix with real entries and trace 1.
\end{proposition}

\begin{proof}
The linear functional
\[
\hat{W}:M_{\mathbb{C}}(n+1)\rightarrow\mathbb{C}\text{, \ }P\mapsto
W_{P}(\mathbf{n}_{0})
\]
is represented by some $K$ such that $\hat{W}(P)=trace(PK)$, since the pairing
$\left\langle K^{\ast},P\right\rangle =trace(PK)$ is a Hermitian inner
product. Therefore, (\ref{W}) holds for $g=1$ and hence also in general by
equivariance
\[
W_{P}(g^{-1}\mathbf{n}_{0})=(W_{P})^{g}(\mathbf{n}_{0})=W_{P^{g}}%
(\mathbf{n}_{0})=trace(P^{g}K)=trace(PK^{g^{-1}})
\]
\qquad On the other hand, for $g\in U(1)$ we have $g\mathbf{n}_{0}%
=\mathbf{n}_{0}$ and then $trace(PK^{g})=trace(PK)$ holds for each $P$,
consequently $K^{g}=K.$ That is, $K$ is fixed by $U(1)$ and hence
$K=diag(\lambda_{1},..,\lambda_{n+1}).$ By choosing $P$ to be the one-element
matrix $\mathcal{E}_{kk}$ it follows that $\lambda_{k}=W_{P}(\mathbf{n}_{0})$
is a real number, due to the reality condition (iii). Finally, by the
normalization condition (iv) $W_{I}=1$ and hence $trace(K)=1.$ \ 
\end{proof}

\begin{definition}\index{Symbol correspondences ! operator kernel}
\label{OK} An operator kernel $K\in M_{\mathbb{C}}(n+1)$ is a diagonal matrix
with the property that the symbol map $W$ defined by (\ref{W}) is a symbol correspondence.
\end{definition}

If follows from Proposition \ref{OK1} that $K$ has an orthogonal
decomposition
\begin{equation}
K=\frac{1}{n+1}I+K_{1}+K_{2}+...+K_{n} \label{Kdec}%
\end{equation}
where for $l\geq1$ each $K_{l}$ is a zero trace real diagonal matrix belonging
to the zero weight (or $m=0)$ subspace of the irreducible summand
$M_{\mathbb{C}}(\varphi_{l})$ of the full matrix space (\ref{sum}), namely
$K_{l}=k_{l}e(l,0)$ for some nonzero real number $k_{l}$.

Conversely, each matrix $K$ of this kind defines a symbol map $P\rightarrow
W_{P}$ by (\ref{W}) whose kernel (as a linear map) is a $G$-invariant
subspace, namely the sum of those $M_{\mathbb{C}}(\varphi_{l})$ for which
$K_{l}=0.$ Therefore, by axiom (i) in (\ref{axiom1}) each $K_{l}\ $must be
nonzero for a symbol correspondence.

\begin{lemma}
Let $K$ be a real diagonal matrix with $trace(K)=1$, and define a symbol map
$W$ by the formula (\ref{W}). Then $W$ satisfies the normalization condition,
namely%
\begin{equation}
\frac{1}{4\pi}\int_{S^{2}}W_{P}dS=\frac{1}{n+1}trace(P) \label{trace2}%
\end{equation}

\end{lemma}

\begin{proof}
Let (\ref{Kdec}) be the orthogonal decomposition of $K$, where $K_{l}\in
M_{\mathbb{C}}(\varphi_{l})$ may possibly be zero, and let $dg$ denote the
normalized measure on $SU(2)$. Then%
\begin{align*}
&  \frac{1}{4\pi}\int_{S^{2}}W_{P}dS =\int_{SU(2)}W_{P}(g\mathbf{n}%
_{0})dg=\int_{SU(2)}trace(PK^{g})dg =\\
&  =trace\left(  P\int_{SU(2)}K^{g}dg\right)  =trace\left[  P\left(  \frac
{1}{n+1}I+\sum_{l=1}^{n}\left(  \int_{SU(2)}\tilde{\varphi}_{l}(g)dg\right)
K_{l}\right)  \right]  ,
\end{align*}
where the operator $\tilde{\varphi}_{l}(g)\in GL(M_{\mathbb{C}}(\varphi_{l}))
$ is the action of $g$ on the vector space $M_{\mathbb{C}}(\varphi_{l})$, in
particular
\[
\tilde{\varphi}_{l}(g)K_{l}=K_{l}^{g}=\varphi_{l}(g)K_{l}\varphi_{l}(g)^{-1}%
\]
Since the representation $g\rightarrow\tilde{\varphi}_{l}(g)$ is irreducible,
it follows by standard representation theory that
\[
\int_{SU(2)}\tilde{\varphi}_{l}(g)dg=0\text{,}%
\]
and this proves the identity (\ref{trace2}).
\end{proof}

Putting together the above results yield the following classification of all
possible symbol correspondences:

\begin{theorem}\index{Symbol correspondences ! operator kernel}
\label{kernel}The construction of symbol maps $W=W^j$ in terms of operator kernels
$K=K^j$ by the formula (\ref{W}) establishes a bijection between symbol
correspondences and real diagonal matrices in $M_{\mathbb{C}}(n+1)$ of type
\begin{equation}
K^j=\frac{1}{n+1}I+%
%TCIMACRO{\tsum \limits_{l=1}^{n}}%
%BeginExpansion
{\textstyle\sum\limits_{l=1}^{n}}
%EndExpansion
c_{l}\sqrt{\frac{2l+1}{n+1}}\mathbf{e}^j(l,0),\text{ \ }c_{l}\neq0\text{ real,}
\label{K}%
\end{equation}
where $\mathbf{e}^j(l,0) \in$ $M_{\mathbb{C}}(\varphi_{l})$ is the traceless
diagonal matrix of unit norm given by
\begin{equation}
\label{explicitelm}\mathbf{e}^j(l,0)= \frac{(-1)^{l}}{l!}\sqrt{2l+1}\sqrt
{\frac{(n-l)!}{(n+l+1)!}} \ {\displaystyle\sum\limits_{k=0}^{l}}(-1)^{k}%
\binom{l}{k}J_{-}^{l-k}J_{+}^{l}J_{-}^{k} \
\end{equation}
(cf. Theorem \ref{explicitbasis}). In particular, the symbol correspondence
$W^j$ defined by formula (\ref{W}) is determined by the $n$-tuple $(c_{1}%
,c_{2},...,c_{n})\in(\mathbb{R}^{\ast})^{n}$.
\end{theorem}

\begin{definition}\index{Symbol correspondences ! characteristic numbers ! of operator kernel}
\label{charactK}The non-zero real numbers $c_{1},c_{2},...,c_{n}$ will be
referred to as the \emph{characteristic numbers }of the operator kernel $K$.
\end{definition}

In the general context of Theorem \ref{kernel}, we now distinguish the symbol 
correspondences of the kind originally defined by Berezin, as follows. 

\begin{definition}\index{Symbol correspondences ! Berezin symbol }
A \emph{Berezin symbol correspondence} is a symbol correspondence whose
operator kernel $K$ is a projection operator $\Pi$ (cf. e.g. 
\cite{Rieffel}).
\end{definition}

Since the trace of a projection operator is its rank, it follows that $\Pi$
must be a one-element matrix $\Pi_{k}=$ $\mathcal{E}_{kk}$ for some $1\leq k$
$\leq n+1=2j+1$. Namely, the expansion (\ref{K}) of $\Pi_{k}$ in the coupled
standard basis reads
\[
\Pi_{k}\ =(-1)^{k+1}\left\vert jmj(-m)\right\rangle =(-1)^{k+1}%
%TCIMACRO{\tsum \limits_{l=0}^{n}}%
%BeginExpansion
{\textstyle\sum\limits_{l=0}^{n}}
%EndExpansion
C_{m,-m,0}^{\text{ }j,\text{\ }j,l}\mathbf{e}(l,0),\text{ \ }m=j-k+1
\]
(cf. (\ref{uncouple}) and (\ref{Clebsch2})) and hence each Clebsch-Gordan
coefficient in the above sum must be nonzero in the Berezin case. We state
this result as follows:

\begin{proposition}
For a spin-j quantum system $\mathcal{H}_{j}$ $=\mathbb{C}^{n+1}$, the Berezin
symbol correspondences are characterized by having as operator kernel a
projection $\Pi_{k}=\mathcal{E}_{kk},$ $1\leq k\leq n+1$, for those $k$ such
that the following Clebsch-Gordan coefficients are nonzero :
\begin{equation}
C_{m,-m,0}^{\text{ }j,\text{ }j,l}\neq0,\text{ \ }m=j-k+1,1\leq l\leq
n=2j\text{ } \label{Berez}%
\end{equation}
Equivalently, the $k^{th}$ entry of each diagonal matrix $\mathbf{e}(l,0)$,
$l=1,2,...,n$, must be nonzero.
\end{proposition}

\begin{remark}
\label{possibleBerezin} The traceless matrix $\mathbf{e}(l,0)$ has the
following symmetry
\[
\mathbf{e}(l,0)=diag(d_{1},d_{2},...,d_{n+1}),\text{ }d_{i}=\pm d_{n+2-i}%
\]

When $n=2j$ is odd, namely half-integral spin $j$, we claim that the Berezin
condition (\ref{Berez}) holds for each $k$. We omit the proof, but for $n<20$
say, the above matrices $\mathbf{e}(l,0)$ can be calculated easily using
computer algebra. However, for integral values of $j$ the Berezin condition
fails for the projection operator $\Pi_{k}$ when $k=n/2+1$. Moreover,
$d_{2}=0$ holds for $n=l(l+1)=6,12,20,..$, so in these dimensions the Berezin
condition also fails for $\Pi_{2}$ and $\Pi_{n}$.

On the other hand, $d_{1}\neq0$ always holds, and consequently the Berezin
condition (\ref{Berez}) is satisfied for $k=1$ and $k=n+1$. Therefore, for all
$n=2j\in\mathbb{N}$, the operators $\Pi_{1}$ and $\Pi_{n+1}$ always yield
Berezin symbol correspondences. $\ $
\end{remark}

\begin{definition}\index{Symbol correspondences ! Berezin symbol ! standard }
\label{standardBerezin} The symbol obtained via the projection operator
$\Pi_{1}$ will be called the \emph{standard Berezin symbol}.
\end{definition}

\begin{remark} The standard Berezin symbol correspondence generalizes to spin systems the 
method of correspondence via ``coherent states''  originally introduced for ordinary quantum 
mechanics \cite{Glau, Sud}. In fact, this method can be applied in more general 
settings (a necessary condition is that the phase space be a K\"{a}hler manifold, cf. e.g. \cite{Per}) 
and the original papers by Berezin were already cast in the more general 
context of complex symmetric spaces   \cite{Berezin1, Berezin2, Berezin}. 
\end{remark}

\subsubsection{General metric relation}

Now, the following proposition 
sets a general \textquotedblleft metric\textquotedblright\ relation, similar
to axiom (v) in Remark \ref{axiom} which is valid for all symbol correspondences.

Let $K$ in (\ref{Kdec}), (\ref{K}) be given, write $\mathbf{e}(0,0)=\frac
{1}{\sqrt{n+1}}I$, let $P\in M_{\mathbb{C}}(n+1)$ and consider the orthogonal
decompositions
\begin{equation}
P=%
%TCIMACRO{\tsum \limits_{l=0}^{n}}%
%BeginExpansion
{\textstyle\sum\limits_{l=0}^{n}}
%EndExpansion
P_{l}=%
%TCIMACRO{\tsum \limits_{l=0}^{n}}%
%BeginExpansion
{\textstyle\sum\limits_{l=0}^{n}}
%EndExpansion%
%TCIMACRO{\tsum \limits_{m=-l}^{l}}%
%BeginExpansion
{\textstyle\sum\limits_{m=-l}^{l}}
%EndExpansion
a_{lm}\mathbf{e}(l,m)\text{, \ }K=%
%TCIMACRO{\dsum \limits_{l=0}^{n}}%
%BeginExpansion
{\displaystyle\sum\limits_{l=0}^{n}}
%EndExpansion
K_{l}=\sum_{l=0}^{n}\gamma_{l}\mathbf{e}(l,0)\text{\ } \label{A}%
\end{equation}

\begin{proposition}\index{Symbol correspondences ! general metric relation}
\label{metricrelation} Any symbol correspondence $W$ satisfies the metric
identity
\begin{equation}
\left\langle W_{P},W_{Q}\right\rangle =%
%TCIMACRO{\tsum \limits_{l=0}^{n}}%
%BeginExpansion
{\textstyle\sum\limits_{l=0}^{n}}
%EndExpansion
\frac{\gamma_{l}^{2}}{2l+1}\left\langle P_{l},Q_{l}\right\rangle =%
%TCIMACRO{\tsum \limits_{l=0}^{n}}%
%BeginExpansion
{\textstyle\sum\limits_{l=0}^{n}}
%EndExpansion
\frac{(c_{l})^{2}}{n+1}\left\langle P_{l},Q_{l}\right\rangle \label{isom}%
\end{equation}
where the $\gamma_{l}$ and $c_{l}$ are related by
\[
\gamma_{l}=c_{l}\sqrt{\frac{2l+1}{n+1}},\text{ cf. }(\ref{K}),(\ref{A}%
),\text{Theorem \ref{kernel} }%
\]
\end{proposition}

\begin{proof}
From (\ref{W}) and (\ref{A}), we have that
\begin{align*}
W_{P}(g\mathbf{n}_{0})  &  =%
%TCIMACRO{\tsum \limits_{l=0}^{n}}%
%BeginExpansion
{\textstyle\sum\limits_{l=0}^{n}}
%EndExpansion
trace(P_{l}K^{g})=%
%TCIMACRO{\tsum \limits_{l=0}^{n}}%
%BeginExpansion
{\textstyle\sum\limits_{l=0}^{n}}
%EndExpansion%
%TCIMACRO{\tsum \limits_{l^{\prime}=0}^{n}}%
%BeginExpansion
{\textstyle\sum\limits_{l^{\prime}=0}^{n}}
%EndExpansion%
%TCIMACRO{\tsum \limits_{m=-l}^{l}}%
%BeginExpansion
{\textstyle\sum\limits_{m=-l}^{l}}
%EndExpansion
a_{lm}trace(\mathbf{e}(l,m)K_{l^{\prime}}^{g})\\
&  =%
%TCIMACRO{\tsum \limits_{l=0}^{n}}%
%BeginExpansion
{\textstyle\sum\limits_{l=0}^{n}}
%EndExpansion%
%TCIMACRO{\tsum \limits_{l^{\prime}=0}^{n}}%
%BeginExpansion
{\textstyle\sum\limits_{l^{\prime}=0}^{n}}
%EndExpansion%
%TCIMACRO{\tsum \limits_{m=-l}^{l}}%
%BeginExpansion
{\textstyle\sum\limits_{m=-l}^{l}}
%EndExpansion
a_{lm}\gamma_{l^{\prime}}trace(\mathbf{e}(l,m)\mathbf{e}(l^{\prime},0)^{g})=%
%TCIMACRO{\tsum \limits_{l=0}^{n}}%
%BeginExpansion
{\textstyle\sum\limits_{l=0}^{n}}
%EndExpansion%
%TCIMACRO{\tsum \limits_{m=-l}^{l}}%
%BeginExpansion
{\textstyle\sum\limits_{m=-l}^{l}}
%EndExpansion
(-1)^{m}a_{lm}\gamma_{l}D_{-m,0}^{l}(g)
\end{align*}
where the inner product
\[
D_{-m,0}^{l}(g)=\left\langle \mathbf{e}(l,-m),\mathbf{e}(l,0)^{g}\right\rangle
=trace((-1)^{m}\mathbf{e}(l,m)\mathbf{e}(l,0)^{g})
\]
is a Wigner D-function, namely a matrix element of the unitary matrix
$D^{l}(g)$ representing the action of $g$ on the irreducible operator subspace
$M_{\mathbb{C}}(\varphi_{l})\simeq\mathbb{C}^{2l+1}$. Consequently, expanding
$Q\in M_{\mathbb{C}}(n+1)$ similarly to (\ref{A}) we obtain
\begin{align*}
\frac{1}{4\pi}%
%TCIMACRO{\tint \limits_{S^{2}}}%
%BeginExpansion
{\textstyle\int\limits_{S^{2}}}
%EndExpansion
W_{P^{\ast}}(\mathbf{n})W_{Q}(\mathbf{n})dS  &  =%
%TCIMACRO{\tint \limits_{G}}%
%BeginExpansion
{\textstyle\int\limits_{G}}
%EndExpansion
W_{P^{\ast}}(g\mathbf{n}_{0})W_{Q}(g\mathbf{n}_{0})dg\\
&  =%
%TCIMACRO{\tsum \limits_{l=0}^{n}}%
%BeginExpansion
{\textstyle\sum\limits_{l=0}^{n}}
%EndExpansion%
%TCIMACRO{\tsum \limits_{l^{\prime}=0}^{n}}%
%BeginExpansion
{\textstyle\sum\limits_{l^{\prime}=0}^{n}}
%EndExpansion%
%TCIMACRO{\tsum \limits_{m=-l}^{l}}%
%BeginExpansion
{\textstyle\sum\limits_{m=-l}^{l}}
%EndExpansion%
%TCIMACRO{\tsum \limits_{m^{\prime}=-l^{\prime}}^{l^{\prime}}}%
%BeginExpansion
{\textstyle\sum\limits_{m^{\prime}=-l^{\prime}}^{l^{\prime}}}
%EndExpansion
\bar{a}_{lm}b_{l^{\prime}m^{\prime}}\gamma_{l}\gamma_{l^{\prime}}%
%TCIMACRO{\tint \limits_{G}}%
%BeginExpansion
{\textstyle\int\limits_{G}}
%EndExpansion
\overline{D_{-m,0}^{l}(g)}D_{-m^{\prime},0}^{l^{\prime}}(g)dg\\
&  =%
%TCIMACRO{\tsum \limits_{l=0}^{n}}%
%BeginExpansion
{\textstyle\sum\limits_{l=0}^{n}}
%EndExpansion%
%TCIMACRO{\tsum \limits_{m=-l}^{l}}%
%BeginExpansion
{\textstyle\sum\limits_{m=-l}^{l}}
%EndExpansion
\bar{a}_{lm}b_{lm}\gamma_{l}^{2}%
%TCIMACRO{\tint \limits_{G}}%
%BeginExpansion
{\textstyle\int\limits_{G}}
%EndExpansion
\left\vert D_{-m,0}^{l}(g)\right\vert ^{2}dg\\
&  =%
%TCIMACRO{\tsum \limits_{l=0}^{n}}%
%BeginExpansion
{\textstyle\sum\limits_{l=0}^{n}}
%EndExpansion%
%TCIMACRO{\tsum \limits_{m=-l}^{l}}%
%BeginExpansion
{\textstyle\sum\limits_{m=-l}^{l}}
%EndExpansion
\bar{a}_{lm}b_{lm}\frac{\gamma_{l}^{2}}{2l+1}=%
%TCIMACRO{\tsum \limits_{l=0}^{n}}%
%BeginExpansion
{\textstyle\sum\limits_{l=0}^{n}}
%EndExpansion
trace(P_{l}^{\ast}Q_{l})\frac{\gamma_{l}^{2}}{2l+1}%
\end{align*}
where we have used the well known Frobenius-Schur orthogonality relations for
the matrix elements $D_{m,m^{\prime}}^{l}(g)$ of irreducible unitary representations.
\end{proof}

From (\ref{isom}) we also deduce the formula
\begin{equation}
(c_{l})^{2}=(n+1)\frac{\left\Vert W_{P}\right\Vert ^{2}}{\left\Vert
P\right\Vert ^{2}}\text{, for any nonzero }P\in M_{\mathbb{C}}(\varphi_{l})
\label{cl2}%
\end{equation}

\begin{corollary}\index{Symbol correspondences ! characteristic numbers ! of Stratonovich-Weyl}
\label{Symbol2}For each $j$, $W^{j}$ is a Stratonovich-Weyl symbol
correspondence if and only if the characteristic numbers are
\begin{equation}
c_{l}=\varepsilon_l=\pm1,\text{ }l=1,...,n=2j \label{cl}%
\end{equation}
and consequently there are precisely $2^{n}$ different symbol maps $W^{j}$ of
this type, in agreement with Theorem 1 of \cite{VG-B}.
\end{corollary}

\begin{remark}
Formula (\ref{cl2}) also gives $(c_{0})^{2}=1,$ but the value $c_{0}=-1$ in
(\ref{cl}) is already excluded by the normalization axiom (iv), cf.
(\ref{axiom1}) and (\ref{c0=1}).
\end{remark}

\begin{summary}\index{Symbol correspondences ! moduli space}
\label{modulispace} It follows from (\ref{isom}) that each symbol
correspondence becomes an isometry by appropriately scaling the
(Hilbert-Schmidt) inner product on each of the irreducible matrix subspaces
$M_{\mathbb{C}}(\varphi_{l})$. The moduli space of all symbol correspondences
is $(\mathbb{R}^{\ast})^{n}$, having $2^{n}$ connected components, and each
symbol correspondence can be continuously deformed to a unique
Stratonovich-Weyl symbol correspondence in the moduli space $(\mathbb{Z}%
_{2})^{n}$. \ 
\end{summary}

\begin{definition}\index{Symbol correspondences ! positive}
\label{positive} A symbol correspondence is \emph{positive} if $c_{l} >0 \ ,
\forall l\leq n$.
\end{definition}

\begin{definition}\index{Symbol correspondences ! Stratonovich-Weyl symbol ! standard}
\label{standardStrat} The unique positive Stratonovich-Weyl symbol
correspondence, for which all characteristic numbers are $1$, i.e.
$c_{l}=1,\text{ }l=1,...,n$, is called the \emph{standard Stratonovich-Weyl
symbol correspondence} and is denoted by
\[
W^{j}_{1} : \mathcal{B}(\mathcal{H}_{j}) \simeq M_{\mathbb{C}}(n+1)\rightarrow
Poly_{\mathbb{C}}(S^{2})_{\leq n}\subset C_{\mathbb{C}}^{\infty}(S^{2}) \ .
\]
\end{definition}

\subsubsection{Covariant-contravariant duality} \index{Symbol correspondences ! covariant-contravariant duality |(}

Given an operator kernel $K$ in the sense of Definition \ref{OK} and Theorem
\ref{kernel}, it is also possible to define a symbol correspondence
$\widetilde{W}$ via the integral equation
\begin{equation}
P\ =\ \frac{n+1}{4\pi}\int_{S^{2}}\widetilde{W}_{P}(g\mathbf{n}_{0}%
)K^{g}dS\ =\ \frac{n+1}{4\pi}\int_{S^{2}}\widetilde{W}_{P}(\mathbf{n}%
)K(\mathbf{n})dS \label{contravariant}%
\end{equation}
where $\mathbf{n}=g\mathbf{n}_{0}\in S^{2}=SO(3)/SO(2)$ and $K^{g}%
=K(\mathbf{n})$, cf. also (\ref{W}).

\begin{definition}
\label{cont} The symbol map
\[
\widetilde{W}=\widetilde{W}^{j,K}:\mathcal{B}(\mathcal{H}_{j})\simeq
M_{\mathbb{C}}(n+1)\rightarrow Poly_{\mathbb{C}}(S^{2})_{\leq n}\subset
C_{\mathbb{C}}^{\infty}(S^{2})
\]
defined implicitly by equation (\ref{contravariant}) is called the
\emph{contravariant} symbol correspondence given by the operator kernel $K$.
On the other hand, the symbol map
\[
W={W}^{j,K}:\mathcal{B}(\mathcal{H}_{j})\simeq M_{\mathbb{C}}(n+1)\rightarrow
Poly_{\mathbb{C}}(S^{2})_{\leq n}\subset C_{\mathbb{C}}^{\infty}(S^{2})
\]
defined explicitly by equation (\ref{W}) is called the \emph{covariant} symbol
correspondence given by the operator kernel $K$.
\end{definition}

\begin{remark}
This terminology of covariant and contravariant symbol correspondences was
introduced by Berezin \cite{Berezin}. See also Remark \ref{CCdual}, below.
\end{remark}

\begin{proposition}
\label{Scc} Let $K$ be the operator kernel for a Stratonovich-Weyl covariant
symbol correspondence ${W}^{j,K}$, that is, $c_{l}=\pm1 \ , 1\leq l\leq n$.
Then, $\widetilde{W}^{j,K}\equiv{W}^{j,K}$, that is, for any operator
$P\in\mathcal{B}(\mathcal{H}_{j})$, $\widetilde{W}_{P}\equiv W_{P}\in
C_{\mathbb{C}}^{\infty}(S^{2})$. In other words, for a Stratonovich-Weyl
symbol correspondence, the covariant and the contravariant symbol
correspondences (defined by K) coincide.
\end{proposition}

\begin{proof}
From (\ref{contravariant}) and (\ref{W}), we have that
\begin{equation}
\label{CCduality}trace(PQ)=\frac{n+1}{4\pi}\int_{S^{2}}\widetilde{W}%
_{P}(\mathbf{n}){W}_{Q}(\mathbf{n})dS\ ,
\end{equation}
but, since $W^{j,K}$ is an isometry (cf. (\ref{innerprod}))
\[
trace(PQ)=\frac{n+1}{4\pi}\int_{S^{2}}{W}_{P}(\mathbf{n}){W}_{Q}%
(\mathbf{n})dS\ .
\]
Since both equations are valid $\forall P,Q\in\mathcal{B}(\mathcal{H}_{j})$,
it follows that $\widetilde{W}_{P}\equiv W_{P}$.
\end{proof}

On the other hand, by analogous reasoning from equations (\ref{contravariant}%
), (\ref{W}) and the metric identity (\ref{isom}), there is the following more
general result:

\begin{theorem}
\label{CC} Let $K$ be determined by the characteristic numbers $c_{1}%
,c_{2},...,c_{n}$, as explained by Theorem \ref{kernel}, and let
\[
{W}^{j,K}:\mathcal{B}(\mathcal{H}_{j}) \simeq M_{\mathbb{C}}(n+1)\rightarrow
Poly_{\mathbb{C}}(S^{2})_{\leq n}\subset C_{\mathbb{C}}^{\infty}(S^{2})
\]
be the covariant symbol correspondence defined explicitly by equation
(\ref{W}). Then,
\begin{equation}
{W}^{j,K}\equiv\widetilde{W}^{j,\widetilde{K}}:\mathcal{B}(\mathcal{H}%
_{j})\rightarrow C_{\mathbb{C}}^{\infty}(S^{2})\ , \label{c-c1}%
\end{equation}
where
\[
\widetilde{W}^{j,\widetilde{K}}:\mathcal{B}(\mathcal{H}_{j}) \simeq
M_{\mathbb{C}}(n+1)\rightarrow Poly_{\mathbb{C}}(S^{2})_{\leq n}\subset
C_{\mathbb{C}}^{\infty}(S^{2})
\]
is the contravariant symbol correspondence defined implicitly by equation
\begin{equation}
\label{contravint}P=\frac{n+1}{4\pi}\int_{S^{2}}\widetilde{W}_{P}%
(\mathbf{n})\widetilde{K}^{g}dS\ ,
\end{equation}
with the operator kernel $\widetilde{K}$ determined by the characteristic
numbers $\tilde{c}_{1},\tilde{c}_{2},...,\tilde{c}_{n}$, where
\begin{equation}
\tilde{c}_{l}=\frac{1}{c_{l}}\ . \label{c-c2}%
\end{equation}
In other words, (\ref{contravint})-(\ref{c-c2}) hold iff \ $\widetilde
{W}^{\widetilde{K}}_{P}(\mathbf{n})=trace(PK^{g})={W}^{K}_{P}(\mathbf{n})$.
\end{theorem}

As a direct consequence of the above theorem, the relation between the
covariant symbol correspondence $W$ given by an operator kernel $K$ and the
associated contravariant symbol correpondence $\widetilde{W}$ given by the
same operator kernel $K$ can be expressed as follows. For any $P=\sum
_{l=0}^{n}P_{l}\in\mathcal{B}(\mathcal{H}_{j})$, we have
\begin{equation}
\widetilde{W}_{P}=\sum_{l=0}^{n}\widetilde{W}_{P_{l}}=\sum_{l=0}^{n}\frac
{1}{(c_{l})^{2}}{W}_{P_{l}}\ , \label{c-c3}%
\end{equation}%
\begin{equation}
{W}_{P}=\sum_{l=0}^{n}{W}_{P_{l}}=\sum_{l=0}^{n}(c_{l})^{2}\widetilde
{W}_{P_{l}}\ . \label{c-c4}%
\end{equation}

\begin{definition}\index{Symbol correspondences ! characteristic numbers ! of the correspondence}
\label{charact} The non-zero real numbers $c_{1},c_{2},...,c_{n}$, which are
the characteristic numbers of the operator kernel $K$ defining the covariant
symbol correspondence $W:\mathcal{B}(\mathcal{H}_{j})\rightarrow
C_{\mathbb{C}}^{\infty}(S^{2})$ explicitly by equation (\ref{W}), will also be
referred to as \emph{the characteristic numbers of the symbol correspondence}
$W$.
\end{definition}

\begin{remark}
\label{CCdual} There is a \emph{duality} $W\longleftrightarrow\widetilde{W}$
between symbol correspondences, namely for a given $K$ the covariant symbol
correspondence $W^{K}$and the contravariant symbol correspondence
$\widetilde{W}^{K}$ are dual to each other. According to Theorem \ref{CC}, the
passage to the dual symbol correspondence is described by inverting the
characteristic numbers, that is, by the replacement $c_{i}\longrightarrow
c_{i}^{-1}$. Thus, if $K$ (resp. $\widetilde{K}$ ) has characteristic numbers
$\{c_{i}\}$ (resp. $\{c_{i}^{-1}\}$), then $\widetilde{W}^{K}=W^{\widetilde
{K}}$ and, as observed in Theorem \ref{CC}, $\widetilde{W}^{\widetilde{K}}%
$coincides with $W^{K}$. The Stratanovich-Weyl symbol correspondences are
precisely the self-dual correspondences for a spin-j system.
\end{remark}\index{Symbol correspondences ! covariant-contravariant duality |)}

We now turn to the formal proof of Proposition \ref{coupledcorrespondence} of section \ref{scvcsb}. 

\

\noindent {\bf Proof of Proposition \ref{coupledcorrespondence}:} To see why the choice $c^{n}_{l}\equiv c_{l}$ in (\ref{symbol2}) yields an
operator kernel $K$ given by the formula (\ref{K}), we apply formula (\ref{W})
with $P=\mathbf{e}(l,0)$ and $g=e$, and use the fact that $Y_{l,0}$ takes the
value $\sqrt{2l+1}$ at the north pole $\mathbf{n}_{0}$ $=(0,0,1)$.
The second part of the proposition is immediate (cf. Corollary \ref{Symbol2}).  

\begin{remark} The numbers $c_{l}\equiv c_{l}^{n}$ of Proposition \ref{coupledcorrespondence} are precisely the characteristic
numbers of $W$ in the sense of Definition \ref{charact}. The notation
$c_{l}^{n}$ instead of $c_{l}$ is to emphasize their dependence on $n=2j$,
whenever this dependence is an important issue. We also denote by $\vec{c}$
the $n$-string $(c_{1},c_{2},...,c_{n})\equiv(c_{1}^{n},c_{2}^{n}%
,...,c_{n}^{n})$, as in (\ref{symbol2}), and by $\frac{1}{\vec{c}}$ 
the n-string $\ (\frac{1}{c_{1}},\frac{1}{c_{2}},...,\frac{1}{c_{n}})$. The
notation $W_{\vec{c}}$ for $W^{K}$ will be heavily
used in the sequel. In particular, the dual of $W_{\vec{c}}$ is $\widetilde
{W}_{\vec{c}}=W_{\frac{1}{\vec{c}}}$.
\end{remark}

\subsection{Symbol correspondences via Hermitian metric}

Let $n=2j$, as usual. We shall construct a symbol map%
\[
B: \mathcal{B}(\mathcal{H}_{j}) \simeq M_{\mathbb{C}}(n+1)\rightarrow
Poly_{\mathbb{C}}(S^{2})_{\leq n}\subset C_{\mathbb{C}}^{\infty}(S^{2})
\]
with the appropriate properties, using the Hermitian metric on the underlying
Hilbert space $\mathcal{H}_{j}\simeq\mathbb{C}^{n+1}$, which we may take to be
the space of binary n-forms.

First of all, consider the following explicit construction of a map\index{Symbol correspondences ! Berezin symbol ! standard |(}
\begin{align}
\Phi_{j}  &  :\mathbb{C}^{2}\rightarrow\mathbb{C}^{n+1}\ ,\ \mathbb{C}%
^{2}\supset\ S^{3}(1)\rightarrow S^{2n+1}(1)\subset\mathbb{C}^{n+1}%
,\label{binforms}\\
\mathbf{z}  &  =(z_{1},z_{2})\mapsto\Phi_{j}(\mathbf{z})=\tilde{Z}=(z_{1}%
^{n},\sqrt{\binom{n}{1}}z_{1}^{n-1}z_{2},..,\sqrt{\binom{n}{k}}z_{1}%
^{n-k}z_{2}^{k},..,z_{2}^{n})\nonumber
\end{align}
where the components of $\tilde{Z}$ can be regarded as normalized binary
n-forms, so that, for $S^{3}(1)$ being the unit sphere in $\mathbb{C}^{2}$,
$S^{2n+1}(1)$ is the unit sphere in $\mathbb{C}^{n+1}$ (cf. the discussion at the end of section 
 \ref{genrotgroup}, in particular equation (\ref{standmon})). As a consequence, as
$SU(2)$ acts on $\mathbb{C}^{2}$ by the standard representation $\varphi
_{1/2}$, the induced action on n-forms is the irreducible unitary
representation
\[
\varphi_{j}:SU(2)\rightarrow SU(n+1)
\]
and the map $\Phi_{j}$ is $\varphi_{j}$-equivariant. Next, consider the Hopf
map
\begin{equation}
\pi:S^{3}(1)\rightarrow S^{2}\simeq\mathbb{C}P^{1}\text{, }\pi(\mathbf{z}%
)=[z_{1},z_{2}]=\mathbf{n,} \label{Hopfmap}%
\end{equation}
also described in (\ref{hopf}), which is equivariant when $SU(2)$ acts on
$S^{2}$ by rotations via the induced homomorphism $\psi:SU(2)\rightarrow
SO(3)$, cf. section 2.2.

And finally, let $h:\mathbb{C}^{n+1}\times\mathbb{C}^{n+1}\rightarrow
\mathbb{C}$ be the usual Hermitian inner product which is conjugate linear in
the first variable. Then we have the following:

\begin{theorem}
\label{Berezincorr1} The map $B$ that associates to each operator $P\in
M_{\mathbb{C}}(n+1)$ the function $B_{P}$ on $S^{2}$ defined by%
\begin{equation}
B_{P}(\mathbf{n})=h(\tilde{Z},P\tilde{Z}), \label{Berezin1}%
\end{equation}
is a symbol correspondence, according to Definition \ref{symbol corr}, whose
operator kernel, according to Definition \ref{OK}, is the projection operator
$\Pi_{1}$.
\end{theorem}

The proof of Theorem \ref{Berezincorr1} follows from the set of lemmas below:

\begin{lemma}
The function $B_{P}$ is well defined and $SU(2)$-equivariant.
\end{lemma}

\begin{proof}
First, note that $B_{P}$ is well defined because
\[
\Phi(e^{i\theta}\mathbf{z})=e^{in\theta}\tilde{Z}\text{, \ }h(e^{in\theta
}\tilde{Z},Pe^{in\theta}\tilde{Z})=h(e^{in\theta}\tilde{Z},e^{in\theta}%
P\tilde{Z})=h(\tilde{Z},P\tilde{Z}).
\]
By the identity $\Phi(g\mathbf{z})=\varphi_{j}(g)\tilde{Z}$, $SU(2)$%
-equivariance of $B$ is seen as follows:%
\[
B_{P^{g}}(\mathbf{n)} =h(\tilde{Z},P^{g}\tilde{Z})=h(\tilde{Z},\varphi
_{j}(g)P\varphi_{j}(g)^{-1}\tilde{Z})=h(\varphi_{j}(g)^{-1}\tilde{Z}%
,P\varphi_{j}(g)^{-1}\tilde{Z})
\]
$=h(\Phi(g^{-1}\mathbf{z}),P\Phi(g^{-1}\mathbf{z}))=B_{P}(\pi(g^{-1}%
\mathbf{z}))=B_{P}(g^{-1}\mathbf{n})=(B_{P})^{g}(\mathbf{n})$\quad
\end{proof}

\begin{lemma}
The map $B: P\mapsto B_{P}$ satisfies the reality condition.
\end{lemma}

\begin{proof}
By the definition of $B_{P}$,
\begin{align*}
B_{P}(\mathbf{n)}  &  =h(\tilde{Z},P\tilde{Z})=%
%TCIMACRO{\dsum \limits_{i=0}^{n}}%
%BeginExpansion
{\displaystyle\sum\limits_{i=0}^{n}}
%EndExpansion%
%TCIMACRO{\dsum \limits_{j=0}^{n}}%
%BeginExpansion
{\displaystyle\sum\limits_{j=0}^{n}}
%EndExpansion
b_{i}b_{j}p_{ji}z_{1}^{n-i}z_{2}^{i}\overline{z_{1}}^{n-j}\overline{z_{2}}%
^{j}\\
&  =%
%TCIMACRO{\dsum \limits_{i}}%
%BeginExpansion
{\displaystyle\sum\limits_{i}}
%EndExpansion
b_{i}^{2}p_{ii}|z_{1}|^{2(n-i)}|z_{2}|^{2i}+\sum_{j<i}b_{i}b_{j}%
|z_{1}|^{2(n-i)}|z_{2}|^{2j}[p_{ij}(\overline{z_{1}}z_{2})^{i-j}+p_{ji}%
(z_{1}\overline{z_{2}})^{i-j}],
\end{align*}
where we have written $b_{i}=\sqrt{\binom{n}{i}}$ for simplicity. Now, the
replacement $P\rightarrow P^{\ast}$ means $p_{ii}\rightarrow\overline{p_{ii}}
$ and
\[
\lbrack p_{ij}(\overline{z_{1}}z_{2})^{i-j}+p_{ji}(z_{1}\overline{z_{2}%
})^{i-j}]\rightarrow\lbrack\overline{p_{ij}(\overline{z_{1}}z_{2}%
)^{i-j}+p_{ji}(z_{1}\overline{z_{2}})^{i-j}}],
\]
and consequently $B_{P^{\ast}}(\mathbf{n})=\overline{B_{P}(\mathbf{n})}$.\quad
\end{proof}

\begin{lemma}
The map $B: P\mapsto B_{P}$ is injective.
\end{lemma}

\begin{proof}
The kernel of $B$ is $Ker(B)= \mathcal{K}=\{P\in M_{\mathbb{C}}(n+1);B_{P}%
=0\},$ which is an $SU(2)$-invariant subspace of $M_{\mathbb{C}}(n+1)=%
%TCIMACRO{\dsum \limits_{l=0}^{n}}%
%BeginExpansion
{\displaystyle\sum\limits_{l=0}^{n}}
%EndExpansion
M_{\mathbb{C}}(\varphi_{l})$, so assuming $\mathcal{K}\neq0$ it splits as a
direct sum of irreducible subspaces
\[
Ker(B)=\mathcal{K=}\sum_{i=1}^{k}\mathcal{K(}\varphi_{l_{i}}) \ ,
\]
then each of the summands has a zero weight vector $|l_{i}0>$, namely there is
some nonzero diagonal matrix%
\[
D=(d_{0},d_{1},...,d_{n})_{0}\in\mathcal{K}%
\]
and consequently, for each $\mathbf{n}\in S^{2}$ there is the following linear
equation%
\[
B_{D}(\mathbf{n})=h(\tilde{Z},D\tilde{Z})=\sum_{i=0}^{n}d_{i}b_{i}^{2}%
z_{1}^{n-i}z_{2}^{i}\overline{z_{1}}^{n-i}\overline{z_{2}}=\sum_{i=0}^{n}%
d_{i}b_{i}^{2}|z_{1}|^{2(n-i)}|z_{2}|^{2i}=0
\]
for the "variables" $d_{i}$. Clearly, the only common solution is $d_{i}=0$
for each $i$, so we must have $\mathcal{K}=0$. \quad
\end{proof}

\begin{lemma}
The map $B$ can be expressed by the formula%
\begin{equation}
B_{Q}(g\mathbf{n}_{0})=trace(Q\Pi_{1}^{g})\text{, }g\in SU(2),\text{\ }
\label{formula}%
\end{equation}
where $\mathbf{n}_{0}=(0,0,1)\in S^{2}$ is the north pole, and $\Pi_{1}$ is
the projection operator%
\[
\Pi_{1}=diag(1,0,0,...,0)
\]

\end{lemma}

\begin{proof}
Referring to (\ref{binforms}) and (\ref{Hopfmap}), we have
\[
\Phi(1,0)=\tilde{Z}_{0}=(1,0,..,0)\text{ and }\pi(1,0)=\mathbf{n}_{0},
\]
consequently for $Q=(q_{ij}),$ $1\leq i,j\leq n+1$,
\[
B_{Q}(\mathbf{n}_{0})=h(\tilde{Z}_{0},Q\tilde{Z}_{0})=q_{11}=trace(Q\Pi_{1})
\]
Therefore (\ref{formula}) holds for $g=e$, and hence by the equivariance of
$B$, $B_{Q}(g\mathbf{n}_{0})=B_{Q^{g^{-1}}}(\mathbf{n}_{0})=trace(Q^{g^{-1}%
}\Pi_{1})=trace(Q\Pi_{1}^{g})$ \quad
\end{proof}

\begin{corollary}
The map $B$ is the standard Berezin symbol correspondence (cf. Definition
\ref{standardBerezin}).
\end{corollary}

\begin{remark}
Berezin \cite{Berezin} introduced his symbol correspondence in a manner
similar to, but not equal to the one presented here. Instead of using the space of homogeneous polynomials in two variables ${ }^h\!P^n_{\mathbb C}(z_1,z_2)$, Berezin used the  representation  on the space of holomorphic polynomials on the sphere $\mathcal Hol^n(S^2)$. As the two representations are intimately connected (see the discussion at the end of section  \ref{genrotgroup}), it is not hard to see that 
the symbols obtained  by
equation (\ref{Berezin1}) coincide with Berezin's original definition of covariant
symbols on the sphere.
\end{remark}\index{Symbol correspondences ! Berezin symbol ! standard |)}

Now, the characteristic numbers of the standard Berezin symbol correspondence
are the characteristic numbers of the symbol map whose operator kernel $K$ is
the projection operator $\Pi_{1}$. Thus for each $n=2j$ there is a string
$\vec{b}$ of $n$ characteristic numbers possibly depending on $n$, which in
this case shall be denoted by $b^{n}_{l}$, namely
\[
\vec{b}=(b_{1}^{n},b_{2}^{n},..,b_{l}^{n},..,b_{n}^{n})
\]
According to (\ref{K}), these numbers are expressed by the following inner
product
\begin{equation}
b_{l}^{n}=\sqrt{\frac{n+1}{2l+1}}\left\langle \Pi_{1},\mathbf{e}%
^{j}(l,0)\right\rangle =\sqrt{\frac{n+1}{2l+1}}\mathbf{e}^{j}(l,0)_{1,1}%
\text{, } \label{Cnl}%
\end{equation}
where $\mathbf{e}^{j}(l,0)_{1,1}$ denotes the first entry of the diagonal
matrix $\mathbf{e}^{j}(l,0)$.

From Remark \ref{CGentries}, equations (\ref{subdiag1}) and
(\ref{CGexplicitentries}), we have that
\begin{equation}
\label{BCG}b_{l}^{n}= \sqrt{\frac{n+1}{2l+1}}C_{j,-j,0}^{j,j,l} \ .
\end{equation}
Thus, the explicit expression for $b^{n}_{l}$ can be obtained from the general
formulae for Clebsh-Gordan coefficients, as (\ref{explicitCG2}). But it can
also be obtained directly by induction, as shown in Appendix \ref{blnproof}.
The result is expressed below.

\begin{proposition}\index{Symbol correspondences ! characteristic numbers ! of Berezin}
\label{bln} The characteristic numbers $b_{l}^{n}$ of the standard Berezin
symbol correspondence are given explicitly by
\begin{equation}
b_{l}^{n} =\sqrt{\frac{n(n-1)...(n-l+1)}{(n+2)(n+3)...(n+l+1)}}=\frac
{\sqrt{\binom{n}{l}}}{\sqrt{\binom{n+l+1}{l}}} = \frac{n!\sqrt{n+1}}%
{\sqrt{(n+l+1)!(n-l)!}} \label{BerezinChar}%
\end{equation}

\end{proposition}

\begin{corollary}
The standard Berezin symbol correspondence is positive and, $\forall
l\in\mathbb{N}$, its characteristic numbers $b^{n}_{l}$ expand in powers of
$n^{-1}$ as:
\begin{align}
b_{l}^{n}  &  =1-\frac{1}{n}\binom{l+1}{2}+\frac{1}{n^{2}}\left\{  \frac{1}%
{2}\binom{l+1}{2}^{2}+\binom{l+1}{2}\right\} \nonumber\\
&  -\frac{1}{n^{3}}\left\{  \frac{1}{6}\binom{l+1}{2}^{3}+\frac{4}{3}%
\binom{l+1}{2}^{2}+\binom{l+1}{2}\right\}  +..... \label{BerezinCharexp}%
\end{align}

\end{corollary}

According to summary \ref{modulispace}, the standard Berezin symbol
correspondence, with operator kernel being the projection operator $\Pi_{1}$,
can be continuously deformed to the standard Stratonovich-Weyl symbol
correspondence. On the other hand, according to remark \ref{possibleBerezin},
$\forall n\geq1$, the projection operator $\Pi_{n+1}$ is also the operator
kernel of a Berezin symbol correspondence. Can it be defined via the Hermitian
metric $h$? What are its characteristic numbers?

\begin{proposition}\index{Symbol correspondences ! Berezin symbol ! alternate |(}
\label{alternateBerezin} Let $\sigma:\mathbb{C}^{2}\to\mathbb{C}^{2}$ be given
by
\[
\sigma: \mathbf{z}=(z_{1},z_{2})\mapsto(-\bar{z}_{2},\bar{z}_{1})
\]
and let $\Phi^{-}=\Phi\circ\sigma: \mathbb{C}^{2}\to\mathbb{C}^{n+1}$, for
$\Phi$ as in (\ref{binforms}). The map
\[
B^{-}:M_{\mathbb{C}}(n+1)\rightarrow C^{\infty}(S^{2}) \ , \ P\mapsto
B^{-}_{P} \ ,
\]
defined by
\begin{equation}
\label{altBermap}B_{P}^{-}(\mathbf{n}) = h(\Phi^{-}(\mathbf{z}), P\Phi
^{-}(\mathbf{z})) \ ,
\end{equation}
is a symbol correspondence with operator kernel $\Pi_{n+1}$ and characteristic
numbers
\begin{equation}
\label{BerezinChar-}b^{n}_{l-} = (-1)^{l} b^{n}_{l} \ ,
\end{equation}
for $b_{l}^{n}$ given by (\ref{BerezinChar}).
\end{proposition}

\begin{proof}
We note that, under the Hopf map $\pi: S^{3}\subset\mathbb{C}^{2}\to S^{2}$
given by (\ref{hopf}), the conjugate of $\sigma$, given by \ $\pi\circ
\sigma=\alpha\circ\pi\ ,$ is the antipodal map $\alpha: S^{2}\to S^{2}$,
\[
\alpha: \mathbf{n}\mapsto-\mathbf{n }\ .
\]
It follows that the function $B_{P}^{-}$ given by (\ref{altBermap}) satisfies
$B_{P}^{-}=B_{P}\circ\alpha$ , that is,
\begin{equation}
\label{antipodalsymbol}\forall P\in M_{\mathbb{C}}(n+1), \ B^{-}%
_{P}(\mathbf{n})=B_{P}(-\mathbf{n}) \ ,
\end{equation}
where $B_{P}$ is the standard Berezin symbol of $P$. Thus, clearly, all
properties of Definition \ref{symbol corr} are satisfied for the map $B^{-}$.

Now, if $P=\mu_{0}\mathbf{e}^{j}(l,m)$, then, according to (\ref{symbol2}),
$B^{-}_{P}(\mathbf{n})=b^{n}_{l-}Y_{l}^{m}(\mathbf{n})$. On the other hand, by
(\ref{symbol2}) and (\ref{antipodalY}), $B_{P}(-\mathbf{n})=b_{l}^{n}Y_{l}%
^{m}(-\mathbf{n}) = (-1)^{l} b_{l}^{n}Y_{l}^{m}(\mathbf{n})$. Therefore,
(\ref{BerezinChar-}) follows from (\ref{antipodalsymbol}).

Finally, $B^{-}_{P}(\mathbf{n}_{0})=h(P\Phi^{-}(1,0), \Phi^{-}%
(1,0))=h(P(0,0,...,0,1),(0,0,...,0,1))=p_{(n+1)(n+1)}$, for $P=[p_{ij}]$.
Thus,
\[
B^{-}_{P}(\mathbf{n}_{0})=trace(P\Pi_{n+1})
\]
and, by equivariance, $\Pi_{n+1}$ is the operator kernel for $B^{-}$.
\end{proof}

\begin{definition}
\label{defaltBer} The symbol correspondence defined by the characteristic
numbers $b^{n}_{l-} = (-1)^{l} b^{n}_{l}$ is called the \emph{alternate
Berezin symbol correspondence}.
\end{definition}\index{Symbol correspondences ! Berezin symbol ! alternate |)}

According to summary \ref{modulispace}, the alternate Berezin symbol
correspondence can be continuously deformed to a Stratonovich-Weyl symbol correspondence:

\begin{definition}\index{Symbol correspondences ! Stratonovich-Weyl symbol ! alternate}
\label{defaltSW} The unique Stratonovich-Weyl symbol correspondence
continuously deformed from the alternate Berezin symbol correspondence is
called the \emph{alternate Stratonovich-Weyl symbol correspondence} and is
given by the characteristic numbers $c^{n}_{l}\equiv\varepsilon_{l}= (-1)^{l}%
$. It shall be denoted by
\[
W^{j}_{1-} : \mathcal{B}(\mathcal{H}_{j}) \simeq M_{\mathbb{C}}%
(n+1)\rightarrow Poly_{\mathbb{C}}(S^{2})_{\leq n}\subset C_{\mathbb{C}%
}^{\infty}(S^{2}) \ .
\]

\end{definition}

\begin{definition}\index{Symbol correspondences ! alternate}
\label{alternatecorresp} If $W_{\vec{c}}^{j}$ is any positive symbol
correspondence given by characteristic numbers $c^{n}_{l}>0$, the symbol
correspondence $W_{\vec{c}-}^{j}$ given by characteristic numbers $c^{n}%
_{l-}=(-1)^{l}c^{n}_{l}$ is an \emph{alternate} symbol correspondence.
\end{definition}\index{Symbol correspondences |)}

%----------------------------------------------------------------------------------------------------------------------------------------------------------------------------------------------------------------------
%-------------------------------------------------------------------------------------- Chapter 7 -----------------------------------------------------------------------------------------------------------------
%----------------------------------------------------------------------------------------------------------------------------------------------------------------------------------------------------------------------

\chapter{Multiplications of symbols on the 2-sphere}
\label{multisymbols}

Given any symbol correspondence $W^j=W^{j}_{\vec c}$, the algebra of operators in
$\mathcal{B}(\mathcal{H}_{j})\simeq M_{\mathbb{C}}(n+1)$ can be imported to
the space of symbols $W^{j}_{\vec c}(\mathcal{B}(\mathcal{H}_{j}))\simeq
Poly_{\mathbb{C}}(S^{2})_{\leq n}\subset C^{\infty}_{\mathbb{C}}(S^{2})$. The
$2$-sphere, with such an algebra on a subset of its function space, has become
known as the ``fuzzy sphere'' \cite{Madore}. However, there is no single
``fuzzy sphere'', as each symbol correspondence defined by characteristic numbers $\vec{c}=(c_1,...,c_n)$ gives rise to a distinct
(although isomorphic) algebra on the space of symbols $Poly_{\mathbb{C}}%
(S^{2})_{\leq n}$, as we shall investigate in some detail, in this chapter. 

This fact has
important bearings on the question of how these various symbol algebras, or $\vec c$-twisted $j$-algebras as they shall be called, relate
to the classical Poisson algebra in the limit $n=2j\to\infty$, as we shall see in Chapter \ref{AsympChapter}.

%-------------------------------------------------------------------------------------- Section 7.1 -----------------------------------------------------------------------------------------------------------------

\section{Twisted products of spherical symbols}

\begin{definition}\index{Twisted products ! definition}
\label{defstarprod} \label{twisted} For a given symbol correspondence
(\ref{symbol}) with symbol map $W^j=W^{j}_{\vec c}$, the \emph{twisted product} $\star$ of
symbols is the binary operation on symbols
\[
\star: W^{j}_{\vec c}(\mathcal{B}(\mathcal{H}_{j})) \times W^{j}_{\vec c}(\mathcal{B}%
(\mathcal{H}_{j})) \to W^{j}_{\vec c}(\mathcal{B}(\mathcal{H}_{j}))\simeq
Poly_{\mathbb{C}}(S^{2})_{\leq n}%
\]
induced by the product of operators via $W^j=W^{j}_{\vec c}$, that is, given by
\begin{equation}
W_{PQ}^{j}=W_{P}^{j}\star W_{Q}^{j} \ , \label{starprod}%
\end{equation}
for every $P, Q\in \mathcal{B}(\mathcal{H}_{j})$, where $W^j_P=W^j_{\vec c}(P)$, etc. 
\end{definition}

\begin{remark}
We often write $\star^{n}$ and $\star^n_{\vec c}$ instead of $\star$ for the twisted product of
symbols to emphasize its dependence on $n=2j$ and $\vec c = (c_1,...,c_n)$. Also, we follow \cite{VG-B} in
calling the product induced by a symbol correspondence a \textit{twisted
product} of symbols and not a \textit{star-product} of functions, or formal
power series, because the latter product, introduced in \cite{Betal}, is
defined from the classical data on $S^{2}$, as opposed to the former. The
relationship between these two different kinds of product, in the limit
$n\rightarrow\infty$, will be briefly commented later on, in the concluding chapter.
\end{remark}

The following properties are immediate from Definitions \ref{symbol corr} and
\ref{twisted}:

\begin{proposition}\index{Twisted products ! general properties}
\label{twistedproperties} The algebra of symbols defined by any twisted
product is:
\begin{equation}%
\begin{array}
[c]{ll}%
(i) & \text{an $SO(3)$-equivariant algebra}: (f_{1}\star f_{2})^{g} =
f_{1}^{g}\star f_{2}^{g}\\
(ii) & \text{an associative algebra} : (f_{1}\star f_{2})\star f_{3} =
f_{1}\star(f_{2}\star f_{3})\\
(iii) & \text{a unital algebra} : 1\star f = f\star1 = f\\
(iv) & \text{a star algebra} : \overline{f_{1}\star f_{2}} = \overline{f_{2}%
}\star\overline{f_{1}} \ \text{ \ }%
\end{array}
\label{axiom3}%
\end{equation}
where $g\in SO(3)$ and $f_{1}, f_{2}, f_{3} \in W^{j}(\mathcal{B}%
(\mathcal{H}_{j}))\simeq Poly_{\mathbb{C}}(S^{2})_{\leq n} \subset
C_{\mathbb{C}}^{\infty}(S^{2}) $.
\end{proposition}

\begin{remark}\index{Twisted products ! induced inner product}
\label{inducedinnprod} In view of Definition \ref{defstarprod}, we can use the
normalization postulate $(iv)$ in Definition \ref{symbol corr} to define an
induced inner product on the space of symbols:
\begin{equation}
\label{inducedinnerproduct}\left\langle W_{P}^{j} \ ,W_{Q}^{j}\right\rangle
_{\star} := \frac{1}{4\pi}\int_{S^{2}} \overline{W_{P}^{j}}\star W_{Q}^{j}
\ dS .
\end{equation}
With this induced inner product, the isometry postulate can be rewritten as:
\[
(v) \ Isometry : \ \left\langle W_{P}^{j} \ ,W_{Q}^{j}\right\rangle =
\left\langle W_{P}^{j} \ ,W_{Q}^{j}\right\rangle _{\star}
\]

\end{remark}

\subsection{Standard twisted products of cartesian symbols on $S^{2}$}

\begin{definition}\index{Twisted products ! standard}
The twisted product obtained via the standard Stratonovich-Weyl correspondence
is called the \emph{standard twisted product} of symbols and is denoted by
$\star_{1}^{n}$, or simply $\star_{1}$.
\end{definition}

Thus, for $f,g\in W_{1}^{j}(\mathcal{B}(\mathcal{H}_{j}))=W^{j}(\mathcal{B}%
(\mathcal{H}_{j})) \simeq Poly_{\mathbb{C}}(S^{2})_{\leq n}\subset
C_{\mathbb{C}}^{\infty}(S^{2})$
\[
\star_{1}=\star_{1}^{n}:(f,g)\mapsto f\star_{1}^{n}g=f\star_{1}g\ \in
W^{j}(\mathcal{B}(\mathcal{H}_{j})) \simeq Poly_{\mathbb{C}}(S^{2})_{\leq
n}\subset C_{\mathbb{C}}^{\infty}(S^{2}) \ .
\]

For a given symbol correspondence $W$, a basic problem is to calculate the
twisted product of the cartesian coordinate functions $x,y,z$ and exhibit its
dependence on the spin parameter $j$. These coordinate functions form a basis
for the linear symbols on $S^{2}$, i.e. homogeneous polynomials in $x,y,z$ of
degree $1$.

Here we shall do this calculation in the basic case $W=$ $W_{1}$, that is, the
standard Stratonovich-Weyl correspondence. In the initial case $n=2j=1$, it
follows from Example \ref{half-j}, the identities (\ref{harmonics}), and
(\ref{symbol1}) with $\varepsilon_{l}=1$,
\[
x\longleftrightarrow\frac{1}{\sqrt{3}}\left(
\begin{array}
[c]{cc}%
0 & 1\\
1 & 0
\end{array}
\right)  ,y\longleftrightarrow\frac{i}{\sqrt{3}}\left(
\begin{array}
[c]{cc}%
0 & -1\\
1 & 0
\end{array}
\right)  ,z\longleftrightarrow\frac{1}{\sqrt{3}}\left(
\begin{array}
[c]{cc}%
1 & 0\\
0 & -1
\end{array}
\right)  ,
\]
from which one deduces%
\[
x\star_{1}^{1}y=-y\star_{1}^{1}x=\frac{i}{\sqrt{3}}z,\text{ \ \ }x\star
_{1}^{1}x=\frac{1}{3}%
\]
and these identities hold after a cyclic permutation of $(x,y,z)$. For
$j\geq1$ there is the following general result.

\begin{proposition}\index{Twisted products ! standard}
\label{stanprodlin} For the standard Stratonovich-Weyl correspondence, let
$n=2j\geq2$, and for $\{a,b,c\}=\{x,y,z\}$ let $\varepsilon_{abc}=\pm1$ be the
sign of the permutation $(x,y,z)\longrightarrow(a,b,c)$. Then
\begin{align}
a\star_{1}^{n}b  &  =\pi_{n}\cdot(ab)+\frac{\varepsilon_{abc}i}{\sqrt{n(n+2)}%
}c\text{\ \ }\label{star}\\
a\star_{1}^{n}a  &  =\pi_{n}\cdot(a^{2})+\frac{1-\pi_{n}}{3}\label{star1}\\
a\star_{1}^{n}a+b\star_{1}^{n}b+c\star_{1}^{n}c  &  =1 \label{star2}%
\end{align}
where
\begin{equation}
\pi_{n}=\frac{\sqrt{30}\mu_{2}}{\sqrt{n+1}n(n+2)}\text{ }=\sqrt{\frac
{(n-1)(n+3)}{n(n+2)}} \label{pin}%
\end{equation}
$\ $
\end{proposition}

The proof of Proposition \ref{stanprodlin} follows readily from the general formula for the standard twisted product of spherical harmonics to be presented in the next section (cf. (\ref{stanprod})), but in Appendix \ref{approdlin} the reader can find a more direct proof.

\begin{remark}\label{xyz1}
According to (\ref{xyz}), via the standard Stratonovich-Weyl correspondence,
the cartesian coordinate functions are identified (modulo a scaling) with the
angular momentum operators in the same coordinate directions, namely%
\[
x=\frac{1}{\sqrt{j(j+1)}}J_{1},\text{ \ }y=\frac{1}{\sqrt{j(j+1)}}J_{2},\text{
\ }z=\frac{1}{\sqrt{j(j+1)}}J_{3}%
\]
On the other hand, for a symbol correspondence with characteristic numbers
$c_{l}^{n}$ one must also divide the operator on the right side by $c_{1}^{n}$
(see below).
\end{remark}

\subsection{Twisted products for general symbol correspondences}

In order to study general twisted products of spherical symbols more
systematically, we start with the twisted product of spherical harmonics.
According to (\ref{symbol2}) and (\ref{coupledproductW}), for a symbol
correspondence $W_{\vec{c}}^{j}$ with characteristic numbers $c_{l}^{n}$ ,
$l=1,2,...n=2j$, denoting the twisted product $\star^{n}$ by $\star_{\vec{c}%
}^{n}$ we have
\begin{align}
Y_{l_{1}}^{m_{1}}\star_{\vec{c}}^{n}Y_{l_{2}}^{m_{2}}\longleftrightarrow &
\ \frac{\mu_{0}^{2}}{c_{l_{1}}^{n}c_{l_{2}}^{n}}\mathbf{e}^{j}(l_{1}%
,m_{1})\mathbf{e}^{j}(l_{2},m_{2})\nonumber\\
=  &  (-1)^{2j+m}\frac{\mu_{0}^{2}}{c_{l_{1}}^{n}c_{l_{2}}^{n}}\sum_{l=0}%
^{2j}\left[
\begin{array}
[c]{ccc}%
l_{1} & l_{2} & l\\
m_{1} & m_{2} & -m
\end{array}
\right]  \!\!{[j]}\ \mathbf{e}^{j}(l,m)\nonumber\\
\longleftrightarrow &  (-1)^{2j+m}\frac{\mu_{0}^{2}}{c_{l_{1}}^{n}c_{l_{2}%
}^{n}}\sum_{l=0}^{2j}\left[
\begin{array}
[c]{ccc}%
l_{1} & l_{2} & l\\
m_{1} & m_{2} & -m
\end{array}
\right]  \!\!{[j]}\ \frac{c_{l}^{n}}{\mu_{0}}Y_{l}^{m},\nonumber
\end{align}
where $m=m_{1}+m_{2}$ and
\[
{\left[
\begin{array}
[c]{ccc}%
l_{1} & l_{2} & l\\
m_{1} & m_{2} & -m
\end{array}
\right]  \!\!{[j]}}\
\]
is the Wigner product symbol given by equations (\ref{coupledproduct5}) and
(\ref{coupledproduct6}) in terms of Wigner $3jm$ and $6j$ symbols (cf.
(\ref{explicitW3jm})-(\ref{explicitW6j2})).Thus, we arrive immediately at the
following main result:

\begin{theorem}\index{Twisted products ! of spherical harmonics ! general formulae}
\label{genprodthm} For $n=2j\in\mathbb{N}$, let $W_{\vec{c}}^{j}%
:\mathcal{B}(\mathcal{H}_{j})\rightarrow Poly_{\mathbb{C}}(S^{2})_{\leq
n}\subset C_{\mathbb{C}}^{\infty}(S^{2})$ be the symbol correspondence with
characteristic numbers $c_{l}^{n}\neq0$, that is,
\[
W_{\vec{c}}^{j}:\mu_{0}\mathbf{e}^{j}(l,m)\longleftrightarrow c_{l}^{n}%
Y_{l}^{m}\text{ };\text{ }-l\leq m\leq l\leq n\ .
\]
Then, the corresponding twisted product of spherical harmonics $Y_{l_{1}%
}^{m_{1}},Y_{l_{2}}^{m_{2}}$, induced by the operator product on
$\mathcal{B}(\mathcal{H}_{j})$, is given by
\begin{equation}
Y_{l_{1}}^{m_{1}}\star_{\vec{c}}^{n}Y_{l_{2}}^{m_{2}}=(-1)^{n+m}\sqrt{n+1}%
\sum_{l=0}^{n}\left[
\begin{array}
[c]{ccc}%
l_{1} & l_{2} & l\\
m_{1} & m_{2} & -m
\end{array}
\right]  \!\!{[j]}\ \frac{c_{l}^{n}}{c_{l_{1}}^{n}c_{l_{2}}^{n}}Y_{l}%
^{m}\ \ \label{genprod}%
\end{equation}
\noindent where $m=m_{1}+m_{2}$. 
\end{theorem}
\begin{corollary}\label{Berezinproduccor}\index{Twisted products ! standard}
The standard twisted product of spherical
harmonics is given by 
\begin{equation}
Y_{l_{1}}^{m_{1}}\star_{1}^{n}Y_{l_{2}}^{m_{2}}=(-1)^{n+m}\sqrt{n+1}\sum
_{l=0}^{n}\left[
\begin{array}
[c]{ccc}%
l_{1} & l_{2} & l\\
m_{1} & m_{2} & -m
\end{array}
\right]  \!\!{[j]}\ Y_{l}^{m}\ .\ \label{stanprod}%
\end{equation}\index{Twisted products ! Berezin}
and the twisted product induced by the standard Berezin correspondence defined by the characteristic numbers $b_{l}^{n}\in
\mathbb{R}^{+}$ as in (\ref{BerezinChar}) is given by equation (\ref{genprod}) above, via the following substitution: 
\begin{align}
\displaystyle{\frac{c_{l_{3}}^{n}}{c_{l_{1}}^{n}c_{l_{2}}^{n}}}\to
\displaystyle{\frac{b_{l_{3}}^{n}}{b_{l_{1}}^{n}b_{l_{2}}^{n}}}  &
=\displaystyle\sqrt{\frac{\binom{n+l_{1}+1}{l_{1}}\binom{n+l_{2}+1}{l_{2}%
}\binom{n}{l_{3}}}{\binom{n}{l_{1}}\binom{n}{l_{2}}\binom{n+l_{3}+1}{l_{3}}}%
}\label{Lb}\\
&  =\displaystyle\sqrt{\frac{(n+l_{1}+1)!(n+l_{2}+1)!(n-l_{1})!(n-l_{2}%
)!}{(n+l_{3}+1)!(n-l_{3})!(n+1)!n!}}\nonumber
\end{align}
\end{corollary}

\begin{remark}
(i) \label{explicitlinearity} By linearity, given two symbols, $f={\sum_{l,m}%
}\phi_{lm}Y_{l}^{m}$ and $g={\sum_{l,m}}\gamma_{lm}Y_{l}^{m}$, their twisted
product expands as
\begin{equation}
f\star_{\vec{c}}^{n}g={\sum_{l_{1},l_{2}=0}^{n}\sum_{m_{1}=-l_{1}}^{l_{1}}%
\sum_{m_{2}=-l_{2}}^{l_{2}}}\phi_{l_{1}m_{1}}\gamma_{l_{2}m_{2}}Y_{l_{1}%
}^{m_{1}}\star_{\vec{c}}^{n}Y_{l_{2}}^{m_{2}}\ , \label{fg}%
\end{equation}
where $Y_{l_{1}}^{m_{1}}\star_{\vec{c}}^{n}Y_{l_{2}}^{m_{2}}$ is given by
formula (\ref{genprod}), and also by the integral formula (\ref{intgensphar})
below. Expanding the product (\ref{fg}) as $f\star_{\vec{c}}^{n}g={\sum_{l,m}%
}\rho_{lm}Y_{l}^{m}$, the coefficients $\rho_{lm}\in\mathbb{C}$ are given by
\begin{align}
(f\star_{\vec{c}}^{n}g)_{lm}  &  =\rho_{lm} \nonumber\\
&  =(-1)^{n+m}\sqrt{n+1}\sum_{l_{1},l_{2}=0}^{n}\sum_{m_{1}=-l_{1}}^{l_{1}}\left[
\begin{array}
[c]{ccc}%
l_{1} & l_{2} & l\\
m_{1} & m_{2} & -m
\end{array}
\right]  \!\!{[j]}\ \frac{c_{l}^{n}}{c_{l_{1}}^{n}c_{l_{2}}^{n}}\ \phi
_{l_{1}m_{1}}\gamma_{l_{2}m_{2}} \label{gendecomp1}%
\end{align}
where $m_{2}=m-m_{1}$ and the sum in $l_{1},l_{2}$ is restricted by
$\delta(l_{1},l_{2},l)=1$.

(ii) By comparing (\ref{genprod}) and (\ref{stanprod}), or directly from
(\ref{gendecomp1}) above, we see that the expansion of any twisted product of
spherical harmonics is immediately obtained from the expansion of the standard
twisted product of spherical harmonics by multiplying each term of the
standard expansion by $\displaystyle{c_{l}^{n}/(c_{l_{1}}^{n}c_{l_{2}}^{n})}$.
\end{remark}

To illustrate Theorem \ref{genprodthm}, Corollary \ref{Berezinproduccor} and Remark \ref{explicitlinearity}, let us first list the standard twisted
product of linear spherical harmonics, which can be obtained directly from
Theorem \ref{genprodthm}, or from Proposition \ref{stanprodlin} and Remark
\ref{explicitlinearity}.
\begin{equation}
Y_{1}^{0}\star_{1}^{n}Y_{1}^{0}\ =\ \pi_{n}\cdot\frac{2}{\sqrt{5}}Y_{2}%
^{0}\ +\ 1 \label{eqstpr1}%
\end{equation}%
\begin{equation}
Y_{1}^{\pm1}\star_{1}^{n}Y_{1}^{\pm1}\ =\ \pi_{n}\cdot\sqrt{\frac{6}{5}}%
Y_{2}^{\pm2}%
\end{equation}%
\begin{equation}
Y_{1}^{0}\star_{1}^{n}Y_{1}^{\pm1}\ =\ \pi_{n}\cdot\sqrt{\frac{3}{5}}%
Y_{2}^{\pm1}\ \pm\ \sqrt{\frac{3}{n(n+2)}}Y_{1}^{\pm1}%
\end{equation}%
\begin{equation}
Y_{1}^{\pm1}\star_{1}^{n}Y_{1}^{0}\ =\ \pi_{n}\cdot\sqrt{\frac{3}{5}}%
Y_{2}^{\pm1}\ \mp\ \sqrt{\frac{3}{n(n+2)}}Y_{1}^{\pm1}%
\end{equation}%
\begin{equation}
Y_{1}^{\pm1}\star_{1}^{n}Y_{1}^{\mp1}\ =\ \pi_{n}\cdot\frac{1}{\sqrt{5}}%
Y_{2}^{0}\ \mp\ \sqrt{\frac{3}{n(n+2)}}Y_{1}^{0}\ -\ 1 \label{eqstpr5}%
\end{equation}
Then, from Remark \ref{explicitlinearity}, or directly from Corollary
\ref{Berezinproduccor} we obtain 
that the standard Berezin twisted product of linear spherical harmonics is
given by
\begin{equation}
Y_{1}^{0}\star_{b}^{n}Y_{1}^{0}\ =\ \frac{n-1}{n}\cdot\frac{2}{\sqrt{5}}%
Y_{2}^{0}\ +\ \frac{n+2}{n} \label{eqstBpr1}%
\end{equation}%
\begin{equation}
Y_{1}^{\pm1}\star_{b}^{n}Y_{1}^{\pm1}\ =\ \frac{n-1}{n}\cdot\sqrt{\frac{6}{5}%
}Y_{2}^{\pm2}%
\end{equation}%
\begin{equation}
Y_{1}^{0}\star_{b}^{n}Y_{1}^{\pm1}\ =\ \frac{n-1}{n}\cdot\sqrt{\frac{3}{5}%
}Y_{2}^{\pm1}\ \pm\ \frac{1}{n}\sqrt{3}Y_{1}^{\pm1}%
\end{equation}%
\begin{equation}
Y_{1}^{\pm1}\star_{b}^{n}Y_{1}^{0}\ =\ \frac{n-1}{n}\cdot\sqrt{\frac{3}{5}%
}Y_{2}^{\pm1}\ \mp\ \frac{1}{n}\sqrt{3}Y_{1}^{\pm1}%
\end{equation}%
\begin{equation}
Y_{1}^{\pm1}\star_{b}^{n}Y_{1}^{\mp1}\ =\ \frac{n-1}{n}\cdot\frac{1}{\sqrt{5}%
}Y_{2}^{0}\ \mp\ \frac{1}{n}\sqrt{3}Y_{1}^{0}\ -\ \frac{n+2}{n}%
\ \label{eqstBpr5}%
\end{equation}
Thus, again from Remark \ref{explicitlinearity}, we finally get

\begin{corollary}\index{Twisted products ! Berezin}
\label{Berezin0}The standard Berezin twisted product of the cartesian
coordinate functions $\{a,b,c\}=\{x,y,z\}$ is given by
\begin{align}
a\star_{\vec{b}}^{n}b  &  =\frac{n-1}{n}ab+\frac{i\varepsilon_{abc}}{n}c\\
a\star_{\vec{b}}^{n}a  &  =\frac{n-1}{n}a^{2}+\frac{1}{n}\\
a\star_{\vec{b}}^{n}a+b\star_{\vec{b}}^{n}b+c\star_{\vec
{b}}^{n}c &  = \frac{n+2}{n} = \frac{j+1}{j}  \label{Bernot1}%
\end{align}
\end{corollary}

\begin{remark} 
Note, in particular, the distinction between equation (\ref{Bernot1}) for the
Berezin product and equation (\ref{star2}) for the standard Stratonovich-Weyl
twisted product. In other words, from Remark \ref{xyz1} and (\ref{BerezinChar}), for the standard Berezin correspondence we have the identifications: 
\[
x=\frac{1}{j+1}J_{1},\text{ \ }y=\frac{1}{j+1}J_{2},\text{
\ }z=\frac{1}{j+1}J_{3}. 
\]
\end{remark}

\subsubsection{The parity property for symbols}

The relationship between twisted products for the standard and the alternate
Stratonovich-Weyl (resp. Berezin) symbol correspondences is described 
as follows:

\begin{proposition}\index{Twisted products ! positive-alternate relation}
\label{relstandalt} Let $W_{1-}^{j}$ be the alternate Stratonovich-Weyl symbol
correspondence given by characteristic numbers $c_{l}^{n}\equiv\varepsilon
_{l}=(-1)^{l}$ (cf. Definition \ref{defaltSW}), and denote by $\star_{1-}^{n}$
its corresponding twisted product. Similarly, let $B_{-}^{j}$ be the alternate Berezin symbol correspondence
$B_{-}^{j}$, with characteristic numbers $b_{l-}^{n}=(-1)^{l}b_{l}^{n}$ (cf.
Definition \ref{defaltBer}), and denote by $\star_{\vec{b}-}^{n}$ its corresponding twisted
product. More generally, denote by $\star_{\vec{c}-}^{n}$ the twisted product induced
by the symbol correspondence with characteristic numbers $c_{l-}^{n}%
=(-1)^{l}c_{l}^{n}$. Then, we have that 
\begin{equation}
Y_{l_{1}}^{m_{1}}\star_{\vec{c}-}^{n}Y_{l_{2}}^{m_{2}}=Y_{l_{2}}^{m_{2}}%
\star_{\vec{c}}^{n}Y_{l_{1}}^{m_{1}}. \label{relstandaltgen}%
\end{equation}
\end{proposition}
\begin{proof}
Using the symmetry property (\ref{WprodSym}) for the Wigner product symbols,
the above identities follow from equation (\ref{genprod}) and the fact that
$c_{l-}^{n}/c_{l_{1}-}^{n}c_{l_{2}-}^{n}=(-1)^{l-l_{1}-l_{2}}c_{l}%
^{n}/c_{l_{1}}^{n}c_{l_{2}}^{n}=(-1)^{l_{1}+l_{2}+l}c_{l}^{n}/c_{l_{1}}%
^{n}c_{l_{2}}^{n}$.
\end{proof}

Now, the parity property for operators (Proposition \ref{Emult}) has a neat
version for any twisted product of symbols, to be stated in the proposition
below. But first, let us make some definitions.

\begin{definition}
For any point $\mathbf{n}\in S^{2}$ let $-\mathbf{n}$ denote its antipodal
point. We say that $f\in C_{\mathbb{C}}^{\infty}(S^{2})$ is \emph{even} if
\[
f(\mathbf{n})=f(-\mathbf{n})\ ,\ \forall\mathbf{n}\in S^{2}%
\]
and $f\in C_{\mathbb{C}}^{\infty}(S^{2})$ is \emph{odd} if
\[
f(\mathbf{n})=-f(-\mathbf{n})\ ,\ \forall\mathbf{n}\in S^{2}%
\]
\end{definition}
In particular, the null function $0$ is the only function that is both even
and odd.

For a given symbol correspondence and associated twisted product $f\star g$ of
symbols, let us denote its symmetric product (or anti-commutator) by
\[
\lbrack\lbrack\ f,g\ ]]_{\star}=f\star g+g\star f
\]
while its commutator is denoted in the usual way by
\[
\lbrack\ f,g\ ]_{\star}=f\star g-g\star f
\]

With these preparations, we state the following result. 

\begin{proposition}\index{Twisted products ! parity property}\index{Parity property ! for symbols}
\label{PPS} (The Parity Property for symbols) With regard to the parity of
symbols, the behavior of the above products expresses as follows:
\begin{align}
\lbrack\lbrack\ even,even \ ]]_{\star}=even\quad &  \quad\lbrack
\ odd,odd\ ]_{\star}=odd\nonumber\\
\lbrack\lbrack\ odd,odd \ ]]_{\star}=even\quad &  \quad\lbrack
\ even,even\ ]_{\star}=odd\label{parity}\\
\lbrack\lbrack\ even,odd \ ]]_{\star}=odd\quad &  \quad\lbrack
\ even,odd\ ]_{\star}=even\nonumber
\end{align}
where, for example, $[[ \ even,even \ ]]_{\star}=even$ \ means that the
anti-commutator of two even symbols is an even symbol, and so on.
\end{proposition}

\begin{proof}
From equation (\ref{antipodalY}), note that if the symbol $f$ is even, then it
can be written as a linear combination of even spherical harmonics,
\[
f\ =\ \sum_{k,m}\ \alpha(k,m)\ Y_{2k}^{m}%
\]
and if the symbol $g$ is odd, then it can be written as a linear combination
of odd spherical harmonics,
\[
g\ =\ \sum_{k,m}\ \beta(k,m)\ Y_{2k+1}^{m}.
\]
Therefore, the relations (\ref{parity}) follow immediately from (\ref{genprod}%
) and (\ref{WprodSym}).
\end{proof}

\subsubsection{Algebra isomorphisms}

The fact that the parity property for operators (Proposition \ref{Emult}) can
be generally re-expressed for symbols (Proposition \ref{PPS}) follows
essentially from the fact that, via any symbol correspondence $W_{\vec{c}}%
^{j}$, the linear space $Poly_{\mathbb{C}}(S^{2})_{\leq n}\subset
C_{\mathbb{C}}^{\infty}(S^{2})$ with its twisted product $\star_{\vec{c}}^{n}$
is isomorphic to the matrix algebra $M_{\mathbb{C}}(n+1)$.

\begin{definition}\index{Twisted j-algebras }
\label{twistedalgebra} For each symbol correspondence $W_{\vec{c}}^{j}$, with
characteristic numbers $\vec{c}=(c_{1}^{n},c_{2}^{n},...,c_{n}^{n})$, the
space of symbols $W_{\vec{c}}^{j}(\mathcal{B}(\mathcal{H}_{j}%
))=Poly_{\mathbb{C}}(S^{2})_{\leq n}\subset C_{\mathbb{C}}^{\infty}(S^{2})$,
with its usual addition and the twisted product $\star_{\vec{c}}^{n}$, shall
be called the \emph{twisted $j$-algebra} associated to the symbol
correspondence $W_{\vec{c}}^{j}$, or simply the \emph{$\vec{c}$-twisted
$j$-algebra}$.$
\end{definition}

\begin{proposition}\index{Twisted j-algebras ! isomorphisms}
All twisted $j$-algebras are isomorphic to each other (same $j$) and no
twisted $j$-algebra is isomorphic to any twisted $j^{\prime}$-algebra, for
$j\neq j^{\prime}$.
\end{proposition}

\begin{proof}
This follows immediately since every $\vec{c}$-twisted $j$-algebra is
isomorphic to $M_{\mathbb{C}}(n+1)$, whereas the latter is not isomorphic to
$M_{\mathbb{C}}(n^{\prime}+1)$, for $n\neq n^{\prime}$.
\end{proof}

\begin{definition}\index{Twisted j-algebras ! transition operators}
\label{TransOperator}If $(Poly_{\mathbb{C}}(S^{2})_{\leq n},\star_{\vec{c}%
}^{n})$ and $(Poly_{\mathbb{C}}(S^{2})_{\leq n},\star_{\vec{c}^{\prime}}^{n})$
are two distinct twisted $j$-algebras, then the following isomorphisms
\begin{align}
U_{\vec{c},\vec{c}^{\prime}}^{j}  &  =W_{\vec{c}^{\prime}}^{j}\circ(W_{\vec
{c}}^{j})^{-1}:(Poly_{\mathbb{C}}(S^{2})_{\leq n},\star_{\vec{c}}%
^{n})\longrightarrow(Poly_{\mathbb{C}}(S^{2})_{\leq n},\star_{\vec{c}^{\prime
}}^{n})\label{kernel-U}\\
V_{\vec{c},\vec{c}^{\prime}}^{j}  &  =(W_{\vec{c}}^{j})^{-1}\circ W_{\vec
{c}^{\prime}}^{j}:M_{\mathbb{C}}(n+1)\rightarrow M_{\mathbb{C}}(n+1)
\label{kernel-V}%
\end{align}
and their inverses $U_{\vec{c}^{\prime},\vec{c}}^{j}$ and $V_{\vec{c}^{\prime
},\vec{c}}^{j}$, respectively, define the \emph{transition
operators}.
\end{definition}

The transition operator $U_{\vec{c},\vec{c}^{\prime}}^{j}$ on  $Poly_{\mathbb C}(S^2)_{\leq n}$, resp. 
$V_{\vec{c},\vec{c}^{\prime}}^{j}$ on $M_{\mathbb{C}}(n+1)$, has
properties reflecting the properties listed in Definition \ref{symbol corr},
such as $SO(3)$-equivariance and preservation of real functions, resp. Hermitian matrices. There is
the following commutative diagram relating the two types of transition
operators
\begin{equation}%
\begin{array}
[c]{ccccc}
& V_{\vec{c},\vec{c}^{\prime}}^{j} &  &  & \\
M_{\mathbb{C}}(n+1) & \longrightarrow & M_{\mathbb{C}}(n+1) &  & \\
&  &  &  & \\
W_{\vec{c}}^{j}\downarrow &  & \downarrow W_{\vec{c}}^{j} &  & \\
& U_{\vec{c},\vec{c}^{\prime}}^{j} &  &  & \\
Poly_{\mathbb{C}}(S^{2})_{\leq n} & \longrightarrow & Poly_{\mathbb{C}}%
(S^{2})_{\leq n} &  & \\
&  &  &  &
\end{array}
\label{transition-diagram}%
\end{equation}

\begin{remark}
Although all twisted $j$-algebras are isomorphic (for fixed $j$), the fact
that distinct symbol correspondences define distinct twisted $j$-algebras has
nontrivial consequences when we consider sequences (in $n=2j\in\mathbb{N}$) of
twisted $j$-algebras and their limits as $n\rightarrow\infty$, as we shall see
further below.
\end{remark}

%-------------------------------------------------------------------------------------- Section 7.2 -----------------------------------------------------------------------------------------------------------------

\section{Integral representations of twisted products}\label{integralsection}

Just as the Moyal product of symbols on $\mathbb{R}^{2n}$ has an integral
version, the Groenewold-von Neumann product \cite{Groen, vN1}, it is
interesting to see that twisted products of spherical symbols can also be
written in integral form:
\begin{equation}
f\star g\ (\mathbf{n})\ =\ \iint_{S^{2}\times S^{2}}f(\mathbf{n}%
_{1})g(\mathbf{n}_{2})\mathbb{L}(\mathbf{n}_{1},\mathbf{n}_{2},\mathbf{n}%
)d\mathbf{n}_{1}d\mathbf{n}_{2}\ . \label{intwistedproduct}%
\end{equation}

Integral forms for twisted products in principle allow for a direct definition
of general twisted products of arbitrary symbols $f,g\in Poly_{\mathbb{C}%
}(S^{2})_{\leq n}$ without the need to decompose them in the basis of
spherical harmonics.

In such an integral representation, all properties of the twisted product are
encoded in the integral trikernel\index{Trikernels |(}\index{Twisted products ! integral representations | see {Trikernels} }
\begin{equation}
\mathbb{L}:S^{2}\times S^{2}\times S^{2}\rightarrow\mathbb{C}. \label{3S2}%
\end{equation}
Therefore, for each symbol correspondence $W_{\vec{c}}^{j}$, there will be a
distinct integral trikernel $\mathbb{L}_{\vec{c}}^{j}$.

\subsection{General formulae and properties of integral trikernels}

\label{genformulaetrk}

For the standard Stratonovich-Weyl symbol correspondence $W_{1}^{j}$, from
equations (\ref{W}), (\ref{contravariant}), and Proposition \ref{Scc}, we
immediately have the following result:

\begin{proposition}\index{Trikernels ! standard Stratonovich}
\label{PropSTrik} Let $f,g\in C_{\mathbb{C}}^{\infty}(S^{2})$ be such that
$f=W^{j}_{1}(F)$, $g=W^{j}_{1}(G)$, for $F,G\in\mathcal{B}(\mathcal{H}_{j})$,
where $W^{j}_{1}$ is determined by the operator kernel $K^{j}_{1}$ with all
characteristic numbers $c_{l}=1$, for $0\leq l\leq n=2j$, in equation
(\ref{K}). Then,
\begin{equation}
\label{intstand}f\star_{1}^{n} g \ (\mathbf{n}) \ = \iint_{S^{2}\times S^{2}}
f(\mathbf{n}_{1})g(\mathbf{n}_{2})\mathbb{L}^{j}_{1}(\mathbf{n}_{1}%
,\mathbf{n}_{2},\mathbf{n})d\mathbf{n}_{1}d\mathbf{n}_{2} \ ,
\end{equation}
where
\begin{equation}
\label{L1}\mathbb{L}^{j}_{1}(\mathbf{n}_{1},\mathbf{n}_{2},\mathbf{n})=\left(
\frac{2j+1}{4\pi}\right)  ^{2} trace(K^{j}_{1}(\mathbf{n}_{1})K^{j}%
_{1}(\mathbf{n}_{2})K^{j}_{1}(\mathbf{n})) \
\end{equation}
is the standard \emph{Stratonovich trikernel}.
\end{proposition}

\begin{corollary}
\label{corstratrik} The Stratonovich trikernel is, in fact, a polynomial
function $\mathbb{L}_{1}^{j}\in(Poly_{\mathbb{C}}(S^{2})_{\leq n})^{3}\subset
C_{\mathbb{C}}^{\infty}(S^{2}\times S^{2}\times S^{2})$ which can be written
as follows
\begin{align}
&  \mathbb{L}_{1}^{j}(\mathbf{n}_{1},\mathbf{n}_{2},\mathbf{n}_{3}%
)\label{L1b}\\
&  =\frac{(-1)^{n}\sqrt{n+1}}{(4\pi)^{2}}\sum_{l_{i},m_{i}}\left[
\begin{array}
[c]{ccc}%
l_{1} & l_{2} & l_{3}\\
m_{1} & m_{2} & m_{3}%
\end{array}
\right]  \!\!{[j]}\ \overline{Y_{l_{1}}^{m_{1}}}(\mathbf{n}_{1})\overline
{Y_{l_{2}}^{m_{2}}}(\mathbf{n}_{2})\overline{Y_{l_{3}}^{m_{3}}}(\mathbf{n}%
_{3})\ \nonumber
\end{align}
\noindent with the summations in $l_{i}$ and $m_{i}$ subject to the
constraints
\begin{equation}
0\leq l_{i}\leq n=2j,\ -l_{i}\leq m_{i}\leq l_{i},\ \delta(l_{1},l_{2}%
,l_{3})=1,\ m_{1}+m_{2}+m_{3}=0. \label{restrictsum}%
\end{equation}

\end{corollary}

\begin{proof}
This follows straightforwardly from the decomposition of symbols in the
standard (orthonormal) basis of spherical harmonics, using formula
(\ref{stanprod}) for the standard twisted product of basis vectors, and
equation (\ref{Ylm1}).
\end{proof}

From the above formula, using (\ref{WprodSym}) we immediately obtain:

\begin{corollary}\index{Trikernels ! symmetric property}
The Stratonovich trikernel is \emph{symmetric}, in the sense that it
satisfies
\begin{equation}
\mathbb{L}_{1}^{j}(\mathbf{n}_{1},\mathbf{n}_{2},\mathbf{n}_{3})=\mathbb{L}%
_{1}^{j}(\mathbf{n}_{3},\mathbf{n}_{1},\mathbf{n}_{2})=\mathbb{L}_{1}%
^{j}(\mathbf{n}_{2},\mathbf{n}_{3},\mathbf{n}_{1})\ . \label{symmetrik}%
\end{equation}

\end{corollary}

\subsubsection{General formulae for trikernels}

For a general correspondence $W_{\vec{c}}^{j}$, with operator kernel
$K_{\vec{c}}^{j}$ as in (\ref{K}) and characteristic numbers $\vec{c}%
=(c^{n}_{1},c^{n}_{2},...,c^{n}_{n})$, we obtain from equation (\ref{W}) and
Theorem \ref{CC} the generalization of Proposition \ref{PropSTrik} and
Corollary \ref{corstratrik}:

\begin{theorem}\index{Trikernels ! general formulae}
\label{geninttwistprod} A general twisted product of $f,g\in Poly_{\mathbb{C}%
}(S^{2})_{\leq n}$ is given by
\begin{equation}
f\star_{\vec{c}}^{n}g\ (\mathbf{n})\ =\iint_{S^{2}\times S^{2}}f(\mathbf{n}%
_{1})g(\mathbf{n}_{2})\mathbb{L}_{\vec{c}}^{j}(\mathbf{n}_{1},\mathbf{n}%
_{2},\mathbf{n})d\mathbf{n}_{1}d\mathbf{n}_{2}\ , \label{intgen}%
\end{equation}
with
\begin{align}
& \mathbb{L}_{\vec{c}}^{j}(\mathbf{n}_{1},\mathbf{n}_{2},\mathbf{n})  
=\left(  \frac{2j+1}{4\pi}\right)  ^{2}trace(\widetilde{K}_{\vec{c}}%
^{j}(\mathbf{n}_{1})\widetilde{K}_{\vec{c}}^{j}(\mathbf{n}_{2})K_{\vec{c}}%
^{j}(\mathbf{n}))\label{Lc}\\
&  =\frac{(-1)^{n}\sqrt{n+1}}{(4\pi)^{2}}\sum_{l_{i},m_{i}}\left[
\begin{array}
[c]{ccc}%
l_{1} & l_{2} & l_{3}\\
m_{1} & m_{2} & m_{3}%
\end{array}
\right]  \!\!{[j]}{\frac{c_{l_{3}}^{n}}{c_{l_{1}}^{n}c_{l_{2}}^{n}}}\overline{Y_{l_{1}%
}^{m_{1}}}(\mathbf{n}_{1})\overline{Y_{l_{2}}^{m_{2}}}(\mathbf{n}%
_{2})\overline{Y_{l_{3}}^{m_{3}}}(\mathbf{n}_{3}) \nonumber
\end{align}
as the corresponding integral trikernel, where ${K}_{\vec{c}}^{j}$ and
$\widetilde{K}_{\vec{c}}^{j}$ are the operator kernels given by equation
(\ref{K}) with characteristic numbers $\vec{c}=(c_{l}^{n})=(c_{1}%
^{n},...,c_{n}^{n})$ and $\frac{1}{\vec{c}}=(\frac{1}{c_{1}^{n}},...,\frac
{1}{c_{n}^{n}})$, respectively, and with the restrictions (\ref{restrictsum})
for the $l_{i},m_{i}$ summations.
\end{theorem}

Note that the trikernel $\mathbb{L}_{\vec{c}}^{j}$ is a polynomial
function, and in general it is not symmetric, namely (\ref{symmetrik}) does
not hold. More explicitly, the trikernel (\ref{Lc})$\mathbb{\ }$has the full
expression
\begin{align}
&  \mathbb{L}_{\vec{c}}^{j}(\mathbf{n}_{1},\mathbf{n}_{2},\mathbf{n}%
_{3})\label{Lc2}\\
&  ={\frac{\sqrt{n+1}}{(4\pi)^{2}}\sum_{l_{i}}}\sqrt{\frac{(n-l_{1}%
)!(n-l_{2})!(n-l_{3})!}{(n+l_{1}+1)!(n+l_{2}+1)!(n+l_{3}+1)!}}\Delta^{2}%
(l_{1},l_{2},l_{3})l_{1}!l_{2}!l_{3}!\ \nonumber\\
&  \cdot{\frac{c_{l_{3}}^{n}}{c_{l_{1}}^{n}c_{l_{2}}^{n}}}\sqrt{(2l_{1}%
+1)(2l_{2}+1)(2l_{3}+1)}\displaystyle{\sum_{k}}\frac{(-1)^{k}(n+1+k)!}%
{(n+k-l_{1}-l_{2}-l_{3})!R(l_{1},l_{2},l_{3};k)}\ \nonumber\\
&  \cdot{\sum_{m_{i}}}S_{m_{1},m_{2},m_{3}}^{\ l_{1},\ \ l_{2},\ \ l_{3}%
}N_{m_{1},m_{2},m_{3}}^{l_{1},\ l_{2},\ l_{3}}\ {Y_{l_{1}}^{m_{1}}}%
(\mathbf{n}_{1}){Y_{l_{2}}^{m_{2}}}(\mathbf{n}_{2}){Y_{l_{3}}^{m_{3}}%
}(\mathbf{n}_{3})\quad\quad\quad\ \nonumber
\end{align}
with $\Delta(l_{1},l_{2},l_{3}),S_{m_{1},m_{2},m_{3}}^{\ l_{1},\ \ l_{2}%
,\ \ l_{3}},N_{m_{1},m_{2},m_{3}}^{l_{1},\ l_{2},\ l_{3}}$, and $R(l_{1}%
,l_{2},l_{3};k)$ given respectively by (\ref{DeltaW}), (\ref{Sjjj}),
(\ref{N123}), and (\ref{explicitW6j2}), with restrictions on the summations
over $l_{i},m_{i},k$ according to (\ref{restrictsum}) and Remark
\ref{summation}.

\

In particular, the explicit expression for the Stratonovich trikernel is obtained from (\ref{Lc2}) by setteing all $c_l^n=1$ and we obtain explicit expressions for the integral trikernel\index{Trikernels ! standard Berezin}
$\mathbb{L}_{\vec{b}}^{j}$ of the standard Berezin twisted product by  performing the substitution (\ref{Lb}) in (\ref{Lc})-(\ref{Lc2}).

\subsubsection{General properties of trikernels}

Since $\mathbb{L}_{\vec{c}}^{j}$ is a polynomial function it also follows from
the general formulae (\ref{intgen})-(\ref{Lc2}) that we cannot use the
integral formulation to extend a twisted product defined on $Poly_{\mathbb{C}%
}(S^{2})_{\leq n}$ to a larger subset of $C_{\mathbb{C}}^{\infty}(S^{2})$,
namely we have:

\begin{corollary}\index{Trikernels ! general properties}
\label{noextension} Let $f,g\in C^{\infty}_{\mathbb{C}}(S^{2})$,
$\mathcal{L}\in C^{\infty}_{\mathbb{C}}(S^{2}\times S^{2}\times S^{2})$, and
define a binary operation $\bullet: C^{\infty}_{\mathbb{C}}(S^{2})\times
C^{\infty}_{\mathbb{C}}(S^{2})\to C^{\infty}_{\mathbb{C}}(S^{2})$ via the
integral formula:
\begin{equation}
\label{extendedproduct}f\bullet g \ (\mathbf{n})=\int_{S^{2}\times S^{2}%
}f(\mathbf{n}_{1})g(\mathbf{n}_{2}) \mathcal{L}(\mathbf{n}_{1}, \mathbf{n}%
_{2}, \mathbf{n})d\mathbf{n}_{1}d\mathbf{n}_{2} \ .
\end{equation}
If $\mathcal{L }= \mathbb{L}^{j}_{\vec{c}} \in(Poly_{\mathbb{C}}(S^{2})_{\leq
n})^{3}\subset C^{\infty}_{\mathbb{C}}(S^{2}\times S^{2}\times S^{2})$ , then
$\bullet= \star^{n}_{\vec{c}} : Poly_{\mathbb{C}}(S^{2})_{\leq n} \times
Poly_{\mathbb{C}}(S^{2})_{\leq n} \to Poly_{\mathbb{C}}(S^{2})_{\leq n} $ and,
in particular, if either $f$ or $g\in C^{\infty}_{\mathbb{C}}(S^{2}) \setminus
Poly_{\mathbb{C}}(S^{2})_{\leq n}$, then $f\bullet g = 0$.
\end{corollary}

\begin{proof}
We prove that, if $\mathcal{L}=\mathbb{L}_{\vec{c}}^{j}\in(Poly_{\mathbb{C}%
}(S^{2})_{\leq n})^{3}$, then $f\bullet g=0$ if either $f$ or $g\in
C_{\mathbb{C}}^{\infty}(S^{2})\setminus Poly_{\mathbb{C}}(S^{2})_{\leq n}$ .
From this it follows immediately that $\bullet=\star_{\vec{c}}^{n}$.

Now, if $g\in C_{\mathbb{C}}^{\infty}(S^{2})\setminus Poly_{\mathbb{C}}%
(S^{2})_{\leq n}$, then $g$ can be expanded as a series of $Y_{l}^{m}$'s, with
all $l>n$. But from (\ref{intgen})-(\ref{Lc}), with the restriction
(\ref{restrictsum}), $f\bullet g=0$ follows from orthogonality of each
$Y_{l}^{m}$ in the $g$ series and every $Y_{l^{\prime}}^{m^{\prime}}$, for
$n<l\neq l^{\prime}\leq n$.
\end{proof}

The integral trikernels ${\mathbb{L}}$ of twisted products have some common
properties that we shall spell out, as follows.

\begin{proposition}
\label{trikernelproperties} If ${\mathbb{L}}=\mathbb{L}_{\vec{c}}^{j}$ is the
integral trikernel of a twisted product according to equation
(\ref{intwistedproduct}), then it satisfies:
\begin{align*}
&  \noindent(i)\quad{\mathbb{L}(\mathbf{n}_{1},\mathbf{n}_{2},\mathbf{n}%
)=\mathbb{L}(g\mathbf{n}_{1},g\mathbf{n}_{2},g\mathbf{n})}\\
&  \noindent(ii)\quad{\int_{S^{2}}\mathbb{L}(\mathbf{n}_{1},\mathbf{n}%
_{2},\mathbf{n})\mathbb{L}(\mathbf{n},\mathbf{n}_{3},\mathbf{n}_{4}%
)d\mathbf{n}\ \ =\ \int_{S^{2}}\mathbb{L}(\mathbf{n}_{1},\mathbf{n}%
,\mathbf{n}_{4})\mathbb{L}(\mathbf{n}_{2},\mathbf{n}_{3},\mathbf{n}%
)d\mathbf{n}}\quad\\
&  \noindent(iii)\quad{\int_{S^{2}}\mathbb{L}(\mathbf{n}_{1},\mathbf{n}%
_{2},\mathbf{n})d\mathbf{n}_{1}\ =R^{j}(\mathbf{n}_{2},\mathbf{n})}%
,\ {\int_{S^{2}}\mathbb{L}(\mathbf{n}_{1},\mathbf{n}_{2},\mathbf{n}%
)d\mathbf{n}_{2}\ =R^{j}(\mathbf{n}_{1},\mathbf{n})}\quad\\
&  \noindent(iv)\quad{\mathbb{L}(\mathbf{n}_{2},\mathbf{n}_{1},\mathbf{n}%
)=\overline{{\mathbb{L}}(\mathbf{n}_{1},\mathbf{n}_{2},\mathbf{n})}}%
\end{align*}
where $g\in SO(3)$, and $R^{j}(\mathbf{n}_{2},\mathbf{n})\in(Poly_{\mathbb{C}%
}(S^{2})_{\leq n})^{2}$ is the \emph{reproducing kernel}\index{Twisted j-algebras ! reproducing kernel} of the truncated
polynomial algebra (twisted j-algebra) $Poly_{\mathbb{C}}(S^{2})_{\leq n}$, characterized by%
\begin{equation}
\int_{S^{2}}R^{j}(\mathbf{n}_{2},\mathbf{n})f(\mathbf{n}_{2})d\mathbf{n}%
_{2}=f(\mathbf{n})\ ,\ \forall f\in Poly_{\mathbb{C}}(S^{2})_{\leq
n}\ \label{reproducing}%
\end{equation}

\end{proposition}

\begin{corollary}
The reproducing kernel has the following expansion
\begin{align}
R^{j}(\mathbf{n}_{1},\mathbf{n}_{2})  &  =\frac{1}{4\pi}\displaystyle{\sum
_{l=0}^{2j}\sum_{m=-l}^{l}}\overline{Y_{l}^{m}}(\mathbf{n}_{1}){Y_{l}^{m}%
}(\mathbf{n}_{2})\ \ \label{Rc}\\
&  =\frac{1}{4\pi}\displaystyle{\sum_{l=0}^{2j}}(2l+1)P_{l}(\mathbf{n}%
_{1}\cdot\mathbf{n}_{2})\ =R^{j}(\mathbf{n}_{2},\mathbf{n}_{1}) \label{leg1}%
\end{align}
where functions $P_{l}$ are the Legendre polynomials (see Chapter
\ref{genspherharm}, also \cite{VMK}), and $\mathbf{n}_{1}\cdot\mathbf{n}_{2}$
denotes the euclidean inner product of unit vectors in $\mathbb{R}^{3}$. In
particular, $R^{j}(\mathbf{n}_{1},\mathbf{n}_{2})=R^{j}(\mathbf{n}%
_{2},\mathbf{n}_{1})$.
\end{corollary}

\begin{proof}
From equation (\ref{intwistedproduct}), properties (i)-(iv) in Proposition
\ref{trikernelproperties} follow directly from properties (i)-(iv) in
Proposition \ref{twistedproperties}, together with Corollary \ref{noextension}.

Next, we observe that (\ref{Rc}) is equivalent to (\ref{reproducing}). But
(\ref{Rc}) follows from equation (\ref{Lc}) and property (iii) in Proposition
\ref{trikernelproperties}, using that%
\[
\int_{S^{2}}Y_{m}^{l}(\mathbf{n})d\mathbf{n}=0\text{, \ if }(l,m)\neq(0,0),\
\]
$\ $ together with the identity
\[
\left(  \ \
\begin{array}
[c]{ccc}%
0 & l & l\\
0 & -m & m
\end{array}
\right)  \left\{  \
\begin{array}
[c]{ccc}%
0 & l & l\\
j & j & j
\end{array}
\right\}  =\frac{(-1)^{2j+m}}{\sqrt{(2j+1)}(2l+1)}\ .
\]
(cf. equations (\ref{CG-W}) and (\ref{coupledproduct6}), and equations 8.5.1
and 9.5.1 in \cite{VMK}).

Finally, to obtain (\ref{leg1}) we use the identity
\begin{equation}
(2l+1)P_{l}(\mathbf{n}_{1}\cdot\mathbf{n}_{2})=\displaystyle{\sum_{m=-l}^{l}%
}\overline{Y_{l}^{m}}(\mathbf{n}_{1}){Y_{l}^{m}}(\mathbf{n_{2}}), \label{leg2}%
\end{equation}
which follows from equations (\ref{spherical}) and (\ref{F1}).
\end{proof}

\begin{remark}\index{Trikernels ! positive-alternate relation}
\label{positive-alternate-trikernels} (i) If $\mathbb{L}_{\vec{c}}^{j}$ is the
integral trikernel of a positive symbol correspondence, then by
(\ref{relstandaltgen}), (\ref{intgen}), and property (iv) in Proposition
(\ref{trikernelproperties}), the integral trikernel of the alternate symbol
correspondence is given by
\begin{equation}
\mathbb{L}_{\vec{c}-}^{j}=\overline{\mathbb{L}_{\vec{c}}^{j}}\ .
\label{pos-alt-trk}%
\end{equation}

(ii) For a symmetric integral trikernel ${\mathbb{L}}^{j}$ (cf. equation
(\ref{symmetrik})), as for example the standard Stratonovich trikernel (cf. equation
(\ref{L1b})), we also have that
\begin{equation}
\displaystyle{\int_{S^{2}}\mathbb{L}_{1}^{j}(\mathbf{n}_{1},\mathbf{n}%
_{2},\mathbf{n})d\mathbf{n}\ =\ R^{j}(\mathbf{n}_{1},\mathbf{n}_{2}),}%
\end{equation}
but this identity does not hold for nonsymmetric integral trikernels.
\end{remark}

\subsubsection{Explicitly $SO(3)$-invariant formulae for trikernels}\index{Trikernels ! $SO(3)$-invariant formulae |(}

Now, we turn to the $SO(3)$-invariance of the integral trikernel, namely,
according to property (i) in Proposition \ref{trikernelproperties},
$\mathbb{L}_{\vec{c}}^{j}$ can also be expressed in terms of $SO(3)$-invariant
functions. First of all, let us recall the following basic fact, whose proof
is found in a more general setting in Weyl \cite{Weyl2}.

\begin{lemma}
\label{Weylinv} Every $SO(3)$-invariant function of three points on the
$2$-sphere $S^{2}$ represented by unit vectors $\mathbf{n}_{i},i=1,2,3$, in
euclidean $3$-space, can be expressed as a function of the three euclidean
inner products $\mathbf{n}_{1}\!\cdot\!\mathbf{n}_{2}$, $\mathbf{n}_{2}%
\!\cdot\!\mathbf{n}_{3}$, $\mathbf{n}_{3}\!\cdot\!\mathbf{n}_{1}$, together
with the determinant%
\begin{equation}\label{det}
\lbrack{\mathbf{n}_{1},\mathbf{n}_{2},\mathbf{n}}_{3}]=\det({\mathbf{n}%
_{1},\mathbf{n}_{2},\mathbf{n}}_{3}) \ .
\end{equation}
\end{lemma}

Thus, we introduce two $SO(3)$-invariant functions of type (\ref{3S2}), namely%
\begin{equation}
\mathcal{T}(\mathbf{n}_{1},\mathbf{n}_{2},\mathbf{n}_{3})=(\mathbf{n}%
_{1}\!\!\cdot\!\mathbf{n}_{2})-(\mathbf{n}_{1}\!\!\cdot\!\mathbf{n}%
_{3})(\mathbf{n}_{2}\!\cdot\!\mathbf{n}_{3})-i[\mathbf{n}_{1},\mathbf{n}%
_{2},\mathbf{n}_{3}] \ , \ \text{ }i=\sqrt{-1} \ , \nonumber\label{tau}%
\end{equation}
and the following function depending on three integers $l_{k}\geq0$:
\begin{align}
&  \mathcal{L}_{l_{1},l_{2},l_{3}}(\mathbf{n}_{1},\mathbf{n}_{2}%
,\mathbf{n}_{3})\label{L_3small}\\
&  =(2l_{1}+1)(2l_{2}+1)(2l_{3}+1)\cdot\Big[\left(
\begin{array}
[c]{ccc}%
l_{1} & l_{2} & l_{3}\\
0 & 0 & 0
\end{array}
\right)  P_{l_{1}}(\mathbf{n}_{1}\cdot\mathbf{n}_{3})P_{l_{2}}(\mathbf{n}%
_{2}\cdot\mathbf{n}_{3})\nonumber\\
&  +%
%TCIMACRO{\dsum \limits_{m=1}^{\min\{l_{1},l_{2}\}}}%
%BeginExpansion
{\displaystyle\sum\limits_{m=1}^{\min\{l_{1},l_{2}\}}}
%EndExpansion
(-1)^{m}{\ }\left(
\begin{array}
[c]{ccc}%
l_{1} & l_{2} & l_{3}\\
m & -m & 0
\end{array}
\right)
%TCIMACRO{\dprod \limits_{i=1}^{2}}%
%BeginExpansion
{\displaystyle\prod\limits_{k=1}^{2}}
%EndExpansion
\sqrt{\frac{(l_{k}-m)!}{(l_{k}+m)!}}\frac{P_{l_{k}}^{m}(\mathbf{n}_{k}%
\!\cdot\!\mathbf{n}_{3})}{(1-(\mathbf{n}_{k}\!\cdot\!\mathbf{n}_{3}%
)^{2})^{m/2}}\nonumber\\
&  \cdot\{(\mathcal{T}(\mathbf{n}_{1},\mathbf{n}_{2},\mathbf{n}_{3}%
))^{m}+(-1)^{L}(\mathcal{T}(\mathbf{n}_{2},\mathbf{n}_{1},\mathbf{n}_{3}%
))^{m}\}\Big]\nonumber
\end{align}
where $L=l_{1}+l_{2}+l_{3}$. In particular, $\mathcal{L}_{l_{1},l_{2},l_{3}}$
has the following property which is obvious from (\ref{L_3small}) and the fact
that
\[
L\ \ \mbox{is odd}\ \Rightarrow\left(
\begin{array}
[c]{ccc}%
l_{1} & l_{2} & l_{3}\\
0 & 0 & 0
\end{array}
\right)  =0\ .
\]

\begin{lemma}
$\mathcal{L}_{l_{1},l_{2},l_{3}}$ is real when $L$ is even and $\mathcal{L}%
_{l_{1},l_{2},l_{3}}$ is purely imaginary when $L$ is odd.
\end{lemma}

With these preparations, we can state the following result.

\begin{theorem}
\label{invtrik} The $\vec{c}$-correspondence trikernel $\mathbb{L}_{\vec{c}%
}^{j}$ can also be written as
\begin{align}
&  \mathbb{L}_{\vec{c}}^{j}(\mathbf{n}_{1},\mathbf{n}_{2},\mathbf{n}%
_{3})\label{invtrikkk}\\
&  =(-1)^{2j}\frac{\sqrt{2j+1}}{(4\pi)^{2}}{\displaystyle\sum
\limits_{\substack{l_{1},l_{2},l_{3}=0\\\delta(l_{1},l_{2},l_{3})=1}}^{2j}%
}\left\{
\begin{array}
[c]{ccc}%
l_{1} & l_{2} & l_{3}\\
j & j & j
\end{array}
\right\}  \frac{c_{l_{3}}^{n}}{c_{l_{1}}^{n}c_{l_{2}}^{n}}\mathcal{L}%
_{l_{1},l_{2},l_{3}}(\mathbf{n}_{1},\mathbf{n}_{2},\mathbf{n}_{3})\nonumber
\end{align}
\end{theorem}

\begin{proof}
By $SO(3)$-invariance, we may assume $\mathbf{n}_{3}=\mathbf{n}_{0}%
=(0,0,1)\in\mathbb{R}^{3}$. Then the formula (\ref{invtrikkk}), including the
expression (\ref{L_3small}) for $\mathcal{L}_{l_{1},l_{2},l_{3}}$, are
obtained in a rather straightforward way from formulae (\ref{Lc})-(\ref{Lc2}).
\end{proof}

\begin{remark}
Formula (\ref{L_3small}) for the function $\mathcal{L}_{l_{1},l_{2},l_{3}}$
looks asymmetric with respect to $l_{3}$ and $\mathbf{n}_{3}$, in comparison
with $l_{1},l_{2}$ and $\mathbf{n}_{1},\mathbf{n}_{2}$. However, such
asymmetry is only apparent and reflects the asymmetrical way in which the
formula for $\mathcal{L}$ was derived. The situation resembles some well-known
formulae for the Wigner $6j$ symbol, which do not manifestly exhibit some of
the symmetries of the symbol. Indeed, we have the following symmetry
properties when $(l_{1},l_{2},l_{3})$ and $(\mathbf{n}_{1},\mathbf{n}%
_{2},\mathbf{n}_{3})$ are permuted covariantly:
\begin{equation}
\mathcal{L}_{l_{1},l_{2},l_{3}}(\mathbf{n}_{1},\mathbf{n}_{2},\mathbf{n}%
_{3})=(-1)^{L}\mathcal{L}_{l_{2},l_{1},l_{3}}(\mathbf{n}_{2},\mathbf{n}%
_{1},\mathbf{n}_{3})=(-1)^{L}\mathcal{L}_{l_{1},l_{3},l_{2}}(\mathbf{n}%
_{1},\mathbf{n}_{3},\mathbf{n}_{2}) \label{covariant}%
\end{equation}
\end{remark}

\begin{example}
\label{exF} In view of equation (\ref{covariant}) above, we list the functions
$\mathcal{L}_{l_{1},l_{2},l_{3}}$ for $l_{k}\leq2$ subject to $\delta
(l_{1},l_{2},l_{3})=1$, as follows:
\begin{align*}
\mathcal{L}_{0,0,0}(\mathbf{n}_{1},\mathbf{n}_{2},\mathbf{n}_{3})  &  =1\\
\mathcal{L}_{1,1,0}(\mathbf{n}_{1},\mathbf{n}_{2},\mathbf{n}_{3})  &
=-{3}\sqrt{3}(\mathbf{n}_{1}\cdot\mathbf{n}_{2})\\
\text{ }\mathcal{L}_{1,1,1}(\mathbf{n}_{1},\mathbf{n}_{2},\mathbf{n}_{3})  &
=i 9\sqrt{\frac{3}{2}}[\mathbf{n}_{1},\mathbf{n}_{2},\mathbf{n}_{3}] \\
\mathcal{L}_{2,1,1}(\mathbf{n}_{1},\mathbf{n}_{2},\mathbf{n}_{3})  &
=3\sqrt{\frac{15}{2}}\left\{  3(\mathbf{n}_{1}\!\!\cdot\!\mathbf{n}%
_{2})(\mathbf{n}_{1}\!\cdot\!\mathbf{n}_{3})-(\mathbf{n}_{2}\!\cdot
\!\mathbf{n}_{3})\right\} \\
\mathcal{L}_{2,2,0}(\mathbf{n}_{1},\mathbf{n}_{2},\mathbf{n}_{3})  &
=5\sqrt{5}P_{2}(\mathbf{n}_{1}\!\cdot\!\mathbf{n}_{2})=\frac{5\sqrt{5}}%
{2}\{3(\mathbf{n}_{1}\!\cdot\!\mathbf{n}_{2})^{2}-1\}\\
\mathcal{L}_{2,2,1}(\mathbf{n}_{1},\mathbf{n}_{2},\mathbf{n}_{3})  &
=-i15\sqrt{\frac{15}{2}}(\mathbf{n}_{1}\!\cdot\!\mathbf{n}_{2})[\mathbf{n}%
_{1},\mathbf{n}_{2},\mathbf{n}_{3}]\ \\
\mathcal{L}_{2,2,2}(\mathbf{n}_{1},\mathbf{n}_{2},\mathbf{n}_{3})  &
=-25\sqrt{\frac{5}{14}}{\Big \{}3{\big \{}[\mathbf{n}_{1},\mathbf{n}%
_{2},\mathbf{n}_{3}]^{2}+(\mathbf{n}_{1}\!\!\cdot\!\mathbf{n}_{2}%
)(\mathbf{n}_{2}\!\cdot\!\mathbf{n}_{3})(\mathbf{n}_{3}\!\cdot\!\mathbf{n}%
_{1}){\big \}}-1 {\Big \}}%
\end{align*}
\end{example}

The following $SO(3)$-invariant function will be used in various formulas below, in this chapter, so we denote it here as 
\begin{equation}\label{U} 
X(\mathbf{n}_{1},\mathbf{n}_{2},\mathbf{n}_{3})    =\mathbf{n}_{1}\cdot\mathbf{n}_{2}+\mathbf{n}_{2}\cdot\mathbf{n}_{3}+\mathbf{n}_{3}\cdot\mathbf{n}_{1} 
\end{equation}

\begin{example}
\label{exL} We can now calculate the Stratonovich trikernel $\mathbb{L}%
_{1}^{j}$ for the two values $j=1/2,1$, by substituting the appropriate
expressions for $\mathcal{L}_{l_{1},l_{2},l_{3}}$ listed in Example \ref{exF}
into equation (\ref{invtrikkk}):
\begin{align}
\mathbb{L}_{1}^{1/2}(\mathbf{n}_{1},\mathbf{n}_{2},\mathbf{n}_{3})  &
=\frac{1}{(4\pi)^{2}}\big\{1+3X(\mathbf{n}_{1},\mathbf{n}_{2},\mathbf{n}%
_{3})+i3\sqrt{3}[\mathbf{n}_{1},\mathbf{n}_{2},\mathbf{n}_{3}%
]\big\}\label{LSt1/2}\\
\mathbb{L}_{1}^{1}(\mathbf{n}_{1},\mathbf{n}_{2},\mathbf{n}_{3})  &  =\frac
{1}{(4\pi)^{2}}\Big\{1+3X(\mathbf{n}_{1},\mathbf{n}_{2},\mathbf{n}_{3}%
)+\frac{15}{2}Z(\mathbf{n}_{1},\mathbf{n}_{2},\mathbf{n}_{3})\label{LSt j=1}\\
&  +\frac{\sqrt{10}}{8}\Big(\big(3X(\mathbf{n}_{1},\mathbf{n}_{2}%
,\mathbf{n}_{3})-1\big)^{2}-24Z(\mathbf{n}_{1},\mathbf{n}_{2},\mathbf{n}%
_{3})-45[\mathbf{n}_{1},\mathbf{n}_{2},\mathbf{n}_{3}]^{2}\Big)\nonumber\\
&  +i\frac{9\sqrt{2}}{4}\big(1+5X(\mathbf{n}_{1},\mathbf{n}_{2},\mathbf{n}%
_{3})\big)[\mathbf{n}_{1},\mathbf{n}_{2},\mathbf{n}_{3}]\Big\}\ \ ,\nonumber
\end{align}
where $X(\mathbf{n}_{1},\mathbf{n}_{2}%
,\mathbf{n}_{3})$ is given by (\ref{U}) and 
\[
Z(\mathbf{n}_{1},\mathbf{n}_{2},\mathbf{n}_{3})    =(\mathbf{n}_{1}%
\cdot\mathbf{n}_{2})^{2}+(\mathbf{n}_{2}\cdot\mathbf{n}_{3})^{2}%
+(\mathbf{n}_{3}\cdot\mathbf{n}_{1})^{2}-1\ .\nonumber
\]

Using equation (\ref{BerezinChar}) for the Berezin characteristic numbers
$c^{n}_{l}=b^{n}_{l}$, we can similarly calculate the Berezin trikernel
$\mathbb{L}_{\vec{b}}^{j}$ for the two lowest values of $j$:
\begin{align}
\mathbb{L}_{\vec{b}}^{1/2}(\mathbf{n}_{1},\mathbf{n}_{2},\mathbf{n}_{3})  
& =\frac{1}{(4\pi)^{2}}\big\{1 + 3X^{1/2}(\mathbf{n}_{1},\mathbf{n}%
_{2},\mathbf{n}_{3}) + i9[\mathbf{n}_{1},\mathbf{n}_{2},\mathbf{n}%
_{3}]\big\}\label{LBer1/2} \\
\mathbb{L}_{\vec{b}}^{1}(\mathbf{n}_{1},\mathbf{n}_{2},\mathbf{n}_{3})  &
=\frac{1}{(4\pi)^{2}}\left(  \frac{9}{4}\right)  \Big\{\big(1 - X^{1}%
(\mathbf{n}_{1},\mathbf{n}_{2},\mathbf{n}_{3})\big)^{2}\\
&  -6\big((\mathbf{n}_{2}\!\cdot\!\mathbf{n}_{3})^{2} + (\mathbf{n}_{3}%
\!\cdot\!\mathbf{n}_{1})^{2} -2(\mathbf{n}_{1}\!\cdot\!\mathbf{n}%
_{2})\big) -25[\mathbf{n}_{1},\mathbf{n}_{2},\mathbf{n}_{3}]^{2}\nonumber\\
&  + i2\big(1+5X^{1}(\mathbf{n}_{1},\mathbf{n}_{2},\mathbf{n}_{3}%
)\big)[\mathbf{n}_{1},\mathbf{n}_{2},\mathbf{n}_{3}]\Big\} \ \ ,\nonumber
\end{align}
where
$$
X^{j}(\mathbf{n}_{1},\mathbf{n}_{2},\mathbf{n}_{3}) = (4j+1)(\mathbf{n}%
_{1}\cdot\mathbf{n}_{2})+(\mathbf{n}_{2}\cdot\mathbf{n}_{3})+(\mathbf{n}%
_{3}\cdot\mathbf{n}_{1}) \ 
$$
\end{example}

Except for $\mathbb{L}_{1}^{1/2}$, which was first presented in \cite{Strat},
and $\mathbb{L}_{\vec{b}}^{1/2}$, which can be inferred from \cite{Wild}, the
formulae in Examples \ref{exF} and \ref{exL} were originally obtained by
Nazira Harb. Further formulae for $\mathcal{L}$, for higher values of $l_{i}$,
and further formulae for the Stratonovich and Berezin trikernels, for higher
values of $j$, can be found in Harb's thesis \cite{NH}.\index{Trikernels ! $SO(3)$-invariant formulae |)}

\subsection{Recursive trikernels and transition kernels}

\label{Wildtrikernels}

\subsubsection{Recursive trikernels}

A recursive trikernel associated to the standard Berezin correspondence has
been defined by Wildberger \cite{Wild}. This trikernel is symmetric, in the
sense of equation (\ref{symmetrik}), but is not the trikernel of a bona-fide
integral equation for the twisted product in the sense of equation
(\ref{intgen}); rather, it is the trikernel of a recursive integral equation
for the standard Berezin twisted product.

Let $B$ be the standard Berezin symbol correspondence, given by equation
(\ref{Berezin1}) in Theorem \ref{Berezincorr1}. We recall from Remark
\ref{inducedinnprod} that, because this symbol correspondence is not
isometric, there are two natural $SU(2)$-invariant inner products on the space
of symbols $Poly_{\mathbb{C}}(S^{2})_{\leq n}$, namely

\begin{itemize}
\item the usual $L^{2}$-inner product on $C_{\mathbb{C}}^{\infty}%
(S^{2})\supset Poly_{\mathbb{C}}(S^{2})_{\leq n}$ (cf. equation (\ref{L2})):
\begin{equation}
\langle f_{1},f_{2}\rangle=\frac{1}{4\pi}\int_{S^{2}}\overline{f_{1}}%
f_{2}\ dS=\frac{1}{4\pi}\int_{S^{2}}\overline{f_{1}(\mathbf{n})}%
f_{2}(\mathbf{n})d\mathbf{n} \label{L2inn}%
\end{equation}

\item the induced inner product on $Poly_{\mathbb{C}}(S^{2})_{\leq n}\subset
C_{\mathbb{C}}^{\infty}(S^{2})$ (cf. equation (\ref{inducedinnerproduct})):
\begin{equation}
\langle f_{1},f_{2}\rangle_{\star_{\vec{b}}^{n}}=\frac{1}{4\pi}\int_{S^{2}%
}\overline{f_{1}}\star_{\vec{b}}^{n}f_{2}\ dS=\frac{1}{4\pi}\int_{S^{2}%
}\overline{f_{1}(\mathbf{n})}\star_{\vec{b}}^{n}f_{2}(\mathbf{n})d\mathbf{n}
\label{indinn}%
\end{equation}

\end{itemize}

Associated with each of these two inner product on $Poly_{\mathbb{C}}%
(S^{2})_{\leq n}$ defined by integration, one can define an integral trikernel
$\mathbb{L}$, cf. (\ref{3S2}). The first one is the Berezin trikernel
$\mathbb{L}_{\vec{b}}^{j}$, also defined implicitly by (cf. equation
(\ref{intgen})):
\begin{equation}
f_{1}\star_{\vec{b}}^{n}f_{2}\ (\mathbf{n})=\iint_{S^{2}\times S^{2}}%
f_{1}(\mathbf{n}_{1})f_{2}(\mathbf{n}_{2})\mathbb{L}_{\vec{b}}^{j}%
(\mathbf{n}_{1},\mathbf{n}_{2},\mathbf{n})d\mathbf{n}_{1}d\mathbf{n}_{2}%
\end{equation}

\begin{definition}
\label{defnWtrik} \emph{Wildberger's recursive trikernel}\index{Trikernels ! Wildberger's recursive} is the function
\[
\mathbb{T}_{\vec{b}}^{j}:S^{2}\times S^{2}\times S^{2}\rightarrow\mathbb{C}%
\]
defined implicitly by
\begin{equation}
f_{1}\star_{\vec{b}}^{n}f_{2}\ (\mathbf{n})=\int_{S^{2}}\left\{  \int_{S^{2}%
}\mathbb{T}_{\vec{b}}^{j}(\mathbf{n}_{1},\mathbf{n}_{2},\mathbf{n})\star
_{\vec{b}}^{n}f_{1}(\mathbf{n}_{1})d\mathbf{n}_{1}\right\}  \star_{\vec{b}%
}^{n}f_{2}(\mathbf{n}_{2})d\mathbf{n}_{2} \label{Wtrik1}%
\end{equation}
where each twisted product $\star_{b}^{n}$ on the r.h.s. of the above equation
is taken with respect to the variables being integrated.
\end{definition}

One obtains, of course, an explicit expression for the Berezin trikernel
$\mathbb{L}_{b}^{j}$ from the equations (\ref{Lc})-(\ref{Lb}). Similarly,
Wildberger's trikernel also has such an explicit expression.

\begin{proposition}
Wildberger's recursive trikernel is given by:
\begin{align}
&  \mathbb{T}_{\vec{b}}^{j}(\mathbf{n}_{1},\mathbf{n}_{2},\mathbf{n}%
_{3})\label{Tb1}\\
&  =\frac{(-1)^{n}\sqrt{n+1}}{(4\pi)^{2}}\sum_{l_{i},m_{i}}\left[
\begin{array}
[c]{ccc}%
l_{1} & l_{2} & l_{3}\\
m_{1} & m_{2} & m_{3}%
\end{array}
\right]  \!\!{[j]}\ b_{l_{1}}^{n}b_{l_{2}}^{n}b_{l_{3}}^{n}\ \overline
{Y_{l_{1}}^{m_{1}}}(\mathbf{n}_{1})\overline{Y_{l_{2}}^{m_{2}}}(\mathbf{n}%
_{2})\overline{Y_{l_{3}}^{m_{3}}}(\mathbf{n}_{3})\quad\quad\nonumber
\end{align}
where $b_{l}^{n}$ are the characteristic numbers of the standard Berezin
correspondence, given by equation (\ref{BerezinChar}).
\end{proposition}

\begin{proof}
From equation (\ref{Wtrik1}), we reason analogously to the proof of Corollary
\ref{corstratrik}, with the difference that, for the induced inner product
(\ref{indinn}), the orthonormal basis vectors are not $Y_{l}^{m}$, but
$b_{l}^{n}Y_{l}^{m}$, with $b_{l}^{n}$ given by (\ref{BerezinChar}).
\end{proof}

Wildberger's recursive trikernel can be generalized to a recursive
trikernel for any symbol correspondence by obvious analogy to the Berezin case:

\begin{corollary}\index{Trikernels ! recursive (general)}
Let $W_{\vec{c}}^{j}$ be the symbol correspondence given by characteristic
numbers $c_{l}^{n}\in\mathbb{R}^{\ast}$. Its recursive integral trikernel is
the polynomial function $\mathbb{T}_{\vec{c}}^{j}\in(Poly_{\mathbb{C}}%
(S^{2})_{\leq n})^{3}$ given by
\begin{align}
&  \mathbb{T}_{\vec{c}}^{j}(\mathbf{n}_{1},\mathbf{n}_{2},\mathbf{n}_{3})\\
&  =\frac{(-1)^{n}\sqrt{n+1}}{(4\pi)^{2}}\sum_{l_{i},m_{i}}\left[
\begin{array}
[c]{ccc}%
l_{1} & l_{2} & l_{3}\\
m_{1} & m_{2} & m_{3}%
\end{array}
\right]  \!\!{[j]}\ c_{l_{1}}^{n}c_{l_{2}}^{n}c_{l_{3}}^{n}\ \overline
{Y_{l_{1}}^{m_{1}}}(\mathbf{n}_{1})\overline{Y_{l_{2}}^{m_{2}}}(\mathbf{n}%
_{2})\overline{Y_{l_{3}}^{m_{3}}}(\mathbf{n}_{3})\quad\quad\nonumber
\end{align}
with the usual (\ref{restrictsum}) restrictions on the $l_{k},m_{k}$
summations. It is symmetric in the sense of (\ref{symmetrik}) and is
implicitly defined by
\begin{equation}
f_{1}\star_{\vec{c}}^{n}f_{2}\ (\mathbf{n})=\int_{S^{2}}\left\{  \int_{S^{2}%
}\mathbb{T}_{\vec{c}}^{j}(\mathbf{n}_{1},\mathbf{n}_{2},\mathbf{n})\star
_{\vec{c}}^{n}f_{1}(\mathbf{n}_{1})d\mathbf{n}_{1}\right\}  \star_{\vec{c}%
}^{n}f_{2}(\mathbf{n}_{2})d\mathbf{n}_{2},\quad\quad\label{recursiveintegral}%
\end{equation}
where each twisted product $\star_{\vec{c}}^{n}$ on the r.h.s. of the above
equation is taken with respect to the variables being integrated.
\end{corollary}

Every recursive trikernel can be written more explicitly, by
performing the following substitution into equation (\ref{Lc2}), 
\begin{equation}\label{subgenrec}
\frac{c_{l_{3}}^{n}}{c_{l_{1}}^{n}c_{l_{2}}^{n}}\rightarrow {c_{l_{1}}^{n}c_{l_{2}}^{n}}{c_{l_{3}}^{n}}  \ .
\end{equation}
Or it can be rewritten in terms of $SO(3)$-invariant functions, by performing the substitution 
 (\ref{subgenrec}) into equation (\ref{invtrikkk}).
For Wildberger's recursive trikernel, the above substitution to be performed in equations  (\ref{Lc2}) and (\ref{invtrikkk}) reads:  
 \begin{equation}
\frac{c_{l_{3}}^{n}}{c_{l_{1}}^{n}c_{l_{2}}^{n}}\rightarrow b_{l_{1}}%
^{n}b_{l_{2}}^{n}b_{l_{3}}^{n}=(n!\sqrt{n+1})^{3}%
%TCIMACRO{\dprod \limits_{i=1}^{3}}%
%BeginExpansion
{\displaystyle\prod\limits_{k=1}^{3}}
%EndExpansion
\frac{1}{\sqrt{(n+l_{k}+1)!(n-l_{k})!}} \ . \label{subWild}%
\end{equation}

Next, we will express the integrand in (\ref{recursiveintegral}) using only
ordinary (commutative) multiplication of functions, that is, no twisted product
involved. To this end, we shall invoke the duality $f\longleftrightarrow
\tilde{f}$ between symbols in $Poly_{\mathbb{C}}(S^{2})_{\leq n}$, with
reference to Remark \ref{CCdual} and equation (\ref{kernel-U}).

\begin{lemma}
If $f,g\in Poly_{\mathbb{C}}(S^{2})_{\leq n}$, then
\[
\int_{S^{2}}f(\mathbf{n})\star_{\vec{c}}^{n}g(\mathbf{n)}d\mathbf{n=}%
\int_{S^{2}}f(\mathbf{n})\tilde{g}(\mathbf{n})d\mathbf{n=}\int_{S^{2}}%
\tilde{f}(\mathbf{n})g(\mathbf{n})d\mathbf{n,}%
\]
where $\tilde{f}$ denotes the contravariant dual of $f$ with respect to
$W_{\vec{c}}^{j}$, that is,%
\begin{equation}
\tilde{f}=U_{\vec{c},\frac{1}{\vec{c}}}^{j}(f)\text{, \ or \ }{f}=U_{\frac
{1}{\vec{c}},\vec{c}}^{j}(\tilde{f})\ \text{ }, \label{dual-f} 
\end{equation}
cf. Definition \ref{TransOperator}. 
\end{lemma}

\begin{proof}
Assume $P,Q\in M_{\mathbb{C}}(n+1)$ correspond to $f,g$ via $W_{\vec{c}}^{j}$,
respectively, and hence by (\ref{CCduality})
\begin{align*}
\int_{S^{2}}f(\mathbf{n})\star_{\vec{c}}^{n}g(\mathbf{n)}d\mathbf{n}  &
\mathbf{=}\int_{S^{2}}W_{\vec{c}}(PQ)dS=\frac{1}{n+1}trace(PQ)=\frac{1}%
{n+1}trace(QP)\\
&  =\int_{S^{2}}\widetilde{W_{\vec{c}}}(Q)W_{\vec{c}}(P)dS=\int_{S^{2}%
}W_{\frac{1}{\vec{c}}}(Q)W_{\vec{c}}(P)dS\\
&  =\int_{S^{2}}\tilde{g}(\mathbf{n})f(\mathbf{n})d\mathbf{n=}\int_{S^{2}%
}\tilde{f}(\mathbf{n})g(\mathbf{n})d\mathbf{n}%
\end{align*}
\end{proof}

By applying this lemma twice, the recursive integral equation
(\ref{recursiveintegral}) can also be written as
\begin{equation}
f_{1}\star_{\vec{c}}^{n}f_{2}\ (\mathbf{n})=\iint_{S^{2}\times S^{2}%
}\mathbb{T}_{\vec{c}}^{j}(\mathbf{n}_{1},\mathbf{n}_{2},\mathbf{n})\tilde
{f}_{1}(\mathbf{n}_{1})\tilde{f}_{2}(\mathbf{n}_{2})d\mathbf{n}_{1}%
d\mathbf{n}_{2}\ . \label{CCintegralproduct}%
\end{equation}

\subsubsection{Transition kernels}\index{Twisted j-algebras ! transition kernels |(}

On the other hand, we shall also make special use of the transition kernel
\[
\mathbb{U}_{\vec{c},\frac{1}{\vec{c}}}^{j}:S^{2}\times S^{2}\rightarrow
\mathbb{C}%
\]
which is associated with the integral form of the transformation
$f\rightarrow\tilde{f}$ (cf. (\ref{dual-f})), as follows:

\begin{proposition}
The duality transformations (\ref{dual-f}) can be written in integral form as
\begin{align}
\tilde{f}(\mathbf{n}_{1})  &  =\int_{S^{2}}\mathbb{U}_{\vec{c},\frac{1}%
{\vec{c}}}^{j}(\mathbf{n}_{1},\mathbf{n}_{2})f(\mathbf{n}_{2})d\mathbf{n}%
_{2}\ ,\label{CCintegraltransf}\\
{f}(\mathbf{n}_{1})  &  =\int_{S^{2}}\mathbb{U}_{\frac{1}{\vec{c}},\vec{c}%
}^{j}(\mathbf{n}_{1},\mathbf{n}_{2})\tilde{f}(\mathbf{n}_{2})d\mathbf{n}%
_{2}\ \label{CCintegraltransf-1}%
\end{align}
with the following transition kernels%
\begin{align}
\mathbb{U}_{\vec{c},\frac{1}{\vec{c}}}^{j}(\mathbf{n}_{1},\mathbf{n}_{2})  &
=\frac{1}{4\pi}\displaystyle{\sum_{l=0}^{2j}}\frac{2l+1}{(c_{l}^{n})^{2}}%
P_{l}(\mathbf{n_{1}}\cdot\mathbf{n}_{2})\ ,\label{Ucc}\\
\mathbb{U}_{\frac{1}{\vec{c}},\vec{c}}^{j}(\mathbf{n}_{1},\mathbf{n}_{2})  &
=\frac{1}{4\pi}\displaystyle{\sum_{l=0}^{2j}}{(c_{l}^{n})^{2}(2l+1)}%
P_{l}(\mathbf{n_{1}}\cdot\mathbf{n}_{2})\ \label{Ucc-1}%
\end{align}
\end{proposition}

\begin{proof}
Let $P\in M_{\mathbb{C}}(n+1)$, $P=\sum_{l=0}^{n} P_{l}$, where each $P_{l}\in
M_{\mathbb{C}}(\varphi_{l})$ decomposes as $P_{l}=\sum_{m=-l}^{l}
p_{lm}\mathbf{e}(l,m)$. Let $f\in Poly_{\mathbb{C}}(S^{2})_{\leq n}$ be the
symbol of $P$ via the correspondence determined by characteristic numbers
$\vec c$. Then, $f=\sum_{l=0}^{n} f_{l}$, where each $f_{l}\in Poly(\varphi
_{l})$ decomposes as $f_{l}=\sum_{m=-l}^{l} \frac{1}{\sqrt{n+1}}c^{n}%
_{l}p_{lm}Y_{lm}$.

It now follows straightforwardly from (\ref{Rc})-(\ref{leg1}),
(\ref{CCintegraltransf}), and (\ref{Ucc}) that $\tilde{f}=\sum_{l=0}^{n}%
\tilde{f}_{l}$, where $\tilde{f}_{l}=\sum_{m=-l}^{l}\frac{1}{\sqrt{n+1}}%
\frac{p_{lm}}{c_{l}^{n}}Y_{lm}$. Therefore, $\tilde{f}$ is the contravariant
dual of $f$. Similarly for (\ref{CCintegraltransf-1}) and (\ref{Ucc-1}).
\end{proof}

Note that in the above proof, we may as well replace $\frac{1}{\vec{c}}$ by
any other string $\vec{c}^{ \ \prime}$of characteristic numbers. In other
words, by similar reasoning we generalize equations (\ref{CCintegraltransf})
and (\ref{CCintegraltransf-1}), as follows:

\begin{proposition}\label{transitionrelations}\index{Twisted j-algebras ! transition operators}
The \emph{transition operator} $U_{\vec{c},\vec{c}^{\prime}}^{j}=W_{\vec
{c}^{\prime}}^{j}\circ(W_{\vec{c}}^{j})^{-1}$, cf. equation (\ref{kernel-U}),
which takes the symbol $f=W_{\vec{c}}^{j}[P]$ of $P$ in correspondence defined
by $\vec{c}$ to the symbol $f^{\prime}=W_{\vec{c}^{\prime}}^{j}[P]$ in
correspondence defined by $\vec{c}^{\ \prime}$, can be written in integral form
as
\begin{equation}
f^{\prime}(\mathbf{n}_{1})=\int_{S^{2}}\mathbb{U}_{\vec{c},\vec{c}^{\prime}%
}^{j}(\mathbf{n}_{1},\mathbf{n}_{2})f(\mathbf{n}_{2})d\mathbf{n}_{2}\ ,
\label{integraltransition}%
\end{equation}
where
\begin{equation}
\mathbb{U}_{\vec{c},{\vec{c}}^{\prime}}^{j}(\mathbf{n}_{1},\mathbf{n}%
_{2})=\frac{1}{4\pi}\displaystyle{\sum_{l=0}^{2j}}\frac{(c_{l}^{n})^{\prime}%
}{c_{l}^{n}}(2l+1)P_{l}(\mathbf{n_{1}}\cdot\mathbf{n}_{2})\ .
\label{integraliso}%
\end{equation}
Moreover, for any pair $(\vec{c},\vec{c}^{\prime})$ the \emph{transition
kernels} $\mathbb{U}_{\vec{c},{\vec{c}}^{\prime}}^{j}$ and $\mathbb{U}%
_{\vec{c}^{\prime},{\vec{c}}}^{j}$ yield the reproducing kernel by the formula%
\begin{equation}
R^{j}(\mathbf{n}_{1},\mathbf{n}_{2})\ =\int_{S^{2}}\mathbb{U}_{\vec{c}%
,{\vec{c}}^{\prime}}^{j}(\mathbf{n}_{1},\mathbf{n})\mathbb{U}_{\vec{c}%
^{\prime},{\vec{c}}}^{j}(\mathbf{n},\mathbf{n}_{2})d\mathbf{n}\ \text{, cf.
(\ref{Rc})-(\ref{leg1}). } \label{integralinverse}%
\end{equation}
Furthermore, the composition of transition operators $U_{\vec{c},{\vec{c}%
}^{\prime}}^{j}$ yields the following ``composition'' rule for transition
kernels
\begin{equation}
\int_{S^{2}}\mathbb{U}_{\vec{c},{\vec{c}}^{\prime}}^{j}(\mathbf{n}%
_{1},\mathbf{n}_{2})\mathbb{U}_{\vec{c}^{\prime},{\vec{c}}^{\prime\prime}}%
^{j}(\mathbf{n}_{2},\mathbf{n}_{3})d\mathbf{n}_{2}=\mathbb{U}_{\vec{c}%
,{\vec{c}}^{\prime\prime}}^{j}(\mathbf{n}_{1},\mathbf{n}_{3})\ .
\label{integralcompose}%
\end{equation}
\end{proposition}

\begin{remark}
\label{Rcl} Note that the reproducing kernel $R^{j}(\mathbf{n}_{1}%
,\mathbf{n}_{2})$ is an $SO(3)$-invariant function on $S^{2}\times S^{2}%
$\ which decomposes as
\begin{equation}
R^{j}(\mathbf{n}_{1},\mathbf{n}_{2})=\displaystyle{\sum_{l=0}^{n}}R_{l}%
^{j}(\mathbf{n}_{1},\mathbf{n}_{2})\text{, \ \ }\ R_{l}^{j}(\mathbf{n}%
_{1},\mathbf{n}_{2})=\frac{2l+1}{4\pi}P_{l}(\mathbf{n}_{1}\cdot\mathbf{n}%
_{2})\text{,} \label{projreprodkernel}%
\end{equation}
so that for every symbol $f\in Poly_{\mathbb{C}}(S^{2})_{\leq n}$ which is
decomposed into the $l$-invariant subspaces as $f=\sum_{l=0}^{n}f_{l}$, we
have
\begin{equation}
\int_{S^{2}}f(\mathbf{n}_{1})R_{l}^{j}(\mathbf{n}_{1},\mathbf{n}%
_{2})d\mathbf{n}_{1}=f_{l}(\mathbf{n}_{2})\ . \label{ProjRc}%
\end{equation}
\end{remark}

Combining equations (\ref{CCintegralproduct})-(\ref{CCintegraltransf-1}), we obtain

\begin{proposition}
\label{bonarer-rel} The bona-fide and the recursive integral trikernels for a
twisted product $\star^{n}_{\vec{c}}$ are related by the integral equations
\begin{equation}
\label{bonarec}\mathbb{L}^{j}_{\vec{c}}(\mathbf{n}_{1}, \mathbf{n}_{2},
\mathbf{n}) = \int_{S^{2}\times S^{2}} \mathbb{U}^{j}_{\vec{c},\frac{1}%
{\vec{c}}} (\mathbf{n}_{1},\mathbf{n}_{1}^{\prime}) \mathbb{U}^{j}_{\vec
{c},\frac{1}{\vec{c}}} (\mathbf{n}_{2},\mathbf{n}_{2}^{\prime}) \mathbb{T}%
_{\vec{c}}^{j}(\mathbf{n}_{1}^{\prime},\mathbf{n}_{2}^{\prime}, \mathbf{n}%
)d\mathbf{n}_{1}^{\prime}d\mathbf{n}_{2}^{\prime}\ ,
\end{equation}
\begin{equation}
\label{recbona}\mathbb{T}^{j}_{\vec{c}}(\mathbf{n}_{1}, \mathbf{n}_{2},
\mathbf{n}) = \int_{S^{2}\times S^{2}} \mathbb{U}^{j}_{\frac{1}{\vec{c}},
\vec{c}} (\mathbf{n}_{1},\mathbf{n}_{1}^{\prime}) \mathbb{U}^{j}_{\frac
{1}{\vec{c}}, \vec{c}} (\mathbf{n}_{2},\mathbf{n}_{2}^{\prime}) \mathbb{L}%
_{\vec{c}}^{j}(\mathbf{n}_{1}^{\prime},\mathbf{n}_{2}^{\prime}, \mathbf{n}%
)d\mathbf{n}_{1}^{\prime}d\mathbf{n}_{2}^{\prime}\ .
\end{equation}
\end{proposition}

The following proposition is obtained straightforwardly from the above
equation (\ref{integraliso}) for the transition kernels $\mathbb{U}_{\vec
{c},{\vec{c}}^{\prime}}^{j}$, see also equation (\ref{leg2}), and the general
formula (\ref{Lc}) for the twisted product integral trikernels $\mathbb{L}%
_{\vec{c}}^{j}$.

\begin{proposition}
\label{propLcc'} The bona-fide integral trikernels of two distinct twisted
products $\star_{\vec{c}}^{n}$ and $\star_{{\vec{c}^{\prime}}}^{n}$ are
related by
\begin{align}
&  \mathbb{L}_{\vec{c}}^{j}(\mathbf{n}_{1},\mathbf{n}_{2},\mathbf{n}%
_{3})\label{Lcc'}\\
&  =\int_{S^{2}\times S^{2}\times S^{2}}\mathbb{U}_{\vec{c},{\vec{c}}^{\prime
}}^{j}(\mathbf{n}_{1},\mathbf{n}_{1}^{\prime})\mathbb{U}_{\vec{c},{\vec{c}%
}^{\prime}}^{j}(\mathbf{n}_{2},\mathbf{n}_{2}^{\prime})\mathbb{U}_{{\vec
{c}^{\prime}},{\vec{c}}}^{j}(\mathbf{n}_{3},\mathbf{n}_{3}^{\prime}%
)\mathbb{L}_{\vec{c}^{\prime}}^{j}(\mathbf{n}_{1}^{\prime},\mathbf{n}%
_{2}^{\prime},\mathbf{n}_{3}^{\prime})d\mathbf{n}_{1}^{\prime}d\mathbf{n}%
_{2}^{\prime}d\mathbf{n}_{3}^{\prime}\ ,\nonumber
\end{align}
\end{proposition}

Similarly for the relation between recursive trikernels of different symbol
correspondences and also for the relation between the bona-fide trikernel of
one correspondence and the recursive trikernel of another correspondence, by
appropriate use of transition kernels (\ref{integraliso}), (\ref{Ucc}%
)-(\ref{Ucc-1}).\index{Twisted j-algebras ! transition kernels |)}

\subsubsection{Duality of twisted products}

In this way it is easy to see that, for any non-self-dual symbol
correspondence, with its twisted product $\star_{\vec{c}}^{n}$, the
covariant-contravariant duality is not a twisted product homomorphism, but
instead, we have the following

\begin{corollary}\index{Twisted products ! covariant-contravariant duality}
\label{CCtwistedquasihomo} The covariant-contravariant duality satisfies
\begin{equation}
\widetilde{f\star_{\vec{c}}^{n}g}=\tilde{f}\star_{\frac{1}{\vec{c}}}^{n}%
\tilde{g}\ . \label{CCtwistedqhomo}%
\end{equation}
\end{corollary}
\begin{proof}
This follows straightforwardly from equations (\ref{CCintegralproduct}%
)-(\ref{Lcc'}).
\end{proof}

\begin{definition}\label{Toeplitz}\index{Twisted products ! Toeplitz}  In view of equation (\ref{CCtwistedqhomo}) above, the twisted product $\star_{\frac{1}{\vec{c}}}^{n}$ shall be called the \emph{dual twisted product} of  $\star_{\vec{c}}^{n}$, also denoted $\tilde{\star}_{\vec{c}}^{n}$. The dual  of the standard Berezin twisted product shall be called the \emph{standard Toeplitz twisted product}, denoted $\star_{\frac{1}{\vec{b}}}^{n}$ , and similarly for the dual of the alternate Berezin twisted product. 
\end{definition} 

\begin{remark}\label{Toeplitz1}  For a given symbol correspondence with characteristic numbers $\vec{c}$, the twisted product $\star_{\vec{c}}^{n}$ is the natural product of covariant symbols, while the dual twisted product  $\star_{\frac{1}{\vec{c}}}^{n}$ is the natural product of contravariant symbols, which are induced from the operator product. In the literature, the association of an operator to a function via the contravariant Berezin symbol correspondence (as given by equation (\ref{contravariant}), but see also Theorem \ref{CC} and Remark \ref{CCdual})   
 is known  as Toeplitz quantization (see, e.g. \cite{BG, BMS} and references therein).   
 \end{remark}
 
 \begin{definition}\label{Toeplitzcorr}\index{Symbol correspondences ! Toeplitz symbol}
 According to Definition \ref{Toeplitz} and Remark \ref{Toeplitz1} above, the symbol correspondence with characteristic numbers $\frac{1}{\vec{b}}$ shall be called the \emph{standard Toeplitz symbol correspondence} and the one with characteristic numbers  $\frac{1}{\vec{b}-}$ shall be called the \emph{alternate Toeplitz symbol correspondence}. 
 \end{definition} 
 
  Explicit expressions for the standard Toeplitz twisted product of spherical harmonics and the standard Toeplitz integral trikernel are obtained from equations  (\ref{genprod}), (\ref{intgen}), (\ref{Lc2}) and (\ref{L_3small})-(\ref{invtrikkk}), via the following substitution: 
 \begin{align}
\displaystyle{\frac{c_{l_{3}}^{n}}{c_{l_{1}}^{n}c_{l_{2}}^{n}}}\to
\displaystyle{\frac{b_{l_{1}}^{n}b_{l_{2}}^{n}}{b_{l_{3}}^{n}}}  &
=\displaystyle\sqrt{\frac{\binom{n}{l_{1}}\binom{n}{l_{2}}\binom{n+l_{3}+1}{l_{3}}}{\binom{n+l_{1}+1}{l_{1}}\binom{n+l_{2}+1}{l_{2}}\binom{n}{l_{3}}}%
}\label{L1/b}\\
&  =\displaystyle\sqrt{\frac{(n+l_{3}+1)!(n-l_{3})!(n+1)!n!}{(n+l_{1}+1)!(n+l_{2}+1)!(n-l_{1})!(n-l_{2}%
)!}}\nonumber
\end{align}

\subsection{Other formulae related to integral trikernels}\label{integralintegral}

\subsubsection{Integral formulae for trikernels}

Returning to Wildberger's recursive trikernel (i.e., for the standard Berezin
correspondence), its nicest feature is its simple closed formula, presented in
the following proposition, whose proof is delegated to Appendix
\ref{proofWild}.

\begin{proposition}[\cite{Wild}]\label{BertrikWild}\index{Trikernels ! Wildberger's recursive} Recalling the notation from Lemma
\ref{Weylinv}, Wildberger's recursive trikernel is the function \ $\mathbb{T}%
_{\vec{b}}^{j}:\ S^{2}\times S^{2}\times S^{2}\ \rightarrow\ \mathbb{C}$
\ given by
\begin{equation}
\mathbb{T}_{\vec{b}}^{j}(\mathbf{n}_{1},\mathbf{n}_{2},\mathbf{n}_{3})=\left(
\frac{n+1}{2^{n}4\pi}\right)  ^{2}\Big(\ 1+X(\mathbf{n}_{1},\mathbf{n}%
_{2},\mathbf{n}_{3})+i[\mathbf{n}_{1},\mathbf{n}_{2},\mathbf{n}_{3}%
]\ \Big)^{n} \label{closedBertrik}%
\end{equation}
where $i=\sqrt{-1}$ and $X(\mathbf{n}_{1},\mathbf{n}%
_{2},\mathbf{n}_{3})$ is defined by (\ref{U}) and $[\mathbf{n}_{1},\mathbf{n}_{2},\mathbf{n}_{3}]$ is the $3\times 3$ determinant. One should compare the
simple closed formula (\ref{closedBertrik}) with the formulae that are
obtained from (\ref{Lc2}) and (\ref{invtrikkk}) via substitution
(\ref{subWild}).
\end{proposition}

Then, an integral equation for general bona-fide trikernels is obtained from
Wildberger's recursive trikernel (\ref{closedBertrik}) and equations
(\ref{BerezinChar}) and (\ref{CCintegraltransf})-(\ref{Lcc'}).

\begin{theorem}\index{Trikernels ! integral formulae |(}
\label{SBW} The $\vec{c}$-correspondence trikernel $\mathbb{L}_{\vec{c}}^{j}$
is itself also expressed by the following sum of integrals:
\begin{align}
& \mathbb{L}_{\vec{c}}^{j}(\mathbf{n}_{1},\mathbf{n}_{2},\mathbf{n}_{3})=\left(
\frac{n+1}{2^{n}4\pi}\right)  ^{2}\displaystyle{\sum_{\substack{l_{1}%
,l_{2},l_{3}=0}}^{n}}\frac{c_{l_{3}}^{n}}{c_{l_{1}}^{n}c_{l_{2}}^{n}}\frac
{1}{b_{l_{1}}^{n}b_{l_{2}}^{n}b_{l_{3}}^{n}}{\small \ \mathcal{I}_{l_{1}%
,l_{2},l_{3}}^{n}(\mathbf{n}_{1},\mathbf{n}_{2},\mathbf{n}_{3})}\nonumber \\
%\end{equation}%
%\begin{equation}
& =\frac{\sqrt{n+1}}{(n!)^{3}4^{n}(4\pi)^{2}}\displaystyle{\sum_{\substack{l_{1}%
,l_{2},l_{3}=0}}^{n}}\frac{c_{l_{3}}^{n}}{c_{l_{1}}^{n}c_{l_{2}}^{n}%
}{\displaystyle\prod\limits_{k=1}^{3}} \sqrt{(n+l_{k}+1)!(n-l_{k})!}
\ \mathcal{I}_{l_{1},l_{2},l_{3}}^{n}(\mathbf{n}_{1},\mathbf{n}_{2}%
,\mathbf{n}_{3}) \label{GenfromWild}%
\end{align}
where
\begin{align}
{\small \mathcal{I}_{l_{1},l_{2},l_{3}}^{n}}(\mathbf{n}_{1},\mathbf{n}%
_{2},\mathbf{n}_{3})  &  =\frac{1}{(4\pi)^{3}}\int_{S^{2}\times S^{2}\times
S^{2}}%
%TCIMACRO{\dprod \limits_{k=1}^{3}}%
%BeginExpansion
{\displaystyle\prod\limits_{k=1}^{3}}
%EndExpansion
(2l_{k}+1)P_{l_{k}}(\mathbf{n}_{k}\cdot\mathbf{n}_{k}^{\prime})\label{Inlll}\\
&  \cdot\Big(\ 1+X\mathbf{(n}_{1}^{\prime},\mathbf{n}_{2}^{\prime}%
,\mathbf{n}_{3}^{\prime})+i[\mathbf{n}_{1}^{\prime},\mathbf{n}_{2}^{\prime
},\mathbf{n}_{3}^{\prime}]\ \Big)^{n}d\mathbf{n}_{1}^{\prime}d\mathbf{n}%
_{2}^{\prime}d\mathbf{n}_{3}^{\prime}\nonumber
\end{align}
\end{theorem}

\

The standard Stratonovich trikernel is obtained from the above formula by
setting all $c_{l}^{n}=1$, whereas for the standard Berezin case, $c_{l}%
^{n}=b_{l}^{n}$ is given by (\ref{BerezinChar}), so that
%\begin{align*}
$$\frac{c_{l_{3}}^{n}}{c_{l_{1}}^{n}c_{l_{2}}^{n}}\frac{1}{b_{l_{1}}^{n}%
b_{l_{2}}^{n}b_{l_{3}}^{n}}    =\frac{1}{(b_{l_{1}}^{n})^{2}(b_{l_{2}}%
^{n})^{2}}
%=\frac{\binom{n+l_{1}+1}{l_{1}}\binom{n+l_{2}+1}{l_{2}}}{\binom{n}{l_{1}}\binom{n}{l_{2}}}\\
=\frac{(n+l_{1}+1)!(n-l_{1})!(n+l_{2}+1)!(n-l_{2})!}{(n+1)^{2}(n!)^{4}}$$
%\end{align*}
and the summation in $l_{3}$ yields the reproducing kernel (cf. equation
(\ref{leg1})), which eliminates one integration on $S^{2}$. Thus, the
standard Berezin trikernel is given by
\begin{align}
&  \mathbb{L}_{\vec{b}}^{j}(\mathbf{n}_{1},\mathbf{n}_{2},\mathbf{n}%
_{3})\label{BerfromWild}\\
&  = \frac{1}{(n!)^{4}4^{n}(4\pi)^{4}} \displaystyle{\sum_{l_{1},l_{2}=0}^{n}}
%\ {\displaystyle\prod\limits_{k=1}^{2}}
(n+l_{1}+1)!(n-l_{1})!(2l_{1}+1)(n+l_{2}+1)!(n-l_{2})!(2l_{2}+1)
\ \nonumber \\
%\end{align}%
%\[
& \cdot\int_{S^{2}\times S^{2}}P_{l_{1}}(\mathbf{n}_{1}\cdot\mathbf{n}%
_{1}^{\prime})P_{l_{2}}(\mathbf{n}_{2}\cdot\mathbf{n}_{2}^{\prime
})\Big(\ 1+X(\mathbf{n}_{1}^{\prime},\mathbf{n}_{2}^{\prime},\mathbf{n}%
_{3})+i[\mathbf{n}_{1}^{\prime},\mathbf{n}_{2}^{\prime},\mathbf{n}%
_{3}]\ \Big)^{n}d\mathbf{n}_{1}^{\prime}d\mathbf{n}_{2}^{\prime}\ . \nonumber
\end{align}

On the other hand, the integral trikernel for the standard Toeplitz twisted product (cf. Definition \ref{Toeplitz}) is obtained from (\ref{GenfromWild}) by setting $c_l^n=\frac{1}{b_l^n}$, so that 
$$\frac{c_{l_{3}}^{n}}{c_{l_{1}}^{n}c_{l_{2}}^{n}}\frac{1}{b_{l_{1}}^{n}%
b_{l_{2}}^{n}b_{l_{3}}^{n}}    =\frac{1}{(b_{l_{3}}^{n})^{2}}
=\frac{(n+l_{3}+1)!(n-l_{3})!}{(n+1)(n!)^{2}}$$
and now both summations in $l_1$ and $l_2$ yield reproducing kernels that eliminate two integrations on $S^2$. Therefore, the standard Toeplitz trikernel is given by  
\begin{align}
&  \mathbb{L}_{\frac{1}{\vec{b}}}^{j}(\mathbf{n}_{1},\mathbf{n}_{2},\mathbf{n}%
_{3})  = \frac{n+1}{(n!)^{2}4^{n}(4\pi)^{3}} \displaystyle{\sum_{l_{3}=0}^{n}}
%\ {\displaystyle\prod\limits_{k=1}^{2}}
(n+l_{3}+1)!(n-l_{3})!(2l_{3}+1)
\ \nonumber \\
%\end{align}%
%\[
& \cdot\int_{S^{2}}P_{l_{3}}(\mathbf{n}_{3}\cdot\mathbf{n}%
_{3}^{\prime})\Big(\ 1+X(\mathbf{n}_{1},\mathbf{n}_{2},\mathbf{n}%
_{3}^{\prime})+i[\mathbf{n}_{1},\mathbf{n}_{2},\mathbf{n}%
_{3}^{\prime}]\ \Big)^{n} d\mathbf{n}_{3}^{\prime}\ . \label{ToeplfromWild}
\end{align}\index{Trikernels ! integral formulae |)}

\subsubsection{Special functional transforms}\index{Special functional transforms |(}

Note that formulae  (\ref{BerfromWild})-(\ref{ToeplfromWild})  could have been obtained directly from equations
(\ref{CCintegraltransf})-(\ref{bonarec}).
Note also that, in these two formulae  above, we have used the
covariant-to-contravariant transition kernel for the standard Berezin symbol
correspondence, namely
\begin{equation}
\mathbb{U}_{\vec{b},\frac{1}{\vec{b}}}^{j}(\mathbf{n},\mathbf{n}^{\prime
})=\frac{1}{4\pi}\displaystyle{\sum_{l=0}^{n}}\frac{\binom{n+l+1}{l}}%
{\binom{n}{l}}(2l+1)P_{l}(\mathbf{n}\cdot\mathbf{n}^{\prime})
\label{co-contra}%
\end{equation}
Its inverse transition kernel has a simple closed formula, as follows:

\begin{proposition}
\label{N} The contravariant-to-covariant transition kernel for the standard
Berezin symbol correspondence is given by
\begin{equation}
\mathbb{U}_{\frac{1}{\vec{b}},\vec{b}}^{j}(\mathbf{n},\mathbf{n}^{\prime
})=\frac{1}{4\pi}\displaystyle{\sum_{l=0}^{n}}\frac{\binom{n}{l}}%
{\binom{n+l+1}{l}}(2l+1)P_{l}(\mathbf{n}\cdot\mathbf{n}^{\prime})=\frac
{n+1}{4\pi}\left(  \frac{1+\mathbf{n}\cdot\mathbf{n}^{\prime}}{2}\right)
^{n}\ \label{NN}%
\end{equation}

\end{proposition}

\begin{remark}\label{triple}
 (i) Formula (\ref{NN}) can be readily obtained from
\cite{Berezin2, Berezin}, but in Appendix \ref{proofN} we present a proof of this
formula following more closely to \cite{Wild}.

(ii) Note that equation (\ref{NN}) is the \textquotedblleft
inverse\textquotedblright\ of equation (\ref{1NN}), that is,
\begin{equation}
\left(  \frac{1+z}{2}\right)  ^{n}=\displaystyle{\sum_{l=0}^{n}}\frac
{\binom{n}{l}}{\binom{n+l+1}{l}}\frac{2l+1}{n+1}P_{l}(z)=(n!)^{2}%
\displaystyle{\sum_{l=0}^{n}}\frac{(2l+1)P_{l}(z)}{(n+l+1)!(n-l)!}\ .
\label{NN1}%
\end{equation}

(iii) However, a simple closed formula for $\mathbb{U}_{\vec{b},\frac{1}%
{\vec{b}}}^{j}(\mathbf{n},\mathbf{n}^{\prime})$ remains unknown to us. Similarly, a
simple closed formula like (\ref{closedBertrik}) for general values of $j$ is not yet known for the
Berezin trikernel or even for the Toeplitz trikernel. The situation is the same for the Stratonovich trikernel.
\end{remark}

\begin{definition}\label{Bertrans1}\index{Special functional transforms ! Berezin}
 Due to the simple closed form for $\mathbb{U}_{\frac{1}{\vec{b}},\vec
{b}}^{j}$ , the integral equation
\begin{equation}
\label{Berezintransform}
{f}(\mathbf{n})=\frac{n+1}{4\pi}\int_{S^{2}} \left(
\frac{1+\mathbf{n}\cdot\mathbf{n}^{\prime}}{2}\right)  ^{n} \tilde
{f}(\mathbf{n}^{\prime})d\mathbf{n}^{\prime} =: \mathcal{B}[\tilde f](\mathbf n)
\end{equation}
is known as the \emph{Berezin transform} of a function $\tilde f$ on $S^{2}$.

In view of (iii) in the above remark, no such simple expression as (\ref{Berezintransform}) is known to us for the inverse transform $f\to\tilde f$.  
Nonetheless, it follows from (\ref{co-contra})  that the \emph{inverse Berezin transform} of a function  $f$ on $S^{2}$ is given by\index{Special functional transforms ! inverse Berezin} 
\begin{equation} 
\mathcal{B}^{-1}[f](\mathbf n)= \tilde f(\mathbf n)= \frac{1}{(n!)^2(n+1)}\displaystyle{\sum_{l=0}^{n}}(n+l+1)!(n-l)!\frac{(2l+1)}{4\pi}
\int_{S^2}P_l(\mathbf n\cdot\mathbf n^{\prime})f(\mathbf n^{\prime})d\mathbf n^{\prime} \ \quad \label{invbertransf}
\end{equation}
\end{definition}

\begin{remark}\label{Bertrans}
(i) The reader should be aware that the term ``Berezin transform'' of an operator  is also sometimes used in the literature to indicate the standard Berezin (covariant) symbol of an operator. Also, the Berezin transform of a function on $S^2$ is more commonly presented using holomorphic coordinates on ${\mathbb C}$, which is identified with $S^2$ via the stereographic projection (\ref{stereogr}), as was originally done by Berezin \cite{Berezin2, Berezin}.

(ii) Clearly, both the Berezin transform of a function on $S^2$, $\mathcal B$ defined by (\ref{Berezintransform}), and its inverse, $\mathcal{B}^{-1}$ defined by (\ref{invbertransf}), depend on $n$. Also, they both  vanish on the
complement of $Poly_{\mathbb{C}}(S^{2})_{\leq n}$, so that composing the Berezin transform 
with its inverse (or vice-versa) amounts to orthogonally projecting any function on $S^2$ to the subspace $Poly_{\mathbb{C}}(S^{2})_{\leq n}$. In fact, this follows as a particular case of equation (\ref{integralinverse}), for $\vec{c}=\vec{b}$, $\vec{c}^{\ \prime}=\frac{1}{\vec{b}}$.
\end{remark}

In the category of positive symbol correspondences, the standard Stratonovich symbol correspondence is naturally singled out, standing in a prominent position together with its twisted product. For this reason, the standard Stratonovich trikernel  is also singled out. We also observe from Theorem \ref{SBW} that the equation which yields the standard Stratonovich trikernel from Wildberger's recursive trikernel can be written more simply as follows: 
\begin{align}
& \mathbb{L}_{1}^{j}(\mathbf{n}_{1},\mathbf{n}_{2},\mathbf{n}_{3})= \label{StantrikfromWild} \\
& \int_{S^2\times S^2\times S^2} \mathbb U^j_{\vec{b},1}(\mathbf n_1,\mathbf n_1^{\prime})\mathbb U^j_{\vec{b},1}(\mathbf n_2,\mathbf n_2^{\prime})\mathbb U^j_{\vec{b},1}(\mathbf n_3,\mathbf n_3^{\prime})\mathbb T^j_{\vec{b}}(\mathbf{n}_{1}^{\prime},\mathbf{n}_{2}^{\prime},\mathbf{n}_{3}^{\prime})d\mathbf{n}_{1}^{\prime}d\mathbf{n}_{2}^{\prime}d\mathbf{n}_{3}^{\prime} \nonumber
 \end{align}
and this, in turn, singles out the importance of the transition kernel $\mathbb U^j_{\vec{b},1}(\mathbf n,\mathbf n^{\prime})$.
\begin{definition}\label{bertostratkernel} The \emph{Berezin-Stratonovich transition kernel} 
$$\mathbb U^j_{\vec{b},1}(\mathbf n,\mathbf n^{\prime})= \frac{1}{4\pi}\displaystyle{\sum_{l=0}^{n}}\frac{2l+1}{b_l^n}P_l(\mathbf n\cdot\mathbf n^{\prime}) $$
defines  the \emph{Berezin-Stratonovich transform}\index{Special functional transforms ! Berezin-Stratonovich}  
$$\mathcal{BS}: Poly_{\mathbb C}(S^2)_{\leq n}\to Poly_{\mathbb C}(S^2)_{\leq n}$$
via the integral equation 
\begin{equation} 
\mathcal{BS}[f](\mathbf n)= \frac{1}{n!\sqrt{n+1}}\displaystyle{\sum_{l=0}^{n}}\sqrt{(n+l+1)!(n-l)!}\frac{(2l+1)}{4\pi}
\int_{S^2}P_l(\mathbf n\cdot\mathbf n^{\prime})f(\mathbf n^{\prime})d\mathbf n^{\prime} \ .\quad \label{BerStrattransf}
\end{equation}
Its inverse, the \emph{Stratonovich-Berezin transform}\index{Special functional transforms ! Stratonovich-Berezin} $\mathcal{SB}=(\mathcal{BS})^{-1}$ is given by 
\begin{equation} 
\mathcal{SB}[f](\mathbf n)= {n!\sqrt{n+1}}\displaystyle{\sum_{l=0}^{n}}\frac{1}{\sqrt{(n+l+1)!(n-l)!}}\frac{(2l+1)}{4\pi}
\int_{S^2}P_l(\mathbf n\cdot\mathbf n^{\prime})f(\mathbf n^{\prime})d\mathbf n^{\prime} \ .\quad \label{StratBertransf}
\end{equation}
\end{definition}

The following lemma, whose straightforward proof is analogous to the proof of the  equations in Proposition  \ref{transitionrelations}, gives the relation between the two special transforms on  $Poly_{\mathbb C}(S^2)_{\leq n}$ defined above.

\begin{lemma} The Berezin transform is the square of the Stratonovich-Berezin transform, 
\begin{equation} \mathcal B = (\mathcal{SB})^2=\mathcal{SB}\circ\mathcal{SB}  \ . \label{B=SB2} 
\end{equation}
\end{lemma}
\begin{remark}
One should note, however, that although $\mathcal{SB}$ is the unique \emph{positive} square root of $\mathcal B$, there are $(\mathbb Z_2)^n$ distinct square roots of $\mathcal B$, one for each choice of Stratonovich correspondence with characteristic numbers $\vec{\varepsilon}=(\varepsilon_l^n)$,  $\varepsilon_l^n=\pm 1$, defining transition kernels $\mathbb U^j_{\vec{b},\vec{\varepsilon}}$ and  $\mathbb U^j_{\vec{\varepsilon},\vec{b}}$ . 
\end{remark}

Again, just as for $\mathcal{B}^{-1}$, simple closed formulae like (\ref{Berezintransform}) for either  $\mathcal{SB}$ or $\mathcal{BS}$ are  yet unknown to us. \index{Special functional transforms |)}

\subsubsection{Integral formulae for twisted products of spherical harmonics}\index{Twisted products ! of spherical harmonics ! integral formulae |(} 

Now, for symbols $f,g\in Poly_{\mathbb{C}}(S^{2})_{\leq n}$
which are decomposed into the $l$-invariant subspaces as $f=\sum_{l=0}%
^{n}f_{l}$, $g=\sum_{l=0}^{n}g_{l}$, we have from Remark \ref{Rcl} (cf.
equations (\ref{projreprodkernel})-(\ref{ProjRc})) and equations
(\ref{intgen}), (\ref{GenfromWild})-(\ref{Inlll}), that
\begin{align}
& f\star_{\vec{c}}^{n}g(\mathbf{n}_{3})=\left(  \frac{n+1}{2^{n}4\pi}\right)
^{2}\displaystyle{\sum_{l_{1},l_{2},l_{3}=0}^{n}}\frac{c_{l_{3}}^{n}}%
{c_{l_{1}}^{n}c_{l_{2}}^{n}}\frac{1}{b_{l_{1}}^{n}b_{l_{2}}^{n}b_{l_{3}}^{n}%
}\frac{2l_{3}+1}{4\pi}\ \label{intgendecomp} \\
& \cdot\int_{S^{2}\times S^{2}\times S^{2}}f_{l_{1}}(\mathbf{n}_{1}^{\prime
})g_{l_{2}}(\mathbf{n}_{2}^{\prime})P_{l_{3}}(\mathbf{n}_{3}\!\cdot\!
\mathbf{n}_{3}^{\prime})\Big( 1+X(\mathbf{n}_{1}^{\prime},\mathbf{n}%
_{2}^{\prime},\mathbf{n}_{3}^{\prime})+i[\mathbf{n}_{1}^{\prime}%
,\mathbf{n}_{2}^{\prime},\mathbf{n}_{3}^{\prime}] \Big)^{n}d\mathbf{n}%
_{1}^{\prime}d\mathbf{n}_{2}^{\prime}d\mathbf{n}_{3}^{\prime} \nonumber
\end{align}
which for the standard Berezin case reduces to
\begin{align}
& f\star_{\vec{b}}^{n}g(\mathbf{n}_{3})=\left(  \frac{n+1}{2^{n}4\pi}\right)
^{2}\displaystyle{\sum_{l_{1},l_{2}=0}^{n}}\frac{\binom{n+l_{1}+1}{l_{1}%
}\binom{n+l_{2}+1}{l_{2}}}{\binom{n}{l_{1}}\binom{n}{l_{2}}}%
\ \label{intberdecomp} \\
& \cdot\int_{S^{2}\times S^{2}}f_{l_{1}}(\mathbf{n}_{1}^{\prime})g_{l_{2}%
}(\mathbf{n}_{2}^{\prime})\Big(\ 1+X(\mathbf{n}_{1}^{\prime},\mathbf{n}%
_{2}^{\prime},\mathbf{n}_{3})+i[\mathbf{n}_{1}^{\prime},\mathbf{n}_{2}%
^{\prime},\mathbf{n}_{3}]\ \Big)^{n}d\mathbf{n}_{1}^{\prime}d\mathbf{n}%
_{2}^{\prime}\ . \nonumber
\end{align}

In particular, equations (\ref{intgendecomp}) and (\ref{intberdecomp}) above
yield alternative expressions for the twisted products of spherical harmonics.

\begin{corollary}
The standard Berezin twisted product of spherical harmonics can also be
written as
\begin{align}
& Y_{l_{1}}^{m_{1}}\star_{\vec{b}}^{n}Y_{l_{2}}^{m_{2}}(\mathbf{n}_{3})=\frac{(n+l_{1}+1)!(n-l_{1})!(n+l_{2}+1)!(n-l_{2})!}{(n!)^{4}4^{n}(4\pi)^{2}} 
 \label{intbersphar} \\
& \cdot\int_{S^{2}\times S^{2}}Y_{l_{1}}^{m_{1}}(\mathbf{n}_{1}^{\prime
})Y_{l_{2}}^{m_{2}}(\mathbf{n}_{2}^{\prime})\Big(\ 1+X(\mathbf{n}_{1}^{\prime
},\mathbf{n}_{2}^{\prime},\mathbf{n}_{3})+i[\mathbf{n}_{1}^{\prime}%
,\mathbf{n}_{2}^{\prime},\mathbf{n}_{3}]\ \Big)^{n}d\mathbf{n}_{1}^{\prime
}d\mathbf{n}_{2}^{\prime}\nonumber
\end{align}
while a general twisted product of spherical harmonics is given by
\begin{align}
& Y_{l_{1}}^{m_{1}}\star_{\vec{c}}^{n}Y_{l_{2}}^{m_{2}}(\mathbf{n}_{3})=\left(
\frac{n+1}{2^{n}4\pi}\right)  ^{2}\frac{1}{c_{l_{1}}^{n}c_{l_{2}}^{n}}\frac
{1}{b_{l_{1}}^{n}b_{l_{2}}^{n}}\ \displaystyle{\sum_{l_{3}=0}^{n}}%
\frac{{c_{l_{3}}^{n}}}{b_{l_{3}}^{n}}\ \frac{2l_{3}+1}{4\pi}%
\ \label{intgensphar} \\
& \cdot\int_{_{S^{2}\times S^{2}\times S^{2}}}\!\!\!Y_{l_{1}}^{m_{1}}(\mathbf{n}%
_{1}^{\prime})Y_{l_{2}}^{m_{2}}(\mathbf{n}_{2}^{\prime})P_{l_{3}}%
(\mathbf{n}_{3}\!\cdot\!\mathbf{n}_{3}^{\prime})\Big( 1+X(\mathbf{n}_{1}^{\prime
},\mathbf{n}_{2}^{\prime},\mathbf{n}_{3}^{\prime})+i[\mathbf{n}_{1}^{\prime
},\mathbf{n}_{2}^{\prime},\mathbf{n}_{3}^{\prime}] \Big)^{n}\!d\mathbf{n}_{1}^{\prime}d\mathbf{n}_{2}^{\prime}d\mathbf{n}_{3}^{\prime} \nonumber 
\end{align}
from which the standard Stratonovich product is obtained by setting all
$c_{l}^{n}=1$, as
\begin{align}
& Y_{l_{1}}^{m_{1}}\star_{1}^{n}Y_{l_{2}}^{m_{2}}(\mathbf{n}_{3}) \label{intstratsphar} \\
& =\frac{\sqrt{n+1}}{(n!)^{3}4^{n}(4\pi)^{3}}{\displaystyle\prod\limits_{k=1}^{2}} \sqrt{(n+l_{k}+1)!(n-l_{k})!}\displaystyle{\sum_{l_{3}=0}^{n}}\sqrt{(n+l_{3}+1)!(n-l_{3})!}(2l_3+1) \nonumber \\
& \cdot\int_{_{S^{2}\times S^{2}\times S^{2}}}\!\!\!Y_{l_{1}}^{m_{1}}(\mathbf{n}%
_{1}^{\prime})Y_{l_{2}}^{m_{2}}(\mathbf{n}_{2}^{\prime})P_{l_{3}}%
(\mathbf{n}_{3}\!\cdot\!\mathbf{n}_{3}^{\prime})\Big( 1+X(\mathbf{n}_{1}^{\prime
},\mathbf{n}_{2}^{\prime},\mathbf{n}_{3}^{\prime})+i[\mathbf{n}_{1}^{\prime
},\mathbf{n}_{2}^{\prime},\mathbf{n}_{3}^{\prime}] \Big)^{n}\!d\mathbf{n}_{1}^{\prime}d\mathbf{n}_{2}^{\prime}d\mathbf{n}_{3}^{\prime} \nonumber 
\end{align}
\end{corollary}\index{Twisted products ! of spherical harmonics ! integral formulae |)} 

\subsubsection{Relationship with spherical geometry}\index{Trikernels ! and spherical geometry |(} 

Finally, we note that the above formulae (\ref{closedBertrik})-(\ref{ToeplfromWild}) and (\ref{intgendecomp})-(\ref{intstratsphar}) can also be rewritten in polar form,  in terms of
the geometry of the vertex spherical triangle spanned by the triple $(\mathbf{n}%
_{1},\mathbf{n}_{2},\mathbf{n}_{3})$ of unit vectors. To this end, let us
write%
\begin{align*}
\cos\alpha_{i}  &  =\mathbf{n}_{j}\cdot\mathbf{n}_{k}\text{, }\alpha_{i}%
\in\lbrack0,\pi]\text{, \ for }\left\{  i,j,k\right\}  =\{1,2,3\},\\
\beta_{i}  &  =\alpha_{i}/2\in\lbrack0,\pi/2]
\end{align*}
and consider the geodesic triangle with vertices $\mathbf{n}_{i}$ and opposite
edges of length $\alpha_{i}$, positive orientation given by $\mathbf{n}%
_{i}\rightarrow\mathbf{n}_{i+1}(i\operatorname{mod}3)$, and oriented
(symplectic) area denoted by $\Theta:$
\[
-2\pi\leq\Theta(\mathbf{n}_{1},\mathbf{n}_{2},\mathbf{n}_{3})\leq2\pi
\]

\begin{lemma}
Let $\mathbf{n}_{k}\in S^{2}(1),k\in\{1,2,3\}$, be points on the standard unit
sphere of total (symplectic) oriented area $4\pi$, and consider the geodesic
triangle with vertices $\mathbf n_k$ and arcs of length $2\beta_{k}$ and signed area $\Theta$, as
explained above. Then
\begin{align}
\cos\beta_{1}\cos\beta_{2}\cos\beta_{3}  &  =\frac{1}{4}\left\vert
\ 1+X(\mathbf{n}_{1},\mathbf{n}_{2},\mathbf{n}_{3})+i[\mathbf{n}%
_{1},\mathbf{n}_{2},\mathbf{n}_{3}]\ \right\vert \label{Berezinorm}\\
\Theta(\mathbf{n}_{1},\mathbf{n}_{2},\mathbf{n}_{3})  &  =2\arg(1+X(\mathbf{n}%
_{1},\mathbf{n}_{2},\mathbf{n}_{3})+i[\mathbf{n}_{1},\mathbf{n}_{2}%
,\mathbf{n}_{3}]) \label{vertextriangle}%
\end{align}
where \  $X(\mathbf{n}_{1},\mathbf{n}_{2},\mathbf{n}_{3})$ is  given by (\ref{U}) and $[\mathbf{n}_{1},\mathbf{n}_{2}%
,\mathbf{n}_{3}]$ is the $3\times 3$ determinant (cf. (\ref{det})).
\end{lemma}

\begin{proof}
The equation (\ref{vertextriangle}) for the area $\Theta(\mathbf{n}%
_{1},\mathbf{n}_{2},\mathbf{n}_{3})$ of the vertex geodesic triangle on $S^{2}$ has
long been known (at least since Euler, probably much earlier, see also
\cite{RT}). To prove equation (\ref{Berezinorm}), we use another known
identity (see \cite{RO1}),
\begin{equation}
\lbrack\mathbf{n}_{1},\mathbf{n}_{2},\mathbf{n}_{3}]^{2}=1-\cos^{2}\alpha
_{1}-\cos^{2}\alpha_{2}-\cos^{2}\alpha_{3}+2\cos\alpha_{1}\cos\alpha_{2}%
\cos\alpha_{3}\ ,
\end{equation}
which is readily seen to imply (\ref{Berezinorm}).
\end{proof}

\begin{corollary}[\cite{Wild}]
\label{Wildtrikamphase} For every $n=2j\in\mathbb{N}$, Wildberger's recursive
trikernel ${\mathbb{T}}_{\vec{b}}^{j}:S^{2}\times S^{2}\times S^{2}%
\rightarrow\mathbb{C}$ , cf. (\ref{closedBertrik}), can be written in polar
form as
\begin{equation}
{\mathbb{T}}_{\vec{b}}^{j}(\mathbf{n}_{1},\mathbf{n}_{2},\mathbf{n}%
_{3})=\left(  \frac{n+1}{4\pi}\right)  ^{2}A_{W}^{n}(\mathbf{n}_{1}%
,\mathbf{n}_{2},\mathbf{n}_{3})\exp{\{i\Phi_{W}^{n}(\mathbf{n}_{1}%
,\mathbf{n}_{2},\mathbf{n}_{3})\}}\  \label{Wildtrikpolar}%
\end{equation}
where the amplitude and phase functions are given respectively by
\begin{align}
A_{W}^{n}(\mathbf{n}_{1},\mathbf{n}_{2},\mathbf{n}_{3})  &  =\cos^{n}\beta
_{1}\cos^{n}\beta_{2}\cos^{n}\beta_{3}\label{Wildampl}\\
&  =\frac{1}{2^{3n/2}}[(1+\mathbf{n}_{1}\cdot\mathbf{n}_{2})(1+\mathbf{n}%
_{2}\cdot\mathbf{n}_{3})(1+\mathbf{n}_{3}\cdot\mathbf{n}_{1})]^{n/2}%
\nonumber\\
\Phi_{W}^{n}(\mathbf{n}_{1},\mathbf{n}_{2},\mathbf{n}_{3})  &  =\frac{n}%
{2}\Theta(\mathbf{n}_{1},\mathbf{n}_{2},\mathbf{n}_{3}) \label{Wildphase}%
\end{align}
\end{corollary}

 In view of the fact that the standard Stratonovich-Weyl twisted product can be seen as the  spin version of the ordinary Moyal-Weyl product, whose integral form obtained by Groenewold and von Neumann \cite{Groen, vN1} 
can be written in terms of the geometry of triangles described  by their midpoints (see e.g. \cite{RO, RT}), one could hope that a result similar to Corollary \ref{Wildtrikamphase} above would hold for the Stratonovich trikernel, with vertex spherical triangles replaced by midpoint spherical triangles (ie. spherical triangles described by their midpoints, with its area function). 
This possibility was first set forth by Weinstein \cite{Wein} and has been partially investigated by Tuynman in collaboration with one of the authors \cite{RT}. 

However, by comparing  with the geometrical formulae for midpoint triangles (area function, etc)  presented in \cite{RO1, RT},   
it is clear from equations  (\ref{LSt1/2})-(\ref{LSt j=1})  in Example \ref{exL},
as well the more general equations  (\ref{L_3small})-(\ref{invtrikkk}),  (\ref{GenfromWild})-(\ref{Inlll}),  
 that such a very simple closed ``midpoint formula'' in the style of Corollary \ref{Wildtrikamphase} only has a chance of holding for the Stratonovich trikernel $\mathbb L_1^j$  asymptotically, as $n\to\infty$, as in a WKB-style approximation. 
 
 But such an asymptotic study of the Stratonovich trikernel   lies outside the scope of this monograph. The more basic asymptotic study to be presented in the next chapter does not touch on this matter. 
\index{Trikernels ! and spherical geometry |)}\index{Trikernels |)}

%---------------------------------------------------------------------------------------------------------------------------------------------------------------------------------------------------------------------
%-------------------------------------------------------------------------------------- Chapter 8 -----------------------------------------------------------------------------------------------------------------
%----------------------------------------------------------------------------------------------------------------------------------------------------------------------------------------------------------------------

\chapter{Beginning asymptotic analysis of twisted products}
\label{AsympChapter}

As mentioned in Remarks \ref{Heisenbdynspinj} and \ref{hamdynsphere}, while
the quantum dynamics of an operator $F\in\mathcal{B}(\mathcal{H}_{j})$ is
governed by Heisenberg's equation (\ref{Heisenbeq}), the classical dynamics of
a function $f\in C^{\infty}_{\mathbb{C}}(S^{2})$ is governed by Hamilton's
equation (\ref{Hameq2}).

Now, via a symbol correspondence $W^{j}_{\vec c}$, Heisenberg's equation can be
transformed into a dynamical equation for the symbol of $F$, $f\in
Poly_{\mathbb{C}}(S^{2})_{\leq n}\subset C_{\mathbb{C}}^{\infty}(S^{2})$, by
substituting the commutator of operators $[H,F]$ by the twisted commutator of
symbols, $[h,f]_{\star^n_{\vec c}}=h\star^n_{\vec c} f-f\star^n_{\vec c} h$, where $h\in Poly_{\mathbb{C}}(S^{2})_{\leq
n}$ is the $W^j_{\vec c}$-symbol of the preferred Hamiltonian operator $H$.

Therefore, a natural
question is whether Hamilton's equation can be obtained from Heisenberg's
equation for symbols in a suitable limit, the so-called (semi)classical limit. 
In other words, whether Poisson dynamics emerge from 
quantum dynamics in a suitable asymptotic limit.

Historically, this question was first addressed in the context of operators
on infinite dimensional Hilbert space $L^{2}%
(\mathbb{R}^{k})$ and functions on affine symplectic space $\mathbb{R}^{2k}%
$. In that context of affine mechanical systems, the (semi)classical limit, the limit of very large
quantum numbers, can be formally treated as the limit $\hbar\rightarrow0$.

However, in the context of spin systems, the (semi)classical limit, the limit of
very large quantum numbers, is the limit $2j=n\rightarrow\infty$ and, in this
context, $\hbar$ is best treated as a constant which can be omitted by scaling
(see Remark \ref{hbar=1}).

Thus, in order to address Bohr's correspondence principle for spin systems, we
must investigate the asymptotic limit and expansions, as $n\rightarrow\infty$,
of the symbol correspondences $W^{j}_{\vec c}$, their twisted products and the symbols themselves.

As we saw in the previous chapter, each symbol correspondence $W^{j}_{\vec c}$ defines a $\vec c$-twisted $j$-algebra on $Poly_{\mathbb{C}}(S^{2})_{\leq
n}$. However, despite the fact that all  $\vec c$-twisted $j$-algebras are isomorphic for each finite $j$, we shall see below that only a subclass of symbol correspondence sequences yield  Poisson dynamics in the asymptotic limit of high spin numbers ($j\to\infty$). This subclass realizes Rieffel's ``strict deformation quantization'' of the $2$-sphere in reverse order (from quantum to classical). However, as we shall see below, this subclass is far from being generic.

\section{Low-$l$ high-$j$-asymptotics of the standard twisted product}\index{Twisted products ! standard! low-l high-j asymptotics|(}

If we recall from Proposition \ref{stanprodlin} the formulae for the standard
twisted product of the cartesian coordinate functions, and try to compute the
asymptotics for these formulae as $n\rightarrow\infty$, the first question is
how to expand these formulae. A natural asymptotic expansion is in power
series, but, power series of what? From equation (\ref{star}), one could guess
that $1/\sqrt{n(n+2)}$, or rather $1/\sqrt{j(j+1)}$, is a natural expansion
parameter (as suggested in \cite{VG-B}). In this case, expanding in negative
powers of $\sqrt{j(j+1)}$ we have
\begin{equation}
a\star_{1}^{n}b\ \longrightarrow\ ab+\frac{i\varepsilon_{abc}}{2\sqrt{j(j+1)}%
}c+O((j(j+1))^{-1}) \label{AS1}%
\end{equation}%
\begin{equation}
a\star_{1}^{n}a\ \longrightarrow\ a^{2}+O((j(j+1))^{-1}) \label{AS2}%
\end{equation}
On the other hand, looking at the formulae for the Berezin twisted product of
the cartesian coordinate functions, as described in Corollary \ref{Berezin0} ,
a more natural expansion parameter seems to be $1/n$. Then, the standard
twisted product of the coordinate functions in negative powers of $n$ becomes
\begin{equation}
a\star_{1}^{n}b\ \longrightarrow\ ab+\frac{i\varepsilon_{abc}}{n}c+O(n^{-2})
\label{AS3}%
\end{equation}%
\begin{equation}
a\star_{1}^{n}a\ \longrightarrow\ a^{2}+O(n^{-2}) \label{AS4}%
\end{equation}

We also observe that in the equivalent expressions (\ref{AS1})-(\ref{AS2}) and
(\ref{AS3})-(\ref{AS4}), the zeroth order term is the classical pointwise
product, while the first order term is the Poisson bracket (multiplied by
$i$). Thus we are led to inquire whether the zeroth and first order expansions
of the standard Stratonovich-Weyl twisted product always coincide with the
pointwise product and the Poisson bracket, respectively. In fact, we have the following:

\begin{theorem}
[\cite{B-T}, \cite{F-K}]\label{asymplimit} For $n=2j>>1$, \ $l_{1}%
,l_{2},l<<2j$,
\begin{equation}
(-1)^{2j+m}\sqrt{\frac{(2j+1)(2l+1)}{(2l_{1}+1)(2l_{2}+1)}}\left[
\begin{array}
[c]{ccc}%
l_{1} & l_{2} & l\\
m_{1} & m_{2} & -m
\end{array}
\right]  \!\!{[j]}\nonumber
\end{equation}%
\begin{equation}
=C_{m_{1},m_{2},m}^{l_{1},l_{2},l}C_{0,0,0}^{l_{1},l_{2},l}+\frac{1}%
{2\sqrt{j(j+1)}}C_{m_{1},m_{2},m}^{l_{1},l_{2},l}P(l_{1},l_{2}%
,l)+O((j(j+1))^{-1}) \label{asympWj}%
\end{equation}%
\begin{equation}
=C_{m_{1},m_{2},m}^{l_{1},l_{2},l}C_{0,0,0}^{l_{1},l_{2},l}+\frac{1}%
{n+1}C_{m_{1},m_{2},m}^{l_{1},l_{2},l}P(l_{1},l_{2},l)+O((n+1)^{-2})
\label{asympWn+1}%
\end{equation}%
\begin{equation}
=C_{m_{1},m_{2},m}^{l_{1},l_{2},l}C_{0,0,0}^{l_{1},l_{2},l}+\frac{1}%
{n}C_{m_{1},m_{2},m}^{l_{1},l_{2},l}P(l_{1},l_{2},l)+O(n^{-2}) \label{asympWn}%
\end{equation}
where $C_{0,0,0}^{l_{1},l_{2},l}$ and $P(l_{1},l_{2},l)$ are given explicitly
by (\ref{C_000}) and (\ref{PPP}), respectively, and satisfy $C_{0,0,0}%
^{l_{1},l_{2},l}\equiv0$ if $l_{1}+l_{2}+l$ is odd, $P(l_{1},l_{2},l)\equiv0$
if $l_{1}+l_{2}+l$ is even.
\end{theorem}

For the proof of Theorem \ref{asymplimit} we refer to Appendix
\ref{proofasymplimit}. It follows that 

\begin{corollary}
\label{asymp} For nonnegative integers $l_{1},l_{2}$, the standard twisted
product of spherical harmonics $Y_{l_{1}}^{m_{1}}$ and $Y_{l_{2}}^{m_{2}}$
satisfies
\begin{align}
\text{(i)}:\  &  \lim_{n\rightarrow\infty}\left(  Y_{l_{1}}^{m_{1}}\star
_{1}^{n}Y_{l_{2}}^{m_{2}}-Y_{l_{2}}^{m_{2}}\star_{1}^{n}Y_{l_{1}}^{m_{1}%
}\right)  =0\label{Poisson01}\\
\text{(ii)}:\  &  \lim_{n\rightarrow\infty}\left(  Y_{l_{1}}^{m_{1}}\star
_{1}^{n}Y_{l_{2}}^{m_{2}}+Y_{l_{2}}^{m_{2}}\star_{1}^{n}Y_{l_{1}}^{m_{1}%
}\right)  =2Y_{l_{1}}^{m_{1}}Y_{l_{2}}^{m_{2}}\label{Poisson02}\\
\text{(iii)}:\  &  \lim_{n\rightarrow\infty}\left(  n[Y_{l_{1}}^{m_{1}}%
\star_{1}^{n}Y_{l_{2}}^{m_{2}}-Y_{l_{2}}^{m_{2}}\star_{1}^{n}Y_{l_{1}}^{m_{1}%
}]\right)  =2i\{Y_{l_{1}}^{m_{1}},Y_{l_{2}}^{m_{2}}\}\label{Poisson03}\\
\text{(iv)}:\  &  \lim_{n\rightarrow\infty}\left(  n[Y_{l_{1}}^{m_{1}}%
\star_{1}^{n}Y_{l_{2}}^{m_{2}}+Y_{l_{2}}^{m_{2}}\star_{1}^{n}Y_{l_{1}}^{m_{1}%
}-2Y_{l_{1}}^{m_{1}}Y_{l_{2}}^{m_{2}}]\right)  =0 \label{Poisson04}%
\end{align}
where the limits above are taken uniformly, i.e. we have uniform convergence of the sequence of functions on the l.h.s. to the function on the r.h.s. 

By linearity, properties (i)-(iv) apply to the product of any $f,g\in
Poly_{\mathbb{C}}(S^{2})_{\leq k}$, where $k\in\mathbb{N}$ is finite.
\end{corollary}

\begin{proof} First, for fixed $l_1,m_1,l_2,m_2$, let us denote $$\mathcal P_{m_1,m_2}^{l_1,l_2}= \{Y_l^m \ , \ |l_1-l_2|\leq l\leq l_1+l_2, \  -l\leq m=m_1+m_2\leq l \}$$ 
and start by rewriting equation (\ref{asympWn}) as 
\begin{equation}
(-1)^{2j+m}\sqrt{\frac{(2j+1)(2l+1)}{(2l_{1}+1)(2l_{2}+1)}}\left[
\begin{array}
[c]{ccc}%
l_{1} & l_{2} & l\\
m_{1} & m_{2} & -m
\end{array}
\right]  \!\!{[j]}\nonumber
\end{equation}%
\begin{equation}
=C_{m_{1},m_{2},m}^{l_{1},l_{2},l}C_{0,0,0}^{l_{1},l_{2},l}+\frac{1}%
{n}C_{m_{1},m_{2},m}^{l_{1},l_{2},l}P(l_{1},l_{2},l)+\frac{1}{n^2}D_{m_{1},m_{2}}^{l_{1},l_{2},l}(n) \ ,  \label{asympWn2}%
\end{equation}
where $D_{m_{1},m_{2}}^{l_{1},l_{2},l}(n)\in\mathbb R$ is such that 
\begin{equation}\label{major1}\left|D_{m_{1},m_{2}}^{l_{1},l_{2},l}(n)\right|\leq K\in\mathbb R^+, \ \forall n\in \mathbb N.\end{equation}

We will show $(iii)$, the others follow similarly. 

Let $||f||= sup(|f(\mathbf n)|, \mathbf n\in S^2)$ denote the sup-norm on the space of  smooth functions on the sphere. From Theorem \ref{asymplimit}, Propositions \ref{PP} and \ref{DP} and equation (\ref{asympWn2}) we have that, $\forall n\geq l_1+l_2$,  
\begin{equation}\label{diffofpoisson} 
||n\left(  Y_{l_{1}}^{m_{1}}\star
_{1}^{n}Y_{l_{2}}^{m_{2}}-Y_{l_{2}}^{m_{2}}\star_{1}^{n}Y_{l_{1}}^{m_{1}%
}\right) - 2i\{Y_{l_{1}}^{m_{1}},Y_{l_{2}}^{m_{2}}\} || = \frac{1}{n}||R^{m_{1},m_{2}}_{l_{1},l_{2}}(n)|| \ , 
\end{equation} 
where 
\begin{equation}\nonumber 
R^{m_{1},m_{2}}_{l_{1},l_{2}}(n) = \displaystyle{\sum_{l=|l_1-l_2|}^{l_1+l_2}}D_{m_{1},m_{2}}^{l_{1},l_{2},l}(n)Y^{m}_l \ , \ m=m_1+m_2 \ .
\end{equation}
Then,  
\begin{equation}\nonumber 
||R^{m_{1},m_{2}}_{l_{1},l_{2}}(n)||\leq\displaystyle{\sum_{l=|l_1-l_2|}^{l_1+l_2}}||D_{m_{1},m_{2}}^{l_{1},l_{2},l}(n)Y^{m}_l||\leq \ K\displaystyle{\sum_{l=|l_1-l_2|}^{l_1+l_2}} ||Y^{m}_l||  \ , 
\end{equation}
where we have used (\ref{major1}). 
Now, let 
$$M=max\{||Y^{m}_l|| , \ Y^{m}_l\in \mathcal P_{m_1,m_2}^{l_1,l_2}\} \ ,$$
then 
\begin{equation}\label{major2}||R^{m_{1},m_{2}}_{l_{1},l_{2}}(n)||\leq \ KM(l_1+l_2+1-|l_1-l_2|)  \end{equation}
and $(iii)$ follows immediately from (\ref{diffofpoisson}) and (\ref{major2}). The other statements $(i), (ii)$ and $(iv)$ are proved similarly. 
\end{proof} 
\begin{remark}
\label{truncatedpoissonremark} For finite $n>>1$, the following is a valid
expansion in powers of $1/n$, as long as $l_{1},l_{2}<<n$:
\begin{equation}
\label{truncatedPoisson}Y_{l_{1}}^{m_{1}}\star_{1}^{n}Y_{l_{2}}^{m_{2}}=
Y_{l_{1}}^{m_{1}}Y_{l_{2}}^{m_{2}}|_{n}+\frac{i}{n}\{Y_{l_{1}}^{m_{1}%
},Y_{l_{2}}^{m_{2}}\}|_{n}+ o(1/n) \ ,
\end{equation}
where $Y_{l_{1}}^{m_{1}}Y_{l_{2}}^{m_{2}}|_{n}$ denotes the $n^{th}$ degree
truncation of $Y_{l_{1}}^{m_{1}}Y_{l_{2}}^{m_{2}}$ and $\{Y_{l_{1}}^{m_{1}%
},Y_{l_{2}}^{m_{2}}\}|_{n}$ denotes the $n^{th}$ degree truncation of
$\{Y_{l_{1}}^{m_{1}},Y_{l_{2}}^{m_{2}}\}$, that is, the truncations $l\leq n$
of the $l$-summations in formulas (\ref{prod2}) and (\ref{Poisson333}), respectively.

However, the asymptotic ($n\rightarrow\infty$) expansion of the Wigner
product symbol presented in Theorem \ref{asymplimit} is invalid without the
assumption $l_{1},l_{2}<<n$ ($l_{1},l_{2}$ remaining finite as $n\to\infty$) and other
asymptotic expansions of the Wigner product symbol are in order if, for
instance, we let $n\rightarrow\infty$ keeping $n/l_{1},\ n/l_{2},\ n/l\ $ finite.
\end{remark}

One can say that the zeroth order term in equations (\ref{asympWj}%
)-(\ref{asympWn}) was first obtained by Brussard and Tolhoek in \cite{B-T},
though not in the form of these equations. The first order term was first
obtained by Freidel and Krasnov in \cite{F-K}. \ 
\index{Twisted products ! standard! low-l high-j asymptotics|)}

\section{Asymptotic types of symbol correspondence sequences}

As the classical products of spherical harmonics appear as the zeroth and
first order terms in the expansion of the standard twisted product in powers
of $n^{-1}$ (or $(j(j+1))^{-1/2}$), we want to investigate which of all
possible symbol correspondences have the same property, namely that their
twisted products are related to the classical products of functions on the
sphere, as $n\rightarrow\infty$. By linearity, it is enough to investigate
this for the products of spherical harmonics.

However, again turning to the standard Berezin twisted product of linear
spherical harmonics or cartesian coordinate functions (cf. Corollary
\ref{Berezin0}), we note that the zeroth order term in the $1/n$ (or
$1/\sqrt{j(j+1)}$) expansion coincides with the pointwise product, whereas the
first order term does not coincide with the Poisson bracket of functions.
Therefore we introduce the following:

\begin{definition}\index{Symbol correspondence sequences |(}
\label{corresequence} Let
\[
\Delta^{+}(\mathbb{N}^{2})=\{(n,l)\in\mathbb{N}^{2}\ |\ n\geq l>0\}
\]
and $\mathcal{C}:\Delta^{+}(\mathbb{N}^{2})\rightarrow\mathbb{R}^{\ast}$ be
any given function. We denote by $\mathbf{W}_{\mathcal{C}}=[W_{\vec{c}}%
^{j}]_{2j=n\in\mathbb{N}}$ the \emph{sequence of symbol correspondences}
defined by characteristic numbers $c_{l}^{n}=\mathcal{C}(n,l)$, $\forall
(n,l)\in\Delta^{+}(\mathbb{N}^{2})$, $c_{0}^{n}=1,\forall n\in\mathbb{N}$. We
denote by
\[
\mathbf{W}_{\mathcal{C}}(S^{2},\star)=[(Poly_{\mathbb{C}}(S^{2})_{\leq
n},\star_{\vec{c}}^{n})]_{n\in\mathbb{N}}%
\]
the associated \emph{sequence of twisted algebras} (cf. Definitions
\ref{charact} and \ref{twistedalgebra}).
\end{definition}

\subsection{Symbol correspondence sequences of Poisson type}

\begin{definition}
\label{producttypes} A symbol correspondence sequence $\mathbf{W}%
_{\mathcal{C}}$, with its associated sequence of twisted algebras
$\mathbf{W}_{\mathcal{C}}(S^{2},\star)$, is of \emph{Poisson type} if,\index{Symbol correspondence sequences ! of (anti) Poisson type}
$\forall l_{1},l_{2}\in\mathbb{N}$,
\begin{align}
\text{(i)}:\  &  \lim_{n\rightarrow\infty}\left(  Y_{l_{1}}^{m_{1}}\star
_{\vec{c}}^{n}Y_{l_{2}}^{m_{2}}-Y_{l_{2}}^{m_{2}}\star_{\vec{c}}^{n}Y_{l_{1}%
}^{m_{1}}\right)  =0\label{Poisson11}\\
\text{(ii)}:\  &  \lim_{n\rightarrow\infty}\left(  Y_{l_{1}}^{m_{1}}%
\star_{\vec{c}}^{n}Y_{l_{2}}^{m_{2}}+Y_{l_{2}}^{m_{2}}\star_{\vec{c}}%
^{n}Y_{l_{1}}^{m_{1}}\right)  =2Y_{l_{1}}^{m_{1}}Y_{l_{2}}^{m_{2}%
}\label{Poisson22}\\
\text{(iii)}:\  &  \lim_{n\rightarrow\infty}\left(  n[Y_{l_{1}}^{m_{1}}%
\star_{\vec{c}}^{n}Y_{l_{2}}^{m_{2}}-Y_{l_{2}}^{m_{2}}\star_{\vec{c}}%
^{n}Y_{l_{1}}^{m_{1}}]\right)  =2i\{Y_{l_{1}}^{m_{1}},Y_{l_{2}}^{m_{2}}\},
\label{Poisson33}%
\end{align}
$\mathbf{W}_{\mathcal{C}}$ is of \emph{anti-Poisson type} if it satisfies
properties (i), (ii) above and 
\begin{equation}
\ \ \quad\ \ \text{(iii')}:\ \lim_{n\rightarrow\infty}\left(  n[Y_{l_{1}%
}^{m_{1}}\star_{\vec{c}}^{n}Y_{l_{2}}^{m_{2}}-Y_{l_{2}}^{m_{2}}\star_{\vec{c}%
}^{n}Y_{l_{1}}^{m_{1}}]\right)  =-2i\{Y_{l_{1}}^{m_{1}},Y_{l_{2}}^{m_{2}}\},
\label{Poisson33'}%
\end{equation}
and $\mathbf{W}_{\mathcal{C}}$ is of \emph{pure-(}resp.
\emph{pure-anti)-Poisson type} if, in addition to properties (i), (ii), and
(iii) (resp. (iii')) above, the following property also holds:\index{Symbol correspondence sequences ! of pure (anti) Poisson type}
\begin{equation}
\ \ \ \quad\ \ \text{(iv)}:\ \lim_{n\rightarrow\infty}\left(  n[Y_{l_{1}%
}^{m_{1}}\star_{\vec{c}}^{n}Y_{l_{2}}^{m_{2}}+Y_{l_{2}}^{m_{2}}\star_{\vec{c}%
}^{n}Y_{l_{1}}^{m_{1}}-2Y_{l_{1}}^{m_{1}}Y_{l_{2}}^{m_{2}}]\right)  =0
\label{s-Poisson}%
\end{equation}
Again, all limits above are taken uniformly, i.e. we consider uniform convergence of the sequence of functions on the l.h.s. to the function on the r.h.s. 
\end{definition}

\begin{remark}
The signs in the r.h.s. of equations (\ref{Poisson33}) and (\ref{Poisson33'})
are related to the choice of orientation of the symplectic form on the sphere
(cf. Remark \ref{signchoice}). Once a choice is fixed, they are also related
by the antipodal map on the sphere (cf. Proposition \ref{alternateBerezin},
equation (\ref{antipodalsymbol}), and Proposition \ref{relstandalt}).
\end{remark}

\begin{remark}
If $\mathbf{W}_{\mathcal{C}}$ is of pure-Poisson type, its twisted product
expands as in equation (\ref{truncatedPoisson}), under the same assumptions of
Remark \ref{truncatedpoissonremark}.
\end{remark}

\begin{proposition}
\label{SWstrictPoisson} The standard Stratonovich-Weyl symbol correspondence
sequence is of pure-Poisson type and the alternate Stratonovich-Weyl symbol
correspondence sequence is of pure-anti-Poisson type.
\end{proposition}

\begin{proof}
For the standard case, this follows from Corollary \ref{asymp}. For the
alternate case, $c^{n}_{l}=\varepsilon_{l}=(-1)^{l}$, it follows from the
above and from Proposition \ref{relstandalt}.
\end{proof}

\begin{proposition}
\label{BerezinPoisson} The standard (resp. alternate) Berezin symbol
correspondence sequence is of Poisson type (resp. anti-Poisson type), but not
of pure-Poisson type (resp. pure-anti-Poisson type). Same for the standard (resp. alternate) Toeplitz symbol correspondence sequences (cf. Definition \ref{Toeplitzcorr}).
\end{proposition}

\begin{proof}
To see that the standard Berezin twisted product is of Poisson type, note from
formula (\ref{BerezinCharexp}) for $b_{l}^{n}=\mathcal{C}(n,l)$ that
\[
\forall l,l_{1},l_{2}<<n,\ \ \frac{b_{l}^{n}}{b_{l_{1}}^{n}b_{l_{2}}^{n}%
}\ \longrightarrow\ 1+O(1/n),\ \text{as}\ \ n\rightarrow\infty
\]
Therefore, by (\ref{genprod}) and (\ref{asympWn}), $\forall l_{1},l_{2}<<n$,
\[
\lim_{n\rightarrow\infty}\left(  Y_{l_{1}}^{m_{1}}\star_{\vec{b}}^{n}Y_{l_{2}%
}^{m_{2}}\right)  \ =\ \lim_{n\rightarrow\infty}\left(  Y_{l_{1}}^{m_{1}}%
\star_{1}^{n}Y_{l_{2}}^{m_{2}}\right)
\]%
\[
\lim_{n\rightarrow\infty}\left(  n[Y_{l_{1}}^{m_{1}}\star_{\vec{b}}%
^{n}Y_{l_{2}}^{m_{2}}-Y_{l_{2}}^{m_{2}}\star_{\vec{b}}^{n}Y_{l_{1}}^{m_{1}%
}]\right)  \ =\ \lim_{n\rightarrow\infty}\left(  n[Y_{l_{1}}^{m_{1}}\star
_{1}^{n}Y_{l_{2}}^{m_{2}}-Y_{l_{2}}^{m_{2}}\star_{1}^{n}Y_{l_{1}}^{m_{1}%
}]\right)
\]
On the other hand, property (iv) already fails for the twisted product of
linear symbols, as shown by Corollary \ref{Berezin0}.\ For the alternate case,
$\mathcal{C}(n,l)=b_{l-}^{n}=(-1)^{l}b_{l}^{n}$ , we use Proposition
\ref{relstandalt}. Analogously for the standard and alternate Toeplitz twisted products, cf. Definition \ref{Toeplitz} and equation (\ref{CCtwistedqhomo}).
\end{proof}

\begin{remark}
The distinction between symbol correspondence sequences of pure-Poisson or
Poisson types is not irrelevant insofar as the Stratonovich-Weyl and the Berezin and Toeplitz 
symbol correspondences satisfy different axioms, namely, the former satisfies
the isometry axiom (v) of Remark \ref{axiom} while the latter ones do not. 
It is
therefore interesting to see that these distinct symbol correspondence
sequences exhibit distinct asymptotics, namely the former satisfies the
property (\ref{s-Poisson}) while the latter ones do not.

We shall say more about their asymptotics below, when a more important
asymptotic distinction between the standard Stratonovich-Weyl and the Berezin and Toeplitz 
correspondences will be highlighted (see Remarks \ref{importantdistinction}
and \ref{high-difference}).
\end{remark}

\subsection{Symbol correspondence sequences of non-Poisson type}\index{Symbol correspondence sequences ! of non-Poisson type|(}

Generic symbol correspondence sequences are not of Poisson or anti-Poisson
type. This can already be seen in the very restrictive case of
Stratonovich-Weyl symbol correspondences. Thus, consider a Stratonovich-Weyl
correspondence sequence given by the characteristic numbers $c_{l}%
^{n}=\varepsilon_{l}^{n}=\pm1$, \ for $1\leq l\leq n.$ Combining equations
(\ref{genprod}) and (\ref{asympWn}), we have for $l_{1},l_{2}<<n$,
\begin{align}
&  Y_{l_{1}}^{m_{1}}\star_{\vec{\varepsilon}}^{n}Y_{l_{2}}^{m_{2}}%
\quad={\displaystyle\sum\limits_{\substack{l=|l_{1}-l_{2}|\\l\equiv
l_{1}+l_{2}(\operatorname{mod}2)}}^{l_{1}+l_{2}}}\sqrt{\frac{(2l_{1}%
+1)(2l_{2}+1)}{2l+1}}\ C_{m_{1},m_{2},m}^{l_{1},l_{2},l}C_{0,0,0}^{l_{1}%
,l_{2},l}\displaystyle{\frac{\varepsilon_{l}^{\infty}}{\varepsilon_{l_{1}%
}^{\infty}\varepsilon_{l_{2}}^{\infty}}}\ Y_{l}^{m}\quad\quad
\label{genstratexp}\\
&  +\quad\displaystyle{\frac{1}{n}}{\displaystyle{\sum
\limits_{_{\substack{l=|l_{1}-l_{2}|+1\\l\equiv l_{1}+l_{2}%
-1(\operatorname{mod}2)}}}^{l_{1}+l_{2}-1}}}\sqrt{\frac{(2l_{1}+1)(2l_{2}%
+1)}{2l+1}}\ C_{m_{1},m_{2},m}^{l_{1},l_{2},l}P({l_{1},l_{2},l}%
)\displaystyle{\frac{\varepsilon_{l}^{\infty}}{\varepsilon_{l_{1}}^{\infty
}\varepsilon_{l_{2}}^{\infty}}}\ Y_{l}^{m}+O(1/n^{2})\nonumber
\end{align}
where
\begin{equation}
\varepsilon_{l}^{\infty}=\displaystyle{\lim_{n\rightarrow\infty}%
\varepsilon_{l}^{n}}\ , \label{epsiloninfty}%
\end{equation}
whenever such limits exist.

Clearly, limit (\ref{epsiloninfty}) does not exist for a random sequence of
strings of $\pm1$, of the form $\vec{\varepsilon}=(\varepsilon_{1}^{n},
\varepsilon_{2}^{n}, \cdots, \varepsilon_{n}^{n})$. Thus, obviously, these
generic Stratonovich-Weyl symbol correspondences are not of Poisson type.

Moreover, by comparison with equations (\ref{prod2}) and (\ref{Poisson333}),
we see that the same can be said of a generic string of $\pm1$ of the form
$\vec{\varepsilon}=(\varepsilon_{1},\varepsilon_{2},\cdots,\varepsilon_{n})$,
where $\varepsilon_{l}^{n}=\varepsilon_{l}=\varepsilon_{l}^{\infty}$ because,
generically, $\varepsilon_{l}/\varepsilon_{l_{1}}\varepsilon_{l_{2}}$ will be
a random assignment
\[
\epsilon_{l_{1},l_{2}}:\mathbb{N}\cap\lbrack|l_{1}-l_{2}|-1,l_{1}%
+l_{2}]\rightarrow\{\pm1\}
\]
and thus, generically, the first sum in (\ref{genstratexp}) will not yield the
pointwise product $Y_{l_{1}}^{m_{1}}Y_{l_{2}}^{m_{2}}$ and the second sum will
not yield the Poisson bracket (times $\pm i$) $\{Y_{l_{1}}^{m_{1}},Y_{l_{2}%
}^{m_{2}}\}$.

In fact, the requirement that, for every $Y_{l_{1}}^{m_{1}}$ and $Y_{l_{2}%
}^{m_{2}}$, $l_{1},l_{2}<<n$, the first sum in (\ref{genstratexp}) yields the
pointwise product $Y_{l_{1}}^{m_{1}}Y_{l_{2}}^{m_{2}}$ and the second sum
yields the Poisson bracket (times $\pm i$) $\{Y_{l_{1}}^{m_{1}},Y_{l_{2}%
}^{m_{2}}\}$, enforces that, $\forall l,l_{1},l_{2}<<n$, either $\varepsilon
_{l}/\varepsilon_{l_{1}}\varepsilon_{l_{2}}=1$ or $\varepsilon_{l}%
/\varepsilon_{l_{1}}\varepsilon_{l_{2}}=(-1)^{l+l_{1}+l_{2}}$ (cf. Proposition
\ref{relstandalt}). Thus, we have:

\begin{proposition}
\label{SWPoisson} A Stratonovich-Weyl symbol correspondence sequence
$[\varepsilon_{l}^{n}]$ which is of Poisson type is a sequence for which
$\varepsilon_{l}^{\infty}=1$, $\forall l\in\mathbb{N}$, and it is of
anti-Poisson type if $\varepsilon_{l}^{\infty}=(-1)^{l}$, $\forall
l\in\mathbb{N}$ (cf. equation (\ref{epsiloninfty})).
\end{proposition}

Although a generic symbol correspondence sequence has no limit as
$n\rightarrow\infty$, there exist symbol correspondence sequences which are
not of Poisson or anti-Poisson type, but still have a well defined limit as
$n\rightarrow\infty$. This is interesting because such a correspondence
sequence defines a dynamics of symbols which is not of Poisson type, in the
limit $n\rightarrow\infty$. We shall illustrate this with two examples:

\begin{example}
\label{ex1} Let ${\mathbf{W}}_{\mathcal{C}}$ be the symbol correspondence
sequence defined by the characteristic numbers $c_{l}^{n}=\mathcal{C}%
(n,l)=n^{-l}$. As $n\rightarrow\infty$, the twisted products of the cartesian
symbols expand as
\begin{equation}
x\star_{\vec{c}}^{n}y=xy+iz+O(1/n)\ ,\ x\star_{\vec{c}}^{n}x=x^{2}%
+1/2+O(1/n)\ \label{NC1lin}%
\end{equation}
and so on for cyclic permutations of $(x,y,z)$. We note that the commutator
\[
x\star_{\vec{c}}^{n}y-y\star_{\vec{c}}^{n}x=2iz+O(1/n)=2i\{x,y\}+O(1/n),
\]
so the Poisson bracket appears as the zeroth order term in the expansion of
the commutator, not as the first order term, as in equation (\ref{Poisson33}).
Thus, this symbol correspondence sequence is not of Poisson type, according to
Definition \ref{producttypes}.

For this symbol correspondence sequence, however, one could try redefining
Poisson dynamics as the zeroth order term in the expansion of the commutator.
But this clearly does not work either, because
\[
\frac{c_{l}^{n}}{c_{l_{1}}^{n}c_{l_{2}}^{n}}=n^{l_{1}+l_{2}-l}%
\]
and therefore, using equation (\ref{genprod}), we can see that the expansion
in powers of $n$, or $1/n$, of the twisted product $Y_{l_{1}}^{m_{1}}%
\star_{\vec{c}}^{n}Y_{l_{2}}^{m_{2}}$ will be completely messed up, with each
expanding term power depending on $l_{1}+l_{2}-l$, for $|l_{1}-l_{2}|\leq
l\leq l_{1}+l_{2}$.
\end{example}

\begin{example}
\label{ex2} Let ${\mathbf{W}}_{\mathcal{C}}$ be the symbol correspondence
sequence defined by the characteristic numbers \ $\displaystyle{c_{l}%
^{n}=\mathcal{C}(n,l)=lb_{l}^{n}-{(l-1)}/{n^{l}}}$, where $b_{l}^{n}$ are the
characteristic numbers of the standard Berezin symbol correspondence, given by
(\ref{BerezinChar}). Then, clearly,
\[
\forall l,l_{1},l_{2}<<n,\ \frac{c_{l}^{n}}{c_{l_{1}}^{n}c_{l_{2}}^{n}%
}\rightarrow\frac{l}{l_{1}l_{2}}+O(1/n)\ ,\ as\ \ n\rightarrow\infty.
\]
Inserting the above estimate in equation (\ref{genprod}), using equation
(\ref{asympWn}), we have
\begin{align}
&  \displaystyle{\lim_{n\rightarrow\infty}Y_{l_{1}}^{m_{1}}\star_{\vec{c}}%
^{n}Y_{l_{2}}^{m_{2}}}\ =\nonumber\label{ex2prod0}\\
&  ={\displaystyle\sum\limits_{\substack{l=|l_{1}-l_{2}|\\l\equiv l_{1}%
+l_{2}(\operatorname{mod}2)}}^{l_{1}+l_{2}}}\sqrt{\frac{(2l_{1}+1)(2l_{2}%
+1)}{2l+1}}\ C_{m_{1},m_{2},m}^{l_{1},l_{2},l}C_{0,0,0}^{l_{1},l_{2}%
,l}\displaystyle{\frac{l}{l_{1}l_{2}}}\ Y_{l}^{m}%
\end{align}%
\begin{align}
&  \displaystyle{\lim_{n\rightarrow\infty}\left(  n[Y_{l_{1}}^{m_{1}}%
\star_{\vec{c}}^{n}Y_{l_{2}}^{m_{2}}-Y_{l_{2}}^{m_{2}}\star_{\vec{c}}%
^{n}Y_{l_{1}}^{m_{1}}]\right)  }\ =\nonumber\label{ex2prod1}\\
&  ={\displaystyle\sum\limits_{_{\substack{l=|l_{1}-l_{2}|+1\\l\equiv
l_{1}+l_{2}-1(\operatorname{mod}2)}}}^{l_{1}+l_{2}-1}}\sqrt{\frac
{(2l_{1}+1)(2l_{2}+1)}{2l+1}}\ C_{m_{1},m_{2},m}^{l_{1},l_{2},l}P({l_{1}%
,l_{2},l})\displaystyle{\frac{l}{l_{1}l_{2}}}\ Y_{l}^{m}%
\end{align}
and from (\ref{prod2}) and (\ref{Poisson333}) we immediately see that the
first order expansion in $1/n$ of the commutator of $\star_{\vec{c}}^{n}$ does
not coincide with the Poisson bracket, nor does the zeroth order expansion in
$1/n$ of $\star_{\vec{c}}^{n}$ coincide with the pointwise product. Therefore,
this symbol correspondence sequence is not of Poisson type, either.
\end{example}

\subsection{Other types of symbol correspondence sequences}

However, although the two examples above exhibit non-Poissonian dynamics of
symbols in the asymptotic limit $n\rightarrow\infty$, the two symbol
correspondence sequences are not of the same asymptotic type. Namely, the
characteristic numbers in Example \ref{ex1} satisfy
\[
\displaystyle{\lim_{n\rightarrow\infty}c_{l}^{n}=0\ ,\ \forall l\in
\mathbb{N},}%
\]
whereas the characteristic numbers in Example \ref{ex2} satisfy
\[
\displaystyle{0\ <\ \lim_{n\rightarrow\infty}c_{l}^{n}\ <\ \infty\ ,\ \forall
l\in\mathbb{N}.}%
\]

\begin{definition}\index{Symbol correspondence sequences ! of limiting type}
\label{def-limiting-types} Let $\mathcal{C}:\Delta^{+}(\mathbb{N}%
^{2})\rightarrow\mathbb{R}^{\ast}$, cf. Definition \ref{corresequence}. The
symbol correspondence sequence ${\mathbf{W}}_{\mathcal{C}}$ determined by
characteristic numbers $c_{0}^{n}=1,\forall n\in\mathbb{N}$, $c_{l}%
^{n}=\mathcal{C}(n,l),\ \forall(n,l)\in\Delta^{+}(\mathbb{N}^{2})$, is of
\emph{limiting type} if
\begin{equation}
\displaystyle{\exists\lim_{n\rightarrow\infty}c_{l}^{n}\ ,\ \forall
l\in\mathbb{N}}, \label{limiting-type}%
\end{equation}
and it is of \emph{strong-limiting type} if, in addition,\index{Symbol correspondence sequences ! of strong-limiting type}
\begin{equation}
\exists\displaystyle{\lim_{(n,l)_{l\leq n}\rightarrow(\infty,\infty)}}%
|c_{l}^{n}|\ . \label{s-limiting-type}%
\end{equation}

\end{definition}

\begin{remark}
\label{limitnotation} In the above definition, the limits are taken in the
usual way. Thus, (\ref{limiting-type}) means that, $\forall l\in\mathbb{N}$,
$\exists\lambda(l)\in\mathbb{R}$, s.t. $\forall\epsilon>0$, $\exists
k(l,\epsilon)\in\mathbb{N}$ s.t. $n>k\Rightarrow|c^{n}_{l}-\lambda
(l)|<\epsilon$. And (\ref{s-limiting-type}) means that, $\exists\eta
\in\mathbb{R}$ s.t. $\forall\delta>0$, $\exists p(\delta), q(\delta
)\in\mathbb{N}$ s.t. $n\geq l, \ n>p, \ l>q \ \Rightarrow||c^{n}_{l}%
|-\eta|<\delta$.
\end{remark}

\begin{definition}
\label{classicalquasiclassical} A limiting-type symbol correspondence sequence
is of \emph{pseudo-classical type} if\index{Symbol correspondence sequences ! of pseudo-classical type}
\begin{equation}
0< \displaystyle{\lim_{n\rightarrow\infty}|c_{l}^{n}| < \infty \ ,\forall l\in
\mathbb{N},} \label{spuriousclassicaltype}%
\end{equation}
and it is of \emph{quasi-classical type} if\index{Symbol correspondence sequences ! of quasi-classical type}
\begin{equation}
\displaystyle{\lim_{n\rightarrow\infty}|c_{l}^{n}|=1\ ,\forall l\in
\mathbb{N}.} \label{quasiclassicaltype}%
\end{equation}

\end{definition}

The symbol correspondence sequence in Example \ref{ex1} is of strong-limiting
type, but not of pseudo-classical type, while the one in Example \ref{ex2} is
of pseudo-classical type, but not of strong-limiting type.

\begin{proposition}
\label{BerezinquasiBohr} Both the standard and the alternate Stratonovich-Weyl
symbol correspondence sequences are of quasi-classical type and of
strong-limiting type. On the other hand, both the standard and the alternate
Berezin and Toeplitz symbol correspondence sequences are of quasi-classical type, but not
of strong-limiting type.
\end{proposition}

\begin{proof}
The first statement is obvious. The quasi-classical property in the Berezin and Toeplitz 
cases follow from the expansion (\ref{BerezinCharexp}) for $b_{l}^{n}$, that
is,
\begin{equation}
\displaystyle{\lim_{n\rightarrow\infty}b_{l}^{n}}=b_{l}^{\infty}=1\ ,\forall
l\in\mathbb{N}. \label{qcBer}%
\end{equation}
On the other hand, note from the closed formula (\ref{BerezinChar}) for
$b_{l}^{n}$ that
\begin{equation}
\displaystyle{\lim_{n\rightarrow\infty}b_{n}^{n}=0,} \label{berezero}%
\end{equation}
and thus, equations (\ref{qcBer})-(\ref{berezero}) imply that
(\ref{s-limiting-type}) is not satisfied for the standard and alternate
Berezin and Toeplitz symbol correspondence sequences.
\end{proof}

\begin{example}
\label{ex3} The symbol correspondence sequence with $c^{n}_{0}=1, \ c^{n}%
_{l}=\mathcal{C}(n,l)=1-\log(1-((l-1)/n)), \ l>0$, and the one with
$\mathcal{C}(n,l)=1+n^{-(1/l)}$, are both of quasi-classical type, but not of
strong-limiting type.
\end{example}

\begin{remark}
\label{asympselfduality} (i) In view of Theorem \ref{CC} and Remark
\ref{CCdual}, we can say that a symbol correspondence sequence of
quasi-classical type is asymptotically self-dual, at least for finite $l$'s.
In view of equation (\ref{CCduality}), this is an important (classical)
feature for a symbol correspondence sequence. Thus, equation (\ref{berezero})
means that the standard and alternate Berezin and Toeplitz symbol correspondence sequences
loose this important asymptotic self-dual property if $l\to\infty$ as
$n\to\infty$.

(ii) On the other hand, from equations (\ref{qcBer}), (\ref{leg1}) and (\ref{NN}) it
follows that the Berezin transform, given by equation (\ref{Berezintransform}),
tends, as $n\to\infty$, to the identity on $Poly_{\mathbb{C}}(S^{2})_{\leq n}$, when applied to functions 
which are decomposable in a finite sum of spherical harmonics (finite $l$'s). The same can be said, of course, 
of its unique positive square root, the Stratonovich-Berezin transform, given by equation (\ref{StratBertransf}), and their respective inverses, given by equations (\ref{invbertransf}) and (\ref{BerStrattransf}). 
\end{remark}

As we have seen from Example \ref{ex1} and \ref{ex2}, symbol correspondence
sequences of limiting type may define asymptotic $(n\rightarrow\infty)$
dynamics of symbols which does not coincide with Hamilton-Poisson dynamics. In
other words, there are limiting-type symbol correspondence sequences of
non-Poisson type.

In fact, the following are necessary and sufficient conditions for a symbol
correspondence sequence to be of Poisson, or pure-Poisson type:

\begin{theorem}
\label{sufficientPoisson} A symbol correspondence sequence $\mathbf{W}%
_{\mathcal{C}}$ is of Poisson type if and only if its characteristic numbers
$c_{l}^{n}=\mathcal{C}(n,l)$ satisfy
\begin{equation}
\displaystyle{\lim_{n\rightarrow\infty}\mathcal{C}(n,l)=1}\ ,
\label{classical1}%
\end{equation}
$\mathbf{W}_{\mathcal{C}}$ is of anti-Poisson type if and only if
\begin{equation}
\displaystyle{\lim_{n\rightarrow\infty}\mathcal{C}(n,l)=(-1)^{l}}\ ,
\label{classical1'}%
\end{equation}
and $\mathbf{W}_{\mathcal{C}}$ is of pure-(resp. anti)-Poisson type if
\begin{equation}
\displaystyle{\lim_{n\rightarrow\infty}n(\mathcal{C}(n,l)-f(l))=0}\ ,
\label{s-classical1}%
\end{equation}
where $f(l)\equiv1$ (resp. $f(l)=(-1)^{l}$).
\end{theorem}

\begin{proof}
For the $f(l)\equiv1$ case, sufficiency of condition (\ref{s-classical1})
follows from equations (\ref{genprod}) and (\ref{asympWn}) and Corollary \ref{asymp}. The weaker
condition (\ref{classical1}) uses the fact that the first term of the
expansion (\ref{asympWn}) is commutative, while the second is
anti-commutative. The case $f(l)=(-1)^{l}$ follows from Proposition
\ref{relstandalt}.

On the other hand, from equations (\ref{WprodSym}), (\ref{genprod}),
(\ref{asympWn}), (\ref{Poisson11})-(\ref{Poisson33}), it follows as in
Proposition \ref{SWPoisson} that $\mathbf{W}_{\mathcal{C}}$ is of Poisson
type, or anti-Poisson type only if, $\forall l,l_{1},l_{2}\in\mathbb{N}$,
either $\displaystyle{\lim_{n\rightarrow\infty}}\frac{c_{l}^{n}}{c_{l_{1}}%
^{n}c_{l_{2}}^{n}}=1$, or $\displaystyle{\lim_{n\rightarrow\infty}}\frac
{c_{l}^{n}}{c_{l_{1}}^{n}c_{l_{2}}^{n}}=(-1)^{l_{1}+l_{2}+l}$, implying
condition (\ref{classical1}); that is, either $\displaystyle{\lim
_{n\rightarrow\infty}}\mathcal{C}(n,l)=f(l)\equiv1$, or $\displaystyle{\lim
_{n\rightarrow\infty}}\mathcal{C}(n,l)=f(l)=(-1)^{l}$. Necessity of
(\ref{s-classical1}) for the pure-Poisson case follows analogously.
\end{proof}

\begin{corollary}
\label{Poissonquasiclassical} Any symbol correspondence sequence of Poisson
type is of quasi-classical type, but the converse is generically not true for
non-positive sequences.
\end{corollary}

\begin{example}
\label{quasinotP} A symbol correspondence sequence such that%
\[
\displaystyle{\lim_{n\rightarrow\infty}}\mathcal{C}(n,l)=h(l)=\left\{
\begin{array}
[c]{cc}%
-1 & \text{if }l\equiv1\text{ }\operatorname{mod}3\\
1 & \text{otherwise}%
\end{array}
\right.
\]
is a symbol correspondence sequence of quasi-classical type which is not of
Poisson type.
\end{example}

\begin{remark}\label{renormalization} 
We thus note that any symbol correspondence sequence of pseudo-classical type can in principle be ``renormalized'' to become a symbol correspondence sequence of Poisson type, if the limit of each and every $l$-sequence of characteristic numbers is known.  
In fact, if $c_l^n=\mathcal{C}(n,l)$ denote the characteristic numbers of a symbol correspondence sequence of pseudo-classical type, then the ``renormalized'' symbol correspondence sequence with characteristic numbers $\chi_l^n=c_l^n/c_l^{\infty}$ is of Poisson type, where $0\neq c_l^{\infty}=\displaystyle{\lim_{n\rightarrow\infty}}c_l^n$. 
\end{remark}

Finally, the following definition expresses the optimal asymptotic conditions
for symbol correspondence sequences.

\begin{definition}\index{Symbol correspondence sequences ! of (pure)-(anti) Bohr type}
\label{defBohrtypes} A symbol correspondence sequence of Poisson (resp.
anti-Poisson) type which is also of strong-limiting type is of \emph{Bohr
}(resp. \emph{anti-Bohr}) type. A symbol correspondence sequence of pure-
(resp. pure-anti)-Poisson type which is also of strong-limiting type is of
\emph{pure-}(resp. \emph{pure-anti})\emph{-Bohr type}.
\end{definition}

\begin{remark}
\label{importantdistinction} The standard (resp. alternate) Berezin and Toeplitz symbol
correspondence sequences are of Poisson (resp. anti-Poisson) type, but not of
Bohr (resp. anti-Bohr) type (cf. Propositions \ref{BerezinPoisson} and
\ref{BerezinquasiBohr}). Similarly, both sequences in Example \ref{ex3} are of
Poisson type, but not of Bohr type.

The standard (resp. alternate) Stratonovich-Weyl symbol correspondence
sequence is of pure-(resp. pure-anti)-Bohr type (cf. Propositions
\ref{SWstrictPoisson} and \ref{BerezinquasiBohr}).
\end{remark}

\begin{example}
If $f(l)\equiv1$ (resp. $f(l)=(-1)^{l}$), any symbol correspondence sequence
with $c_{l}^{n}=\mathcal{C}(n,l)=f(l)g(n)\neq0,\ \forall(n,l)\in\Delta
^{+}(\mathbb{N}^{2})$, is of Bohr (resp. anti-Bohr) type if
$\displaystyle{\lim_{n\rightarrow\infty}}g(n)=1$, and it is of pure-(resp.
pure-anti)-Bohr type if $g(n)=1+ o(n^{-1})$, $n\rightarrow\infty$.
\end{example}

\begin{remark}
\label{high-difference} The distinction between symbol correspondence
sequences of Poisson type or Bohr type is manifest in the high-$l$-asymptotic dynamics of symbols, that is, asymptotic analysis of the
dynamics of symbols when $l\rightarrow\infty$ as $j\rightarrow\infty$ ; in
other words, highly oscillatory symbols and their twisted products.
\end{remark}

Due to space and time constraints, in this monograph we shall not study the
high-$l$-asymptotics of symbol correspondence sequences and twisted products, deferring this study to a later opportunity.
\index{Symbol correspondence sequences |)} 

\section{Final remarks and considerations}

The attentive reader may now ask  the following question: since every twisted $j$-algebra of polynomial functions is isomorphic to the quantum algebra of finite $(n+1)$-square matrices, while the classical Poisson algebra of smooth functions is infinite dimensional,  how can one pose statements about the $n\to\infty$ limit of twisted products without addressing the problem of passing from finite to infinite matrices? In other words, we have so far been clever to avoid addressing this problem, but is this really justified? If so, in what sense?         

In this respect, first we point out that various general works from different authors have studied the passage from finite matrices to infinite matrices and differential operators on infinite dimensional Hilbert spaces, which, as one should expect, is far from being a trivial problem, although it is facilitated in the case where all functions have compact support, as is the case of functions on the sphere.    

But, by working directly with the symbols of the operators, we were able to bypass this subtle problem altogether, in the sense of looking at the $2j=n\to\infty$ limit as an asymptotic limit of twisted $j$-algebras. Let us expand on this point.  

Knowing that  $Poly_{\mathbb C}(S^2)_{\leq n} = \{Y_l^m\}_{-l\leq m\leq l\leq n}$ densely approximates $\mathcal C^{\infty}_{\mathbb C}(S^2)$ as $n\to\infty$, we treat $1/n$, for $n\in\mathbb N$, as an asymptotic expansion parameter and look at the asymptotic expansions in this parameter of the expressions obtained for each sequence of twisted $j$-algebras, associated to each symbol correspondence sequence. 
As long as $1/n\neq 0$, we have $n\in\mathbb N$ and each $\vec c$-twisted $j$-algebra $(Poly_{\mathbb C}(S^2)_{\leq n}, \star^n_{\vec c})$   is a finite-dimensional algebra isomorphic to the matrix algebra of the spin-$j$ system. 
The case $1/n=0$ is never really considered, just as $\infty$ is neither a real nor a natural number. Thus, for each $n\in\mathbb N$ and general $Y_{l_1}^{m_1}, Y_{l_2}^{m_2}\in Poly_{\mathbb C}(S^2)_{\leq n}$, 
\begin{eqnarray}\label{normineq1} & (Y_{l_1}^{m_1} \star^n_{\vec c} Y_{l_2}^{m_2}  +  Y_{l_2}^{m_2} \star^n_{\vec c} Y_{l_1}^{m_1})/2 - Y_{l_1}^{m_1} Y_{l_2}^{m_2}\neq 0 \ , \nonumber \\   
& \label{normineq2}n(Y_{l_1}^{m_1} \star^n_{\vec c} Y_{l_2}^{m_2}  -  Y_{l_2}^{m_2} \star^n_{\vec c} Y_{l_1}^{m_1})/2 - i\{ Y_{l_1}^{m_1}, Y_{l_2}^{m_2}\} \neq 0 \ . \nonumber\end{eqnarray} 

However,  the errors, i.e. differences from zero in the r.h.s. of these expressions, become smaller and smaller as $n$ increases, \emph{if and only if} we have $\forall l\leq n$    
that $|c^n_l-1|$ is either zero or becomes smaller and smaller as $n$ increases keeping the $l$'s fixed (similarly for considering $|c^n_l-(-1)^l|$ in the anti-Poisson case). 

In other words, taking the sup-norm in 
 the space of smooth functions on the sphere, $||f||=sup(|f(\mathbf n)|, \mathbf n\in S^2)$,   then we can rewrite the above expressions as 
\begin{eqnarray}\label{normineq11} & ||(Y_{l_1}^{m_1} \star^n_{\vec c} Y_{l_2}^{m_2}  +  Y_{l_2}^{m_2} \star^n_{\vec c} Y_{l_1}^{m_1})/2 - Y_{l_1}^{m_1} Y_{l_2}^{m_2}||= S[\vec c]_{l_1,l_2}^{m_1,m_2}(n) \ ,  \\   
& \label{normineq21}||n(Y_{l_1}^{m_1} \star^n_{\vec c} Y_{l_2}^{m_2}  -  Y_{l_2}^{m_2} \star^n_{\vec c} Y_{l_1}^{m_1})/2 - i\{ Y_{l_1}^{m_1}, Y_{l_2}^{m_2}\}||= A[\vec c]_{l_1,l_2}^{m_1,m_2}(n)  \ , \quad \ \end{eqnarray}
where $\ S[\vec c]_{l_1,l_2}^{m_1,m_2}, \ A[\vec c]_{l_1,l_2}^{m_1,m_2}: \mathbb N\to\mathbb R^+$ are sequences of nonnegative real numbers satisfying 
\begin{equation}\label{Poissonfinal} \displaystyle{\lim_{n\to\infty}}S[\vec c]_{l_1,l_2}^{m_1,m_2}(n)=  \displaystyle{\lim_{n\to\infty}}A[\vec c]_{l_1,l_2}^{m_1,m_2}(n) = 0 \ ,  
\end{equation} 
for general $l_1,l_2$, \emph{if and only if} the sequence of characteristic numbers $c^n_l$ satisfies the conditions of Theorem \ref{sufficientPoisson}, cf. eq. (\ref{classical1}) (or eq. (\ref{classical1'}) for anti-Poisson).

Therefore, if this asymptotic condition on the characteristic numbers $c^n_l$ is satisfied, then we can say that the corresponding sequence of $\vec c$-twisted $j$-algebras approximates better and better the infinite-dimensional Poisson algebra of smooth functions on the sphere, knowing that for each $n\in\mathbb N$ we are only considering a finite-dimensional algebra, the $\vec c$-twisted $j$-algebra $(Poly_{\mathbb C}(S^2)_{\leq n}, \star^n_{\vec c})$ which is isomorphic to the operator algebra of the corresponding spin-$j$ quantum system.

In this respect, our approach is somewhat similar in spirit to the approach developed by Rieffel \cite{Rieffel, Rieffel2} that uses the standard Berezin (and Toeplitz) symbol correspondences to show that the operator algebra of spin-$j$ systems ``converges in quantum Gromov-Hausdorff distance'', as $n=2j\to\infty$, to the algebra of continuous functions on the sphere (under pointwise product, i.e. concerning expressions like (\ref{normineq1}) only). Much of Rieffel's work centers on precisely defining and proving this metric convergence, for which he relies on the Berezin transform (see Definition \ref{Bertrans1}). 
We note, however, that the standard Berezin (and Toeplitz) symbol correspondence sequences satisfy  Theorem \ref{sufficientPoisson}, meaning that the above necessary and sufficient condition on the characteristic numbers $c^n_l$ is satisfied (and as a consequence, we have Remark \ref{asympselfduality} $(ii)$, on the Berezin transform).

We should also point out, however, that the asymptotic relation obtained via equations (\ref{normineq11})-(\ref{Poissonfinal}) under the condition of Theorem \ref{sufficientPoisson}, between the sequence of operator algebras of spin-$j$ systems and the Poisson algebra on $S^2$, is considerably simpler than any similar asymptotic relation involving the algebra of bounded operators on $L^2_{\mathbb C}(\mathbb R)$ and the Poisson algebra on $\mathbb R^2$, for instance, where both algebras are infinite dimensional. In such cases, any sequence of finite dimensional algebras that aims at approximating the Poisson algebra asymptotically also needs to approximate the operator algebra at each stage, and controlling what happens to the latter (infinite-dimensional) kernel is another issue.      

On the other hand, by working with sequences of spin-$j$ operator algebras and their corresponding sequences of twisted $j$-algebras $(Poly_{\mathbb C}(S^2)_{\leq n}, \star^n_{\vec c})$, we were able to discover an interesting phenomenon: although all $\vec c$-twisted $j$-algebras are isomorphic for each fixed $n=2j\in\mathbb N$, any sequence of $\vec c$-twisted $j$-algebras that does not satisfy the condition that all $|c^n_l-1|$ are zero or become smaller and smaller as $n$ increases (or similarly for $|c^n_l-(-1)^l|$), will not provide a better and  better approximation to the Poisson algebra on $S^2$. Again, this is an asymptotic statement that is verified for $n$ larger and larger, though always finite.   

But, as the generic condition on sequences of characteristic numbers $c^n_l$ is the \emph{failure} of the condition expressed by Theorem \ref{sufficientPoisson}, cf. equations (\ref{classical1}) or  (\ref{classical1'}), the interested reader could 
then wonder whether such generic failure of sequences of $\vec c$-twisted $j$-algebras to approximate the classical Poisson algebra can have measurable consequences. We now turn to this question.

\subsubsection{``Empirical'' considerations}

In standard quantum mechanics, the operators themselves are not measurable
quantities, rather, measurable quantities are some expectation values which,
for spin systems, are expressed by the Hilbert-Schmidt inner product
(\ref{hilb}), i.e.
\begin{equation}
\langle P,Q\rangle=trace(P^{\ast}Q)\ . \label{measurable}%
\end{equation}

In view of equation (\ref{W}), 
the value at any point
$\mathbf{n}\in S^{2}$ of the symbol $W_{\vec{c}}^{j}(P)$ of any operator $P\in
M_{\mathbb{C}}(n+1)$ can be written in the form of equation (\ref{measurable}%
) as an expectation value, for any symbol correspondence $W_{\vec{c}}^{j}$.

Though most expectation values are quantum measurable quantities, the real
ones are more easily so. 
Therefore, the value of any real symbol (i.e. symbol of any hermitian
operator) at any point on $S^{2}$ can be assumed to be a quantum measurable
quantity for a spin system. On the other hand, the value of any real function
at any point on $S^{2}$ is a classical measurable quantity for a spin system.
Thus, under this assumption, the results of the previous section on the
asymptotic dynamics of symbols acquire a definite measurable significance.

However, even if the assumption on the quantum measurability of every real
symbol is considerably relaxed, the results of the previous section still have
far reaching measurability. To see this, let us focus on a simple instance,
namely assume $P=P^{\ast}$ and $trace(P)=1$, $H=H^{\ast}$, $A=A^{\ast}$, and
set $Q=[H,A]=i\hbar\dot{A}$. Then the equation
\begin{equation}
\langle P,\dot{A}\rangle=trace(P\dot{A})=\frac{1}{i\hbar}trace(P[H,A])
\label{timederivative}%
\end{equation}
has a standard empirical meaning in quantum mechanics with $H$ as the
Hamiltonian: $P$ is a generalized \textquotedblleft state\textquotedblright%
\ (a density operator, which is a pure state if $P$ is the projector onto a
one-dimensional subspace) and equation (\ref{timederivative}) measures the
time derivative of operator $A$ in state $P$, or the time derivative
of the expectation value of $A$ in state $P$, in the Heisenberg picture
(see, for instance, \cite{CT-D-L}).

For any symbol correspondence $W_{\vec{c}}^{j}$, the expectation value
(\ref{timederivative}) can be written in either of the equivalent forms given
by equations (\ref{CCduality}) and (\ref{inducedinnerproduct}),
\[
\frac{n+1}{4\pi}\int_{S^{2}}{\widetilde{W_{\vec{c}}^{j}}(P)}\ W_{\vec{c}}%
^{j}(\dot{A})\ dS=\frac{n+1}{4\pi}\int_{S^{2}}{W_{\vec{c}}^{j}(P)}\star
_{\vec{c}}^{n}W_{\vec{c}}^{j}(\dot{A})\ dS
\]%
\begin{align}
&  =\frac{n+1}{4\pi i\hbar}\int_{S^{2}}{W_{\vec{c}}^{j}(P)}\star_{\vec{c}}%
^{n}[W_{\vec{c}}^{j}(H),W_{\vec{c}}^{j}(A)]_{\star_{\vec{c}}^{n}%
}\ dS\label{measurable1}\\
&  =\frac{n+1}{4\pi i\hbar}\int_{S^{2}}\widetilde{W_{\vec{c}}^{j}(P)}%
[W_{\vec{c}}^{j}(H),W_{\vec{c}}^{j}(A)]_{\star_{\vec{c}}^{n}}\ dS\ \nonumber
\end{align}
where $\widetilde{W_{\vec{c}}^{j}}$ is the symbol correspondence dual to
$W_{\vec{c}}^{j}$ (see Remark \ref{CCdual}) and $[W_{\vec{c}}^{j}%
(H),W_{\vec{c}}^{j}(A)]_{\star_{\vec{c}}^{n}}$ is the twisted commutator of
$W_{\vec{c}}^{j}(H)$ and $W_{\vec{c}}^{j}(A)$.

Now, suppose that all above symbols have well defined $j\rightarrow\infty$
asymptotic limits of non highly-oscillatory type (decomposable into finite
sums of spherical harmonics). In fact, let us also suppose that the symbol
correspondence sequence of $W_{\vec{c}}^{j}$ is of quasi-classical type (cf.
Definition \ref{classicalquasiclassical}) and denote by $p,h,a$ the asymptotic
limits of $W_{\vec{c}}^{j}(P)$, $W_{\vec{c}}^{j}(H)$, $W_{\vec{c}}^{j}(A)$,
respectively. From the asymptotic self-dual property of symbol correspondence
sequences of quasi-classical type (cf. Remark \ref{asympselfduality}), the
last expression in (\ref{measurable1}) has asymptotic limit
\begin{equation}
\frac{1}{4\pi}\int_{S^{2}}p[h,a]_{\infty}dS\ , \label{measurable3}%
\end{equation}
where
\begin{equation}
\lbrack h,a]_{\infty}=\lim_{n\rightarrow\infty}\frac{n+1}{i\hbar}[W_{\vec{c}%
}^{j}(H),W_{\vec{c}}^{j}(A)]_{\star_{\vec{c}}^{n}}
\nonumber\label{measurable4}%
\end{equation}
(In order to relate this with the equations of the previous Chapter, set
$\hbar=2$).

On the other hand, from Hamilton-Poisson dynamics, the classical limit of
(\ref{timederivative}) should be
\begin{equation}
\frac{1}{4\pi}\int_{S^{2}}p\{h,a\}dS\ , \label{measurable5}%
\end{equation}
where $\{h,a\}$ is the Poisson bracket of $h,a$. However, if the
quasi-classical symbol correspondence sequence of $W_{\vec{c}}^{j}$ is not of
Poisson type, as in Example \ref{quasinotP} for instance, then $[h,a]_{\infty
}\neq\{h,a\}$ and, in this case, the integrals (\ref{measurable3}) and
(\ref{measurable5}) will in general not coincide.

Similar asymptotics of (\ref{timederivative}) can be performed for other kinds
of limiting-type symbol correspondence sequences which are not of Poisson
type, with similar conclusions. For instance, it is not too difficult to see
that, for a positive symbol correspondence sequence $W_{\vec{c}}^{j}$ of
pseudo-classical type (cf. Definition \ref{classicalquasiclassical}), if
$p,h,a$ are the asymptotic limits as before, then we have in general
\[
\lim_{n\rightarrow\infty}\frac{n+1}{4\pi i\hbar}\int_{S^{2}}{W_{\vec{c}}%
^{j}(P)}\star_{\vec{c}}^{n}[W_{\vec{c}}^{j}(H),W_{\vec{c}}^{j}(A)]_{\star
_{\vec{c}}^{n}}\ dS\ \neq\ \frac{1}{4\pi}\int_{S^{2}}p\{h,a\}dS\ .
\]
And this is also typical for other instances of quantum expectation values. 

On the other hand, in view of Remark \ref{renormalization}, every symbol correspondence sequence of 
pseudo-classical type (with characteristic numbers $c_l^n$) can in principle be mapped to a ``renormalized'' symbol correspondence 
sequence of Poisson type (with characteristic numbers $\chi_l^n$), so  that none of the measurable idiosyncrasies  discussed above applies to the asymptotics of the ``renormalized'' symbols. However, 
such a ``renormalization'' requires knowing $c_l^{\infty}=\displaystyle{\lim_{n\rightarrow\infty}}c_l^n \ , \ \forall l\in\mathbb N$, and this may not be the case, or it may  be impractical to perform such an ``asymptotic renormalization'' on an already pre-established symbol correspondence sequence.    

Now, recalling that the standard Stratonovich-Weyl and Berezin symbol
correspondence sequences are of Poisson type, it would be interesting to see whether the
non-pure-Poisson property of Berezin symbols can have nontrivial
measurability in low-$l$ high-$j$-asymptotics. 
Furthermore, the identity
(\ref{berezero}) suggests that their high-$l$-asymptotics are quite different, in
the Stratonovich-Weyl and Berezin (and Toeplitz) cases. But this is a considerably harder question, which could perhaps be better addressed by using the integral formulations studied in section  \ref{integralsection}. However, these would be much more useful, in this respect, if we had been able to obtain closed formulae for the integral trikernels under consideration.  
Therefore, further investigations in this direction could turn out to be profitable.  

An alternative standpoint  could be established by obtaining adequate asymptotical approximations for these trikernels which could be used in high-$l$ asymptotic investigations. A possible approach to this goal is to introduce appropriate $j$-dependent scalings on the spheres so that, in the $n\to\infty$ limit, spherical patches tend to the symplectic plane and the spherical  trikernels tend to  well-known affine ones in small neighborhoods. In terms of the symmetry groups, this path leads to  a contraction of the Lie algebra of $SU(2)$ to the Lie algebra of the Heisenberg group. 
Some work along these lines has been carried out in the context of the group $SU(1,1)$, instead of $SU(2)$, where, by first performing a contraction of the Lie algebra and later  performing a ``quantized decontraction'', some closed formulae for trikernels on the hyperbolic plane have been obtained \cite{BDS}. One wonders, however, if the completely different topologies of $SU(2)$ and the Heisenberg group can allow for any useful outcome of this kind of procedure, in the spherical case.  

Another approach is to use other asymptotic formulae for the Wigner $3jm$ and $6j$ symbols, that can be used to obtain high-$l$ asymptotic approximations of the products and/or the trikernels. Particular formulae for these Wigner symbols are known in the asymptotic limit  when the $l$'s tend to $\infty$ linearly with $j$, that is, keeping all fractions $l/j$ fixed. The respective nonuniform formulae have long been known, cf.  \cite{PR}, but uniform formulae are also known, cf. \cite{Lit2}, \cite{Lit3}.          
Still another approach is to work with the integral formulae for the trikernels, which were obtained in section \ref{integralintegral}, either to obtain asymptotic formulae for the trikernels themselves, or directly to the products of highly oscillatory functions.  

Finally, we end this subsection  with an important clarification of its context:  at this point, all of the above ``empirical'' considerations are purely theoretical  and whether any of these can eventually be actually observed in a real physical laboratory in some possible future, is at present totally unknown to us.

%----------------------------------------------------------------------------------------------------------------------------------------------------------------------------------------------------------------------
%-------------------------------------------------------------------------------------- Chapter 9 -----------------------------------------------------------------------------------------------------------------
%----------------------------------------------------------------------------------------------------------------------------------------------------------------------------------------------------------------------

\chapter{Conclusion}

Since the work Bayen, Flato,  Frondsal, Lichnerowicz and 
Sternheimer on deformation quantization
\cite{Betal}, much emphasis has been placed on a class of problems initially
known as quantization of Poisson manifolds. At first, the deformation
quantization program, which started in \cite{Betal} but was inspired by the
much older work of Moyal \cite{Moy}, seemed to promise a definitive approach
towards a precise mathematical relationship between quantum and classical
mechanics in a unique and general setting. And soon, approaches to invariant deformation quantization were set forth, as the early work of Bayen and Frondsal \cite{BF} on the \emph{formal} deformation quantization of the $2$-sphere. Moreover, this promise of a general formalism showed
itself stronger after the works of Fedosov \cite{Fed} and Kontsevich
\cite{Kon}, thus inspiring many to enlarge the program to ever more general settings.

However, as could have been clear from the start of the program, already for
the case of affine symplectic spaces, the deformation quantization approach is
not so well suited to handle highly oscillatory functions. These are common in
some WKB semiclassical approximation of certain types of operators in ordinary
quantum mechanics, particularly projectors or evolution operators in the Weyl
representation (see, for instance, the discussion in \cite{RO}).

Furthermore, the promise of a very general framework for quantization took a
hard blow with the work of Rieffel \cite{Rieffel1}, based on the work of 
Wassermann \cite{Was}, which showed that, in the
simple case of the homogeneous $2$-sphere, any $SO(3)$-invariant
\textquotedblleft strict deformation quantization\textquotedblright\ of
$S^{2}$ has to be isomorphic to some $SU(2)$-invariant finite matrix algebra,
or some sequence of $SU(2)$-invariant matrix algebras in reverse order, i.e.
of finite dimensions decreasing from infinity. Here, by \emph{strict} deformation
quantization, one should understand a closed associative noncommutative
algebra in some function subspace of ${C}_{\mathbb C}(S^{2})$, with all
the required properties of a deformation quantization. Thus, in particular,
equations analogous to (\ref{Poisson11})-(\ref{Poisson33'}) have to be
satisfied (see \cite{Rieffel1} for more details; see also \cite{BMS, Land}).

In this way, the understanding that the path from classical to quantum
mechanics can be quite more subtle than straightforward, and quite more
peculiar than generic, once again could not be dismissed.

On the other hand, the path from quantum to classical mechanics has often been
thought to be unique, at least in principle. Despite the various methods of
semiclassical approximation in affine mechanical systems, these have often
been thought of as different approximations pertaining to an underlying unique
limiting procedure. Thus, it is commonly believed that semiclassical
approximations to the Weyl-Wigner formalism should, for instance, not be
essentially different from semiclassical approximations to the coherent-state
formalism, as the two approximations are commonly believed to be empirically equivalent.

The case of spin systems studied in this monograph shows that, on the contrary,
the path from quantum to classical mechanics is very far from being unique.
Different symbol correspondence sequences yield different semiclassical limits, 
when such a limit actually exists, which does not
always happen. 

Therefore, a generic symbol correspondence sequence defines a
 ``quantization of $S^{2}$ in reverse order'',  i.e. from quantum to ``classical'', or better a sequence of ``fuzzy
spheres'',   in the sense of defining a sequence of function
algebras satisfying Proposition \ref{twistedproperties}, what we have called in this book a sequence of ``$\vec c$-twisted $j$-algebras''. However, generically
this is not a reversed-order deformation of the classical sphere, in the sense that to be a reversed-order  
deformation of the classical sphere equations (\ref{Poisson11}%
)-(\ref{Poisson33'}) must also be satisfied.

Only a subclass of symbol correspondence sequences yield Poisson dynamics on
$S^{2}$ in the asymptotic $n\to\infty$ limit. This subclass, the subclass of symbol correspondence
sequences of Poisson (or anti-Poisson) type, realizes strict deformation
quantizations of the classical two-sphere in reverse order. To
this subclass belong the standard and the alternate Stratonovich-Weyl, as well as
the standard and the alternate Berezin symbol correspondences, which are the
spherical analogues of the Weyl-Wigner and the coherent-state representations
of affine quantum mechanics, whose classical limits yield Poisson dynamics
in affine symplectic space, at least for non-highly-oscillatory functions
(again, see \cite{RO} for more details).

Thus, it is important to emphasize that the symbol correspondence sequences
outside this subclass define symbolic dynamics on $S^{2}$ which need not
be empirically equivalent to Poisson dynamics in the asymptotic $n\to\infty$ limit. 
On the other hand, further investigations are in order, to assert the possibility of 
empirical distinctions within the subclass of symbol correspondence sequences of Poisson type. 

In this respect, we could benefit from  a more detailed asymptotical comparison between the standard 
Stratonovich-Weyl and the standard Berezin (and Toeplitz) symbol correspondence sequences, 
particularly from the point of view of possible quantum measurability of their distinctions.  
This could take the form of: (i) understanding possible measurable consequences of the 
higher order terms in the expansion (\ref{BerezinCharexp}) for 
the standard Berezin characteristic numbers, in low-$l$ high-$j$ asymptotics, 
or: (ii) some substantial understanding of the high-$l$ asymptotics of these correspondences, since (\ref{berezero}) is an indication that the high-$l$ asymptotical dynamics of the standard Berezin (and Toeplitz) symbol correspondence could be distinguishable from the high-$l$ asymptotical dynamics of the standard Stratonovich-Weyl correspondence (or some other symbol correspondence sequence of Bohr type, arguably the ``best type'', of which the standard and alternate Stratonovich-Weyl correspondences are the supreme prototypes). Much help for  (ii) could come from obtaining closed formulas for these trikernels, which we haven't yet been able to acquire, but a less ambitious goal would be deriving adequate asymptotical expressions for these trikernels which could be used for (ii). 

Finally, we end this chapter with a more philosophical conclusion, perhaps the most important conclusion of this monograph.

It has long been 
recognized by many, mainly physicists but also mathematicians, that quantum mechanics ``carries more information'' or ``is
actually bigger'' than classical mechanics, meaning that one cannot produce full
quantum dynamics unambiguously, solely on the basis of classical data (the problem of
strictly quantizing the sphere mentioned above being an instance of this
general principle).

However,  the abounding existence in spin systems of symbol correspondence
sequences of non-Poisson type means that, in order to guarantee classical
Poisson dynamics of spherical symbols, some limiting constraints must be
placed upon the symbol correspondence sequences. In other words, classical
information must be added to the quantum data, as well. 
This fact brings forth
the realization, for spin systems, that a full consistent theory relating
quantum and classical mechanics is actually bigger than either of these two
theories alone.

%----------------------------------------------------------------------------------------------------------------------------------------------------------------------------------------------------------------------
%-------------------------------------------------------------------------------------- Appendix -----------------------------------------------------------------------------------------------------------------
%----------------------------------------------------------------------------------------------------------------------------------------------------------------------------------------------------------------------

\chapter*{Appendix: further proofs}\label{App2}\addcontentsline{toc}{chapter}{Appendix}
\markboth{Appendix}{Appendix}\thispagestyle{empty}
\setcounter{chapter}{0}

\section{A proof of Proposition \ref{WigD-CG}}
\label{ProofWDCG}

We shall derive the coupling rule (\ref{couple1}) and its inversion
(\ref{couple2}) by formal reasoning with the Clebsch-Gordan coefficients.
Starting from the formula (\ref{Clebsch2}), let $g\in SU(2)$ act on both
sides, which yields (with $m=m_{1}+m_{2}$)
\begin{equation}%
%TCIMACRO{\dsum \limits_{\mu_{1},\mu_{2}}}%
%BeginExpansion
{\displaystyle\sum\limits_{\mu_{1},\mu_{2}}}
%EndExpansion
D_{\mu_{1},m_{1}}^{j_{1}}D_{\mu_{2},m_{2}}^{j_{2}}\left\vert j_{1}\mu_{1}%
j_{2}\mu_{2}\right\rangle =\sum_{j}C_{m_{1},m_{2},m}^{j_{1},j_{2},j}%
%TCIMACRO{\dsum \limits_{\mu}}%
%BeginExpansion
{\displaystyle\sum\limits_{\mu}}
%EndExpansion
D_{\mu,m}^{j}\left\vert (j_{1}j_{2})j\mu\right\rangle \text{ }
\label{Clebsch4}%
\end{equation}
Now, substitute the expansion of type (\ref{Invers}) for the coupled basis
vector $\left\vert (j_{1}j_{2})j\mu\right\rangle $ in (\ref{Clebsch4}) and
obtain (with $\mu=\mu_{1}^{\prime}+\mu_{2}^{\prime}$)
\begin{align}%
%TCIMACRO{\dsum \limits_{\mu_{1},\mu_{2}}}%
%BeginExpansion
{\displaystyle\sum\limits_{\mu_{1},\mu_{2}}}
%EndExpansion
D_{\mu_{1},m_{1}}^{j_{1}}D_{\mu_{2},m_{2}}^{j_{2}}\left\vert j_{1}\mu_{1}%
j_{2}\mu_{2}\right\rangle  &  =\sum_{j}%
%TCIMACRO{\dsum \limits_{\mu_{1}^{\prime},\mu_{2}^{\prime}}}%
%BeginExpansion
{\displaystyle\sum\limits_{\mu_{1}^{\prime},\mu_{2}^{\prime}}}
%EndExpansion
C_{m_{1},m_{2},m}^{j_{1},j_{2},j}C_{\mu_{1}^{\prime},\mu_{2}^{\prime},\mu
}^{j_{1},j_{2},j}D_{\mu,m}^{j}\left\vert j_{1}\mu_{1}^{\prime}j_{2}\mu
_{2}^{\prime}\right\rangle \text{ }\label{Clebsch5}\\
&  =%
%TCIMACRO{\dsum \limits_{\mu_{1},\mu_{2}}}%
%BeginExpansion
{\displaystyle\sum\limits_{\mu_{1},\mu_{2}}}
%EndExpansion
\sum_{j}C_{m_{1},m_{2},m}^{j_{1},j_{2},j}C_{\mu_{1},\mu_{2},\mu}^{j_{1}%
,j_{2},j}D_{\mu,m}^{j}\left\vert j_{1}\mu_{1}j_{2}\mu_{2}\right\rangle
,\nonumber
\end{align}
where the last expression follows from the second by the change of notation
$\mu_{i}^{\prime}\rightarrow\mu_{i}$. Since both sides of (\ref{Clebsch5}) are
linear combinations of the uncoupled basis, corresponding coefficients are
identical, consequently formula (\ref{couple1}) must hold.

Next, by applying $g\in SU(2)$ to both sides of the formula (\ref{Invers}),%
\begin{align*}
\sum_{\mu}D_{\mu,m}^{j}\left\vert (j_{1}j_{2})j\mu\right\rangle  &
=\sum_{m_{1}}C_{m_{1},m_{2},m}^{j_{1},j_{2},j}\sum_{\mu_{1},\mu_{2}}D_{\mu
_{1},m_{1}}^{j_{1}}D_{\mu_{2},m_{2}}^{j_{2}}\left\vert j_{1}\mu_{1}j_{2}%
\mu_{2}\right\rangle \\
&  =\sum_{m_{1}}C_{m_{1},m_{2},m}^{j_{1},j_{2},j}\sum_{\mu_{1},\mu_{2}}%
D_{\mu_{1},m_{1}}^{j_{1}}D_{\mu_{2},m_{2}}^{j_{2}}\sum_{k}C_{\mu_{1},\mu
_{2},\mu}^{j_{1},j_{2},k}\left\vert (j_{1}j_{2})k\mu\right\rangle
\end{align*}
Now, choosing $k=j$ and fixing the value of $\mu$, comparison of the
coefficient of $\left\vert (j_{1}j_{2})j\mu\right\rangle $ on both sides of
the previous identity yields%
\[
D_{\mu,m}^{j}=\sum_{\mu_{1},\mu_{2}}D_{\mu_{1},m_{1}}^{j_{1}}D_{\mu_{2},m_{2}%
}^{j_{2}}\sum_{m_{1}}C_{\mu_{1},\mu_{2},\mu}^{j_{1},j_{2},j}C_{m_{1},m_{2}%
,m}^{j_{1},j_{2},j}%
\]
which is, in fact, the identity (\ref{couple2}) since only terms with $\mu
_{2}=\mu-\mu_{1}$ can give a nonzero contribution.

\section{A proof of Proposition \ref{Emult} \label{parity prop}}

The parity property for the product of operators, Proposition \ref{Emult},
follows straightfowardly from the product rule for the coupled basis of
operators, Corollary \ref{emult2 copy(1)}, and the symmetry properties of the
Wigner $3jm$ and $6j$ symbols, as stated at the end of Chapter 2. However, it
is possible to prove the parity property in an independent way, which
highlights the large amounts of combinatorics that are encoded in the Wigner
$3jm$ and $6j$ symbols. Thus, we now present this direct proof of Proposition
\ref{Emult}, namely the parity property for the matrices%

\[
E(l,m)=(-1)^{l}\mu_{l,m}^{n}\mathbf{e}(l,m)
\]
introduced in Chapter \ref{invmatrixbasis}. This independent proof was worked
out in collaboration with Nazira Harb.

First, we need some preliminary results. For greater clarity, we shall retain
the following notation used in Chapter 5.1.1
\begin{equation}
A=J_{+}\in\Delta(1)\text{, }B=A^{T}=J_{-}\in\Delta(-1). \label{AB}%
\end{equation}
so that all matrices in $M_{\mathbb{R}}(n+1)$ are expressible as linear
combinations of monomials or "words" in the letters $A$ and $B$. Observe that
each monomial%

\begin{equation}
P=A^{a_{1}}B^{b_{1}}A^{a_{2}}B^{b_{2}}...A^{a_{p}}B^{b_{p}}\text{, \ }%
a_{i}\geq0,b_{i}\geq0 \label{5f}%
\end{equation}
is an $m$-subdiagonal matrix for some $m$ in the range $-n\leq m\leq n$,
\begin{equation}
P=(x_{1},x_{2},..,x_{k})_{m}\text{ , }k=n+1-\left\vert m\right\vert \label{5j}%
\end{equation}
with nonnegative entries $x_{i}$. We shall refer to $m$ as the \emph{weight}
$\mu(P)$ of $P$, and consequently the monomial (\ref{5f}) has weight
\[
m=\mu(P)=%
%TCIMACRO{\dsum }%
%BeginExpansion
{\displaystyle\sum}
%EndExpansion
a_{i}-%
%TCIMACRO{\dsum }%
%BeginExpansion
{\displaystyle\sum}
%EndExpansion
b_{i}%
\]
In particular, $\mu(P^{T})=-\mu(P$), $\mu(A)=$ $1$, $\mu(B)=-1$, and diagonal
matrices has weight zero. Moreover,%

\[
\mu(PQ)=\mu(P)+\mu(Q)\text{, \ }trace(P)\neq0\text{\ }\Longrightarrow
\mu(P)=0.
\]

We shall also compare a monomial with its \emph{reverse} monomial, namely the
reverse of $X$ in (\ref{5f}) is by definition
\begin{equation}
P^{rev}=B^{b_{p}}A^{a_{p}}...B^{b_{2}}A^{a_{2}}B^{b_{1}}A^{a_{1}} \label{5h}%
\end{equation}

\begin{lemma}
\label{trace} A monomial matrix (\ref{5f}) and its reverse (cf. (\ref{5h}))
are related as follows :%
\[
P=(x_{1},x_{2},..,x_{k})_{m}\text{, }P^{rev}=(x_{k},x_{k-1},..,x_{1})_{m}%
\]
In particular, $trace(P)=trace(P^{rev}).$
\end{lemma}

\begin{proof}
Define the \emph{height} of the monomial in (\ref{5f}) to be the number $h(P)=%
%TCIMACRO{\dsum }%
%BeginExpansion
{\displaystyle\sum}
%EndExpansion
a_{i}+%
%TCIMACRO{\dsum }%
%BeginExpansion
{\displaystyle\sum}
%EndExpansion
b_{i}$.We shall prove the lemma by induction on the height (rather than
weight). The lemma holds for monomials of height $1$, namely $A$ and $B$,
which are their own reverse. For example,
\begin{equation}
A=(\alpha_{1},\alpha_{2},..,\alpha_{n})_{1}\text{, \ }\alpha_{1}=\alpha
_{n},\alpha_{2}=\alpha_{n-1},... \label{5i}%
\end{equation}
Now, assume the lemma holds for all monomials of height $h$, and let $Y$ be a
monomial of height $h+1$. Then $Q=AP$ or $BP$, say $Q=AP$ where $P$ is the
$m$-subdiagonal matrix (\ref{5j}) and $0\leq m<n$. By assumtion,
$P^{rev}=(x_{k},x_{k-1},..,x_{1})_{m}$ and we calculate
\[
Q=(\alpha_{1}x_{2},\alpha_{2}x_{3},..,\alpha_{k-1}x_{k})_{m+1},\text{ }%
Q^{rev}=P^{rev}A=(\alpha_{m+1}x_{k},\alpha_{m+2}x_{k-1},..,\alpha_{n}%
x_{2})_{m+1}%
\]
By the symmetry of $A$ illustrated in (\ref{5i}), the lemma also holds for $Q
$. The case $Q=BP$ is similar and hence it is omitted.
\end{proof}

Now, we turn to the proof of the parity property. The two cases (i) and (ii)
of Theorem \ref{Emult} are similar, so let us choose case (i) and give a
detailed proof, which amounts to show the following inner product
\begin{align*}
&  \left\langle E(l,m),\left[  E(l_{1},m_{1}),E(l_{2},m_{2})\right]
\right\rangle \\
&  =%
%TCIMACRO{\dsum _{k=0}^{l-m}}%
%BeginExpansion
{\displaystyle\sum_{k=0}^{l-m}}
%EndExpansion%
%TCIMACRO{\dsum _{i=0}^{l_{1}-m\,_{1}}}%
%BeginExpansion
{\displaystyle\sum_{i=0}^{l_{1}-m\,_{1}}}
%EndExpansion%
%TCIMACRO{\dsum _{j=0}^{l_{2}-m_{2}}}%
%BeginExpansion
{\displaystyle\sum_{j=0}^{l_{2}-m_{2}}}
%EndExpansion
(-1)^{i+j+k}\binom{l-m}{k}\binom{l_{1}-m_{1}}{i}\binom{l_{2}-m_{2}}{j}\\
&
\begin{array}
[c]{c}%
\cdot\{trace(A^{k}B^{l}A^{l-m-k}B^{l_{1}-m_{1}-i}A^{l_{1}}B^{l_{2}-m_{2}%
-j+i}A^{l_{2}}B^{j})\\
-trace(A^{k}B^{l}A^{l-m-k}B^{l_{2}-m_{2}-j}A^{l_{2}}B^{l_{1}-m_{1}%
-i+j}A^{l_{1}}B^{i})\}
\end{array}
\end{align*}
vanishes when we assume $l$ $\equiv l_{1}+l_{2}(\operatorname{mod}2)$ and
$m=m_{1}+m_{2}$. The vanishing of the inner product is immediate when $m\neq
m_{1}+m_{2}$ since each $E(l,m)$ is an $m$-subdiagonal matrix. Thus, for the
proof, let us consider separately the two cases : either $l-m$ is odd or $l-m$
is even.

First, assume $l-m$ is odd and divide the summation over $k$ into two sums :%
\[
\Sigma=%
%TCIMACRO{\dsum _{k=0}^{l-m}}%
%BeginExpansion
{\displaystyle\sum_{k=0}^{l-m}}
%EndExpansion
[...]=\Sigma_{1}+\Sigma_{2}=%
%TCIMACRO{\dsum \limits_{k=0}^{(l-m-1)/2}}%
%BeginExpansion
{\displaystyle\sum\limits_{k=0}^{(l-m-1)/2}}
%EndExpansion
\ [....]+%
%TCIMACRO{\dsum \limits_{\ \ k=(l-m+1)/2}^{l-m}}%
%BeginExpansion
{\displaystyle\sum\limits_{\ \ k=(l-m+1)/2}^{l-m}}
%EndExpansion
[...]
\]
In the second sum we make the substitution $(k,i,j)\rightarrow(t,r,s)$ by
setting $t=l-m-k,r=l_{1}-m_{1}-i,s=l_{2}-m_{2}-j$; in particular%
\[
i+j+k=(l-m-t)+(l_{1}-m_{1}-r)+(l_{2}-m_{2}-s)\equiv r+s+t(\operatorname{mod}%
2)
\]
because of the assumption $l\equiv l_{1}+l_{2}$. Consequently, the second sum
becomes%
\begin{align*}
\Sigma_{2}  &  =%
%TCIMACRO{\dsum \limits_{t=0}^{\frac{l-m-1}{2}}}%
%BeginExpansion
{\displaystyle\sum\limits_{t=0}^{\frac{l-m-1}{2}}}
%EndExpansion%
%TCIMACRO{\dsum _{r=0}^{l_{1}-m\,_{1}}}%
%BeginExpansion
{\displaystyle\sum_{r=0}^{l_{1}-m\,_{1}}}
%EndExpansion%
%TCIMACRO{\dsum _{s=0}^{l_{2}-m_{2}}}%
%BeginExpansion
{\displaystyle\sum_{s=0}^{l_{2}-m_{2}}}
%EndExpansion
(-1)^{r+s+t}\binom{l-m}{t}\binom{l_{1}-m_{1}}{r}\binom{l_{2}-m_{2}}{s}\\
&  \cdot\{trace(A^{l-m-t}B^{l}A^{t}B^{r}A^{l_{1}}B^{l_{1}-m_{1}+s-r}A^{l_{2}%
}B^{l_{2}-m_{2}-s})\\
&  -trace(A^{l-m-t}B^{l}A^{t}B^{s}A^{l_{2}}B^{l_{2}-m_{2}-s+r}A^{l_{1}%
}B^{l_{1}-m_{1}-r})\},
\end{align*}
and by the change of notation $(t,r,s)\rightarrow(k,i,j)$ in the expression
$\Sigma_{2}$, we can write
\[
\Sigma=%
%TCIMACRO{\dsum _{k=0}^{\frac{l-m-1}{2}}}%
%BeginExpansion
{\displaystyle\sum_{k=0}^{\frac{l-m-1}{2}}}
%EndExpansion%
%TCIMACRO{\dsum _{i=0}^{l_{1}-m\,_{1}}}%
%BeginExpansion
{\displaystyle\sum_{i=0}^{l_{1}-m\,_{1}}}
%EndExpansion%
%TCIMACRO{\dsum _{j=0}^{l_{2}-m_{2}}}%
%BeginExpansion
{\displaystyle\sum_{j=0}^{l_{2}-m_{2}}}
%EndExpansion
(-1)^{i+j+k}\binom{l-m}{k}\binom{l_{1}-m_{1}}{i}\binom{l_{2}-m_{2}}{j}%
\cdot\{...\}
\]
where
\begin{align*}
&  \{...\}=\{trace(A^{k}B^{l}A^{l-m-k}B^{l_{1}-m_{1}-i}A^{l_{1}}B^{l_{2}%
-m_{2}-j+i}A^{l_{2}}B^{j})\\
&  -trace(A^{k}B^{l}A^{l-m-k}B^{l_{2}-m_{2}-j}A^{l_{2}}B^{l_{1}-m_{1}%
-i+j}A^{l_{1}}B^{i})\\
&  +trace(A^{l-m-k}B^{l}A^{k}B^{i}A^{l_{1}}B^{l_{1}-m_{1}+j-i}A^{l_{2}%
}B^{l_{2}-m_{2}-j})\\
&  -trace(A^{l-m-k}B^{l}A^{k}B^{j}A^{l_{2}}B^{l_{2}-m_{2}-j+i}A^{l_{1}%
}B^{l_{1}-m_{1}-i})\}\\
&  =\{trace(A^{k}B^{l}A^{l-m-k}B^{l_{1}-m_{1}-i}A^{l_{1}}B^{l_{2}-m_{2}%
-j+i}A^{l_{2}}B^{j})\\
&  -trace(B^{j}A^{l_{2}}B^{l_{2}-m_{2}-j+i}A^{l_{1}}B^{l_{1}-m_{1}-i}%
A^{l-m-k}B^{l}A^{k})\\
&  +trace(B^{i}A^{l_{1}}B^{l_{1}-m_{1}+j-i}A^{l_{2}}B^{l_{2}-m_{2}-j}%
A^{l-m-k}B^{l}A^{k})\\
&  -trace(A^{k}B^{l}A^{l-m-k}B^{l_{2}-m_{2}-j}A^{l_{2}}B^{l_{1}-m_{1}%
-i+j}A^{l_{1}}B^{i})
\end{align*}
To obtain the last expression of $\{...\}$\ we have rearranged the four trace
terms of $\{...\}$ in the new order $1,4,3,2$, and we have also made use of
the cyclic property of the trace.\ Now, it follows from Lemma \ref{trace} that
the expression $\{...\}$ vanishes identically, for each triple $(k,i,j)$ of indices.

Next, if $l-m$ is an even integer, we break the sum $\Sigma$ over $k$ into two
sums $\Sigma_{1}$ and $\Sigma_{2}$. In the first sum, $k=0,1..,\frac{l-m}%
{2}-1$, plus the first trace term of $\{...\}$ in (\ref{5f}) for $k=\frac
{l-m}{2}$ (and summation over $i,j$, of course). In the second sum, $k=$
$\frac{l-m}{2}+1,..,l-m$, plus the second trace term for $k=$ $\frac{l-m}{2}$.
Then the proof of the vanishing of $\Sigma$ follows analogously, and this
completes the proof of property (i). Property (ii) is proven analogously.

Finally, to complete the proof of Theorem \ref{Emult} it remains to show that
the product $E(l_{1},m_{1})E(l_{2},m_{2})$ is a linear combination of terms
$E(l,m)$ with $\left\vert l_{1}-l_{2}\right\vert \leq l\leq l_{1}+l_{2}$.

But, in the linear expansion of the product, a typical term $E(l,m)$ belongs
to the matrix subspace $M_{\mathbb{C}}(\varphi_{l})$. The operator $A=J_{+}$
acts as the derivation $ad_{A}$ on matrices and leaves the subspace invariant,
so by repeated application the term $E(l,m)$ is mapped to nonzero multiples of
$E(l,m^{\prime})$ with $|m^{\prime}|\leq l$. In particular, if the expansion
has a term $E(l,m)$ with $l>l_{1}+l_{2}$, say $l$ is maximal, application of
the operator $A$ will map the expansion to a non-zero multiple of $E(l,l)$. On
the other hand, the above product is an $m$-subdiagonal matrix and the action
of $A$ yields $m^{\prime}$-subdiagonal matrices with $m^{\prime}$ at most
equal to $l_{1}+l_{2}$. This is a contradiction.

Next, let us assume $l_{1}\geq l_{2}$, and suppose the expansion has the term
$E(l,m)$ where $l$ lies in the range $0\leq l<l_{1}-l_{2}$. Application of $A$
to this term can only yield terms $E(l,m^{\prime})$ with $m^{\prime}\leq l$.
However, application of $A$ to the product also yield the term $E(l_{1}%
,l_{1})E(l_{2},m_{2})$, which is $m^{\prime}$-subdiagonal with $m^{\prime
}=l_{1}+m_{2}$. On the other hand,%
\[
m^{\prime}=l_{1}+m_{2}\geq l_{1}-l_{2}>l
\]
and this is a contradiction.

\section{A proof of Proposition \ref{symWPS}}

\label{proofsymWPS}

Proposition \ref{symWPS} follows straight from the explicit formulae
(\ref{explicitW6j})-(\ref{explicitW6j2}). However, it is interesting to see
how it can be obtained directly from the general equation (\ref{W6j2})
defining the Wigner $6j$ symbol, as shown below.

Thus, we start from the following formula, which is a particular case of
(\ref{W6j2}):
\begin{equation}
\left\{
\begin{array}
[c]{ccc}%
l_{1} & l_{2} & l_{3}\\
j & j & j
\end{array}
\right\}  =\sum(-1)^{3j+\delta+\epsilon+\phi}{}_{\cdot} \label{W6jjj}%
\end{equation}%
\[
\cdot\left(
\begin{array}
[c]{ccc}%
l_{1} & l_{2} & l_{3}\\
\alpha & \beta & \gamma
\end{array}
\right)  \left(
\begin{array}
[c]{ccc}%
l_{1} & j & j\\
\alpha & \epsilon & -\phi
\end{array}
\right)  \left(
\begin{array}
[c]{ccc}%
j & l_{2} & j\\
-\delta & \beta & \phi
\end{array}
\right)  \left(
\begin{array}
[c]{ccc}%
j & j & l_{3}\\
\delta & -\epsilon & \gamma
\end{array}
\right)
\]
where, again, the sum is taken over all possible values of $\alpha
,\beta,\gamma,\delta,\epsilon,\phi$, and only three of these are independent.
Therefore,
\begin{equation}
\left\{
\begin{array}
[c]{ccc}%
l_{1} & l_{3} & l_{2}\\
j & j & j
\end{array}
\right\}  =\sum(-1)^{3j+\delta^{\prime}+\epsilon^{\prime}+\phi^{\prime}}%
{}_{\cdot} \label{W6jjj2}%
\end{equation}%
\[
\cdot\left(
\begin{array}
[c]{ccc}%
l_{1} & l_{3} & l_{2}\\
\alpha^{\prime} & \beta^{\prime} & \gamma^{\prime}%
\end{array}
\right)  \left(
\begin{array}
[c]{ccc}%
l_{1} & j & j\\
\alpha^{\prime} & \epsilon^{\prime} & -\phi^{\prime}%
\end{array}
\right)  \left(
\begin{array}
[c]{ccc}%
j & l_{3} & j\\
-\delta^{\prime} & \beta^{\prime} & \phi^{\prime}%
\end{array}
\right)  \left(
\begin{array}
[c]{ccc}%
j & j & l_{2}\\
\delta^{\prime} & -\epsilon^{\prime} & \gamma^{\prime}%
\end{array}
\right)
\]
Using (\ref{cyclic}) and re-naming $\alpha^{\prime}=\alpha,\ \beta^{\prime
}=\gamma,\ \gamma^{\prime}=\beta$, from (\ref{W6jjj2}) we get
\begin{equation}
\left\{
\begin{array}
[c]{ccc}%
l_{1} & l_{3} & l_{2}\\
j & j & j
\end{array}
\right\}  =\sum(-1)^{3j+\delta^{\prime}+\epsilon^{\prime}+\phi^{\prime}%
+l_{1}+2l_{2}+2l_{3}+4j}{}_{\cdot} \label{W6jjj3}%
\end{equation}%
\[
\cdot\left(
\begin{array}
[c]{ccc}%
l_{1} & l_{2} & l_{3}\\
\alpha & \beta & \gamma
\end{array}
\right)  \left(
\begin{array}
[c]{ccc}%
l_{1} & j & j\\
\alpha & \epsilon^{\prime} & -\phi^{\prime}%
\end{array}
\right)  \left(
\begin{array}
[c]{ccc}%
j & j & l_{3}\\
-\delta^{\prime} & \phi^{\prime} & \gamma
\end{array}
\right)  \left(
\begin{array}
[c]{ccc}%
j & l_{2} & j\\
\delta^{\prime} & \beta & -\epsilon^{\prime}%
\end{array}
\right)
\]
Renaming $\delta^{\prime}=-\delta\ ,\ \epsilon^{\prime}=-\phi\ ,\ \phi
^{\prime}=-\epsilon$, from (\ref{W6jjj3}) we get
\begin{equation}
\left\{
\begin{array}
[c]{ccc}%
l_{1} & l_{3} & l_{2}\\
j & j & j
\end{array}
\right\}  =\sum(-1)^{3j-\delta-\epsilon-\phi+l_{1}+2l_{2}+2l_{3}+4j}{}_{\cdot}
\label{W6jjj4}%
\end{equation}%
\[
\cdot\left(
\begin{array}
[c]{ccc}%
l_{1} & l_{2} & l_{3}\\
\alpha & \beta & \gamma
\end{array}
\right)  \left(
\begin{array}
[c]{ccc}%
l_{1} & j & j\\
\alpha & -\phi & \epsilon
\end{array}
\right)  \left(
\begin{array}
[c]{ccc}%
j & j & l_{3}\\
\delta & -\epsilon & \gamma
\end{array}
\right)  \left(
\begin{array}
[c]{ccc}%
j & l_{2} & j\\
-\delta & \beta & \phi
\end{array}
\right)
\]
Again using (\ref{cyclic}), from (\ref{W6jjj4}) we get
\begin{equation}
\left\{
\begin{array}
[c]{ccc}%
l_{1} & l_{3} & l_{2}\\
j & j & j
\end{array}
\right\}  =\sum(-1)^{3j-\delta-\epsilon-\phi+2l_{1}+2l_{2}+2l_{3}+6j}{}%
_{\cdot} \label{W6jjj5}%
\end{equation}%
\[
\cdot\left(
\begin{array}
[c]{ccc}%
l_{1} & l_{2} & l_{3}\\
\alpha & \beta & \gamma
\end{array}
\right)  \left(
\begin{array}
[c]{ccc}%
l_{1} & j & j\\
\alpha & \epsilon & -\phi
\end{array}
\right)  \left(
\begin{array}
[c]{ccc}%
j & l_{2} & j\\
-\delta & \beta & \phi
\end{array}
\right)  \left(
\begin{array}
[c]{ccc}%
j & j & l_{3}\\
\delta & -\epsilon & \gamma
\end{array}
\right)
\]
But $(-1)^{3j-\delta-\epsilon-\phi+2l_{1}+2l_{2}+2l_{3}+6j}=(-1)^{3j+\delta
+\epsilon+\phi}(-1)^{2(l_{1}+l_{2}+l_{3}+2j)}(-1)^{2(j-\delta-\epsilon-\phi)}$
and $(-1)^{2(l_{1}+l_{2}+l_{3}+2j)}=(-1)^{2(j-\delta-\epsilon-\phi)}=1$, so
the values of (\ref{W6jjj5}) and (\ref{W6jjj}) are identical.

Similarly, permutation of any other two columns in (\ref{W6jjj}) leaves the
value invariant.

\section{A proof of Proposition \ref{bln}}

\label{blnproof}

The formula (\ref{BerezinChar}) in Proposition \ref{bln} follows directly from
equation (\ref{BCG}) and the explicit formulae (\ref{explicitCG2}%
)-(\ref{Sjjj}) for the Clebsch-Gordan coefficients. Nonetheless, it is
interesting to see how it can be obtained more directly and independently of
these formulae.

To begin with, for $l=1$ we have by (\ref{rel}) and (\ref{basis})
\[
\mathbf{e}(1,0)=\frac{1}{\beta_{1,1}}\left[  J_{-},\mathbf{e}(1,1)\right]
=\frac{1}{\beta_{1,1}\mu_{1}}\left[  J_{+},J_{-}\right]  =\frac{\sqrt{2}}%
{\mu_{1}}J_{3}%
\]
and consequently%
\[
b_{1}^{n}=\sqrt{\frac{n}{n+2}}%
\]
We shall work out the general formula
\begin{equation}
\mathbf{e}^{j}(l,0)_{1,1}=\frac{1}{\mu_{l}\sqrt{(2l)!}}x_{1}^{2}x_{2}%
^{2}...x_{l}^{2}\text{, \ where }x_{k}=\sqrt{k(n-k+1)}, \label{e1}%
\end{equation}
and then, formula (\ref{BerezinChar}) follows immediately from
\begin{equation}
b_{l}^{n}=\sqrt{\frac{n+1}{2l+1}}\frac{x_{1}^{2}x_{2}^{2}...x_{l}^{2}}%
{\sqrt{(2l)!}\mu_{l}}\ . \label{Berezin2}%
\end{equation}

So, let us focus on this formula, which can be proved by induction on $l$. The
underlying calculations are simpler by working with the matrices $E(l,0)$
rather than the normed matrices $\mathbf{e}(l,0)$, so let us illustrate the
idea by taking $l=3$ and formally calculate subdiagonal matrices
\begin{align*}
A  &  =(x_{1},x_{2},...,x_{n})_{1},B=(x_{1},x_{2},...,x_{n})_{-1}\\
A^{2}  &  =(y_{1},y_{2},..,y_{n-1})_{2},\text{ }y_{k}=x_{k}x_{k+1}\\
\ A^{3}  &  =(z_{1},z_{2},..,z_{n-2})_{3},\text{ }z_{k}=x_{k+2}y_{k}%
=x_{k}x_{k+1}x_{k+2}\\
E(3,2)  &  =[B,A^{3}]=(u_{1},u_{2},..,u_{n-1})_{2},u_{1}=-x_{3}z_{1}%
=-x_{1}x_{2}x_{3}^{2}\\
E(3,1)  &  =[B,[B,A^{3}]]=(v_{1},v_{2},..,v_{n})_{1},\text{ }v_{1}=-x_{2}%
u_{1}=x_{1}x_{2}^{2}x_{3}^{2}\\
E(3,0)  &  =[B,[B,[B,A^{3}]]]=(w_{1},w_{2},..,w_{n+1})_{0},\text{ }%
w_{1}=-x_{1}v_{1}=-x_{1}^{2}x_{2}^{2}x_{3}^{2}%
\end{align*}

Thus, the first entry of the diagonal matrix $E(l,0)$ is seen to be
\[
E(l,0)_{1,1}=(-1)^{l}x_{1}^{2}x_{2}^{2}...x_{l}^{2},
\]
and on the other hand (cf. Chapter 2.4.1)
\begin{equation}
\mathbf{e}(l,0)=\frac{(-1)^{l}}{\mu_{l,0}^{n}}E(l,0), \label{eE}%
\end{equation}
where by (\ref{my6})
\[
\mu_{l,0}^{n}=\mu_{l}\sqrt{l!}\sqrt{(2l)(2l-1)(2l-2)...(l+1)}=\mu_{l}%
\sqrt{(2l)!}%
\]
is the norm of $E(l,0)$. Now, formula (\ref{Berezin2}) follows from
(\ref{Cnl}) and (\ref{eE}).

\section{A proof of Proposition \ref{stanprodlin}}\label{approdlin} 

Proposition \ref{stanprodlin} follows from (\ref{stanprod}). But it can be proved more directly, without resorting to the formulas for the Wigner product symbol, as follows.

From the identities (\ref{harmonics}) we deduce the following formulae%
\begin{align*}
x  &  =\frac{-1}{\sqrt{6}}(Y_{1,1}-Y_{1,-1}),y=\frac{i}{\sqrt{6}}%
(Y_{1,1}+Y_{1,-1}),z=\frac{1}{\sqrt{3}}Y_{1,0}\\
(x\pm iy)^{2}  &  =\sqrt{\frac{8}{15}}Y_{2,\pm2},xy=\frac{-i}{\sqrt{30}%
}(Y_{2,2}-Y_{2,-2}),x^{2}-y^{2}=\sqrt{\frac{2}{15}}(Y_{2,2}+Y_{2,-2})
\end{align*}
Therefore, recalling the definition of the coupled standard basis
$\{\mathbf{e}(l,m)\}$, the symbol correspondence
\[
W_{1}:\mu_{0}\mathbf{e}(l,m)\longleftrightarrow Y_{l,m}\
\]
yields the specific correspondences
\begin{align}
x  &  \longleftrightarrow\frac{1}{\sqrt{6}}\frac{\mu_{0}}{\mu_{1}%
}(A+B),y=\frac{i}{\sqrt{6}}\frac{\mu_{0}}{\mu_{1}}(-A+B),z=\frac{1}{\sqrt{6}%
}\frac{\mu_{0}}{\mu_{1}}(AB-BA)\nonumber\\
xy  &  \longleftrightarrow\frac{-i}{\sqrt{30}}\frac{\mu_{0}}{\mu_{2}}%
(A^{2}-B^{2}),x^{2}-y^{2}\longleftrightarrow\sqrt{\frac{2}{15}}\frac{\mu_{0}%
}{\mu_{2}}(A^{2}+B^{2}) \label{corr1}%
\end{align}
where we have used the notation $A=J_{+},B=J_{-}$. Consequently,
\begin{align*}
x\star_{1}^{n}y  &  \longleftrightarrow\frac{1}{\sqrt{6}}\frac{\mu_{0}}%
{\mu_{1}}(A+B)\frac{i}{\sqrt{6}}\frac{\mu_{0}}{\mu_{1}}(-A+B)\\
&  =\frac{\sqrt{30}\mu_{2}}{n(n+2)\mu_{0}}\frac{-i}{\sqrt{30}}\frac{\mu_{0}%
}{\mu_{2}}(A^{2}-B^{2})+\frac{i}{\sqrt{n(n+2)}}\frac{1}{\sqrt{6}}\frac{\mu
_{0}}{\mu_{1}}(AB-BA)\\
&  \longleftrightarrow\frac{\sqrt{30}\mu_{2}}{n(n+2)\mu_{0}}(xy)+\frac
{i}{\sqrt{n(n+2)}}z
\end{align*}
This gives the product formula (\ref{star}) for $(x,y,z)=(a,b,c)$, and
similarly one verifies the formula for a cyclic permutation of the coordinate functions.

On the other hand, by (\ref{lower}) - (\ref{rel}) we also have%
\begin{equation}
x\longleftrightarrow\frac{2}{\sqrt{6}}\frac{\mu_{0}}{\mu_{1}}J_{1}%
\ ,\ y\longleftrightarrow\frac{2}{\sqrt{6}}\frac{\mu_{0}}{\mu_{1}}%
J_{2}\ ,\ z\longleftrightarrow\frac{2}{\sqrt{6}}\frac{\mu_{0}}{\mu_{1}}J_{3},
\label{xyz}%
\end{equation}
which by (\ref{J2sum}) yields
\[
x\star_{1}^{n}x+y\star_{1}^{n}y+z\star_{1}^{n}z\longleftrightarrow\frac
{4}{n(n+2)}(J_{1}^{2}+J_{2}^{2}+J_{3}^{2})=I,
\]
and this proves the third identity (\ref{star2}).

Finally, let us calculate the three products $a\ast a$, for $a=x,y,z,$ from
three linear equations relating them. To this end, we start with the
correspondences%
\[
A\longleftrightarrow\frac{\mu_{1}}{\mu_{0}}\sqrt{\frac{3}{2}}%
(x+iy),B\longleftrightarrow\frac{\mu_{1}}{\mu_{0}}\sqrt{\frac{3}{2}}(x-iy)
\]
which yield%
\[
A^{2}+B^{2}\longleftrightarrow\frac{3\mu_{1}^{2}}{\mu_{0}^{2}}(x\star_{1}^{n}
x-y\star_{1}^{n} y).
\]
Combining this with (\ref{corr1}) we obtain the identity%
\[
x\star_{1}^{n} x-y\star_{1}^{n} y=\frac{1}{3}\sqrt{\frac{15}{2}}\frac{\mu
_{0}\mu_{2}}{\mu_{1}^{2}}(x^{2}-y^{2}),
\]
and similarly, there is the identity%
\[
y\star_{1}^{n} y-z\star_{1}^{n} z=\frac{1}{3}\sqrt{\frac{15}{2}}\frac{\mu
_{0}\mu_{2}}{\mu_{1}^{2}}(y^{2}-z^{2})
\]
These two equations together with equation (\ref{star2}) yield the solution
(\ref{star1}).

\section{A proof of Proposition \ref{BertrikWild}}

\label{proofWild}

In order to keep par with the convention used in \cite{Wild}, in this appendix we assume the
Hermitian inner product $h_{n+1}(\cdot,\cdot)=\ <\cdot,\cdot>$ on
$\mathcal{H}_{j}\equiv\mathbb{C}^{n+1}$, as well as all other inner products,
to be conjugate linear in the second entry, not the first. Then, the standard
Berezin symbol of an operator $T:\mathcal{H}_{j}\rightarrow\mathcal{H}_{j}$ is
given by
\begin{equation}
B_{T}(\mathbf{n})=h_{n+1}(T\tilde{Z},\tilde{Z})=\ <T\tilde{Z},\tilde
{Z}>\nonumber
\end{equation}
where
\[
\mathbf{z}=(z_{1},z_{2})\in SU(2)=S^{3}\subset\mathbb{C}^{2},\text{ }\Phi
_{j}(\mathbf{z})=\tilde{Z}\in\mathcal{M}_{j}\subset S^{2n+1}\subset
\mathbb{C}^{n+1}%
\]%
\begin{equation}
\Phi_{j}(\mathbf{z})=\tilde{Z}=\left(  z_{1}^{n},\sqrt{\binom{n}{1}}%
z_{1}^{n-1}z_{2},..,\sqrt{\binom{n}{k}}z_{1}^{n-k}z_{2}^{k},..,z_{2}%
^{n}\right)  \label{tildeZ}%
\end{equation}
and $\mathbf{n}=\pi(\mathbf{z})=[z_{1},z_{2}]\in S^{2}$, $\pi$ being the
projection in the Hopf fibration
\[
S^{1}\rightarrow S^{3}\rightarrow S^{2}\ ,\ \pi:S^{3}\rightarrow S^{2}\
\]

The map $\Phi_{j}:$ $S^{3}\rightarrow S^{2n+1}$ is $SU(2)$-equivariant and is
an embedding if $j$ is half-integral, in which case its image is the orbit
$\mathcal{M}_{j}\simeq SU(2)\simeq S^{3}$, whereas in the case of integral $j$
the orbit is a manifold $\mathcal{M}_{j}$ $\simeq SO(3)\simeq P^{3}$. To the
Hopf fibration there is a related $S^{1}$ principal fibre bundle depending on
$j$
\begin{equation}
S^{1}\rightarrow\mathcal{M}_{j}\rightarrow S^{2}\ ,\ \pi_{j}:\mathcal{M}%
_{j}\rightarrow S^{2}\ ,\ \pi=\pi_{j}\circ\Phi_{j}\ . \label{PFB}%
\end{equation}

In what follows, it is important to highlight the explicit relation between
the Hermitian metrics $h_{n+1}:\mathcal{M}_{j}\times\mathcal{M}_{j}%
\rightarrow\mathbb{C}$ and $h_{2}:S^{3}\times S^{3}\rightarrow\mathbb{C}$ that
is immediate from the explicit expression (\ref{tildeZ}) of the map $\Phi_{j}%
$, namely
\begin{equation}
h_{n+1}(\tilde{Z},\tilde{Z}^{\prime})=h_{n+1}(\Phi_{j}(\mathbf{z}),\Phi
_{j}(\mathbf{z}^{\prime}))=(h_{2}(\mathbf{z},\mathbf{z}^{\prime}))^{n}\ .
\label{scaling}%
\end{equation}
To simplify and keep close to the notation in \cite{Wild}, we shall denote
points $\tilde{Z},\tilde{Z}_{1},...$ on the orbit $\mathcal{M}_{j}$ by
$m,m_{1},m_{2},...$ and general vectors in $\mathbb{C}^{n+1}$ by
$v,v_{1},v_{2},w,...$.

Now, the manifold $\mathcal{M}_{j}$ inherits from the Hilbert space
$\mathcal{H}_{j}$, viewed as a euclidean space $\mathbb{R}^{2n+2}$, an
$SU(2)$-invariant Riemannian metric, and hence an $SU(2)$-invariant measure
$dm$ as well as an invariant $L^{2}$-inner product%
\begin{equation}
<f,g>_{\mathcal{M}_{j}}\ :=\int_{\mathcal{M}_{j}}f(m)\overline{g(m})dm
\label{L2b}%
\end{equation}
Thus, the idea set forth in \cite{Wild} is to work out most of what is related
to the standard Berezin correspondence at the level of the orbit
$\mathcal{M}_{j}$, which is possible because of the explicit use of the
Hermitian structure for this correspondence, as follows.

First, we may and shall assume the measure $dm$ is ``normalized'' so that
\begin{equation}
v=\int_{\mathcal{M}_{j}}<v,m>mdm,\text{ for all }v\in\mathbb{C}^{n+1}
\label{21}%
\end{equation}
In particular, this implies
\begin{equation}
<v_{1},v_{2}>=\int_{\mathcal{M}_{j}}<v_{1},m><m,v_{2}>dm \label{22}%
\end{equation}

\begin{lemma}
For any $T\in M_{n+1}(\mathbb{C)}$,
\begin{equation}
trace(T)=\int_{\mathcal{M}_{j}}<Tm,m>dm \label{23}%
\end{equation}

\end{lemma}

\begin{proof}
Let $\{e_{i},i=1,2,..n+1\}$ be an orthonormal basis of $\mathbb{C}^{n+1}$, and
hence%
\[
\forall m\in\mathcal{M}_{j},\ m=\sum_{i=1}^{n+1}<m,e_{i}>e_{i}%
\ ,\ \mbox{and}\ e_{i}=\int_{\mathcal{M}_{j}}<e_{i},m>mdm\ ,
\]%
\begin{align*}
trace(T)  &  =\sum<Te_{i},e_{i}>\text{ }=\sum_{i}<T(\int_{\mathcal{M}_{j}%
}<e_{i},m>mdm,e_{i}>\\
&  =\sum_{i}<\int_{\mathcal{M}_{j}}<e_{i},m>Tmdm,e_{i}>\text{ }=\sum_{i}%
\int_{\mathcal{M}_{j}}<Tm,\overline{<e_{i},m>}e_{i}>dm\\
&  =\int_{\mathcal{M}_{j}}<Tm,\sum_{i}<m,e_{i}>e_{i}>dm=\int_{\mathcal{M}_{j}%
}<Tm,m>dm
\end{align*}

\end{proof}

\quad\quad By choosing $T=Id$ in the above lemma we deduce the following :
%\begin{corollary}
%The ``normalization'' convention (\ref{21}) for the measure $dm$ on the orbit $%
%\mathcal{M}_{j}$ means that the orbit has volume%
\begin{equation}
Vol(\mathcal{M}_{j})=\int_{\mathcal{M}_{j}}dm=n+1 \ . \label{24}%
\end{equation}
%\end{corollary}

For any operator $T$, following \cite{Wild} we define the function
\begin{equation}
K_{T}:\mathcal{H}_{j}\times\mathcal{H}_{j}\rightarrow\mathbb{C}\text{,
\ \ }K_{T}(v,w)=\text{ }<Tv,w> \label{13a}%
\end{equation}
and express $T$ as an integral operator by integration over $\mathcal{M}_{j}%
$:
\begin{equation}
Tv=\int_{\mathcal{M}_{j}}K_{T}(v,m)mdm \label{13b}%
\end{equation}
The validity of identity (\ref{13b}) follows from the normalization (\ref{21})
of the measure $dm$. In view of (\ref{13b}), $K_{T}$ is the integral kernel of
$T$.

Clealy, the kernel of the composition $T_{2}T_{1}$, as a function
$\mathcal{H}_{j}\times\mathcal{H}_{j}\rightarrow\mathbb{C}$, can be expressed
directly via inner product using (\ref{13a}), but also as an integral
\begin{equation}
K_{T_{2}T_{1}}(v,w)=\text{ }<T_{2}T_{1}v,w>\text{ }=\int_{\mathcal{M}_{j}%
}K_{T_{1}}(v,m)K_{T_{2}}(m,w)dm \label{13c}%
\end{equation}

Now, we remind that for the standard Berezin correspondence determined by
characteristic numbers $\vec{b}$, the covariant-to-contravariant transition
operator on symbols
\[
U_{\vec{b},\frac{1}{\vec{b}}}^{j}:Poly_{\mathbb{C}}(S^{2})_{\leq n}\rightarrow
Poly_{\mathbb{C}}(S^{2})_{\leq n}\ ,\ f\mapsto\tilde{f}%
\]
corresponds to the transition operator (cf. Definition \ref{TransOperator})
\[
V_{\vec{b},\frac{1}{\vec{b}}}^{j}:M_{\mathbb{C}}(n+1)\rightarrow
M_{\mathbb{C}}(n+1)\ ,\ F\mapsto F^{\prime}%
\]
in such a way that
\begin{equation}
f=W_{\vec{b}}^{j}(F)=B_{F}\iff\tilde{f}=W_{\vec{b}}^{j}(F^{\prime
})=B_{F^{\prime}}\ . \label{transitions}%
\end{equation}
To keep par with the notation used in \cite{Wild}, we shall denote these
transition operators and their respective inverses, as follows:
\[
\eta^{-1}\equiv V_{\vec{b},\frac{1}{\vec{b}}}^{j}\ ,\ \eta\equiv V_{\frac
{1}{\vec{b}},\vec{b}}^{j}\quad,\quad\tilde{\eta}^{-1}\equiv U_{\vec{b}%
,\frac{1}{\vec{b}}}^{j}\ ,\ \tilde{\eta}\equiv U_{\frac{1}{\vec{b}},\vec{b}%
}^{j}%
\]

According to equations (\ref{CCduality}) and (\ref{inducedinnerproduct}),
denoting the induced inner product on $Poly_{\mathbb{C}}(S^{2})_{\leq n}$ by
$<\cdot,\cdot>_{\star_{\vec{b}}^{n}}$ and the (usual) $L^{2}$-inner product on
$\mathcal{C}^{\infty}(S^{2})$ by $<\cdot,\cdot>$ ,
\begin{equation}
\frac{1}{n+1}trace(FG^{\ast})=\ <f,g>_{\star_{\vec{b}}^{n}}\ =\ <\tilde
{f},g>=<\tilde{\eta}^{-1}(B_{F}),g>=<B_{\eta^{-1}(F)},B_{G}>\ ,\nonumber
\end{equation}%
\begin{align}
\frac{1}{n+1}trace(\eta(F)G^{\ast})  &  =\ <B_{\eta(F)},g>_{\star_{\vec{b}%
}^{n}}=<\tilde{\eta}(f),g>_{\star_{\vec{b}}^{n}}\ \label{L1L2}\\
&  =\text{ }<{f},g>=<B_{F},B_{G}>.\nonumber
\end{align}

In what follows, we denote by $\sigma_{T}$ the trivial lift of a standard
Berezin symbol $B_{T}$ on $S^{2}$ to the orbit $\mathcal{M}_{j}$, namely
\[
\sigma_{T}(m)=B_{T}(\pi_{j}(m))\ .
\]

\begin{lemma}
\label{prop.3.2} For $T\in M_{\mathbb{C}}(n+1)$, the operator $\eta(T)$ can be
expressed as follows:
%\end{subequations}%
\begin{equation}
\eta(T)(v)=\int_{\mathcal{M}_{j}}B_{T}(\pi_{j}(m))<v,m>mdm \label{29}%
\end{equation}

\end{lemma}

\begin{proof}
Define the operator $T_{1}$ $\in M_{\mathbb{C}}(n+1)$ by the right side of
(\ref{29}), namely%
\begin{equation}
T_{1}(v)=\int_{\mathcal{M}_{j}}\sigma_{T}(m)<v,m>mdm \label{30}%
\end{equation}
We need to show that $T_{1}=\eta(T)$. Now, $K_{T_{1}}(v,w)=<T_{1}v,w>$, thus,
by (\ref{30}),
\[
K_{T_{1}}(v,w)=\int_{\mathcal{M}_{j}}\sigma_{T}(m)<v,m><m,w>dm
\]
But, given another operator $S$, the composition $T_{1}S^{\ast}$ has by
(\ref{13c}) kernel
\begin{align}
K_{T_{1}S^{\ast}}(v,w)  &  =\int_{\mathcal{M}_{j}}K_{S^{\ast}}(v,m)K_{T_{1}%
}(m,w)dm\label{31}\\
&  =\int_{\mathcal{M}_{j}}\int_{\mathcal{M}_{j}}<S^{\ast}v,m>\sigma
_{T}(m^{\prime})<m,m^{\prime}><m^{\prime},w>dm^{\prime}dm\nonumber\\
&  =\int_{\mathcal{M}_{j}}[\int_{\mathcal{M}_{j}}<v,Sm><m,m^{\prime}%
>dm]\sigma_{T}(m^{\prime})<m^{\prime},w>dm^{\prime}\nonumber\\
&  =\int_{\mathcal{M}_{j}}<v,Sm^{\prime}>\sigma_{T}(m^{\prime})<m^{\prime
},w>dm^{\prime}\nonumber
\end{align}
where in the last step we have applied (\ref{22}). Now, by (\ref{23}),
(\ref{31}), and applying (\ref{22}) again, we obtain
\begin{align}
trace(T_{1}S^{\ast})  &  =\int_{\mathcal{M}_{j}}K_{T_{1}S^{\ast}%
}(m,m)dm\label{32}\\
&  =\int_{\mathcal{M}_{j}}\int_{\mathcal{M}_{j}}\sigma_{T}(m^{\prime
})<m,Sm^{\prime}><m^{\prime},m>dm^{\prime}dm\nonumber\\
&  =\int_{\mathcal{M}_{j}}\sigma_{T}(m^{\prime})<m^{\prime},Sm^{\prime
}>dm^{\prime}=\int_{\mathcal{M}_{j}}\sigma_{T}(m^{\prime})\overline
{<Sm^{\prime},m^{\prime}>}dm^{\prime}\nonumber\\
&  =\int_{\mathcal{M}_{j}}\sigma_{T}(m^{\prime})\overline{\sigma_{S}%
(m^{\prime})}=<\sigma_{T},\sigma_{S}>_{\mathcal{M}_{j}}=(n+1)<B_{T}%
,B_{S}>_{S^{2}}\nonumber\\
&  \Rightarrow\frac{1}{n+1}trace(T_{1}S^{\ast})=<B_{T},B_{S}>,
\end{align}
and since $S$ is arbitrary it follows from (\ref{L1L2}) that $T_{1}=\eta(T)$.
\end{proof}

Now, let us denote by $\omega_{1}, \omega_{2}$ the standard Berezin symbols of
$T_{1}, T_{2}$ lifted up to the orbit $\mathcal{M}_{j}$. Then, denoting by
$\omega_{1}\star\omega_{2}$ the standard Berezin symbol of $T_{1}T_{2}$ lifted
to $\mathcal{M}_{j}$, we have the following result.

\begin{lemma}
\label{comporbit} For Berezin symbols $\omega_{i}=\sigma_{T_{i}},i=1,2$,
\begin{equation}
\omega_{1}\star\omega_{2}(m)=\iint_{\mathcal{M}_{j}\times\mathcal{M}_{j}}%
B_{1}^{\prime}(m,m_{2},m_{1})\tilde{\omega}_{1}(m_{1})\tilde{\omega}_{2}%
(m_{2})dm_{1}dm_{2} \label{48}%
\end{equation}
where $\tilde{\omega_{i}}=\tilde{\eta}^{-1}(\omega_{i})$ is the contravariant
Berezin symbol of $T_{i}$ lifted to $\mathcal{M}_{j}$, and
\begin{equation}
B_{1}^{\prime}(m,m_{1},m_{2})\ =\ <m,m_{1}><m_{1},m_{2}><m_{2},m>
\label{Wild1'}%
\end{equation}

\end{lemma}

\begin{proof}
Start with
\begin{align}
&  <T_{2}T_{1}v,v>=<T_{1}v,T_{2}^{\ast}v>=\int_{\mathcal{M}}<T_{1}%
v,m><m,T_{2}^{\ast}v>dm\nonumber\\
&  =\int_{\mathcal{M}_{j}}<T_{1}v,m><T_{2}m,v>dm \label{49}%
\end{align}
Using Lemma \ref{prop.3.2}, let us choose $S=\eta^{-1}(T_{1})$, so that
\begin{align*}
T_{1}v  &  =\int_{\mathcal{M}_{j}}\sigma_{\eta^{-1}(T_{1})}(m_{1}%
)<v,m_{1}>m_{1}dm_{1}=\int_{\mathcal{M}_{j}}<\eta^{-1}(T_{1})m_{1,}%
m_{1}><v,m_{1}>m_{1}dm_{1}\\
&  \Rightarrow<T_{1}v,m>=\int_{\mathcal{M}_{j}}<\eta^{-1}(T_{1})m_{1,}%
m_{1}><v,m_{1}><m_{1},m>dm_{1}\\
&  \Rightarrow<T_{2}T_{1}v,v>=\int_{\mathcal{M}_{j}}\{[\int_{\mathcal{M}_{j}%
}<\eta^{-1}(T_{1})m_{1,}m_{1}><v,m_{1}><m_{1},m>dm_{1}]\cdot\\
&  \cdot\lbrack\int_{\mathcal{M}_{j}} <\eta^{-1}(T_{2})m_{2,}m_{2}%
><m,m_{2}><m_{2},v>dm_{2}]\}dm\\
&  \Rightarrow<T_{2}T_{1}v,v>=\iint_{\mathcal{M}_{j}\times\mathcal{M}_{j}%
}<\eta^{-1}(T_{1})m_{1,}m_{1}><\eta^{-1}(T_{2})m_{2,}m_{2}>\cdot\label{51}\\
&  \cdot<v,m_{1}><m_{1},m_{2}><m_{2},v>dm_{1}dm_{2}%
\end{align*}

Now, we write $v=m$ and obtain the desired result.
\end{proof}

\ 

Because $B_{1}^{\prime}$ clearly descends to the level of $S^{2}$ (cf.
(\ref{tildeZ}) and (\ref{PFB})), the following result is immediate (cf.
equation (\ref{CCintegralproduct})):

\begin{corollary}
Set $\pi_{j}(m_{i})=\mathbf{n}_{i}\in S^{2}$, $i=1,2,3$. Then,
\begin{equation}
\overline{B_{1}^{\prime}(m_{1},m_{2},m_{3})}=\left(  \frac{4\pi}{n+1}\right)
^{2}\mathbb{T}_{\vec{b}}^{j}(\mathbf{n}_{1},\mathbf{n}_{2},\mathbf{n}_{3})\ ,
\label{B1W}%
\end{equation}
where conjugation on the l.h.s. is necessary to account for the different
conventions of Hermitian product used for defining $B_{1}^{\prime}$ and
$\mathbb{T}_{\vec{b}}^{j}$ (cf. Remark \ref{conjugatefirstentry}).
\end{corollary}

Now, from equations (\ref{scaling}) and (\ref{Wild1'}), we have immediately
\begin{equation}
B_{1}^{\prime}(m_{1},m_{2},m_{3})=\left(  h_{2}(\mathbf{z}_{1},\mathbf{z}%
_{2})h_{2}(\mathbf{z}_{2},\mathbf{z}_{3})h_{2}(\mathbf{z}_{3},\mathbf{z}%
_{1})\right)  ^{n}, \label{immediateWild}%
\end{equation}
and by a straightforward computation using formula (\ref{hopf}) for the Hopf
map and the convention that $h_{2}$ is conjugate linear in the second entry,
we finally get
\begin{align}
&  h_{2}(\mathbf{z}_{1},\mathbf{z}_{2})h_{2}(\mathbf{z}_{2},\mathbf{z}%
_{3})h_{2}(\mathbf{z}_{3},\mathbf{z}_{1})\label{lastWild}\\
&  =\frac{1}{4}(1+\mathbf{n}_{1}\cdot\mathbf{n}_{2}+\mathbf{n}_{2}%
\cdot\mathbf{n}_{3}+\mathbf{n}_{3}\cdot\mathbf{n}_{1}-i[\mathbf{n}%
_{1},\mathbf{n}_{2},\mathbf{n}_{3}])\ \nonumber
\end{align}

Equations (\ref{B1W})-(\ref{lastWild}) are equivalent to equation
(\ref{closedBertrik}).

%\newpage

\section{A proof of Proposition \ref{N}}

\label{proofN}

We use notations and conventions from Appendix \ref{proofWild}. Proposition
\ref{N} is equivalent to Lemma \ref{normtransition} below, which is a
consequence of the following lemma.

\begin{lemma}
At the level of $\mathcal{M}_{j}$, the kernel of the
contravariant-to-covariant symbol transformation $\tilde{\eta}$ for the
standard Berezin correspondence is the function
\[
N^{\prime}(m,m^{\prime})=|<m,m^{\prime}>|^{2}%
\]

\end{lemma}

\begin{proof}
Let $\omega=\sigma_{T}$. The operator $\tilde{\eta}$ is defined by \quad%
\[
\tilde{\eta}(\omega)(m)=\int_{\mathcal{M}_{j}}N^{\prime}(m,m^{\prime}%
)\omega(m^{\prime})dm^{\prime}%
\]
Then, by Lemma \ref{prop.3.2},
\begin{align*}
\tilde{\eta}(\omega)(m)  &  =\sigma_{\eta(T)}(m)=\text{ }<\eta(T)m,m>\\
&  =\int_{\mathcal{M}_{j}}\sigma_{T}(m^{\prime})<m,m^{\prime}><m^{\prime
},m>dm^{\prime}\\
&  =\int_{\mathcal{M}_{j}}\omega(m^{\prime})[<m,m^{\prime}><m^{\prime
},m>]dm^{\prime}\\
&  =\int_{\mathcal{M}_{j}}N^{\prime}(m,m^{\prime})\omega(m^{\prime}%
)dm^{\prime}%
\end{align*}
where \quad$N^{\prime}(m,m^{\prime}) \ = \ <m,m^{\prime}><m^{\prime},m> \ =
\ |<m,m^{\prime}>|^{2}$ \ \ 
\end{proof}

\ 

Clearly, $N^{\prime}$ descends to the level of $S^{2}$ and we have the
following result.

\begin{lemma}
\label{normtransition} Let $\mathbf{n}=\pi_{j}(m), \mathbf{n}^{\prime}=\pi
_{j}(m^{\prime})\in S^{2}$. Recalling that $\eta= V_{\frac{1}{\vec{b}},\vec
{b}}$ ,
\begin{equation}
\label{normtransition1}N^{\prime}(m,m^{\prime}) = |<m,m^{\prime}>|^{2} =
\left(  \frac{1+\mathbf{n}\cdot\mathbf{n}^{\prime}}{2}\right)  ^{n}=\frac
{4\pi}{n+1}\mathbb{U}_{\frac{1}{\vec{b}},\vec{b}}(\mathbf{n},\mathbf{n}%
^{\prime})
\end{equation}

\end{lemma}

\begin{proof}
It follows immediately from equation (\ref{scaling}) that
\[
|<m,m^{\prime}>|^{2}=|h_{n+1}(\tilde{Z},\tilde{Z}^{\prime})|^{2}%
=|(h_{2}(\mathbf{z},\mathbf{z}^{\prime}))^{n}|^{2}=(|h_{2}(\mathbf{z}%
,\mathbf{z}^{\prime})|^{2})^{n},
\]
and by a straightforward computation using equation (\ref{hopf}) for the Hopf
map,
\[
|h_{2}(\mathbf{z},\mathbf{z}^{\prime})|^{2}=(1+\mathbf{n}\cdot\mathbf{n}%
^{\prime})/2\ .
\]
The last equality in (\ref{normtransition1}) is immediate from the definitions.
\end{proof}

\section{A proof of Theorem \ref{asymplimit}}

\label{proofasymplimit}

Clearly, equations (\ref{asympWj})-(\ref{asympWn}) are equivalent, so here we
will focus on the expansion of type (\ref{asympWn+1}) in inverse powers of
$n+1$, namely,
\begin{align}
&  (-1)^{2j+m_{3}}\sqrt{\frac{(2j+1)(2l_{3}+1)}{(2l_{1}+1)(2l_{2}+1)}}\left[
\begin{array}
[c]{ccc}%
l_{1} & l_{2} & l_{3}\\
m_{1} & m_{2} & -m_{3}%
\end{array}
\right]  \!\!{[j]}\label{LS}\\
&  =C_{0,0,0}^{l_{1},l_{2},l_{3}}C_{m_{1},m_{2},m_{3}}^{l_{1},l_{2},l_{3}%
}+\frac{1}{n+1}C_{m_{1},m_{2},m_{3}}^{l_{1},l_{2},l_{3}}P({l_{1},l_{2},l_{3}%
})+O((n+1)^{-2})\ \label{RS}%
\end{align}
for $n=2j>>1$, \ $l_{1},l_{2},l_{3}<<n$ (we emphasize that, in what follows,
this is equivalent to letting $n\rightarrow\infty$ keeping $l_{1},l_{2},l_{3}$ finite).

\bigskip The Wigner product symbol, as expressed in (119), decomposes as
follows:
\begin{align*}
&  \left[
\begin{array}
[c]{ccc}%
l_{1} & l_{2} & l_{3}\\
m_{1} & m_{2} & -m_{3}%
\end{array}
\right]  \!\!{[j]}\\
&  =\sqrt{(2l_{1}+1)(2l_{2}+1)(2l_{3}+1)}\left(
\begin{array}
[c]{ccc}%
l_{1} & l_{2} & l_{3}\\
-m_{1} & -m_{2} & +m_{3}%
\end{array}
\right)  \left\{
\begin{array}
[c]{ccc}%
l_{1} & l_{2} & l_{3}\\
j & j & j
\end{array}
\right\} \\
&  =\sqrt{(2l_{1}+1)(2l_{2}+1)(2l_{3}+1)}(-1)^{l_{1}+l_{2}+l_{3}}\left(
\begin{array}
[c]{ccc}%
l_{1} & l_{2} & l_{3}\\
m_{1} & m_{2} & -m_{3}%
\end{array}
\right)  \left\{
\begin{array}
[c]{ccc}%
l_{1} & l_{2} & l_{3}\\
j & j & j
\end{array}
\right\} \\
&  =(-1)^{l_{1}+l_{2}+l_{3}}\sqrt{(2l_{1}+1)(2l_{2}+1)(2l_{3}+1)}%
\frac{(-1)^{-l_{1}+l_{2}-m_{3}}}{\sqrt{2l_{3}+1}}C_{m_{1},m_{2},m_{3}}%
^{l_{1},l_{2},l_{3}}\left\{
\begin{array}
[c]{ccc}%
l_{1} & l_{2} & l_{3}\\
j & j & j
\end{array}
\right\} \\
&  =(-1)^{\ l_{3}-m_{3}}\sqrt{(2l_{1}+1)(2l_{2}+1)}C_{m_{1},m_{2},m_{3}%
}^{l_{1},l_{2},l_{3}}\left\{
\begin{array}
[c]{ccc}%
l_{1} & l_{2} & l_{3}\\
j & j & j
\end{array}
\right\}  ,
\end{align*}
and consequently we can write the expression (\ref{LS}) as
\begin{align}
&  (-1)^{n+m_{3}}\sqrt{\frac{(2j+1)(2l_{3}+1)}{(2l_{1}+1)(2l_{2}+1)}}\left[
\begin{array}
[c]{ccc}%
l_{1} & l_{2} & l_{3}\\
m_{1} & m_{2} & -m_{3}%
\end{array}
\right]  \!\!{[j]}\nonumber\\
&  =C_{m_{1},m_{2},m_{3}}^{l_{1},l_{2},l_{3}}\Phi(l_{1,}l_{2},l_{3};n+1)
\label{LSx}%
\end{align}
where the function $\Phi$ is defined by
\begin{equation}
\Phi(l_{1,}l_{2},l_{3};n+1)=(-1)^{n}(-1)^{l_{3}}\sqrt{n+1}\sqrt{2l_{3}%
+1}\ \left\{
\begin{array}
[c]{ccc}%
l_{1} & l_{2} & l_{3}\\
j & j & j
\end{array}
\right\}  \label{phi}%
\end{equation}

Next, using the expression (123) for the 6j-symbol $\left\{  {}\right\}  $ on
the right side of (\ref{phi}), we can express the above function as
\[
\Phi(l_{1,}l_{2},l_{3};n+1)=(-1)^{l_{3}}\sqrt{2l_{3}+1}l_{1}!l_{2}%
!\ l_{3}!\Delta(l_{1},l_{2},l_{3})\Upsilon(l_{1,}l_{2},l_{3};n+1)\
\]
where $\Upsilon$ is the only function depending on $n$, namely,
\begin{equation}
\Upsilon(l_{1},l_{2},l_{3};n+1)= \label{explicitW6jbis}%
\end{equation}%
\begin{equation}
=\sqrt{\frac{(n+1)\cdot(n-l_{1})!(n-l_{2})!(n-l_{3})!}{(n+l_{1}+1)!(n+l_{2}%
+1)!(n+l_{3}+1)!}}\displaystyle{\sum_{k}}\frac{(-1)^{k}(n+1+k)!}%
{(n+k-l_{1}-l_{2}-l_{3})!R(l_{1},l_{2},l_{3};k)}\ ,\nonumber
\end{equation}
with summation index $k$ assuming all integral values for which all factorial
arguments in $(n+k-l_{1}-l_{2}-l_{3})!R(l_{1},l_{2},l_{3};k)$ are nonnegative,
where
\begin{equation}
R(l_{1},l_{2},l_{3};k)=%
%TCIMACRO{\dprod \limits_{i=1}^{3}}%
%BeginExpansion
{\displaystyle\prod\limits_{i=1}^{3}}
%EndExpansion
(k-l_{i})!%
%TCIMACRO{\dprod \limits_{i<j}^{3}}%
%BeginExpansion
{\displaystyle\prod\limits_{i<j}^{3}}
%EndExpansion
(l_{i}+l_{j}-k)!\text{ ,\ cf. (\ref{explicitW6j2})} \label{explicitW6j2bis}%
\end{equation}
We note that, for $n>>l_{i}$, the restriction on summation index $k$ amounts
to demanding all factorial arguments in $R(l_{1},l_{2},l_{3};k)$ being
nonnegative, namely
\begin{equation}
\max\left\{  l_{i}\right\}  \leq k\leq\min\left\{  l_{i}+l_{j}\right\}  \ \ .
\label{k}%
\end{equation}

Now, writing $\mu=n+1$, $L=l_{1}+l_{2}+l_{3}$, we can re-express $\Upsilon
(l_{1},l_{2},l_{3};n+1)$ as follows:
\begin{align}
&  \Upsilon(l_{1},l_{2},l_{3};\mu)=\mu^{1/2}\left[
%TCIMACRO{\dprod \limits_{i=1}^{3}}%
%BeginExpansion
{\displaystyle\prod\limits_{i=1}^{3}}
%EndExpansion
(\mu-l_{i})...(\mu+l_{i})^{-1/2}\right]  ^{-1/2}{\sum_{k}}\frac{(-1)^{k}%
(\mu+k)!}{(\mu+k-L-1)!R(l_{1},l_{2},l_{3};k)}\nonumber\\
&  =\mu^{-L-1}\left[
%TCIMACRO{\dprod \limits_{i=1}^{3}}%
%BeginExpansion
{\displaystyle\prod\limits_{i=1}^{3}}
%EndExpansion
\left(  1-\frac{l_{i}}{\mu}\right)  ...\left(  1+\frac{l_{i}}{\mu}\right)
\right]  ^{-1/2}{\sum_{k}}\frac{(-1)^{k}(\mu+k)!}{(\mu+k-L-1)!R(l_{1}%
,l_{2},l_{3};k)}\quad\nonumber\\
&  =\left[
%TCIMACRO{\dprod \limits_{i=1}^{3}}%
%BeginExpansion
{\displaystyle\prod\limits_{i=1}^{3}}
%EndExpansion
\prod_{p_{i}=0}^{l_{i}}\left(  1-\left(  \frac{p_{i}}{\mu}\right)
^{2}\right)  \right]  ^{-1/2}\frac{1}{\mu^{L+1}}{\sum_{k}}\frac{(-1)^{k}%
(\mu+k)!}{(\mu+k-L-1)!\ R(l_{1},l_{2},l_{3};k)} \label{Theta1}%
\end{align}

Note that the inverse square root factor in (\ref{Theta1}) expands as
$1+O(\mu^{-2})$. Therefore, the first two terms in powers of $1/\mu\ $ in the
expansion of $\Upsilon(l_{1},l_{2},l_{3};\mu)$ are given by the first two terms
in powers of $1/\mu$ in the expansion of
\begin{align}
\Psi(l_{1},l_{2},l_{3};\mu)  &  =\frac{1}{\mu^{L+1}}{\sum_{k}}\frac
{(-1)^{k}(\mu+k)!}{(\mu+k-L-1)!R(l_{1},l_{2},l_{3};k)}\ \label{Psi}\\
&  ={\sum_{k}}\frac{(-1)^{k}S(k,L;\mu)}{R(l_{1},l_{2},l_{3};k)}%
\end{align}
where we have written%

\begin{equation}
S(k,L;\mu)=\left(  1+\frac{k}{\mu}\right)  \left(  1+\frac{k-1}{\mu}\right)
...\left(  1+\frac{k-L}{\mu}\right)  \label{SE}%
\end{equation}
Clearly, the latter expands when $\mu\rightarrow\infty$ as
\begin{equation}
S(k,L;\mu)=1+\frac{1}{\mu}\left[  \frac{L+1}{2}\left(  2k-L\right)  \right]
+O(\mu^{-2})\ .\label{SE1}%
\end{equation}
Consequently, the asymptotic expansion of $\Upsilon$ begins as follows :
\begin{align}
\Upsilon(l_{1},l_{2},l_{3};\mu) &  ={\sum_{k}}\frac{(-1)^{k}}{R(l_{1}%
,l_{2},l_{3};k)}\label{Theta2}\\
&  +\frac{1}{\mu}\ \frac{(l_{1}+l_{2}+l_{3}+1)}{2}\ {\sum_{k}}\frac
{(-1)^{k}(2k-(l_{1}+l_{2}+l_{3}))}{R(l_{1},l_{2},l_{3};k)}\nonumber\\
&  +O(\mu^{-2})\ \nonumber
\end{align}
where the summation index $k$ is subject to the constraint (\ref{k}). Thus,
the expression (\ref{LS}), presented as (\ref{LSx}), has the asymptotic
expansion
\begin{align}
&  C_{m_{1},m_{2},m_{3}}^{l_{1},l_{2},l_{3}}\Phi(l_{1,}l_{2},l_{3}%
;\mu)=C_{m_{1},m_{2},m_{3}}^{l_{1},l_{2},l_{3}}(-1)^{l_{3}}\sqrt{2l_{3}%
+1}l_{1}!l_{2}!\ l_{3}!\Delta(l_{1},l_{2},l_{3})\nonumber\\
&  \cdot\left\{  {\sum_{k}}\frac{(-1)^{k}}{R(l_{1},l_{2},l_{3};k)}+\frac
{1}{\mu}\ \frac{(l_{1}+l_{2}+l_{3}+1)}{2}\ {\sum_{k}}\frac{(-1)^{k}%
(2k-(l_{1}+l_{2}+l_{3}))}{R(l_{1},l_{2},l_{3};k)}+O(\mu^{-2})\ \right\}
\nonumber\\
&  =C_{m_{1},m_{2},m_{3}}^{l_{1},l_{2},l_{3}}\Phi_{0}(l_{1},l_{2},l_{3}%
)+\frac{1}{\mu}C_{m_{1},m_{2},m_{3}}^{l_{1},l_{2},l_{3}}\Phi_{1}(l_{1}%
,l_{2},l_{3})+O(\mu^{-2})\label{5b}%
\end{align}
where
\begin{align}
\Phi_{0}(l_{1},l_{2},l_{3}) &  =(-1)^{l_{3}}\sqrt{2l_{3}+1}l_{1}!l_{2}%
!\ l_{3}!\Delta(l_{1},l_{2},l_{3}){\sum_{k}}\frac{(-1)^{k}}{R(l_{1}%
,l_{2},l_{3};k)}\label{Phizero}\\
\Phi_{1}(l_{1},l_{2},l_{3}) &  =(-1)^{l_{3}}\sqrt{2l_{3}+1}l_{1}!l_{2}%
!\ l_{3}!\Delta(l_{1},l_{2},l_{3})\label{Phi_1}\\
&  \cdot\frac{(l_{1}+l_{2}+l_{3}+1)}{2}\ {\sum_{k}}\frac{(-1)^{k}%
(2k-(l_{1}+l_{2}+l_{3}))}{R(l_{1},l_{2},l_{3};k)}\nonumber
\end{align}
Thus, the analysis of the first two terms in (\ref{5b}) amounts to a closer
look at the above functions $\Phi_{0}$ and $\Phi_{1}$, and below we shall
divide into two cases accordingly.  

\subsection{The $0^{th}$ order term}

Let us set
\begin{equation}
\Sigma_{0}(l_{1},l_{2},l_{3})=\displaystyle{\sum_{k}}\frac{(-1)^{k}}%
{R(l_{1},l_{2},l_{3};k)}\ , \label{Sigma0}%
\end{equation}
where as before, the summation index runs over the string of nonnegative
integers $k$ as in (\ref{k}). Then, for the zeroth order term,
%$$ \displaystyle{\sum_k}\frac{(-1)^{l_1-l_2 + k}}{R(l_1,l_2,l_3,k)} = (-1)^{l_1-l_2}\Sigma_0 \ , $$
if we perform the change of variables $k\rightarrow L-k$ in the summation
$\Sigma_{0}$, we obtain
\begin{equation}
\Sigma_{0}(l_{1},l_{2},l_{3})=(-1)^{L}\Sigma_{0}(l_{1},l_{2},l_{3})
\label{symSigma0}%
\end{equation}
%$$(-1)\Sigma_0 =(-1)^{l_3}\Sigma_0 $$
%$$ \displaystyle{\sum_k}\frac{(-1)^{l_1-l_2 +k}}{R(l_1,l_2,l_3,k)} = \displaystyle{\sum_t}\frac{(-1)^{t-l_3}}{R(l_1,l_2,l_3,t)} \ ,$$
because
\begin{equation}
R(l_{1},l_{2},l_{3};k)=R(l_{1},l_{2},l_{3},L-k) \label{symR}%
\end{equation}
which implies that
\begin{equation}
\Sigma_{0}\equiv\Phi_{0}\equiv0\ ,\ \ \mbox{if}\quad L=l_{1}+l_{2}+l_{3}\text{
is an odd number,} \label{vanishSigma0}%
\end{equation}
so the expression (\ref{Phizero}) can be nonzero only when $L$ is even.

Recall the summation in equation (\ref{explicitCG2}), which defines explicitly
the Clebsch-Gordan coefficients,
\begin{align}
&  C_{m_{1},m_{2},m_{3}}^{l_{1},l_{2},l_{3}}\ =\ \delta(m_{1}+m_{2}%
,m_{3})\sqrt{2l_{3}+1}\ \Delta(l_{1},l_{2},l_{3})\ S_{m_{1},m_{2},m_{3}%
}^{\ l_{1},\ \ l_{2},\ \ l_{3}}\nonumber\label{explicitCG2bis}\\
&  \cdot\displaystyle{\sum_{z}}\frac{(-1)^{z}}{z!(l_{1}+l_{2}-l_{3}%
-z)!(l_{1}-m_{1}-z)!(l_{2}+m_{2}-z)!(l_{3}-l_{2}+m_{1}+z)!(l_{3}-l_{1}%
-m_{2}+z)!}\nonumber
\end{align}
with $\Delta(l_{1},l_{2},l_{3})$ and $S_{m_{1},m_{2},m_{3}}^{\ l_{1}%
,\ \ l_{2},\ \ l_{3}}$ given respectively by (\ref{DeltaW}) and (\ref{Sjjj}).
Then, setting $m_{1}=m_{2}=0$ and $z=k-l_{3}$, the summation above becomes
\begin{align*}
&  \displaystyle{\sum_{k}}\frac{(-1)^{k-l_{3}}}{(k-l_{1})!(k-l_{2}%
)!(k-l_{3})!(l_{1}+l_{2}-k)!(l_{2}+l_{3}-k)!(l_{3}+l_{1}-k)!}  \\
&  =(-1)^{l_{3}}{\sum_{k}}\frac{(-1)^{k}}{R(l_{1},l_{2},l_{3};k)}=(-1)^{l_{3}}\Sigma_{0}(l_{1},l_{2},l_{3}) 
%\label{summmm}
\end{align*}
Thus we have the formula
\begin{equation}
C_{0,0,0}^{l_{1},l_{2},l_{3}}=(-1)^{l_{3}}\sqrt{2l_{3}+1}\Delta(l_{1}%
,l_{2},l_{3})l_{1}!l_{2}!l_{3}!\Sigma_{0}(l_{1},l_{2},l_{3}) \label{CSigma0}%
\end{equation}
where the r.h.s. expression is the same as (\ref{Phizero}), consequently
\begin{equation}
\Phi_{0}(l_{1},l_{2},l_{3})=C_{0,0,0}^{l_{1},l_{2},l_{3}}\ . \label{+-}%
\end{equation}

We remark that the symmetry of Clebsch-Gordan coefficients (cf.
(\ref{vanishCG000})), besides the above discussion, implies that the quantity
(\ref{+-}) vanishes when $L$ is odd. We also note that the coefficient
(\ref{+-}) is given in closed form by equation (\ref{C_000}) and, in view of
equation (\ref{CSigma0}), this is equivalent to the following closed formula
for $\Sigma_{0}$ when \ $L=l_{1}+l_{2}+l_{3}$\ is even, namely
\begin{equation}
\Sigma_{0}(l_{1},l_{2},l_{3})=(-1)^{L/2}\frac{(\frac{L}{2})!}{l_{1}!(\frac
{L}{2}-l_{1})!\ l_{2}!(\frac{L}{2}-l_{2})!\ l_{3}!(\frac{L}{2}-l_{3})!}\ .
\label{Sigma0closed}%
\end{equation}

\subsection{The $1^{st}$ order term}

Let us set
\begin{equation}
\Sigma_{1}(l_{1},l_{2},l_{3})=\displaystyle{\sum_{k}}\frac{(-1)^{k}%
\ k}{R(l_{1},l_{2},l_{3};k)}\ ,\label{Sigma1}%
\end{equation}
with summation over nonnegative integers $k$ as in (\ref{k}). Now, using the
symmetry (\ref{symR}), we calculate
\begin{align*}
{\sum_{k}}\frac{(-1)^{k}(L-k)\ }{R(l_{1},l_{2},l_{3};k)} &  =(-1)^{L}{\sum
_{k}}\frac{(-1)^{L-k}(L-k)\ }{R(l_{1},l_{2},l_{3},L-k)}=(-1)^{L}\Sigma_{1}\\
&  =L{\sum_{k}}\frac{(-1)^{k}\ }{R(l_{1},l_{2},l_{3};k)}-{\sum_{k}}%
\frac{(-1)^{k}\ k}{R(l_{1},l_{2},l_{3};k)}=L\Sigma_{0}-\Sigma_{1}%
\end{align*}
and deduce the identity%
\begin{equation}
\lbrack1+(-1)^{L}]\Sigma_{1}=L\Sigma_{0}\label{b2}%
\end{equation}
and hence%
\begin{align}
L\text{ odd } &  \Longrightarrow\Sigma_{0}(l_{1,}l_{2},l_{3})=0\quad
\mbox{(cf. (\ref{vanishSigma0}))}\nonumber\\
L\text{ even } &  \Longrightarrow\Sigma_{1}(l_{1},l_{2},l_{3})=\frac{L}%
{2}\Sigma_{0}(l_{1,}l_{2},l_{3})\label{b4}%
\end{align}
Therefore, using (\ref{b2}), $\Phi_{1}$ given by (\ref{Phi_1}) can be
re-expressed as
\begin{align}
\Phi_{1}(l_{1},l_{2},l_{3}) &  =(-1)^{l_{3}}\sqrt{2l_{3}+1}\Delta(l_{1}%
,l_{2},l_{3})l_{1}!l_{2}!l_{3}!\frac{(L+1)}{2}\label{split6}\\
&  \cdot\ \left(  {\sum_{k}}\frac{(-1)^{k}\ k}{R(l_{1},l_{2},l_{3};k)}%
-{\sum_{k}}\frac{(-1)^{k}(L-k)}{R(l_{1},l_{2},l_{3};k)}\right)  \nonumber\\
&  =(-1)^{l_{3}}\sqrt{2l_{3}+1}\Delta(l_{1},l_{2},l_{3})l_{1}!l_{2}%
!l_{3}!\frac{(L+1)}{2}\left(  \Sigma_{1}-(-1)^{L}\Sigma_{1}\right)
\nonumber\\
&  =(-1)^{l_{3}}\frac{1+(-1)^{L+1}}{2}\sqrt{2l_{3}+1}\Delta(l_{1},l_{2}%
,l_{3})l_{1}!l_{2}!l_{3}!(L+1)\cdot\Sigma_{1}\nonumber
\end{align}

\begin{corollary}
\label{Phi111} With $\Sigma_{1}$ defined as in (\ref{Sigma1}) and
$L=l_{1}+l_{2}+l_{3}$, we have that $\Phi_{1}(l_{1},l_{2},l_{3})=0$ if $L$ is
even, and
\begin{equation}
\Phi_{1}(l_{1},l_{2},l_{3})=[(-1)^{l_{3}}\sqrt{2l_{3}+1}\Delta(l_{1}%
,l_{2},l_{3})l_{1}!l_{2}!l_{3}!(L+1)]\cdot\Sigma_{1}(l_{1},l_{2}%
,l_{3})\nonumber
\end{equation}
if $L$ is odd.
\end{corollary}

Now, just as $\Sigma_{0}$ has the closed formula given by (\ref{Sigma0closed}%
), $\Sigma_{1}$ has a similar closed formula which we shall seek. Denote by
$[x]$ the integer part of a positive rational number, namely the largest
integer $\leq x$. Set
\begin{equation}
Q(l_{1},l_{2},l_{3})=\frac{[\frac{L}{2}]!}{l_{1}!([\frac{L}{2}]-l_{1}%
)!\ l_{2}!([\frac{L}{2}]-l_{2})!\ l_{3}!([\frac{L}{2}]-l_{3})!}\ ,\label{Q}%
\end{equation}
so that $\Sigma_{0}\equiv(-1)^{L/2}Q$, when $L=l_{1}+l_{2}+l_{3}$ is even. 

\begin{proposition}
The functions $\Sigma_{1}$ and $Q$ defined by (\ref{Sigma1}) and (\ref{Q}),
respectively, are proportional to each other. More precisely:
\begin{align}
\Sigma_{1}(l_{1},l_{2},l_{3}) &  =\frac{(-1)^{[\frac{L+1}{2}]}}{2}\left(
1+\left[  \frac{L+1}{2}\right]  +(-1)^{L}\left[  \frac{L-1}{2}\right]
\right)  Q(l_{1},l_{2},l_{3})\quad\label{Sigma1closed}\\
&  =\left\{
\begin{array}
[c]{c}%
(-1)^{\frac{L}{2}}\displaystyle{\frac{L}{2}}Q(l_{1},l_{2},l_{3})\ ,\quad\text{if}\ L\text{ is even}\\
\\
(-1)^{\frac{L+1}{2}}Q(l_{1},l_{2},l_{3})\ ,\quad\ \text{if}\ L\text{ is odd}%
\end{array}
\right.  \quad\quad\ \quad\quad\ \nonumber
\end{align}

\end{proposition}

For $L$ even, the above expression for $\Sigma_{1}$  is immediate from
(\ref{Sigma0closed}) and (\ref{b4}), namely it follows from simple symmetry
considerations. For $L$ odd, a formal (combinatorial) proof of the above
formula for $\Sigma_{1}$ has so far eluded us. The reader may easily verify
the formula with a computer program, say for odd  $L=l_1+l_2+l_3$ up to $\simeq 40000$, recalling that the
numbers $l_{i}$ are subject to the triangle inequalities $\delta(l_{1}%
,l_{2},l_{3})=1$, cf. (\ref{triangleineq}). 

Therefore, from the equations in Corollary \ref{Phi111} and Proposition \ref{DP} we
finally obtain that, regardless of whether $L$ is even or odd, we always have
\[
\Phi_{1}\equiv P
\]
and (cf. equations (\ref{RS}) and (\ref{5b})) this concludes the proof of
Theorem \ref{asymplimit}.

\newpage

\addcontentsline{toc}{chapter}{Index}
\printindex

\end{document}